\documentclass[preprint,3p,number,sort&compress]{elsarticle}
\usepackage{amsmath,amssymb,paralist,color,pifont,stmaryrd,multirow,url}
\usepackage[standard,thmmarks]{ntheorem}
\usepackage[ngerman,english]{babel}
\usepackage{graphicx,tikz}

\useshorthands{"}
\addto\extrasenglish{\languageshorthands{ngerman}}

\journal{[To be determined]}

\title{Multiple Context-Free Tree Grammars: \\ Lexicalization and
  Characterization}

\author[je]{Joost Engelfriet}
\ead{j.engelfriet@liacs.leidenuniv.nl}

\author[am]{Andreas Maletti}
\ead{maletti@informatik.uni-leipzig.de}

\author[sm]{Sebastian Maneth}
\ead{maneth@uni-bremen.de}

\address[je]{LIACS, Leiden University, P.O.~Box~9512,
  2300~RA~Leiden, The Netherlands}

\address[am]{Institute of Computer Science, Universit\"at Leipzig,
  P.O.~Box 100\,920, 04009~Leipzig, Germany \\[.5ex]}

\address[sm]{Department of Mathematics and Informatics, Universit{\"a}t Bremen,
  P.O.~Box 330\,440, 28334 Bremen, Germany}

\allowdisplaybreaks

\DeclareMathOperator{\pos}{pos}
\DeclareMathOperator{\yield}{yd}
\DeclareMathOperator{\rk}{rk}
\DeclareMathOperator{\init}{in}
\DeclareMathOperator{\dec}{dec}
\DeclareMathOperator{\rem}{rem}
\DeclareMathOperator{\val}{val}

\DeclareMathOperator{\var}{var}
\DeclareMathOperator{\sequ}{seq}
\DeclareMathOperator{\ren}{ren}
\DeclareMathOperator{\renh}{\ren_{\hole}}
\DeclareMathOperator{\tree}{tree}
\DeclareMathOperator{\dtr}{dtr}
\DeclareMathOperator{\lex}{lex}
\DeclareMathOperator{\bign}{bign}
\DeclareMathOperator{\pf}{pf}
\DeclareMathOperator{\sts}{sts}
\DeclareMathOperator{\lhs}{lhs}
\DeclareMathOperator{\rhs}{rhs}

\DeclareMathOperator{\num}{num}
\DeclareMathOperator{\alp}{occ}
\DeclareMathOperator{\pr}{pr}

\DeclareMathOperator{\pair}{pair}
\DeclareMathOperator{\gr}{gr}
\DeclareMathOperator{\ext}{ext}
\DeclareMathOperator{\inc}{inc}
\DeclareMathOperator{\cut}{cut}
\DeclareMathOperator{\wid}{\theta}
\DeclareMathOperator{\skel}{\kappa}

\newcommand{\plus}{^{\scriptscriptstyle +}}
\newcommand{\der}{\mathrm{der}}
\newcommand{\red}{\mathrm{red}}
\newcommand{\e}{\,\triangleright}
\newcommand{\N}{\mathcal{N}}
\newcommand{\LL}{\mathcal{L}}
\newcommand{\X}{\mathcal{X}}
\newcommand{\F}{\mathcal{F}}
\newcommand{\M}{\mathcal{M}}
\newcommand{\D}{\mathcal{D}}
\newcommand{\DL}{\mathit{DL}}
\newcommand{\LDTR}{\textup{LDT}$^\textup{R}$}
\newcommand{\DTRfc}{\textup{DT}$_\textup{fc}^\textup{R}$}
\newcommand{\Deltadl}{\Delta_{\text{dl}}}
\newcommand{\Dec}{\mathrm{Skel}}
\newcommand{\PF}{\mathit{PF}}
\newcommand{\rlab}{\mathrm{rlab}}
\newcommand{\flab}{\mathrm{flab}}
\newcommand{\new}{\text{new}}
\newcommand{\mr}{\mathit{mrk}}
\newcommand{\fin}{\hfill $\Box$}
\newcommand{\hole}{\SBox}

\providecommand*{\SBox}[0]{\ensuremath{\mathchoice
    {{\scriptstyle \Box}}%
    {{\scriptstyle \Box}}%
    {{\scriptscriptstyle \Box}}%
    {{\scriptscriptstyle \Box}}%
}}

\providecommand*{\nat}[0]{\ensuremath{\mathbb{N}}}
\providecommand*{\seq}[3]{\ensuremath{#1_{#2}, \dotsc, #1_{#3}}}
\providecommand*{\word}[3]{\ensuremath{#1_{#2} \dotsm #1_{#3}}}
\providecommand*{\abs}[1]{\ensuremath{\lvert #1 \rvert}}

\renewtheorem{lemma}{Lemma}
\renewtheorem{definition}[lemma]{Definition}
\renewtheorem{example}[lemma]{Example}
\renewtheorem{theorem}[lemma]{Theorem}
\renewtheorem{corollary}[lemma]{Corollary}
\renewtheorem{proposition}[lemma]{Proposition}
\renewtheorem{remark}[lemma]{Remark}

\divide \topmargin by 2 \advance \topmargin -1in
\divide \leftmargin by 2 \advance \leftmargin -1in
\oddsidemargin \leftmargin \evensidemargin \leftmargin

\begin{document}

\begin{abstract}
Multiple (simple) context-free tree grammars are investigated,
where ``simple'' means ``linear and nondeleting''. 
Every multiple context-free tree grammar that is finitely ambiguous can be lexicalized; 
i.e., it can be transformed into an equivalent one (generating the same tree language)
in which each rule of the grammar contains a lexical symbol. 
Due to this transformation, the rank of the nonterminals increases at most by~1,
and the multiplicity (or fan-out) of the grammar increases 
at most by the maximal rank of the lexical symbols; 
in particular, the multiplicity does not increase when all lexical symbols have rank~0.
Multiple context-free tree grammars have the same tree generating power as
multi-component tree adjoining grammars (provided the latter can use a root-marker).
Moreover, every multi-component tree adjoining grammar that is finitely ambiguous 
can be lexicalized. 
Multiple context-free tree grammars have the same string generating power 
as multiple context-free (string) grammars and polynomial time parsing algorithms.
A tree language can be generated by a multiple context-free tree grammar if and only if
it is the image of a regular tree language under a deterministic finite-copying macro tree transducer.
Multiple context-free tree grammars can be used as a synchronous translation device. 
\end{abstract}

\maketitle
 
{\scriptsize \tableofcontents}
 
\section{Introduction}

Multiple context-free (string) grammars (MCFG) were introduced in~\cite{sekmatfujkas91} and, 
independently, in~\cite{shaweijos87} where they are called (string-based) 
linear context-free rewriting systems (LCFRS). 
They are of interest to computational linguists because they can model cross-serial dependencies,
whereas they can still be parsed in polynomial time and generate semi-linear languages.   
Multiple context-free \emph{tree} grammars were introduced in~\cite{kan10},
in the sense that it is suggested in~\cite[Section~5]{kan10} that they are the 
hyperedge-replacement context-free graph grammars in tree generating normal form, 
as defined in~\cite{engman00}. 
Such graph grammars generate the same string languages as MCFGs~\cite{enghey91,wei92}. 
It is shown in~\cite{kan10} that they 
generate the same tree languages as second-order abstract categorial grammars (2ACG),
generalizing the fact that MCFGs generate the same string languages as 2ACGs \cite{sal07}. 
It is also observed in~\cite{kan10} that the set-local multi-component 
tree adjoining grammar (MC"~TAG, see~\cite{wei88,kall09}), well-known to computational linguists, 
is roughly the monadic restriction 
of the multiple context-free tree grammar, 
just as the tree adjoining grammar (TAG, see~\cite{joslevtak75,jossch97}) is roughly the monadic restriction 
of the (linear and nondeleting) context-free tree grammar, see~\cite{moe98,fujkas00,keprog11}. 
We note that the multiple context-free tree grammar could also be called the \emph{tree-based} LCFRS;
such tree grammars were implicitly envisioned already in~\cite{shaweijos87}.  

In this paper we define the multiple context-free tree grammars (MCFTG) 
in terms of familiar concepts from tree language theory (see, e.g., \cite{gecste84,gecste97}),
and we base our proofs on elementary properties of trees and tree homomorphisms. 
Thus, we do not use other formalisms such as graph grammars, $\lambda$-calculus, or logic programs. 
Since the relationship between MCFTGs and the above type of graph grammars is quite straightforward,
it follows from the results of~\cite{engman00} that the tree languages generated by 
MCFTGs can be characterized as the images of the regular tree languages under 
deterministic finite-copying macro tree transducers (see~\cite{engvog85,fulvog98,engman99}).
However, since no full version of~\cite{engman00} ever appeared in a journal, 
we present that characterization here (Theorem~\ref{thm:charact}). 
It generalizes the well-known fact that the 
string languages generated by MCFGs can be characterized as the yields of the images 
of the regular tree languages under deterministic finite-copying top-down tree transducers, 
cf.~\cite{wei92}.
These two characterizations imply (by a result from~\cite{engman99}) that the MCFTGs
have the same string generating power as MCFGs, through the yields of their tree languages. 
We also give a direct proof of this fact (Corollary~\ref{cor:ymcft}), 
and show how it leads to polynomial time parsing algorithms for MCFTGs (Theorem~\ref{thm:mcftgparse}). 
All trees that have a given string as yield, can be viewed as ``syntactic trees'' of that string. 
A parsing algorithm computes, for a given string, one syntactic tree (or all syntactic trees) 
of that string in the tree language generated by the grammar.
It should be noted that, due to its context-free nature, an MCFTG, like a TAG, 
also has derivation trees (or parse trees), which show the way in which a tree 
is generated by the rules of the grammar. A derivation tree can be viewed as a meta level tree
and the derived syntactic tree as an object level tree, cf.~\cite{jossch97}. 
In fact, the parsing algorithm computes a derivation tree (or all derivation trees) 
for the given string, and then computes the corresponding syntactic tree(s). 

We define the MCFTG as a straightforward generalization of the MCFG, based on  
tree substitution rather than string substitution, 
where a (second-order) tree substitution is a tree homomorphism.
However, our formal syntactic definition of the MCFTG is closer to the one 
of the context-free tree grammar (CFTG) 
as in, e.g., \cite{rou70,gecste97,engsch77,keprog11,fujkas00,kan16,staott07}. 
Just as for the MCFG, the semantics of the MCFTG is a least fixed point semantics, 
which can easily be viewed as a semantics based on parse trees (Theorem~\ref{thm:dtree}). 
Moreover, we provide a rewriting semantics for MCFTGs
(similar to the one for CFTGs and similar to the one 
in~\cite{ramsat99} for MCFGs) leading to a usual notion of derivation,
for which the derivation trees then equal the parse trees (Theorem~\ref{thm:deriv}).
Intuitively, an MCFTG $G$ is a \emph{simple} 
(i.e., linear and nondeleting) context-free tree grammar (spCFTG)
in which several nonterminals are rewritten in one derivation step. Thus every rule 
of $G$ is a sequence of rules of an spCFTG, and the left-hand side nonterminals of these rules 
are rewritten simultaneously. However, a sequence of nonterminals can only be rewritten 
if (earlier in the derivation) they were introduced explicitly as such 
by the application of a rule of $G$. Therefore, each rule of $G$ must also specify the 
sequences of (occurrences of) nonterminals in its right-hand side that may later be rewritten. 
This restriction is called ``locality'' in~\cite{wei88,ramsat99,kall09}.

Apart from the above-mentioned results (and some related results), our main result is 
that MCFTGs can be lexicalized (Theorem~\ref{thm:main}). 
Let us consider an MCFTG $G$ that generates a tree language $L(G)$ over the ranked alphabet $\Sigma$,
and let $\Delta\subseteq\Sigma$ be a given set of \emph{lexical items}. 
We say that $G$ is \emph{lexicalized} (with respect to $\Delta$) if every rule of $G$ 
contains at least one lexical item (or anchor). 
Lexicalized grammars are of importance for several reasons. First, a lexicalized grammar 
is often more understandable, because the rules of the grammar can be grouped around the lexical items.
Each rule can then be viewed as lexical information on its anchor, 
demonstrating a syntactical construction in which the anchor can occur.
Second, a lexicalized grammar defines a so-called dependency structure on the lexical items 
of each generated object, allowing to investigate certain aspects 
of the grammatical structure of that object, see~\cite{kuh10}. 
Third, certain parsing methods can take significant advantage of the fact 
that the grammar is lexicalized,
see, e.g., \cite{schabejos88}. In the case where each lexical item is a
symbol of the string alphabet (i.e., has rank~0), each rule of a lexicalized
grammar produces at least one symbol of the generated string.
Consequently, the number of rule applications (i.e., derivation steps) 
is clearly bounded by the length of the input string.  In addition, the lexical items in the
rules guide the rule selection in a derivation, which
works especially well in scenarios with large
alphabets (cf.\@ the detailed account in~\cite{che01}). 

We say that $G$ is \emph{finitely ambiguous} (with respect to $\Delta$) if, for every $n\geq 0$,
$L(G)$ contains only finitely many trees with $n$ occurrences of lexical items. For simplicity,
let us also assume here that every tree in $L(G)$ contains at least one lexical item. 
Obviously, if $G$ is lexicalized, then it is finitely ambiguous.
Our main result is that for a given MCFTG $G$ it is decidable whether or not $G$ is finitely ambiguous,
and if so, a lexicalized MCFTG $G'$ can be constructed 
that is (strongly) equivalent to $G$, i.e., $L(G')=L(G)$.
Moreover, we show that $G'$ is grammatically similar to $G$, in the sense that 
their derivation trees are closely related: every derivation tree of $G'$ 
can be translated by a finite-state tree transducer into a derivation tree of $G$ 
for the same syntactic tree, and vice versa. To be more precise, 
this can be done by a linear deterministic top-down tree transducer 
with regular look-ahead (\LDTR"~transducer). 
We say that $G$ and $G'$ are \LDTR"~equivalent.  
Since the class of \LDTR"~transductions is closed under composition, 
this is indeed an equivalence relation for MCFTGs.  
Note that, due to the \LDTR"~equivalence of $G'$ and $G$, 
any parsing algorithm for $G'$ can be turned into a parsing algorithm for $G$
by translating the derivation trees of $G'$ in linear time into derivation trees of $G$,
using the \LDTR"~transducer. Thus, the notion of \LDTR"~equivalence is similar to 
the well-known notion of cover for context-free grammars (see, e.g., \cite{grahar72,nij80}).
For context-free grammars, no \LDTR"~transducer can handle the derivation tree translation  
that corresponds to the transformation into Greibach Normal Form. In fact, our lexicalization of MCFTGs
generalizes the transformation of a context-free grammar into Operator Normal Form 
as presented in~\cite{grahar72}, which is much simpler than the transformation into Greibach Normal Form. 

The \emph{multiplicity} (or \emph{fan-out}) of an MCFTG is the maximal number of nonterminals 
that can be rewritten simultaneously in one derivation step. The lexicalization of MCFTGs,
as discussed above, increases the multiplicity of the grammar by at most the maximal rank of 
the lexical symbols in $\Delta$. 
When viewing an MCFTG as generating a string language, consisting of the yields of the generated trees,
it is natural that all lexical items are symbols of rank~0, 
which means that they belong to the alphabet of that string language. 
The lexicalization process is then called strong lexicalization, 
because it preserves the generated tree language (whereas weak lexicalization just requires 
preservation of the generated string language). 
Thus, strong lexicalization of MCFTGs does not increase the multiplicity. 
In particular spCFTGs, which are MCFTGs of multiplicity~1, can be strongly lexicalized 
as already shown in~\cite{maleng12}. Note that 
all TAG tree languages can be generated by spCFTGs \cite{keprog11}. 
Although TAGs can be weakly lexicalized (see~\cite{fuj05}), 
they cannot be strongly lexicalized, which was unexpectedly shown in~\cite{kuhsat12}. 
Thus, from the lexicalization point of view, spCFTGs have a significant advantage over TAGs. 
The strong lexicalization of MCFTGs (with lexical symbols of rank~0) is presented without proof 
(and without the notion of \LDTR"~equivalence) in~\cite{maleng17}. 

The \emph{width} of an MCFTG is the maximal rank of its nonterminals. 
The lexicalization of MCFTGs increases the width of the grammar by at most~1. 

In addition to the above results we compare the MCFTGs with the MC"~TAGs
and prove that they have (``almost'') the same tree generating power,
as also presented in~\cite{maleng17}.
It is shown in~\cite{keprog11} that ``non"~strict'' TAGs, which 
are a slight generalization of TAGs, generate the same tree languages as
monadic spCFTGs, where `monadic' means width at most~1; i.e., 
all nonterminals have rank~1 or~0.
We confirm and strengthen the above-mentioned observation in~\cite{kan10}
by showing that both MCFTGs and monadic MCFTGs have the same tree generating power 
as non"~strict~MC"~TAGs (Theorems~\ref{thm:footed} and~\ref{thm:monadic}),
with a polynomial increase of multiplicity. 
Since the constructions preserve lexicalized grammars, we obtain that 
non"~strict~MC"~TAGs can be (strongly) lexicalized. 
Note that by a straightforward generalization of~\cite{kuhsat12} it can be shown that 
non"~strict TAGs cannot be strongly lexicalized.
Then we show that even (strict)~MC"~TAGs have the same tree generating power as MCFTGs
(Theorem~\ref{thm:mctal}). 
To be precise, if $L$ is a tree language generated by an MCFTG, then the tree language 
$\#(L)=\{\#(t)\mid t\in L\}$ can be generated by an MC"~TAG, where $\#$ is a ``root-marker'' of rank~1.
This result settles a problem stated in~\cite[Section~4.5]{wei88}.\footnote{In the first paragraph 
of that section, Weir states that ``it would be interesting to investigate whether there exist
LCFRS's with object level tree sets that cannot be produced by any MCTAG.''}
It also implies that, as opposed to TAGs, 
MC"~TAGs can be (strongly) lexicalized (Theorem~\ref{thm:lexmctag}). 

It is shown in~\cite{yoskan05,yos06} that 2ACGs, and in particular tree generating 2ACGs, 
can be lexicalized (for $\Delta=\Sigma$). Although 2ACGs and MCFTGs generate the same tree languages,
this does not imply that MCFTGs can be lexicalized. 
It is shown in~\cite{sanaksgra16} that multi-dimensional TAGs can be strongly lexicalized.
Although it seems that for every multi-dimensional TAG there is an MCFTG 
generating the same tree language (see the Conclusion of~\cite{kan16}), 
nothing else seems to be known about the relationship between multi-dimensional TAGs and MC"~TAGs or MCFTGs. 

The structure of this paper is as follows. Section~\ref{sec:prelim} consists of preliminaries, 
mostly on trees and tree homomorphisms. Since a sequence of nonterminals of an MCFTG generates 
a sequence of trees, we also consider sequences of trees, called forests. 
The substitution of a forest for a sequence of symbols in a forest is realized by a tree homomorphism. 
In Section~\ref{sec:mcftg} we define the MCFTG, 
its least fixed point semantics (in terms of forest substitution), its derivation trees, and its derivations. 
Every derivation tree yields a tree, called its value, and the tree language generated by the grammar 
equals the set of values of its derivation trees. 
The set of derivation trees is itself a regular tree language. 
We recall the notion of an \LDTR"~transducer, and we define two MCFTGs to be \LDTR"~equivalent
if there is a value-preserving \LDTR"~transducer from the derivation trees of one grammar to the other,
and vice versa.  
Section~\ref{sec:norm} contains a number of normal forms. For every MCFTG we construct an \LDTR"~equivalent
MCFTG in such a normal form. In Section~\ref{sub:basicnf} we discuss some basic normal forms, 
such as permutation-freeness which means that application of a rule cannot permute subtrees. 
In Section~\ref{sub:lexnf} we prove that every MCFTG can be transformed into Growing Normal Form
(generalizing the result of~\cite{staott07,sta09} for spCFTGs). This means that 
every derivation step increases the sum of the number of terminal symbols and the number of 
``big nonterminals'' (which are the sequences of nonterminals that 
form the left-hand sides of the rules of the MCFTG). 
It even holds for finitely ambiguous MCFTGs, with `terminal' replaced by `lexical'
(Theorem~\ref{thm:dec-growing}). 
Thus, this result is already part of our lexicalization procedure. 
Moreover, we prove that finite ambiguity is decidable. 
Section~\ref{sec:lex} is devoted to the remaining, main part of the lexicalization procedure.
It shows that every MCFTG in (lexical) Growing Normal Form can be transformed into an \LDTR"~equivalent
lexicalized MCFTG. The intuitive idea is to transport certain lexical items from positions
in the derivation tree that contain more than one lexical item (more precisely, that 
are labeled with a rule of the grammar that contains more than one lexical item), 
up to positions that do not contain any lexical item. 
In Section~\ref{sub:mctag} we 
prove that MCFTGs have the same tree generating power as non-strict MC"~TAGs.
We define non-strict MC"~TAGs as a special type of MCFTGs, namely ``footed'' ones, 
which (as in~\cite{keprog11}) are permutation-free MCFTGs such that in every rule 
the arguments of each left-hand side nonterminal are all passed to one node in the 
right-hand side of the rule.
Then we prove in Section~\ref{sub:mcftismctal} that (strict) MC"~TAGs have the same tree 
generating power as MCFTGs, as explained above, and we show that MC"~TAGs can be strongly lexicalized. 
In Section~\ref{sub:monadic} we observe that every MC"~TAG (and hence every MCFTG) can be transformed 
into an equivalent MCFTG of width at most~1, which is in contrast to the fact that spCFTGs (and 
arbitrary context-free tree grammars) give rise to a strict hierarchy with respect to width,
as shown in~\cite[Theorem~6.5]{engrozslu80} (see also~\cite[Lemma~24]{lohmansch12}). 
In all the results of Section~\ref{sec:monadic} the constructed grammar is 
\LDTR"~equivalent to the given one. 
In Section~\ref{sub:stringgen} we define the multiple context-free (string) grammar (MCFG) as 
the ``monadic case'' of the MCFTG, which means that all terminal and nonterminal symbols 
have rank~1, except for a special terminal symbol and the initial nonterminal symbol that have rank~0. 
We prove (using permutation-freeness) that every tree language $L(G)$ that is generated by an MCFTG $G$
can also be generated by an MCFG, provided that we view every tree as a string in the usual way
(Theorem~\ref{thm:mcft-to-mcf}). 
Using this we show that $\yield(L(G))$, which is the set of yields of the trees in $L(G)$, 
can also be generated by an MCFG $G'$ and, in fact, every MCFG string language is of that form. 
Since, moreover, the derivation trees of $G$ and $G'$ are related by \LDTR"~transducers
(in a way similar to \LDTR"~equivalence), 
this result can be used to transform any polynomial time parsing algorithm for MCFGs into 
a polynomial time parsing algorithm for MCFTGs, as discussed in Section~\ref{sub:parsing}.
In Section~\ref{sec:charact} we recall the notion of macro tree transducer, and show that
the tree translation that computes the value of a derivation tree of an MCFTG $G$
can be realized by a deterministic finite-copying macro tree transducer 
(DMT$_\textrm{fc}$"~transducer). This implies that $L(G)$ is 
the image of a regular tree language (viz.\@ the set of derivation trees of $G$) under 
a DMT$_\textrm{fc}$"~transduction. 
Vice versa, every such image can be generated by an MCFTG that can be obtained by 
a straightforward product construction. From this characterization of the MCFTG tree languages 
we obtain a number of other characterizations (including those for the MCFG string languages), 
known from the literature. Thus, they are the tree/string languages generated by context-free graph grammars, 
they are the tree/string languages generated by 2ACGs, and 
they are the tree/string languages obtained as images of the regular tree languages under 
deterministic MSO"~definable tree/tree-to-string transductions 
(where MSO stands for Monadic Second-Order logic). 
Section~\ref{sec:trans} is based on the natural idea that,  
since every ``big nonterminal'' of an MCFTG generates a forest, i.e., a sequence of trees, 
we can also use an MCFTG to generate a set of pairs of trees (i.e., a tree translation)
and hence, taking yields, to realize a string translation. 
We study the resulting translation device in Section~\ref{sec:trans} and call it an MCFT"~transducer. 
It generalizes the (binary) rational tree translation of~\cite{rao97} 
(called synchronous forest substitution grammar in~\cite{mal13}) and the 
synchronous context-free tree grammar of~\cite{nedvog12}. 
We prove two results similar to those in~\cite{nedvog12}.
The first result characterizes the MCFT"~transductions in terms of macro tree transducers,
generalizing the characterization of the MCFTG tree languages of Section~\ref{sec:charact}. 
We show that the MCFT"~transductions are the bimorphisms determined by the
DMT$_\textrm{fc}$"~transductions as morphisms (Theorem~\ref{thm:bimorphism}). 
The second result generalizes the parsing result for MCFTGs in Section~\ref{sec:mcfg}.
It shows that any polynomial time parsing algorithm for MCFGs can be transformed into 
a polynomial time parsing algorithm for MCFT"~transducers (Theorem~\ref{thm:transparse}). 
For an MCFT"~transducer $M$, the algorithm parses a given input string $w$ and 
translates it into a corresponding output string; 
more precisely, the algorithm computes all pairs $(t_1,t_2)$ in the transduction of $M$ 
such that the yield of $t_1$ is $w$. 
Finally, in Section~\ref{sec:parallel}, we consider two generalizations of the MCFTG
for which the basic semantic definitions are essentially still valid.
In both cases the generalized MCFTG is able to generate 
an unbounded number of copies of a subtree, 
by allowing several occurrences of the same nonterminal (in the first case) or the same variable 
(in the second case) to appear in the right-hand side of a rule. 
Consequently, the resulting tree languages need not be semi-linear
anymore.
The first generalization is the \emph{parallel} MCFTG (or PMCFTG), which is the obvious generalization of 
the well-known parallel MCFG of~\cite{sekmatfujkas91}. Roughly speaking, in a parallel MCFTG 
(or parallel MCFG),
whenever two occurrences of the same nonterminal are introduced in a derivation step,
these occurrences must be rewritten in exactly the same way in the remainder of the derivation. 
We did not study the lexicalization 
of PMCFTGs, but for all the other results on MCFTGs there are analogous results for PMCFTGs
with almost the same proofs.
The second generalization, which we briefly consider, is the \emph{general} 
(P)MCFTG, for which we drop the restriction that the rules must be linear (in the variables). 
Thus a general (P)MCFTG can copy subtrees during one derivation step. 
General MCFTGs are discussed in~\cite{boukalsal12}. 
The general MCFTGs of multiplicity~1 are the classical IO~context-free tree grammars.
The synchronized-context-free tree languages of~\cite{boicharet15}
(which are defined by logic programs) lie between the MCFTG tree languages and 
the general PMCFTG tree languages. 
The general PMCFTG tree languages can be characterized as the images of 
the regular tree languages under arbitrary deterministic macro tree transductions,
but otherwise we have no results for general (P)MCFTGs. 

As observed above, part of the results in this contribution were first
presented in~\cite{engman00}, \cite{maleng12},
  and~\cite{maleng17}.

\section{Preliminaries}
\label{sec:prelim}
\noindent We denote the set~$\{1, 2, 3, \dotsc\}$ of positive
integers by~$\nat$ and the set of nonnegative integers by~$\nat_0 =
\nat \cup \{0\}$.  For every~$n 
\in \nat_0$, we let~$[n] = \{i \in \nat \mid i \leq n\}$.
For a set $A$, we denote its cardinality by~$\abs A$.
A partition of $A$ is a set $\Pi$ of subsets of $A$ 
such that each element of~$A$ is contained in exactly one element of~$\Pi$;
we allow the empty set $\emptyset$ to be an element of $\Pi$. 
For two functions $f \colon A \to B$ and $g \colon B \to C$ 
(where $A$, $B$, and $C$ are sets), 
the composition $g \circ f \colon A \to C$ of $f$ and $g$ 
is defined as usual by~$(g \circ f)(a) = g(f(a))$ for every~$a \in A$.

\subsection{Sequences and strings}
\label{sub:seqs}
\noindent Let $A$~be a (not necessarily finite) set.
When we view~$A$ as
a set of basic (i.e., indecomposable) elements, we call $A$~an
alphabet and each of its elements a symbol.  Note that we do not
require alphabets to be finite; finiteness will be explicitly
mentioned.\footnote{Infinite alphabets are sometimes convenient.
For instance, it is natural to view the infinite set $\{x_1,x_2,\dotsc\}$ of 
variables occurring in trees as an alphabet, see
Section~\protect{\ref{sub:sub}}.  We will use grammars with infinite
alphabets as a technical tool in Section~\protect{\ref{sub:deriv}} to
define the derivations of usual grammars, which of course have finite
alphabets.}
For every $n \in \nat_0$, we denote by~$A^n$ the $n$"~fold
Cartesian product of~$A$ containing sequences over~$A$;
i.e., $A^n = \{(\seq a1n) \mid \seq a1n \in A\}$ and $A^0 = \{(\,)\}$
contains only the empty sequence~$(\,)$, which we also denote
by~$\varepsilon$.  Moreover, we let
$A^{\scriptscriptstyle +} = \bigcup_{n \in \nat} A^n$ and
$A^* = \bigcup_{n \in \nat_0} A^n$.  
When $A$ is viewed as an alphabet, the sequences in $A^*$ are also called strings. 
Let $w = (\seq a1n)$ be a sequence (or string). 
Its length~$n$ is denoted by~$\abs w$.
For $i\in[n]$, the $i$"~th~element of~$w$ is~$a_i$.
The elements of $w$ are said to occur in~$w$.
The set $\{\seq a1n\}$ of elements of~$w$ will be denoted by~$\alp(w)$.
The sequence~$w$ is \emph{repetition-free}
if no element of $A$ occurs more than once in~$w$; i.e., $\abs{\alp(w)} = n$.
A \emph{permutation} of $w$ is a sequence $(\seq a{i_1}{i_n})$ of the same length
such that $\{\seq i1n\}=[n]$. 
Given another sequence $v = (\seq {a'}1m)$ the
concatenation~$w \cdot v$, also written just~$wv$, is simply
$(\seq a1n, \seq {a'}1m)$.  Moreover, for every~$n \in \nat_0$, the
$n$"~fold concatenation of~$w$ with itself is denoted by~$w^n$, in
particular~$w^0 = \varepsilon$.  As usual, we identify the
sequence~$(a)$ of length~$1$ with the element $a \in A$ it contains,
so~$A = A^1 \subseteq A^{\scriptscriptstyle +}$.  Consequently, we often write the
sequence~$(\seq a1n)$ as~$\word a1n$.  However, if the
$\seq a1n$ are themselves sequences, then $\word a1n$~will
always denote their concatenation and never the sequence~$(\seq a1n)$
of sequences.

\paragraph{Notation} In the following we will often denote sequences over a
set~$A$ by the same letters as the elements of~$A$.  For instance, we
will write $a = (\seq a1n)$ with~$a \in A^{\scriptscriptstyle +}$ and
$a_i \in A$ for all $i \in [n]$.  It should hopefully always be clear
whether a sequence over~$A$ or an element of~$A$ is meant.  We will
consider sequences over several different types of sets, and it would
be awkward to use different letters, fonts, or decorations (like
$\overline{a}$~and~$\vec{a}$) for all of them.

\paragraph{Homomorphisms}
Let $A$ and $B$ be sets.  
A (string) homomorphism from $A$ to $B$ is a mapping $h\colon A\to B^*$. It determines 
a mapping $h^*\colon A^*\to B^*$ which is also called a (string) homomorphism and 
which is defined inductively as follows for $w\in A^*$: 
\[ h^*(w) = 
\begin{cases} 
  \varepsilon & \text{if $w = \varepsilon$} \\
  h(a)\cdot h^*(v) & \text{if $w = av$
    with $a \in A$ and $v\in A^*$.}
\end{cases} \]
We note that $h^*$~and~$h$ coincide on~$A$ and that
$h^*(wv) = h^*(w) \cdot h^*(v)$ 
for all $w,v\in A^*$. 
In certain particular cases, which will be explicitly mentioned, we will denote $h^*$
simply by~$h$, for readability.\footnote{There will be four such cases only: 
yield functions `$\yield$' (see the remainder of this paragraph), 
rank functions `$\rk$' (see the first paragraph of Section~\protect{\ref{sub:trees}}), 
injections `$\init$' (see the first paragraph of Section~\protect{\ref{sub:sub}}),
and tree homomorphisms $\hat{h}$ (see the third paragraph of Section~\protect{\ref{sub:sub}}).
}
A~homomorphism \emph{over $A$} is a homomorphism from $A$ to itself.
We will often use the following homomorphism from $A$ to $B$, 
in the special case where $B\subseteq A$. For a string $w$ over $A$, 
the \emph{yield} of $w$ with respect to $B$,
denoted $\yield_B(w)$, is the string over $B$ that is obtained from $w$ by erasing 
all symbols not in $B$. Formally, $\yield_B$ is the homomorphism from $A$ to $B$ 
such that $\yield_B(a)=a$ if $a\in B$ and $\yield_B(a)=\varepsilon$ otherwise,
and we define $\yield_B(w)=\yield_B^*(w)$. Thus, 
\[ \yield_B(w) = 
\begin{cases}
  \varepsilon & \text{if $w = \varepsilon$} \\
  a \yield_B(v)
  & \text{if $w = av$ with $a \in B$ and $v\in A^*$} \\
  \yield_B(v) 
  & \text{if $w = av$ with $a \in A\setminus B$ and $v\in A^*$.}
\end{cases} \]
Note that $\yield_A$ is the identity on $A^*$. 

\paragraph{Context-free grammars} 
We assume that the reader is familiar with context-free
grammars~\cite{aubebo97}, which are presented here as systems
$G = (N, \Sigma, S, R)$ containing a finite alphabet~$N$ of
nonterminals, a finite alphabet~$\Sigma$ of terminals that is disjoint
to~$N$, an initial nonterminal~$S \in N$, and a finite set~$R$ of
rules of the form~$A \to w$ with a nonterminal $A \in N$ and a
string $w \in (N \cup \Sigma)^*$.  Each nonterminal~$A$
generates a language~$L(G, A)$, which is given by
$L(G, A) =\{w \in \Sigma^* \mid A \Rightarrow_G^* w\}$ using the
reflexive, transitive closure~$\Rightarrow_G^*$ of the usual rewriting
relation~$\mathord{\Rightarrow_G} = \{(uAv, uwv) \mid u, v \in (N \cup
\Sigma)^*,\, A \to w \in R\}$
of the context-free grammar~$G$.  The language generated by~$G$
is~$L(G) = L(G, S)$.  The nonterminals $A, A' \in N$ are
\emph{aliases} if $\{w \mid A \to w \in R\} = \{w \mid A'
\to w \in R\}$, which yields that $L(G, A) = L(G, A')$.  It is well
known that for every context-free grammar~$G = (N, \Sigma, S, R)$
there is an equivalent one $G' = (N', \Sigma, S_1, R')$ such that
$w$~does not contain any nonterminal more than once for every rule
$A \to w \in R'$.  This can be achieved by introducing sufficiently
many aliases as follows.  Let $m$~be the maximal number of occurrences
of a nonterminal in the right-hand side of a rule in~$R$.  We replace
each nonterminal~$A$ by new nonterminals~$\seq A1m$ with initial
nonterminal~$S_1$.  In addition, we replace each rule~$A \to w$ by all
the rules~$A_i \to w'$, where~$i \in [m]$ and $w'$~is obtained
from~$w$ by replacing the $j$"~th~occurrence of each nonterminal~$B$
in~$w$ by~$B_j$.  Thus, $\seq A1m$ are aliases. 
As an example, the grammar~$G$ with rules
$S \to \sigma SS$ and $S \to a$ is transformed into the grammar~$G'$
with rules $S_1 \to \sigma S_1S_2$, $S_2 \to \sigma S_1S_2$,
$S_1 \to a$, and $S_2 \to a$.  It should be clear that~$L(G') = L(G)$,
and in fact, the derivation trees of $G$~and~$G'$ are closely related
(by simply introducing appropriate subscripts in the derivation trees of $G$ 
or removing the introduced subscripts from the derivation trees of~$G'$).

\subsection{Trees and forests}
\label{sub:trees}
\noindent
A \emph{ranked set}, or \emph{ranked alphabet}, is a pair~$(\Sigma,
\mathord{\rk}_\Sigma)$, where $\Sigma$~is a (possibly infinite) set and 
$\mathord{\rk}_\Sigma \colon \Sigma \to \nat_0$ is a mapping 
that associates a rank to every element of~$\Sigma$. 
In what follows the elements of $\Sigma$ will be called symbols. 
For all $k\in \nat_0$, 
we let $\Sigma^{(k)} = \{ \sigma \in \Sigma \mid
\mathord{\rk}_\Sigma(\sigma) = k\}$ be the set of all symbols of rank~$k$.  
We sometimes indicate the rank~$k$ of a
symbol~$\sigma\in\Sigma$ explicitly, as in~$\sigma^{(k)}$.
Moreover, as usual, we just write~$\Sigma$ for the ranked alphabet~$(\Sigma,
\mathord{\rk}_\Sigma)$, and whenever $\Sigma$ is clear from the context, 
we write `$\mathord{\rk}$' instead of `$\mathord{\rk}_\Sigma$'. 
If $\Sigma$ is finite, then we denote by $\mr_\Sigma$ the maximal rank 
of the symbols in $\Sigma$; i.e., $\mr_\Sigma=\max\{\rk(\sigma)\mid \sigma\in\Sigma\}$.
The mapping $\mathord{\rk}^*$ from~$\Sigma^*$
to~$\nat_0^*$, as defined in the paragraph on homomorphisms in Section~\ref{sub:seqs}, 
will also be denoted by `$\mathord{\rk}$'.
It associates a multiple rank (i.e., a sequence of ranks) 
to every sequence of elements of~$\Sigma$. 
The union of ranked alphabets 
$(\Sigma,\mathord{\rk}_\Sigma)$~and~$(\Delta, \mathord{\rk}_\Delta)$ 
is~$(\Sigma \cup \Delta, \mathord{\rk}_\Sigma \cup \mathord{\rk}_\Delta)$;
it is again a ranked alphabet provided that
the same rank~$\rk_\Sigma(\gamma) = \rk_\Delta(\gamma)$ is assigned to all
symbols~$\gamma \in \Sigma \cap \Delta$.  

We build trees over the ranked alphabet~$\Sigma$ such that the nodes
are labeled by elements of~$\Sigma$ and the rank of the node label
determines the number of its children.  Formally we define trees as 
nonempty strings over $\Sigma$ as follows. The set~$T_\Sigma$
of \emph{trees over~$\Sigma$} is the smallest set~$T \subseteq
\Sigma^{\scriptscriptstyle +}$ such that $\sigma \word t1k \in T$ for
all $k \in \nat_0$, $\sigma \in \Sigma^{(k)}$, and $\seq t1k \in T$.
As usual, we will also denote the string~$\sigma \word t1k$ by the term~$\sigma(\seq
t1k)$.  If we know that~$t \in T_\Sigma$ and $t = \sigma(\seq t1k)$,
then it is clear that $k \in \nat_0$, $\sigma \in \Sigma^{(k)}$, and
$\seq t1k \in T_\Sigma$, so unless we need stronger assumptions, we
will often omit the quantifications of~$k$, $\sigma$, and~$\seq t1k$.
It is well known that if $\sigma w \in T_\Sigma$ with $k \in \nat_0$,
$\sigma \in \Sigma^{(k)}$, and $w \in \Sigma^*$,  then there 
are unique trees~$\seq t1k \in T_\Sigma$ such that~$w = \word t1k$.
Any subset of~$T_\Sigma$ is called a tree language over~$\Sigma$.  A
detailed treatment of trees and tree languages is presented
in~\cite{gecste84} (see also~\cite{eng75b,gecste97}).

Trees can be viewed as node-labeled graphs in a well-known way.  As
usual, we use \textsc{Dewey} notation to address the nodes of a tree;
these addresses will be called positions.  Formally, a \emph{position}
is an element of~$\nat^*$.  Thus, it is a sequence of positive
integers, which, intuitively, indicates successively in which subtree
the addressed node can be found.  More precisely, the root is at
position~$\varepsilon$, and the position~$pi$ with $p \in\nat^*$~and~$i \in \nat$
refers to the $i$"~th child of the node at position~$p$.  
The set~$\pos(t) \subseteq \nat^*$ of positions of a tree~$t
\in T_\Sigma$ with $t = \sigma(\seq t1k)$ is defined inductively by
$\pos(t) = \{\varepsilon\} \cup \{ ip \mid i \in [k],\, p \in
\pos(t_i)\}$.  The tree~$t$ associates a label to each of its
positions, so it induces a mapping~$t \colon \pos(t) \to \Sigma$ such
that $t(p)$~is the \emph{label} of~$t$ at position~$p$.  Formally,
if~$t = \sigma(\seq t1k)$, then~$t(\varepsilon) = \sigma$ and $t(ip) =
t_i(p)$.  
For nodes $p,p'\in\pos(t)$, we say as usual that 
$p'$ is an \emph{ancestor} of $p$ if $p'$ is a prefix of $p$; 
i.e., there exists $w\in\nat^*$ such that $p=p'w$. 
A \emph{leaf} of~$t$ is a position~$p \in \pos(t)$
with~$t(p) \in \Sigma^{(0)}$.  
The \emph{yield} of~$t$, denoted
by~$\yield(t)$, is the sequence of labels of its leaves, read from
left to right. However, as usual, we assume the existence of a special symbol $e$ of rank~0
that represents the empty string and is omitted from~$\yield(t)$. 
Formally $\yield(t) = \yield_{\Sigma^{(0)}\setminus\{e\}}(t)$,
where $\yield_B$ is defined in the paragraph on homomorphisms in Section~\ref{sub:seqs}. 

A \emph{forest} is a sequence of trees; i.e., an element of~$T_\Sigma^*$.
Note that every tree of~$T_\Sigma$ is a forest of length~$1$.  
A forest can be viewed as a node-labeled graph in a natural way, 
for instance by connecting the roots of its trees by ``invisible'' $\#$"~labeled directed edges,
in the given order. This leads to the following obvious extension of \textsc{Dewey} notation
to address the nodes of a forest. Formally, from now on, a \emph{position}
is an element of the set~$\{\#^n p\mid n\in\nat_0, \,p\in\nat^*\}\subseteq (\nat\cup\{\#\})^*$,
where $\#$ is a special symbol not in~$\nat$.
Intuitively, the root of the $j$"~th tree of a forest is at position $\#^{j-1}$ and,
as before, the position~$pi$ refers to the $i$"~th child of the node at position~$p$. 
For each forest~$t = (\seq 
t1m)$ with $m \in \nat_0$~and~$\seq t1m \in T_\Sigma$, the
set~$\pos(t)$ of positions of~$t$ is defined by~$\pos(t) =
\bigcup_{j = 1}^m \{ \#^{j-1}p \mid p \in \pos(t_j)\}$.  Moreover, for
every $j \in [m]$~and~$p \in \pos(t_j)$, we let $t(\#^{j-1}p) =
t_j(p)$ be the label of~$t$ at position~$\#^{j-1}p$.\footnote{These definitions 
are consistent with those given in the previous paragraph 
for trees, which are forests of length~1.   
}
Let~$\Omega \subseteq\Sigma$ be a selection of symbols.  For every~$t \in
T_\Sigma^*$, we let $\pos_\Omega(t) = \{ p \in
\pos(t) \mid t(p) \in \Omega\}$ be the set of all $\Omega$"~labeled
positions of~$t$.  For every $\sigma \in \Sigma$, we simply
write~$\pos_\sigma(t)$ instead of~$\pos_{\{\sigma\}}(t)$, and we say
that $\sigma$~\emph{occurs in}~$t$ if $\pos_\sigma(t) \neq \emptyset$.
The set of symbols in~$\Omega$ that occur in~$t$ is denoted
by~$\alp_\Omega(t)$; i.e., $\alp_\Omega(t) = \{t(p) \mid p \in
\pos_\Omega(t)\}$.\footnote{Note that 
$\alp(t)=\{\seq t1m\}$ by Section~\protect{\ref{sub:seqs}}. This will, however, never be used.
} 
The forest~$t$ is \emph{uniquely $\Omega$"~labeled} if no symbol
in~$\Omega$ occurs more than once in~$t$; i.e.,
$\abs{\pos_\omega(t)} \leq 1$ for every $\omega \in \Omega$.  It is
well known, and can easily be proved by induction on the structure
of~$t$, that $\abs{\pos(t)} + m \leq 2 \cdot
\abs{\pos_{\Sigma^{(0)}}(t)} + \abs{\pos_{\Sigma^{(1)}}(t)}$ for every
forest~$t \in T_\Sigma^*$ of length~$m$.

\paragraph{Regular tree grammars}
A \emph{regular tree grammar} (in short, RTG) over~$\Sigma$ is a
context-free grammar $G = (N, \Sigma, S, R)$ such that $N$~is a ranked
alphabet with~$\rk(A) = 0$ for every~$A \in N$, $\Sigma$~is a ranked alphabet, 
and $w$~is a tree in~$T_{N \cup \Sigma}$ for every rule $A
\to w$ in~$R$. Throughout this contribution we assume that $G$ is in normal form;
i.e., that all its rules are of the form $A\to \sigma(\seq A 1k)$ with $k\in\nat_0$, 
$A,\seq A 1k \in N$, and $\sigma\in\Sigma^{(k)}$.
The language~$L(G)$ generated by an RTG~$G$ is a
\emph{regular tree language}.  The class of all regular tree languages
is denoted by~$\text{RT}$.  We assume the reader to be familiar with
regular tree grammars~\cite[Section~6]{gecste97}, and also more or less familiar
with (linear, nondeleting) context-free tree 
grammars~\cite[Section~15]{gecste97}, 
which we formally define in Section~\ref{sec:mcftg}.

\subsection{Substitution}
\label{sub:sub}
\noindent
In this subsection we define and discuss first- and second-order
substitution of trees and forests.  To this end, we use a fixed
countably infinite alphabet~$X = \{x_1, x_2, \dotsc\}\cup\{\hole\}$ of \emph{variables},
which is disjoint to the ranked alphabet~$\Sigma$, and for every $k
\in \nat_0$ we let $X_k = \{ x_i \mid i \in [k]\}$ be the first
$k$~variables from~$X$.  Note that $X_0 = \emptyset$. 
The use of the special variable $\hole$ will be explained in Section~\ref{sec:lex} 
(before Lemma~\ref{lem:main}). 
For $Z\subseteq X$, the set~$T_\Sigma(Z)$ of trees over~$\Sigma$ with variables in~$Z$ is
defined by~$T_\Sigma(Z) = T_{\Sigma \cup Z}$, where every variable
$x \in Z$ has rank~$0$.  Thus, the variables can only occur at the leaves.  
We will be mainly interested in the substitution of patterns.  
For every $k \in \nat_0$, we define the set~$P_\Sigma(X_k)$ of $k$"~ary
\emph{patterns} to consist of all trees~$t \in T_\Sigma(X_k)$ such
that each variable of $X_k$ occurs exactly once in~$t$; 
i.e., $\abs{\pos_x(t)} = 1$ for every~$x \in X_k$.\footnote{Note that 
the variable $\hole$ does not occur in patterns.
} 
Consequently, $P_\Sigma(X_0) = T_\Sigma(X_0) = T_\Sigma$, and for all distinct $i,j \in
\nat_0$ the sets $P_\Sigma(X_i)$~and~$P_\Sigma(X_j)$ are disjoint.
This allows us to turn the set~$P_\Sigma(X)
= \bigcup_{k \in \nat_0} P_\Sigma(X_k)$ of all patterns into a ranked
set such that~$P_\Sigma(X)^{(k)} = P_\Sigma(X_k)$ for every $k \in 
\nat_0$; in other words, for every $t \in P_\Sigma(X)$ let $\rk(t)$~be
the unique integer~$k \in \nat_0$ such that $t \in P_\Sigma(X_k)$.\footnote{Since
$P_\Sigma(X)\subseteq (\Sigma\cup X)^*$ by definition, every pattern 
$t\in P_\Sigma(X)$ also has a multiple rank $\rk_{\Sigma\cup X}(t)\in\nat_0^*$. 
This will, however, never be used. We also observe that we will not consider 
trees over the ranked set $P_\Sigma(X)$.
}
Since `$\mathord{\rk}$' also denotes $\mathord{\rk}^*$ 
(see the first paragraph of Section~\ref{sub:trees}), `$\mathord{\rk}$'~is also a mapping
from~$P_\Sigma(X)^*$ to~$\nat^*_0$.
There is a natural rank-preserving \emph{injection} $\init \colon
\Sigma \to P_\Sigma(X)$ of the alphabet~$\Sigma$ into the set of
patterns, which is given by~$\init(\sigma) = \sigma(\seq x1k)$ for
every $k \in \nat_0$~and~$\sigma \in \Sigma^{(k)}$.  Note that
$\init(\sigma) = \sigma$ if $k = 0$.  
The mapping $\init^*$ from~$\Sigma^*$
to~$P_\Sigma(X)^*$, as defined in Section~\ref{sub:seqs}, 
will also be denoted by `$\init$'.
It is a rank-preserving injection that associates a sequence of patterns 
to every sequence of elements of~$\Sigma$. 

We start with first-order substitution, in which variables are replaced by
trees.  For a tree~$t \in T_\Sigma(X)$, a set~$Z \subseteq X$ of variables, 
and a mapping $f \colon Z \to T_\Sigma(X)$, the
\emph{first-order substitution}~$t[f]$, also written as $t[z \gets
f(z) \mid z \in Z]$, yields the tree in~$T_\Sigma(X)$ obtained by
replacing in~$t$ every occurrence of~$z$ by~$f(z)$ for every $z \in
Z$.  Formally, $t[f]$~is defined by induction on the structure of~$t$
as follows:
\[ t[f] =
\begin{cases}
  f(z) & \text{if $t = z$ with $z \in Z$} \\
  \sigma(t_1[f], \dotsc, t_k[f]) & \text{if $t = \sigma(\seq t1k)$
    with $\sigma \in \Sigma\cup X$, $\sigma \notin Z$.}
\end{cases} \]
We note that~$t[f] = h^*(t)$, where $h$~is the string homomorphism
over~$\Sigma \cup X$ such that $h(\alpha) = f(\alpha)$
if~$\alpha \in Z$ and $h(\alpha) = \alpha$ otherwise.

Whereas we replace $X$"~labeled nodes (which are leaves) in first-order substitution, 
in second-order substitution we replace $\Sigma$"~labeled nodes (which can also be
internal nodes); i.e., nodes with a label in~$\Sigma^{(k)}$ for some
$k \in \nat_0$.  Such a node is replaced by a $k$-ary pattern, in which the
variables~$\seq x1k$ are used as unique placeholders for the
$k$~children of the node.  In fact, second-order substitutions are
just tree homomorphisms.  Let $\Sigma$~and~$\Delta$ be ranked
alphabets.  A \emph{(simple) tree homomorphism} from~$\Sigma$
to~$\Delta$ is a rank-preserving mapping $h \colon \Sigma \to P_\Delta(X)$; i.e.,
$\rk(h(\sigma)) = \rk(\sigma)$ for every $\sigma \in
\Sigma$.\footnote{Since $h(\sigma)$~is a pattern for every~$\sigma \in
  \Sigma$, the tree homomorphism~$h$ is simple; i.e., linear and
  nondeleting.  This is the only type of tree homomorphism considered
  in this paper (except briefly in the last section).}  
It determines a mapping~$\hat{h} \colon T_\Sigma(X)
\to T_\Delta(X)$, and we will use~$\hat{h}$ also to denote the mapping 
$(\hat{h})^* \colon
T_\Sigma(X)^* \to T_\Delta(X)^*$ as defined in 
the paragraph on homomorphisms in Section~\ref{sub:seqs}. 
Roughly speaking, for a tree (or forest)~$t$, the tree (or
forest)~$\hat{h}(t)$ is obtained from~$t$ by replacing, for every $p
\in \pos_\sigma(t)$ with label~$\sigma \in \Sigma^{(k)}$, the subtree at
position~$p$ by the pattern~$h(\sigma)$, into which the $k$~subtrees
at positions $p1, \dotsc, pk$ are (first-order) substituted for the
variables $\seq x1k$, respectively.  Since $h(\sigma)$~is a pattern,
these subtrees can neither be copied nor deleted, but they can be
permuted.  Thus, the pattern~$h(\sigma)$ is ``folded'' into~$t$ at
position~$p$.  Formally, the mapping~$\hat{h}$, which we also call
tree homomorphism, is defined inductively as follows for $t\in T_\Sigma(X)$:
\[ \hat h(t) =
\begin{cases}
  x & \text{if $t = x$ with $x \in X$} \\
  h(\sigma)[x_i \gets \hat{h}(t_i) \mid 1 \leq i \leq k] 
  & \text{if $t = \sigma(\seq t1k)$ with $\sigma \in \Sigma$}.
\end{cases} \]
Clearly, $\hat{h}(t)$~only depends on the values of~$h$ for the symbols
occurring in~$t$; in other words, if $g$~is another tree 
homomorphism from~$\Sigma$ to~$\Delta$ such that $g(\sigma) =
h(\sigma)$ for every $\sigma \in \alp_\Sigma(t)$, then $\hat{g}(t) =
\hat{h}(t)$.  We additionally observe that $\hat{h}(t) =
\delta(\hat{h}(t_1), \dotsc, \hat{h}(t_k))$ if $t = \sigma(\seq t1k)$
and $h(\sigma) = \init(\delta)$ for some $\delta \in  \Delta$.  A tree
homomorphism~$h$ is a \emph{projection} if for every~$\sigma \in
\Sigma$ there exists~$\delta \in \Delta$ such that $h(\sigma) =
\init(\delta)$.  Thus, a projection is just a relabeling of the nodes
of the trees.  For a ranked alphabet~$\Sigma$, a tree homomorphism
\emph{over~$\Sigma$} is a tree homomorphism from~$\Sigma$ to
itself.

The following lemma states elementary properties of (simple) tree
homomorphisms.  They can easily be proved by induction on the structure
of trees in~$T_\Sigma(X)$ and then extended to forests
in~$T_\Sigma(X)^*$.

\begin{lemma}
  \label{lem:treehom}
  Let $h$~be a tree homomorphism from~$\Sigma$ to~$\Delta$, and let
  $t \in T_\Sigma(X)^*$ and $u = \hat{h}(t)$.
  \begin{compactenum}[\upshape (1)] 
  \item $\abs{\pos_x(u)} = \abs{\pos_x(t)}$ for every $x \in X$.
  \item $\abs{\pos_\delta(u)} = \sum_{\sigma \in \Sigma}
    \abs{\pos_\sigma(t)} \cdot \abs{\pos_\delta(h(\sigma))}$ for every
    $\delta \in \Delta$. 
  \end{compactenum}
\end{lemma}

By the first statement of this lemma, tree homomorphisms preserve
patterns and their ranks; i.e., $\hat{h}(t) \in P_\Delta(X_k)$ for
all~$t \in P_\Sigma(X_k)$.  Moreover, $\hat{h}(t) \in
P_\Delta(X)^*$ and $\rk(\hat{h}(t)) = \rk(t)$ for
all $t \in P_\Sigma(X)^*$.  

Next, we recall two
other easy properties of tree homomorphisms.  Namely, they distribute
over first-order substitution, and they are closed under composition
(see \protect{\cite[Corollary~8(5)]{bak79}}).

\begin{lemma}
  \label{lem:treehomsub}
  Let $h$~be a tree homomorphism from~$\Sigma$ to~$\Delta$, let $t
  \in T_\Sigma(X)$, and let $f \colon Z \to
  T_\Sigma(X)$ for some $Z \subseteq X$.  Then
  $\hat{h}(t[f]) = \hat{h}(t)[\hat{h} \circ f]$. 
\end{lemma}

\begin{lemma}
  \label{lem:treehomcomp}
  Let $h_1$~and~$h_2$ be tree homomorphisms from~$\Sigma$ to~$\Omega$
  and from~$\Omega$ to~$\Delta$, respectively, and let $h = \hat{h}_2 \circ
  h_1$, which is a tree homomorphism from~$\Sigma$ to~$\Delta$.
  Then~$\hat{h} = \hat{h}_2 \circ \hat{h}_1$.
\end{lemma}

These lemmas have straightforward proofs.
Lemma~\ref{lem:treehomsub} can be proved by induction on the
structure of~$t$, and then Lemma~\ref{lem:treehomcomp} 
can be proved by showing that $\hat{h}(t) =
\hat{h}_2(\hat{h}_1(t))$, again by induction on the structure of~$t$, 
using Lemma~\ref{lem:treehomsub} in the induction step.

In the remainder of this subsection we consider tree homomorphisms
over~$\Sigma$.  Let $t$~be a forest
in~$T_\Sigma(X)^*$ and let $\sigma = (\seq \sigma1n)
\in \Sigma^n$ with $n \in \nat_0$ be a repetition-free sequence of
symbols in~$\Sigma$.  Moreover, let $u = (\seq u1n)$ be a forest
in~$P_\Sigma(X)^n$ such that $\rk(u) = \rk(\sigma)$.\footnote{Recall
that this means that $u_i \in P_\Sigma(X_{\rk(\sigma_i)})$ for every $i\in[n]$.} 
The \emph{second-order substitution} $t[\sigma \gets u]$ yields the
forest $\hat{h}(t) \in T_\Sigma(X)^*$, where
$h$~is the tree homomorphism over~$\Sigma$ corresponding to~$[\sigma
\gets u]$, which is defined by $h(\sigma_i) = u_i$ for~$i \in [n]$ and
$h(\tau) = \init(\tau)$ for $\tau \in \Sigma \setminus
\{\seq\sigma1n\}$.  If~$t \in P_\Sigma(X)^*$,
then $t[\sigma \gets u] \in P_\Sigma(X)^*$ and
$\rk(t[\sigma \gets u]) = \rk(t)$ by Lemma~\ref{lem:treehom}(1).
Obviously, the order of the  symbols and trees in $\sigma$~and~$u$ is
irrelevant:  if $\sigma' = (\seq\sigma{i_1}{i_n})$ and $u' = (\seq
u{i_1}{i_n})$, where $(\seq i1n)$~is a permutation of~$(1,\dotsc,n)$, then
$t[\sigma' \gets u'] = t[\sigma \gets u]$.  Thus, the use of sequences 
is just a way of associating each symbol~$\sigma_i$ with its replacing
tree~$u_i$.  Clearly, $t[\sigma \gets u] = t$ if no symbol of~$\sigma$
occurs in~$t$; i.e., if $\alp_\Sigma(t) \cap \alp(\sigma) =
\emptyset$.  We also note that $t[\sigma \gets \init(\sigma)] = t$ and 
$\init(\sigma)[\sigma \gets u] = u$.  Finally $t[\sigma \gets u] =
t_1[\sigma \gets u] \cdot t_2[\sigma \gets u]$ if $t = t_1t_2$ for
forests $t_1$~and~$t_2$.

In the next lemma, we state some additional elementary properties of
second-order substitution.

\begin{lemma}
  \label{lem:comm-assoc}
  Let $t \in T_\Sigma(X)^*$ be a forest and
  $\sigma_1, \sigma_2 \in \Sigma^*$ be
  repetition-free sequences of symbols.  Moreover, let $u_1, u_2 \in
  P_\Sigma(X)^*$ be forests of patterns such that
  $\rk(u_1) = \rk(\sigma_1)$ and $\rk(u_2) = \rk(\sigma_2)$.
  \begin{compactenum}[\upshape (1)]
  \item If $\alp(\sigma_1) \cap \alp(\sigma_2) = \emptyset$
    (i.e., $\sigma_1\sigma_2$~is repetition-free), then
\begin{quote} $t[\sigma_1 \gets u_1][\sigma_2 \gets u_2] =
    t[\sigma_1\sigma_2 \gets u_1[\sigma_2 \gets u_2] \cdot u_2]$.
\end{quote} 
  \item If $\alp(\sigma_1) \cap
    \alp(\sigma_2) = \emptyset$ and $\alp_\Sigma(u_1) \cap
    \alp(\sigma_2) = \emptyset$, then
\begin{quote} $t[\sigma_1 \gets u_1][\sigma_2 \gets u_2] =
    t[\sigma_1\sigma_2 \gets u_1u_2]$.
\end{quote}
  \item If
    $\alp(\sigma_1) \cap \alp(\sigma_2) = \emptyset$ and
    $\alp_\Sigma(u_2) \cap \alp(\sigma_1) = \emptyset$, then
\begin{quote} $t[\sigma_1 \gets u_1][\sigma_2 \gets u_2] = t[\sigma_2 \gets
    u_2][\sigma_1 \gets u_1[\sigma_2 \gets u_2]]$. 
\end{quote}
  \item If $\alp_\Sigma(t) \cap \alp(\sigma_2)
    \subseteq \alp(\sigma_1)$, then
\begin{quote} $t[\sigma_1 \gets u_1][\sigma_2 \gets u_2] = t[\sigma_1 \gets
    u_1[\sigma_2 \gets u_2]]$. 
\end{quote}
  \end{compactenum} 
\end{lemma}

\begin{proof} 
  Let $h_1$~and~$h_2$ be the tree homomorphisms over $\Sigma$ that correspond to
  $[\sigma_1 \gets u_1]$~and~$[\sigma_2 \gets u_2]$, as defined above.
  Moreover, let $h$~be the tree homomorphism that corresponds to
  $[\sigma_1\sigma_2 \gets u_1[\sigma_2 \gets u_2] \cdot u_2]$.
  Provided that $\sigma_1\sigma_2$ is repetition-free, it is easy to
  check that $h = \hat{h}_2 \circ h_1$, and hence $\hat{h} = \hat{h}_2
  \circ \hat{h}_1$ by Lemma~\ref{lem:treehomcomp}.  This shows the
  first equality.  If additionally no symbol of~$\sigma_2$ occurs
  in~$u_1$, then $u_1[\sigma_2 \gets u_2] = u_1$, which shows the
  second equality.  The third equality is a direct consequence of the
  first two because $t[\sigma_1\sigma_2 \gets u_1[\sigma_2 \gets u_2]
  \cdot u_2] = t[\sigma_2\sigma_1 \gets u_2 \cdot u_1[\sigma_2 \gets
  u_2]]$.  To prove the fourth equality, let $g$~be the tree
  homomorphism that corresponds to~$[\sigma_1 \gets u_1[\sigma_2 \gets
  u_2]]$.  By Lemma~\ref{lem:treehomcomp}, it now suffices to show
  that $\hat{h}_2(h_1(\sigma)) = g(\sigma)$ for every $\sigma \in
  \alp_\Sigma(t)$.  This is obvious for $\sigma \in \alp(\sigma_1)$.
  If $\sigma \in \alp_\Sigma(t) \setminus \alp(\sigma_1)$ then, by
  assumption, $\sigma \notin \alp(\sigma_2)$, and so both sides of the
  equation are equal to~$\init(\sigma)$.
\end{proof}

In particular, Lemma~\ref{lem:comm-assoc}(3) 
implies that $t[\sigma_1 \gets
u_1][\sigma_2 \gets u_2] = t[\sigma_2 \gets u_2][\sigma_1 \gets u_1]$
provided that $\alp(\sigma_1) \cap \alp(\sigma_2) = \emptyset$,
$\alp_\Sigma(u_2) \cap \alp(\sigma_1) = \emptyset$, and
$\alp_\Sigma(u_1) \cap \alp(\sigma_2) = \emptyset$.  This is called
the confluence or commutativity of substitution in~\cite{cou87}.
Similarly, Lemma~\ref{lem:comm-assoc}(4) 
is called the associativity of substitution 
in~\cite{cou87}.  As shown in the proof above, 
these two properties of substitution are
essentially special cases of the composition of tree homomorphisms
as characterized in Lemma~\ref{lem:treehomcomp}.

Above, we have defined the substitution of a forest (of patterns) for a
repetition-free sequence over~$\Sigma$. In the next section we also
need to simultaneously substitute several forests for several such
sequences.  That leads to the following formal definitions, which may
now seem rather superfluous.  Let $\LL = \{\seq \sigma1k\}$ be a
finite subset of~$\Sigma^*$ such that $\word
\sigma1k$ is repetition-free, where $\word \sigma1k = \varepsilon$
if~$k = 0$.  A (second-order) \emph{substitution function} for~$\LL$
is a mapping~$f \colon \LL \to P_\Sigma(X)^*$
such that $\rk(f(\sigma)) = \rk(\sigma)$ for every $\sigma \in \LL$.
For a forest~$t \in P_\Sigma(X)^*$, the
\emph{simultaneous second-order substitution}~$t[f]$, also written as 
$t[\sigma \gets f(\sigma) \mid \sigma \in \LL]$, yields
$t[f] = t[\word \sigma1k \gets f(\sigma_1) \cdots f(\sigma_k)]$.
Clearly, $t[f]$~does not depend on the given order of the elements
in~$\LL$.  In the special case $\LL \subseteq \Sigma$ we obtain a
notion of second-order substitution that does not involve sequences,
with $f \colon \LL \to P_\Sigma(X)$.  In that case we have $t[f] =
t[(\seq \sigma1k) \gets (f(\sigma_1), \dotsc, f(\sigma_k))]$.

\section{Multiple context-free tree grammars}
\label{sec:mcftg}
\noindent
In this section we introduce the main formalism discussed in this
contribution: the multiple context-free tree grammars. 
In the first subsection we define their syntax and least fixed point semantics and
in the second and third subsection we discuss two alternative semantics, 
namely their derivation trees and their derivations, respectively. 
In the second subsection we also define the notion of \LDTR"~equivalence of 
multiple context-free tree grammars, which formalizes grammatical similarity. 

\subsection{Syntax and least fixed point semantics}
\label{sub:defmcftg}
\noindent
We start with the syntax of multiple context-free tree grammars, which we
explain after the formal definition.  The definition of their
semantics follows after that explanation. Then we give two examples.

\begin{definition}
  \label{def:mcftg}
  \upshape
  A \emph{multiple context-free tree grammar} (in short, MCFTG) is a
  system $G = (N, \N, \Sigma, S, R)$ such that 
  \begin{compactitem}
  \item $N$~is a finite ranked alphabet of \emph{nonterminals},
  \item $\N \subseteq N^{\scriptscriptstyle +}$~is a finite set of
    \emph{big nonterminals}, which are nonempty repetition-free sequences 
    of nonterminals, such that $\alp(A)\neq \alp(A')$ for all distinct $A,A'\in\N$,
  \item $\Sigma$~is a finite ranked alphabet of \emph{terminals} such
    that $\Sigma \cap N = \emptyset$ and $\mr_\Sigma \geq 1$,\footnote{To avoid trivialities, 
  we do not consider the case where all symbols of $\Sigma$ have rank~0.} 
  \item $S \in \N \cap N^{(0)}$~is the \emph{initial (big) nonterminal} (of length~1 and rank~$0$), and
  \item $R$~is a finite set of \emph{rules} of the form~$A \to (u,
    \LL)$, where $A \in \N$ is a big nonterminal, $u \in P_{N \cup
      \Sigma}(X)^{\scriptscriptstyle +}$ is a uniquely $N$"~labeled
    forest (of patterns) such that $\rk(u) = \rk(A)$, and $\LL
    \subseteq \N$ is a set of big nonterminals such that~$\{\alp(B) \mid B \in \LL\}$ 
    is a partition of~$\alp_N(u)$.\footnote{Thus, $\alp_N(u) = \bigcup_{B
        \in \LL} \alp(B)$ and $\alp(B) \cap \alp(B') = \emptyset$ for
      all distinct $B, B' \in \LL$.} \fin
  \end{compactitem}
\end{definition}

For a given rule~$\rho = A \to (u, \LL)$, the big nonterminal~$A$,
denoted by~$\lhs(\rho)$, is called the left-hand side of~$\rho$, the
forest~$u$, denoted by~$\rhs(\rho)$, is called the right-hand side
of~$\rho$, and the big nonterminals of~$\LL$, denoted by~$\LL(\rho)$,
are called the \emph{links} of $\rho$.  

The \emph{multiplicity} (or \emph{fan-out}) of the MCFTG~$G$,
which is denoted by~$\mu(G)$, is the maximal length of its big
nonterminals. The \emph{width} of~$G$, which is denoted
by~$\wid(G)$, is the maximal rank of its nonterminals.
And the \emph{rule-width} (or \emph{rank}) of $G$, which is denoted by $\lambda(G)$, is 
the maximal number of links of its rules. Thus 
$\mu(G)=\max\{\abs{A}\mid A\in\N\}$, $\,\wid(G)=\mr_N=\max\{\rk(A)\mid A\in N\}$, 
and $\lambda(G)=\max\{\abs{\LL(\rho)}\mid \rho\in R\}$.

Next, we define two syntactic restrictions.  An MCFTG~$G$ is a \emph{multiple
  regular tree grammar} (in short, MRTG) if~$\wid(G) = 0$, and it is
a \emph{(simple) context-free tree grammar} (in short, spCFTG)
if~$\mu(G) = 1$; i.e.,~$\N \subseteq N$.  In an~MRTG all nonterminals
thus have rank~$0$, and in an~spCFTG all big nonterminals are
nonterminals since their length is exactly~$1$.  Consequently, in
an~spCFTG we may simply assume that~$\N = N$, and thus there is no
need to specify~$\N$ for it.  In the literature, a rule~$A \to (u, \LL)$ of
an~spCFTG is usually written as~$\init(A) \to u$, in which~$\init(A) =
A(\seq x1{\rk(A)})$ and~$\LL$ can be omitted because it must be equal
to~$\alp_N(u)$.  Since the right-hand side~$u$ of this rule is a
pattern, our context-free tree grammars are simple; i.e., linear
and nondeleting.

Let us discuss the requirements on the components of~$G$ in more
detail.  Each big nonterminal is a nonempty repetition-free
sequence~$A = (\seq A1n)$ of nonterminals from~$N$.
Repetition-freeness of~$A$ requires that all these nonterminals~$A_i$
are distinct (cf.\@ Section~\ref{sub:seqs}).  
The requirement that `$\alp$' is injective on $\N$ 
(i.e., that $\alp(A)\neq \alp(A')$ for all distinct $A,A'\in\N$)
means that $\N$ can be viewed as consisting of sets of nonterminals,
where each set is equipped with a fixed linear order
(viz.\@ the set $\alp(A)=\{\seq A1n\}$ with the order $\sqsubseteq$ 
such that $A_1 \sqsubset \cdots \sqsubset A_n$). 
Moreover, since the alphabet $N$ is ranked, every big
nonterminal~$A$  has a (multiple) rank $\rk(A) = (\rk(A_1), \dotsc, 
\rk(A_n)) \in \nat_0^n$ (cf.\@ Section~\ref{sub:trees}), and
similarly, every forest~$u = (\seq u1n)$ with $\seq u1n \in P_{N \cup
  \Sigma}(X)$ has a (multiple) rank $\rk(u) = (\rk(u_1), \dotsc,
\rk(u_n)) \in \nat_0^n$ (cf.\@ Section~\ref{sub:sub}).  Thus, a rule
$A \to (u, \LL)$ of~$G$ is of the form~$(\seq A1n) \to ((\seq u1n),
\LL)$ where $n\in\nat_0$, $A_i \in N$ and $u_i \in P_{N \cup
  \Sigma}(X_{\rk(A_i)})$ for every~$i \in [n]$, and $\LL \subseteq \N$.
The use of sequences is irrelevant; it is just a way of associating   
each $A_i\in\alp(A)$ with the corresponding pattern $u_i$, thus 
facilitating the formal description of the syntax and semantics of~$G$.
Additionally, in the above rule, $u$~is uniquely $N$"~labeled, which means that also
in~$u$ no nonterminal occurs more than once (cf.\@
Section~\ref{sub:trees}).  This requirement, which is not essential
but technically convenient, is similar to the restriction discussed
for context-free grammars at the end of Section~\ref{sub:seqs}.
Moreover, the set~$\{\alp(B) \mid B\in \LL\}$ forms a partition
of~$\alp_N(u)$.  Since each big nonterminal~$B$ is repetition-free,
`$\alp$' is injective on $\N$, and
$u$~is uniquely $N$"~labeled, we obtain that each big nonterminal
from~$\LL$ occurs ``spread-out'' exactly once in~$u$ and no other
nonterminals occur in~$u$.  More precisely, for each big nonterminal
$B = (\seq C1m) \in \LL$ with~$\seq C1m \in N$, there is a unique
repetition-free sequence $p_B = (\seq p1m) \in \pos_N(u)^m$ of 
positions such that $(u(p_1), \dotsc, u(p_m)) = (\seq C1m)$,
and we have that $\alp(p_B)\cap \alp(p_{B'})=\emptyset$ for every other $B'\in\LL$
and $\pos_N(u) = \bigcup_{B\in \LL} \alp(p_B)$.
Note that if $\LL = \{\seq B1k\}$
with $\seq B1k \in \N$, then the concatenation $\word B1k \in
N^*$ of the elements of~$\LL$ is repetition-free
and $\alp(\word B1k) = \alp_N(u)$. 

\begin{figure}
  \centering
  \includegraphics{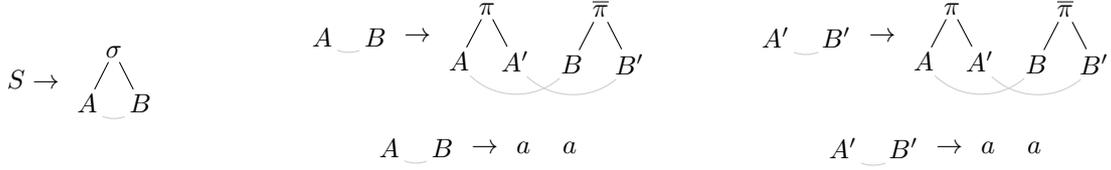}
  \caption{Rules of the MRTG~$G$ of
    Example~\protect{\ref{exa:copy}}.}
  \label{fig:links}
\end{figure}

Intuitively, the application of the above rule~$\rho = A \to (u, \LL)$
consists of the simultaneous application of the $n$~spCFTG rules
$A_i(\seq x1{\rk(A_i)}) \to u_i$ to an occurrence of the
``spread-out'' big nonterminal~$A = (\seq A1n)$ and the
introduction of (occurrences of) the new ``spread-out'' big
nonterminals from~$\LL$.  Every big nonterminal~$B = (\seq C1m) \in
\LL$, as above, can be viewed as a link between the positions~$\seq
p1m$ of~$u$ with labels~$\seq C1m$ as well as a link between the
corresponding positions after the application of~$\rho$~(see
Figure~\ref{fig:links}).  The rule~$\rho$ can only be applied to
positions with labels~$\seq A1n$ that are joined by such 
a link.  Thus, rule applications are~``local'' in the sense that a
rule can rewrite only nonterminals that were previously introduced
together in a single step of the derivation, just as for the local
unordered scattered context grammar of~\cite{ramsat99}, which is
equivalent to the multiple context-free (string) grammar.  However,
since it is technically a bit problematic to define such derivation
steps between trees in~$T_{N \cup \Sigma}$ that are not necessarily
uniquely $N$"~labeled (because it additionally requires to keep track
of each link as a sequence of positions rather than as a big
nonterminal), we prefer to define the language generated by the
MCFTG~$G$ through a least fixed point semantics similar to that of
multiple context-free (string) grammars in~\cite{sekmatfujkas91}.  As
will be discussed in Section~\ref{sub:dtrees}, this is closely related
to a semantics in terms of derivation trees, similar to that of
(string-based) linear context-free rewriting systems
in~\cite{shaweijos87}.  The derivations of an MCFTG will be considered
in Section~\ref{sub:deriv}.

In an spCFTG, a nonterminal~$A$ of rank~$k$ can be viewed as a
generator of trees in~$P_\Sigma(X_k)$ using derivations that start
with~$A(\seq x1k)$.  In the same fashion, a big nonterminal~$A$ of an
MCFTG generates nonempty forests in~$P_\Sigma(X)^*$ of the
same rank as~$A$, as defined next.  Let $G = (N, \N, \Sigma, S, R)$ be
an MCFTG.  For every big nonterminal~$A \in \N$ we define the
\emph{forest language generated by~$A$}, denoted by~$L(G, A)$, as
follows.  For all big nonterminals~$A \in \N$ simultaneously,
$L(G, A) \subseteq P_\Sigma(X)^*$ is the smallest
set of forests such that for every rule~$A \to (u, \LL) \in R$, if
$f \colon \LL \to P_\Sigma(X)^*$ is a
substitution function for~$\LL$ such that $f(B) \in L(G, B)$ for
every~$B \in \LL$, then $u[f] \in L(G, A)$.  Note that $u[f]$~is a
simultaneous second-order substitution as defined at the end of
Section~\ref{sub:sub}.  The fact that $f$~is a substitution function
for~$\LL$ means that $\rk(f(B)) = \rk(B)$ for every~$B \in \LL$, which
implies that $\rk(t) = \rk(A)$ for every~$t \in L(G, A)$; in particular, 
$t$ is a nonempty forest of the same length as $A$.  The
\emph{tree language~$L(G)$ generated by~$G$} is defined
by~$L(G) = L(G, S) \subseteq T_\Sigma$.  Two MCFTGs $G_1$~and~$G_2$
are \emph{equivalent} if $L(G_1) = L(G_2)$.\footnote{When viewing
  $G_1$~and~$G_2$ as specifications of the string languages
  $\yield(L(G_1))$~and~$\yield(L(G_2))$, they are \emph{strongly
    equivalent} if~$L(G_1) = L(G_2)$ and \emph{weakly equivalent} if
  $\yield(L(G_1)) = \yield(L(G_2))$.}  A tree language is
\emph{multiple context-free} (\emph{multiple regular}, (simple) \emph{context-free})
if it is generated by an~MCFTG (MRTG, spCFTG).  The corresponding
class of generated tree languages is denoted by~$\text{MCFT}$
($\text{MRT}$, $\text{CFT}_\text{sp}$).

As observed above, each big nonterminal can be viewed as a nonempty subset of~$N$,  
together with a fixed linear order on its elements.  It is easy to see that 
the tree language $L(G)$ generated by~$G$ does
not depend on that order.  For a given big nonterminal~$A = (\seq
A1n)$ and a given permutation~$A' = (\seq A{i_1}{i_n})$ of~$A$, we can
change every rule $A \to ((\seq u1n), \LL)$ into the rule $A' \to ((\seq
u{i_1}{i_n}), (\LL \setminus \{A\}) \cup \{A'\})$, provided that we also
change~$\LL(\rho)$ into~$(\LL(\rho) \setminus \{A\}) \cup \{A'\}$ for
every other rule~$\rho \in R$.

The restriction that the right-hand side of a rule of~$G$ must be
uniquely $N$"~labeled can be compensated for by the appropriate use of
aliases.  Two big nonterminals~$A, A' \in \N$ are said to be
\emph{aliases} if $\{(u, \LL) \mid A \to (u, \LL) \in R\} = \{(u, \LL)
\mid A' \to (u, \LL) \in R\}$.  It is not difficult to see that~$L(G,
A) = L(G, A')$ for aliases $A$~and~$A'$.  Of course, in examples, we
need not specify the rules of an alias (but we often will).
Additionally, to improve the readability of examples, we will write a rule~$A \to
(u, \LL)$ as~$\init(A) \to u$ and specify~$\LL$ separately.  Recall
from Section~\ref{sub:sub} that if~$A = (\seq A1n)$ and~$\rk(A_i) =
k_i$ for every~$i \in [n]$, then
\[ \init(A) = (A_1(\seq x1{k_1}), \dotsc, A_n(\seq x1{k_n}))
  \enspace. \]
If all the big nonterminals of $G$ are mutually disjoint, in the sense that
they have no nonterminals in common (i.e., $\alp(B) \cap \alp(B') =
\emptyset$ for all distinct~$B, B' \in \N$), then it is not even
necessary to specify~$\LL$ because it clearly is equal to $\{B \in \N
\mid \alp(B) \subseteq \alp_N(u)\}$.  

\begin{example}
  \label{exa:copy}
  \upshape
  We first consider the MRTG $G = (N, \N, \Sigma, S, R)$ such that
  (i)~$N = \{S, A, B, A', B'\}$, (ii)~$\N = \{S, (A, B), (A', B')\}$,
  and (iii)~$\Sigma = \{\sigma^{(2)}, \pi^{(2)}, \bar{\pi}^{(2)},
  a^{(0)}\}$.  Thus, $\mu(G) = 2$. And $\wid(G) = 0$ because $G$~is a
  multiple regular tree grammar.  The big nonterminal~$(A', B')$ is an
  alias of~$(A, B)$.  The set~$R$ contains the rules (illustrated in
  Figure~\ref{fig:links})
  \begin{alignat*}{5}
    S &\to \sigma(A, B) \qquad &
    (A, B) &\to (\pi(A, A'), \bar{\pi}(B, B')) \qquad
    & (A', B') &\to (\pi(A, A'), \bar{\pi}(B, B')) \\
    && (A,B) &\to (a,a) 
    & (A', B') &\to (a,a) \enspace.
  \end{alignat*}
  Since the big nonterminals in~$\N$ are mutually disjoint, the
  set~$\LL$ of links of each rule is uniquely determined.  In fact,
  $\LL = \{(A, B)\}$~for the leftmost rule in the first line, $\LL =
  \{(A, B), (A', B')\}$~for the two remaining rules in the first line,
  and $\LL = \emptyset$~for the two rules in the second line.  The tree
  language~$L(G)$ generated by~$G$ consists of all trees~$\sigma(t,
  \bar{t}\,)$, where $t$~is a tree over~$\{\pi, a\}$ and $\bar{t}$~is
  the same tree with every~$\pi$ replaced by~$\bar{\pi}$.  For readers
  familiar with the multiple context-free grammars
  of~\cite{sekmatfujkas91} we note that this tree language can be
  generated by such a grammar with nonterminals $S$~and~$C$, where 
  $C$~corresponds to our big nonterminal~$(A,B)$ and its alias, using
  the three rules
  \begin{compactitem}
  \item $S \to f[C]$ with $f(x_{11}, x_{12}) = \sigma x_{11} x_{12}$,
  \item $C \to g[C, C]$ with $g(x_{11}, x_{12}, x_{21}, x_{22}) = (\pi
    x_{11} x_{21},\, \bar{\pi} x_{12} x_{22})$, and 
  \item $C \to (a, a)$.
  \end{compactitem}
  Note that the variables~$x_{11}$, $x_{12}$, $x_{21}$, and~$x_{22}$
  of~\cite{sekmatfujkas91} correspond to our nonterminals~$A$, $B$,
  $A'$, and~$B'$, respectively.  In fact, every tree language
  in~$\text{MRT}$ can be generated by a multiple context-free grammar,
  just as every regular tree language can be generated by a
  context-free grammar (see Section~\ref{sub:trees}).  We will prove
  in Section~\ref{sec:mcfg} (Theorem~\ref{thm:mcft-to-mcf}) that this
  even holds for MCFT.  \fin
\end{example}

\begin{example}
  \label{exa:main}
  \upshape
  As a second example we consider the MCFTG~$G = (N, \N, \Sigma, S,
  R)$ such that
  \begin{compactitem}
  \item $N = \{S^{(0)}, A^{(0)}, B^{(1)}, B'^{(1)}, T_1^{(1)},
    T_2^{(0)}, T_3^{(0)}\}$ and $\N = \{S, A, B, B', (T_1, T_2,
    T_3)\}$, and
  \item $\Sigma = \{\sigma^{(2)}, \alpha^{(1)}, \beta^{(1)},
    \gamma^{(1)}, \tau^{(0)}, \nu^{(0)}\}$.
  \end{compactitem}
  Consequently, $\mu(G) = 3$~and~$\wid(G) = 1$.  The (big)
  nonterminal~$B'$ is an alias of~$B$.  The set~$R$ consists of the
  following rules~$\seq \rho16$ and the two rules
  $\rho'_3$~and~$\rho'_4$ with left-hand side~$B'$ (illustrated in
  Figure~\ref{fig:links2}).
  \begin{alignat*}{7}
  \rho_1 \colon&& S &\to \alpha(A)
  &\rho_2 \colon&& A &\to T_1(\sigma(B(T_2), T_3)) \\ 
  \rho_3 \colon&& B(x_1) &\to \sigma(B(x_1), B'(A))
  &\rho'_3 \colon&& B'(x_1) &\to \sigma(B(x_1), B'(A)) \\
  \rho_4 \colon&& B(x_1) &\to x_1  
  &\rho'_4 \colon&& B'(x_1) &\to x_1 \\
  \rho_5 \colon&& (T_1(x_1), T_2, T_3) &\to
  (\alpha(T_1(\beta(x_1))),\, \alpha(T_2),\, \gamma(T_3)) \qquad
  &\rho_6 \colon&& (T_1(x_1), T_2, T_3) &\to (x_1, \tau, \nu)
  \enspace.
  \end{alignat*}
  Since, again, all big nonterminals in~$\N$ are mutually disjoint,
  the sets of links of these rules are uniquely determined.  They are,
  in fact, as follows: 
  \begin{alignat*}{5}
    \LL(\rho_1) &= \{A\} \qquad
    &\LL(\rho_3) &= \LL(\rho'_3) = \{B, B', A\} \qquad
    & \LL(\rho_2) &= \{B, (T_1, T_2, T_3)\} \\
    &&\LL(\rho_4) &= \LL(\rho'_4) = \LL(\rho_6) = \emptyset
    & \LL(\rho_5) &= \{(T_1, T_2, T_3)\} \enspace. 
  \end{alignat*}
  Let~$T = (T_1, T_2, T_3)$.  The rule~$\rho_6$ shows that~$(x_1,
  \tau, \nu) \in L(G, T)$.  We can write the rule~$\rho_5$ also as $T
  \to (\alpha T_1 \beta x_1,\, \alpha T_2,\, \gamma T_3)$.
  Substituting~$(x_1, \tau, \nu)$ for~$T$ in~$u_5 = \rhs(\rho_5)$ we
  obtain that $L(G, T)$~also contains the forest $u_5[(T_1, T_2, T_3)
  \gets (x_1, \tau, \nu)] = (\alpha \beta x_1,\, \alpha \tau,\, \gamma
  \nu)$.  Then, substituting this forest for~$T$ in~$u_5$ we obtain
  that~$L(G, T)$ also contains $(\alpha \alpha \beta \beta x_1,\,
  \alpha \alpha \tau,\, \gamma \gamma \nu)$.  Continuing in this way
  we see that $L(G, T) = \{(\alpha^n \beta^n x_1, \alpha^n \tau,
  \gamma^n \nu) \mid n \in \nat_0\}$.  If we temporarily view~$A$ as a
  terminal, then $B(x_1)$~generates all trees~$t \in T_{\{\sigma, A,
    x_1\}}$ such that the left-most leaf of~$t$ has label~$x_1$ and
  all other leaves have label~$A$.  The right-hand side~$u_2 =
  T_1(\sigma(B(T_2), T_3))$ of~$\rho_2$ generates all trees~$u_2[B
  \gets t, T \gets t']$ with~$t$ as above and~$t' \in L(G, T)$; i.e., all trees
  $\alpha^n \beta^n \sigma(t[x_1 \gets \alpha^n \tau], \gamma^n \nu)$.
  This should give an idea of the form of the trees in~$L(G, A)$, and
  hence of the trees in~$L(G)$.  \fin
\end{example}

\begin{figure}
  \centering
  \includegraphics{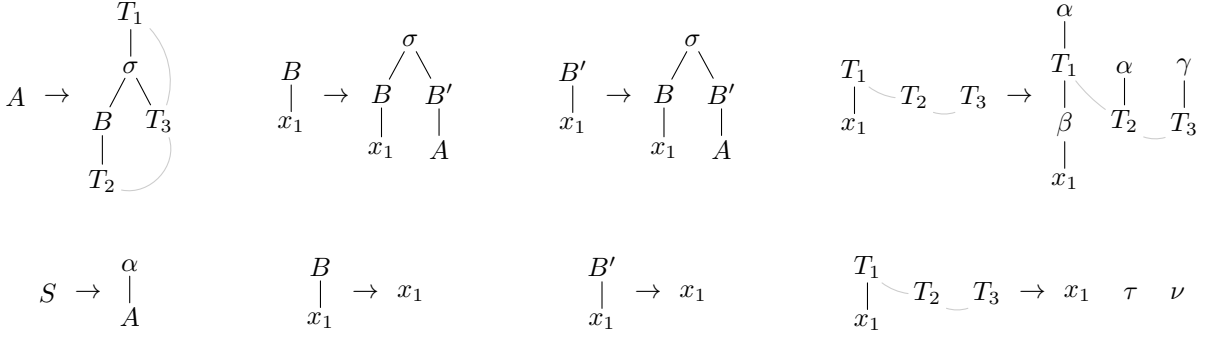}
  \caption{Rules of the MCFTG~$G$ of
    Example~\protect{\ref{exa:main}}.}
  \label{fig:links2}
\end{figure}

\subsection{Derivation trees}
\label{sub:dtrees}
\noindent The least fixed point semantics of an MCFTG~$G = (N, \N,
\Sigma, S, R)$ naturally leads to the notion of a derivation tree
of~$G$ that we define now.  We assume that for every rule~$\rho$
of~$G$, the links in~$\LL(\rho)$ are linearly ordered by an arbitrary,
fixed order~$\sqsubseteq$.  Whenever we write~$\LL(\rho) = \{\seq
B1k\}$ with $B_i \in \N$ for all $i \in [k]$, we will assume
that~$B_1 \sqsubset \dotsb \sqsubset B_k$.  The \emph{derivation tree
  grammar} of~$G$ is the RTG~$G_\der = (N_\der, R, S, R_\der)$ defined
as follows.\footnote{See Section~\protect{\ref{sub:trees}} for the definition of
  a regular tree grammar (RTG). Note that, in this contribution, 
RTGs are in normal form.}  First, $N_\der = \N$; i.e., its
nonterminals (of rank~$0$) are the big nonterminals of~$G$. Its
initial nonterminal is~$S$, which is the initial (big) nonterminal of~$G$.
Second, its terminal ranked alphabet is the set~$R$ of rules of~$G$
such that the rule~$\rho$ has rank~$\rk(\rho) = \abs{\LL(\rho)}$.\footnote{Note that, 
therefore, the rule-width of $G$ (as defined after Definition~\protect{\ref{def:mcftg}}) is
$\lambda(G)=\mr_R$, the maximal rank of its rules.}
Finally, the set~$R_\der$ consists of all rules~$A \to \rho(\seq B1k)$
such that~$\rho \in R$, $\lhs(\rho) = A$, and $\LL(\rho) = \{\seq
B1k\}$.  For~$A \in \N$, a \emph{derivation tree} of~$G$ of type~$A$
is a tree~$d \in T_{\N \cup R}$ such that~$A \Rightarrow^*_{G_\der}
d$.  Obviously, every derivation tree has a unique type, viz.~$\lhs(d(\varepsilon))$; 
i.e., the left-hand side of the rule that labels
its root.  We will denote the set of derivation trees of~$G$ of
type~$A$ by~$\DL(G_\der, A)$.  Note that $L(G_\der, A) = \DL(G_\der,
A) \cap T_R$.  To capture the semantics of~$G$, only the derivation
trees in~$L(G_\der) \subseteq T_R$ are relevant, but we will need the
other derivation trees for technical reasons in proofs.
As in the case of context-free grammars, it can be checked locally
whether a tree~$d \in T_{\N \cup R}$ is a derivation tree.  In fact,
let us say that the type of a position~$p \in \pos(d)$ is either~$d(p)$
if~$d(p) \in \N$, or~$\lhs(d(p))$ if~$d(p) \in R$.  Then $d$~is a
derivation tree if and only if for every position~$p \in
\pos_R(d)$ with $\LL(d(p)) = \{\seq B1k\}$,
the child~$pi$ of~$p$ has type~$B_i$ for every $i\in[k]$.

The \emph{value} of a
derivation tree~$d$ of type~$A$, denoted by~$\val(d)$, is a forest
in~$P_{N \cup \Sigma}(X)^{\scriptscriptstyle +}$ of the same rank
as~$A$ in~$G$, and is defined inductively as follows.  If~$d = A \in
\N$, then $\val(d) = \init(A)$.  If~$d = \rho(\seq d1k)$ for some
$\rho = A \to (u, \LL) \in R$ with $\LL = \{\seq B1k\}$ (and thus
$d_i$~is of type~$B_i$ for every~$i \in [k]$), then $\val(d) = u[B_i
\gets \val(d_i) \mid 1 \leq i \leq k]$.  The value~$\val(d)$ 
of the derivation tree $d$ can clearly be
computed in linear time. We also observe here that its computation 
can be realized by a macro tree transducer~\cite{coufra82,engvog85} (see
Lemma~\ref{lem:valismtt} in Section~\ref{sec:charact}).  Since that
macro tree transducer is finite-copying, `$\val$'~can also be realized
by a deterministic MSO-transducer (see~\cite{engman99}). 

\begin{example}
  \label{exa:dtree}
  \upshape
  The derivation tree grammar~$G_\der$ of the grammar~$G$ of
  Example~\ref{exa:main} has the following eight rules, where $T =
  (T_1, T_2, T_3)$ and the linear order of the links of each rule
  of~$G$ is fixed as indicated in Example~\ref{exa:main}:
  \begin{center}
    \vspace{-3ex}
    \begin{minipage}{0.4\linewidth}
      \begin{alignat*}{7}
        S &\to \rho_1(A) & A &\to \rho_2(B, T) \\
        B &\to \rho_3(B, B', A) \qquad & B' &\to \rho'_3(B, B', A) \\
        B &\to \rho_4 & B' &\to \rho'_4 \\
        T &\to \rho_5(T) \qquad & T &\to \rho_6 \enspace.
      \end{alignat*}
    \end{minipage}
    \resizebox{0.55\linewidth}{!}{
      \framebox{
        \begin{minipage}{.91\linewidth}
          Rules of Example~\ref{exa:main}:
          \begin{alignat*}{7}
            \rho_1 \colon&& S &\to \alpha(A)
            &\rho_2 \colon&& A &\to T_1(\sigma(B(T_2), T_3)) \\ 
            \rho_3 \colon&& B(x_1) &\to \sigma(B(x_1), B'(A))
            &\rho'_3 \colon&& B'(x_1) &\to \sigma(B(x_1), B'(A)) \\
            \rho_4 \colon&& B(x_1) &\to x_1  
            &\rho'_4 \colon&& B'(x_1) &\to x_1 \\
            \rho_5 \colon&& (T_1(x_1), T_2, T_3) &\to
            (\alpha(T_1(\beta(x_1))),\, \alpha(T_2),\, \gamma(T_3)) \qquad
            &\rho_6 \colon&& (T_1(x_1), T_2, T_3) &\to (x_1, \tau, \nu)
            \enspace.
          \end{alignat*}
        \end{minipage}
      }}
  \end{center}
  An example of a derivation tree of type~$A$ is~$d =
  \rho_2(\rho_3(\rho_4, B', A), \rho_5(\rho_6))$, which is shown in
  Figure~\ref{fig:deriv}.  Obviously, $\val(\rho_4) = x_1$~and we
  have~$\val(\rho_6) = (x_1, \tau, \nu)$.  Then
  $\val(\rho_5(\rho_6))$~is obtained by substituting~$(x_1, \tau,
  \nu)$ for~$T = (T_1, T_2, T_3)$ in the right-hand side of
  rule~$\rho_5$.  We saw in Example~\ref{exa:main} that the result is
  $(\alpha \beta x_1,\, \alpha \tau,\, \gamma \nu)$.  Similarly,
  $\val(\rho_3(\rho_4, B', A))$~is obtained from~$\rhs(\rho_3)$ by
  substituting~$\val(\rho_4) = x_1$ for~$B$ (and simultaneously
  substituting~$\init(B')$ for~$B'$ and~$\init(A)$ for~$A$, without
  effect).  The result is~$\sigma(x_1, B'(A))$.  Finally, $\val(d)$~is
  obtained from~$\rhs(\rho_2)$ by substituting~$\sigma(x_1, B'(A))$
  for~$B$ and $(\alpha \beta x_1,\, \alpha \tau,\, \gamma \nu)$
  for~$T$.  Hence $\val(d) = \alpha \beta(\sigma(\sigma(\alpha \tau,
  B'(A)), \gamma \nu))$.  The process is illustrated in
  Figure~\ref{fig:deriv}.  An example of a derivation tree
  in~$L(G_\der, S)$ is  
  \[ d' = \rho_1(d[(B', A) \gets (\rho'_4,\, \rho_2(\rho_4,
    \rho_6))]) \enspace, \] which equals 
  $\rho_1(\rho_2(\rho_3(\rho_4, \rho'_4, \rho_2(\rho_4, \rho_6)),
  \rho_5(\rho_6)))$.  Clearly, $\val(\rho'_4) =
  x_1$~and~$\val(\rho_2(\rho_4, \rho_6)) = \sigma(\tau, \nu)$.  It is
  straightforward to compute $\val(d') = \alpha \alpha
  \beta(\sigma(\sigma(\alpha \tau, \sigma(\tau, \nu)), \gamma \nu)) =
  \alpha(\val(d)[(B', A) \gets (x_1,\, \sigma(\tau, \nu))])$, which shows
  that `$\val$'~distributes over substitution. \fin
\end{example}

\begin{figure}
  \centering
  \includegraphics{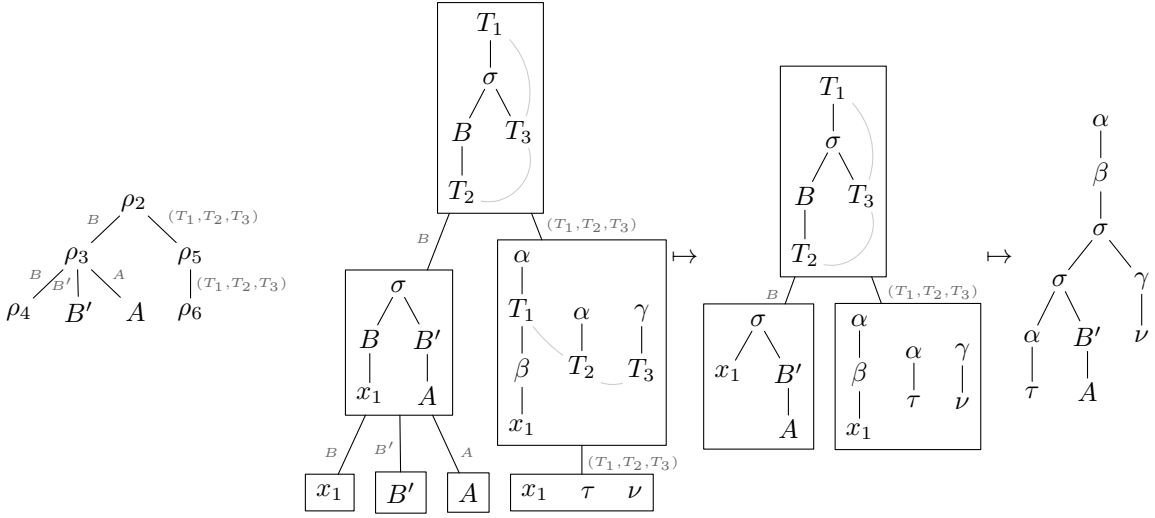}
  \caption{Derivation tree of the MCFTG~$G$ of
    Example~\protect{\ref{exa:main}} and illustration of the
    (bottom-up) computation of its value.} 
  \label{fig:deriv}
\end{figure}

From the least fixed point semantics we immediately obtain a
characterization by derivation trees.

\begin{theorem} 
  \label{thm:dtree}
  $L(G, A) = \val(L(G_\der, A))$ for every~$A \in \N$.  In particular,
  $L(G) = \val(L(G_\der))$.  
\end{theorem}

\begin{proof}
  Obviously, the sets~$\val(L(G_\der, A))$ satisfy the fixed point
  requirement for all~$A \in \N$, which says that for every rule~$\rho = A \to (u,
  \LL) \in R$ and substitution function~$f$ for~$\LL$ such that
  $f(B)$~is in $\val(L(G_\der, B))$ for every~$B \in \LL$, we have that
  $u[f] \in \val(L(G_\der, A))$.  In fact, if~$\LL = \{\seq B1k\}$
  and~$f(B_i) = \val(d_i)$ for all $i \in [k]$, then $u[B\gets f(B)\mid B\in\LL]$~is equal
  to $\val(\rho(\seq d1k))$ by definition of~`$\val$'.  This shows
  that $L(G, A) \subseteq \val(L(G_\der, A))$ for every~$A \in \N$.
  In the other direction, it is easy to show that $\val(d) \in L(G,
  A)$ for every~$d \in L(G_\der, A)$ and every~$A \in \N$ by induction
  on the structure of the derivation tree~$d$.
\end{proof}

This theorem implies that the emptiness problem is decidable for
$L(G)$~and~$L(G, A)$.  In fact, $L(G) = \emptyset$ if and only if
$L(G_\der) = \emptyset$, which is decidable because $G_\der$~is an
RTG; and similarly for~$L(G,A)$.  It is now also very easy to see that
$L(G, A) = L(G, A')$ for aliases $A$~and~$A'$: if $\rho = A \to (u,
\LL)$ and $\rho' = A' \to (u, \LL)$ are rules and $d = \rho(\seq d1k)$
is in~$L(G_\der, A)$, then $d' = \rho'(\seq d1k)$ is in~$L(G_\der,
A')$ and $\val(d) = \val(d')$, under the assumption that $\LL$~has the
same linear order in $\rho$~and~$\rho'$.

We will need three simple properties of derivation trees, 
which are stated in the next three lemmas.  The first
is a generalization of Lemma~\ref{lem:treehom}(2) and states that for
every derivation tree of~$G$, the number of occurrences of a terminal
in~$\val(d)$ is the sum of its occurrences in the right-hand sides of
the rules that occur in~$d$.  Also, the number of occurrences of a 
nonterminal in~$\val(d)$ is equal to the number of its ``occurrences''
(as part of a big nonterminal) in~$d$.  

\begin{lemma}
  \label{lem:valocc} 
  Let~$d \in \DL(G_\der, A)$ with~$A \in \N$, and let $\sigma \in
  \Sigma$~and~$C \in N$.
  \begin{compactenum}[\upshape (1)]
  \item $\abs{\pos_\sigma(\val(d))} = \sum_{p \in \pos_R(d)}
    \abs{\pos_\sigma(\rhs(d(p)))}$.  
  \item $\abs{\pos_C(\val(d))} = \sum_{B \in \N_C} \abs{\pos_B(d)}$,
    where $\N_C =\{B \in \N \mid C \in \alp(B)\}$. 
  \item $\val(d) \in T_\Sigma$ if and only if~$d \in T_R$.
  \end{compactenum}
\end{lemma}

\begin{proof}
  The proofs of (1)~and~(2) can be achieved by induction on the
  structure of~$d$.  The statements are obvious for~$d = A \in \N$
  because we obtain~$0 = 0$ in~(1), $1 = 1$~in~(2) if~$C
  \in \alp(A)$, and~$0 = 0$ in~(2) otherwise.  Let us now consider $d
  =  \rho(\seq d1k)$ for some rule~$\rho = A \to (u, \LL)$ with $\LL =
  \{\seq B1k\}$.  By the definition of~`$\val$' we have $\val(d) =
  u[B_i \gets \val(d_i) \mid 1 \leq i \leq k]$, which equals
  the second-order substitution $u[\word B1k \gets \val(d_1) \dotsm
  \val(d_k)]$ by the definition of simultaneous second-order
  substitution.  Let $h$~be the tree homomorphism over~$N \cup \Sigma$
  corresponding to~$[\word B1k \gets \val(d_1) \dotsm \val(d_k)]$.  It
  is now straightforward to prove (1)~and~(2) using
  Lemma~\ref{lem:treehom}(2) and the induction hypotheses for~$\seq
  d1k$.  It follows from~(2) that $\alp_N(\val(d)) = \bigcup_{B \in
    \alp_\N(d)} \alp(B)$, which proves~(3).
\end{proof}

The second property is that `$\val$'~distributes over second-order
substitution, of which an example was presented at the end of
Example~\ref{exa:dtree}.  It can be viewed as a generalization of
Lemma~\ref{lem:comm-assoc}(4).  For convenience, and because it is all
we will need, we only prove this for the case where just one big
nonterminal is replaced.

\begin{lemma}
  \label{lem:valsub} 
  Let $A, B \in \N$, and let $d \in \DL(G_\der, A)$ and $d' \in
  \DL(G_\der, B)$ be derivation trees of type $A$ and $B$ 
  such that $B\in\alp_{\N}(d)$.  
  Then $\val(d[B \gets d']) = \val(d)[B \gets \val(d')]$.
\end{lemma}

\begin{proof}
  As in Lemma~\ref{lem:valocc}, we proceed by induction on the
  structure of~$d$.  For~$d = A \in \N$ both sides of the equation are
  equal to~$\val(d')$ if~$B = A$ and equal to~$\init(A)$ otherwise.  Now 
  we consider~$d = \rho(\seq d1k)$ for some~$\rho = A \to (u, \LL)$
  with $\LL = \{\seq B1k\}$.  Then 
  \begin{alignat*}{3}
    &\phantom{{}={}} \val(d)[B \gets \val(d')] \\
    &= u[\word B1k \gets \val(d_1) \dotsm
    \val(d_k)] \, [B \gets \val(d')] = u \bigl[\word B1k \gets
    (\val(d_1) \dotsm \val(d_k))[B \gets \val(d')]\, \bigr] \\
    &= u \bigl[\word B1k \gets \val(d_1)[B \gets \val(d')] \dotsm
    \val(d_k)[B \gets \val(d')]\, \bigr] \\
    &= u \bigl[\word B1k \gets \val(d_1[B \gets d']) \dotsm \val(d_k[B
    \gets d']) \bigr] = \val(\rho(d_1[B \gets d'], \dotsc, d_k[B \gets
    d'])) \\ 
    &= \val(d[B \gets d']) \enspace,
  \end{alignat*}
  where the second equality is by Lemma~\ref{lem:comm-assoc}(4) and
  the fourth by the induction hypotheses.
\end{proof}

We will use the following simple third property in the proofs of
Lemmas \ref{lem:terminal-removal}~and~\ref{lem:d-chain-free}.  

\begin{lemma}
  \label{lem:easy-dtree} 
  Let~$F \subseteq R$, $\,\N' \subseteq \N$, and~$\D_B = \DL(G_\der, B)
  \cap T_{\N' \cup F}$ for every~$B \in \N$.  Moreover, let~$A \in
  \N$, $t \in \val(\D_A)$, and $L_{\langle A,t\rangle} = \{d \in \D_A
  \mid \val(d) = t\}$.  If $\val(\D_B)$~is finite for every~$B \in \N$,
  then $L_{\langle A, t\rangle}$~is a regular tree language. 
\end{lemma}

\begin{proof}
  An RTG for~$L_{\langle A, t\rangle}$ has the nonterminals~$\langle
  B, v\rangle$ with $B \in \N$~and~$v \in \val(\D_B)$, of which 
  the nonterminal~$\langle A, t\rangle$ is initial.  For every rule
  $\rho = B \to (u, \LL)$ of~$G$ with $\rho \in F$ and $\LL = \{\seq
  B1k\}$, it has all the rules $\langle B, v\rangle \to \rho(\langle
  B_1, v_1\rangle, \dotsc, \langle B_k, v_k\rangle)$ such that $v_i
  \in \val(\D_{B_i})$ for every~$i \in [k]$, and $v = u[B_i \gets v_i
  \mid 1 \leq i \leq k]$.  Moreover, for every~$B \in \N'$ it has the
  rule~$\langle B, \init(B)\rangle \to B$.  This grammar can be viewed
  as a deterministic bottom-up finite tree
  automaton~\cite{gecste84,gecste97} that, for every derivation
  tree~$d \in T_{\N' \cup F}$, computes the type of~$d$ and its
  value~$\val(d)$.
\end{proof}

Let us turn to the comparison of the derivation trees of two MCFTGs $G$ and $G'$.
We can define $G$ and $G'$ to be ``$\X$"~equivalent'', where $\X$~is a
class of tree transductions, if there are value-preserving tree
transductions in~$\X$ from the derivation trees of each grammar to
those of the other grammar.  The idea here is that $G$~and~$G'$ are
grammatically closely related if $\X$~is a relatively simple class of
tree transductions.  For that purpose we choose the class~$\X =
\text{\LDTR}$, which we define now.  To define tree transducers we
use the infinite alphabet~$Y = \{y_1, y_2, \dotsc\}$ of \emph{input
  variables} to avoid confusion with the set~$X$ of
variables used in MCFTGs (the set~$X$ will also be used
as output variables, or 
parameters, for macro tree transducers in Section~\ref{sec:charact}).
For every~$k \in \nat_0$, we let~$Y_k = \{y_i \mid i \in [k]\}$.

A \emph{linear deterministic top-down tree transducer with regular
  look-ahead} (in short, \LDTR"~transducer) from~$\Omega$ to~$\Sigma$
is a system $M = (Q, \Omega, \Sigma, q_0, R)$, where $Q$~is a finite
set of \emph{states}, $\Omega$~and~$\Sigma$ are finite ranked
alphabets of \emph{input} and \emph{output symbols} with $Q \cap
\Sigma = \emptyset$, $q_0 \in Q$ is the \emph{initial state}, and
$R$~is a finite set of \emph{rules}.  Each rule in~$R$ is of the
form 
\[ \langle q,\, \omega(y_1\colon L_1, \dotsc, y_k \colon L_k) \colon
L_0 \rangle \to \zeta \enspace, \] 
where $q \in Q$, $k \in \nat_0$, $\omega \in \Omega^{(k)}$, $L_0,\seq L1k$~are
regular tree languages over~$\Omega$ (specified, e.g., by RTGs),
and~$\zeta\in T_{(Q \times Y_k) \cup \Sigma}$ using the ranked
alphabet~$Q \times Y_k$, in which every element has rank~$0$.
Additionally, we require that each~$y \in Y_k$ occurs at most once
in~$\zeta$ (linearity property), and that if $\langle q,\,
\omega(y_1 \colon L'_1, \dotsc, y_k \colon L'_k) \colon L'_0 \rangle
\to \zeta'$ is another rule in~$R$ (for the same $q$~and~$\omega$),
then there exists an index~$0 \leq i \leq k$ such that~$L_i \cap L'_i
= \emptyset$ (determinism property).  If~$L_i = T_\Omega$ in the
above rule, then we omit~`${} \colon L_i$'.  An \LDTR"~transducer
is called an LDT"~transducer (without regular look-ahead) if~$L_i
= T_\Omega$ for every~$0 \leq i \leq k$ in every rule.

For every input tree~$s \in T_\Omega$ and every state~$q \in Q$, we
define the \emph{$q$"~translation} of~$s$ by~$M$, denoted by~$M_q(s)$,
inductively as follows.  If $s = \omega(\seq s1k)$, the above rule is
in~$R$, $s \in L_0$, and~$s_i \in L_i$ for every~$i \in 
[k]$, then 
\[ M_q(s) = \zeta[\langle q', y_i \rangle \gets M_{q'}(s_i)
\mid q' \in Q, \,1 \leq i \leq k] \enspace. \]  
We observe that $M_q(s)$~is undefined if there does not exist an
appropriate rule or, using the rule above, $M_{q'}(s_i)$~is undefined
for some~$\langle q', y_i \rangle$ that occurs in~$\zeta$.  Moreover,
the \emph{tree transduction realized by~$M$}, also denoted by~$M$, is
the partial function~$M \colon T_\Omega \to T_\Sigma$, which is given
by~$M(s) = M_{q_0}(s)$ for every $s \in T_\Omega$.  The tree~$M(s)$,
provided it is defined, is also called the \emph{translation} of~$s$
by~$M$.  We denote by~\LDTR\@ the class of all tree transductions
realized by \LDTR"~transducers.  Note that every tree
homomorphism~$\hat{h}$ from~$\Omega$ to~$\Sigma$ can be realized by an
LDT"~transducer with one state~$q$ and with the rules $\langle q,\,
\omega(\seq y1k) \rangle \to h(\omega)[x_i \gets \langle q, y_i
\rangle \mid 1 \leq i\leq k]$ for every~$k\in\nat_0$ and $\omega \in
\Omega^{(k)}$.  We need the following two basic properties of \LDTR.

\begin{proposition}
  \label{pro:LDTRcomp}
  \LDTR\@ is closed under composition.
\end{proposition}

\begin{proof}
  This is stated after~\cite[Theorem~2.11]{eng77}.  Part~(2) of its
  proof shows the statement because the constructions in the proofs
  of~\cite[Lemmas~2.9 and~2.10]{eng77} preserve linearity.
\end{proof}

An \LDTR"~transducer $M$ is a \emph{finite-state relabeling} if, in each of its rules as above,
$\zeta$ is of the form $\sigma(\langle q_1,y_1\rangle,\dotsc,\langle q_k,y_k\rangle)$ for some
$\sigma\in \Sigma^{(k)}$ and $\seq q 1k \in Q$. Such a transducer just changes the labels 
of the nodes of the input tree. Note that every projection is a finite-state relabeling. 

\begin{proposition}
  \label{pro:LDTRinv}
  For every \LDTR"~transducer $M = (Q, \Omega, \Sigma, q_0, R)$ 
  there is a polynomial time algorithm that, for every 
  RTG~$H$ over~$\Sigma$ as input, outputs 
  an RTG~$H'$ over~$\Omega$ such that~$L(H') = M^{-1}(L(H))$.
  If~$M$ is a finite-state relabeling, then there is a linear time algorithm 
  for the same task. 
\end{proposition}

\begin{proof}
  It is well known that the class~RT is closed under inverse
  \LDTR"~transductions~\cite[Lemma~1.2 and Theorem~2.6]{eng77}.  We
  now show that the transformation can be realized in polynomial time,
  for fixed~$M$.  By~\cite[Theorem~2.8]{eng77} and (the proof
  of)~\cite[Theorem~3.5]{eng75}, the transduction~$M$ can be written
  as the ``bimorphism''~$\{(\hat{\pi}(t), \hat{h}(t)) \mid t \in K\}$, where
  $K$~is a regular tree language over a finite alphabet~$\Delta$,
  $\pi$~is a projection from~$\Delta$ to~$\Omega$, and $h$~is a
  tree homomorphism from~$\Delta$ to~$\Sigma$.  Therefore
  $M^{-1}(L(H)) = \hat{\pi}(\hat{h}^{-1}(L(H)) \cap K)$.  Hence, since the intersection
  with~$K$ and the projection~$\hat{\pi}$ can be realized in linear time
  because $K$~and~$\pi$ are fixed, we may assume in the remainder of this proof that $M$~is a
  tree homomorphism $h$ from $\Omega$ to $\Sigma$.

  Now let $H = (N, \Sigma, S, R_H)$.  As mentioned at the end of Section~\ref{sub:trees},
  we assume that $H$~is in
  normal form; i.e., that the rules in~$R_H$ are of the form~$A_0 \to
  \sigma(\seq A1m)$ with~$m \in \nat_0$, $\sigma \in \Sigma^{(m)}$,
  and~$\seq A1m \in N$.  We
  construct $H' = (N, \Omega, S, R')$ such that for every $k \in
  \nat_0$, $\omega \in \Omega^{(k)}$, and $A_0,\seq A1k \in N$, 
  if $A_0
  \Rightarrow^*_H h(\omega)[x_i \gets A_i \mid 1 \leq i
  \leq k]$, then the rule~$A_0 \to \omega(\seq A1k)$ is in~$R'$.
  It is straightforward to show that $L(H', A) = \hat{h}^{-1}(L(H, A))$ for
  every~$A \in N$.  It should be clear that the construction of~$H'$
  takes polynomial time (in the size of~$H$). In fact, it
  takes time $O(n^k)$ where $n$ is the size of $H$ and $k=\mr_\Omega+1$ 
  (and recall that $\mr_\Omega$ is the maximal rank of the symbols in~$\Omega$). 
  If $M$ is a finite-state relabeling, then it can be checked that $h$ 
  is also a projection.
  Hence the set $R'$ can be constructed such that if $h(\omega)=\init(\sigma)$
  and $A_0 \to \sigma(\seq A 1k)$ is in~$R_H$, then $A_0 \to \omega(\seq A1k)$ is in~$R'$.
  That construction only takes linear time. 
\end{proof}

We now define $\X$"~equivalence of MCFTGs $G$~and~$G'$ for
$\X = \text{\LDTR}$.  However, for future use, we give a more general
definition that  involves a tree transformation~$\varphi$ and implies
that $L(G') = \varphi(L(G))$.

\begin{definition}
  \label{def:LDTReq}
  \upshape
  Let $G = (N, \N, \Sigma, S, R)$ and $G' = (N', \N', \Sigma', S',
  R')$ be MCFTGs, and  let $\varphi$~be a mapping from~$T_\Sigma$
  to~$T_{\Sigma'}$.  The grammar~$G'$ is
  \emph{\LDTR"~$\varphi$"~equivalent} to the grammar~$G$ if there
  exist tree transductions $M \colon T_R \to T_{R'}$~and~$M' \colon
  T_{R'} \to T_R$ in~\LDTR\@ such that
  \begin{compactenum}[\indent\upshape (1)]
  \item $M(d) \in L(G'_\der)$~and~$\val(M(d)) = \varphi(\val(d))$ for
    every~$d \in L(G_\der)$, and vice versa,   
  \item $M'(d') \in L(G_\der)$~and~$\varphi(\val(M'(d')))=\val(d')$
    for every~$d' \in L(G'_\der)$.
  \end{compactenum}
  In particular, $M(d)$~must be defined for every~$d \in L(G_\der)$,
  and similarly for~$M'(d')$.  

  The grammars $G$~and~$G'$ are
  \emph{\LDTR"~equivalent} if~$\Sigma = \Sigma'$ and $\varphi$~is the
  identity on~$T_\Sigma$. \fin
\end{definition}

It directly follows from item~(1) and Theorem~\ref{thm:dtree} that
$\varphi(L(G)) \subseteq L(G')$, and $L(G')\subseteq \varphi(L(G))$ 
follows from item~(2).  Hence $L(G') = \varphi(L(G))$. 
In particular, \LDTR"~equivalent MCFTGs are
equivalent.  Since \LDTR\@ is closed under composition by
Proposition~\ref{pro:LDTRcomp}, \LDTR"~equivalence of~MCFTGs is an
equivalence relation.  That is, of course, not true for
\LDTR"~$\varphi$"~equivalence in general.

It should be noted that the notion of \LDTR"~$\varphi$"~equivalence is
independent of the linear order of the links in the rules of
$G$~and~$G'$.  In fact, if $\rho = A \to (u, \LL)$~is a rule of~$G$
with~$\LL = \{\seq B1k\}$ and we change that order into~$\{\seq
B{i_1}{i_k}\}$, where $(\seq i1k)$~is a permutation of~$(1,\dots,k)$, then a 
tree homomorphism~$h$ over~$R$ can transform the old derivation trees
into the new ones via~$h(\rho) = \rho(\seq x{i_1}{i_k})$.  That proves
the observation because tree homomorphisms are in~\LDTR\@ and \LDTR~is
closed under composition.  Thus, whenever we construct a new 
grammar $G$~or~$G'$,  we can choose those orders in a convenient way.

As observed above, \LDTR"~equivalent grammars $G$~and~$G'$ are
grammatically closely related by means of the \LDTR"~transducers
$M$~and~$M'$.  Consequently, their parsing problems are closely
related as well because the transducer~$M'$ transforms a derivation
tree of~$G'$ with value~$t \in T_\Sigma$ in linear time into one
of~$G$ with the same value~$t$.  Moreover, if $H'$~is an RTG that
generates all derivation trees of~$G'$ with value~$t$, then an 
RTG~$H$ can be constructed in polynomial time that generates all
derivation trees of~$G$ with value~$t$.  This follows from
Proposition~\ref{pro:LDTRinv} because $L(H) = M^{-1}(L(H')) \cap
L(G_\der)$. The parsing problem for MCFTGs will be discussed in more detail 
in Section~\ref{sub:parsing}.

An important example of a tree transduction that \emph{cannot} be
realized by an \LDTR"~transducer is the transformation of a
left-recursive tree into a right-recursive tree with the same yield.
\begin{figure}
  \centering
  \includegraphics{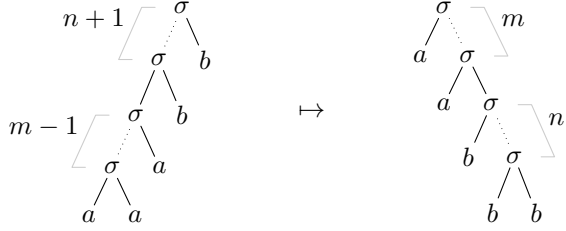}
  \caption{Translating left-recursive into right-recursive trees.}
  \label{fig:leftrighttransform}
\end{figure}
More precisely, for the ranked alphabet $\Sigma = \{\sigma^{(2)},
a^{(0)}, b^{(0)}\}$, the tree transformation 
\[ \tau = \{(\sigma^{m+n}a^mb^{n+1},\, (\sigma a)^m(\sigma b)^nb) \mid
m, n \in \nat_0\} \enspace, \] which translates the left-recursive
tree~$\sigma^{m+n}a^mb^{n+1}$ into the right-recursive tree~$(\sigma
a)^m(\sigma b)^nb$ (see Figure~\ref{fig:leftrighttransform}), cannot
be realized by any \LDTR"~transducer; i.e., $\tau \notin
\text{\LDTR}$.  This can be proved by a classical pumping argument; if
there would be such a transducer, then the language~$\{b^{n+1} a^m
b^{n+1} a^m \mid m, n \in \nat_0\}$ would be linear context-free.  By
a similar argument one can show that there is no yield-preserving
\LDTR"~transducer that translates the derivation trees of the
left-recursive context-free grammar with rules~$S\to Sa$, $S\to Sb$,
$S\to a$, and~$S\to b$ into the derivation trees of an equivalent
context-free grammar in \textsc{Greibach} Normal Form.  In other
words, the transformation of a context-free grammar into
\textsc{Greibach} Normal Form involves a grammatical transformation of
derivation trees that cannot be realized by any \LDTR"~transducer.

\subsection{Derivations}
\label{sub:deriv}
\noindent
In this subsection we present a rewriting semantics of MCFTGs,
inspired by the level grammars of~\cite{skyum74}.  The definitions and
results of this subsection will not be utilized in the other sections,
but we hope that they improve the intuition of the reader concerning
MCFTGs.

Let $G = (N, \N, \Sigma, S, R)$ be an MCFTG.  In a naive approach
we would define the derivation steps of~$G$ on trees~$t \in T_{N \cup
  \Sigma}$ and the application of a rule~$A \to (u, \LL)$ to~$t$
leading to a derivation step~$t \Rightarrow t[A \gets u]$, 
provided that $\alp(A)\subseteq \alp_N(t)$.  
Such a naive derivation is shown in Figure~\ref{fig:derivation} 
for the grammar~$G$ of Example~\protect{\ref{exa:main}}.
\begin{figure}
  \centering
  \includegraphics{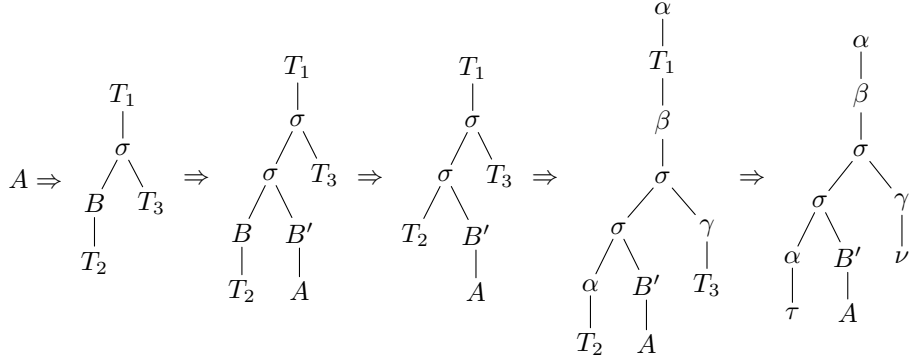}
  \caption{A naive (leftmost) derivation of the grammar~$G$ of
    Example~\protect{\ref{exa:main}}, corresponding to the derivation tree 
    in Example~\protect{\ref{exa:dtree}} and in
    Figure~\protect{\ref{fig:deriv}}.  All big nonterminals of $G$ are
    mutually disjoint and all the trees in this derivation are
    uniquely $N$-labeled.}
  \label{fig:derivation}
\end{figure}
Assuming that all big nonterminals of $G$ are mutually disjoint 
(as in Example~\protect{\ref{exa:main}}),
this naive derivation step works
if $A$~occurs exactly once in~$t$ (e.g., when $t$~is uniquely
$N$"~labeled). However, it fails if $A$ occurs  several times in~$t$ because the rule
is then applied to all occurrences simultaneously.  Moreover, if~$A =
(A_1, A_2)$ with~$A_1, A_2 \in N$, then it is unclear which
occurrences of $A_1$~and~$A_2$ are linked.  
If not all big nonterminals of $G$ are mutually disjoint, 
then it is not clear at all which nonterminals in $t$ are linked 
(even when $t$~is uniquely $N$"~labeled). 
Thus, we additionally have
to keep track of how the nonterminal occurrences in~$t$ are linked
together to form occurrences of big nonterminals.  To facilitate this,
we change for each position~$p \in \pos_N(t)$ of~$t$ the label~$t(p)$
into an appropriate label~$\langle t(p), \ell\rangle$, where $\ell \in
\nat^*$~is a position, which is also called \emph{link identifier}.
Nonterminal occurrences with the same link identifier~$\ell$ are
linked, and we only derive uniquely $(N \times \nat^*)$"~labeled
trees.  We note that the positions $p$~and~$\ell$ need not coincide.
In fact, $\ell$~is a position of the derivation tree corresponding to
the derivation. 

We need additional notation for the formalization.  As in the previous
subsection, we assume that for every rule~$\rho$ of~$G$ the
set~$\LL(\rho)$ of links is linearly ordered.  For a big
nonterminal~$A = (\seq A1n) \in \N$ and a link identifier~$\ell \in
\nat^*$, we define $A \otimes \ell = (\langle A_1, \ell \rangle,
\dotsc, \langle A_n, \ell \rangle) \in (N \times
\nat^*)^{\scriptscriptstyle +}$.  Moreover, for~$\ell \in \nat^*$ and
a rule~$\rho = A \to (u, \LL) \in R$ with $\LL = \{\seq B1k\}$, we
define $(u, \LL) \otimes \ell = u[B_i \gets \init(B_i \otimes \ell i)
\mid 1 \leq i \leq k]$.  Note that $(u, \LL) \otimes \ell$~is a
forest obtained from~$u$ by appropriately relabeling its $N$"~labeled
positions.

Now let $t_1, t_2 \in T_{(N \times \nat^*) \cup \Sigma}$ be trees,
$\rho = A \to (u, \LL) \in R$ be a rule, and $\ell \in \nat^*$ be a
link identifier.  We define the \emph{derivation step}~$t_1
\Rightarrow_G^{\rho, \ell} t_2$ if $\alp(A \otimes \ell) =
\alp_{N \times \{\ell\}}(t_1)$ and $t_2 = t_1[A \otimes \ell \gets (u, \LL) \otimes
\ell]$.  Intuitively, $A \otimes \ell$~occurs in~$t_1$ 
(and no other nonterminals with link identifier $\ell$ occur in~$t_1$)
and the occurrence of $A \otimes \ell$ is replaced by~$(u, \LL) \otimes \ell$.  
We write $t_1 \Rightarrow_G t_2$ if there exist $\rho$~and~$\ell$ such that $t_1
\Rightarrow_G^{\rho, \ell} t_2$.

\begin{figure}
  \centering
  \includegraphics{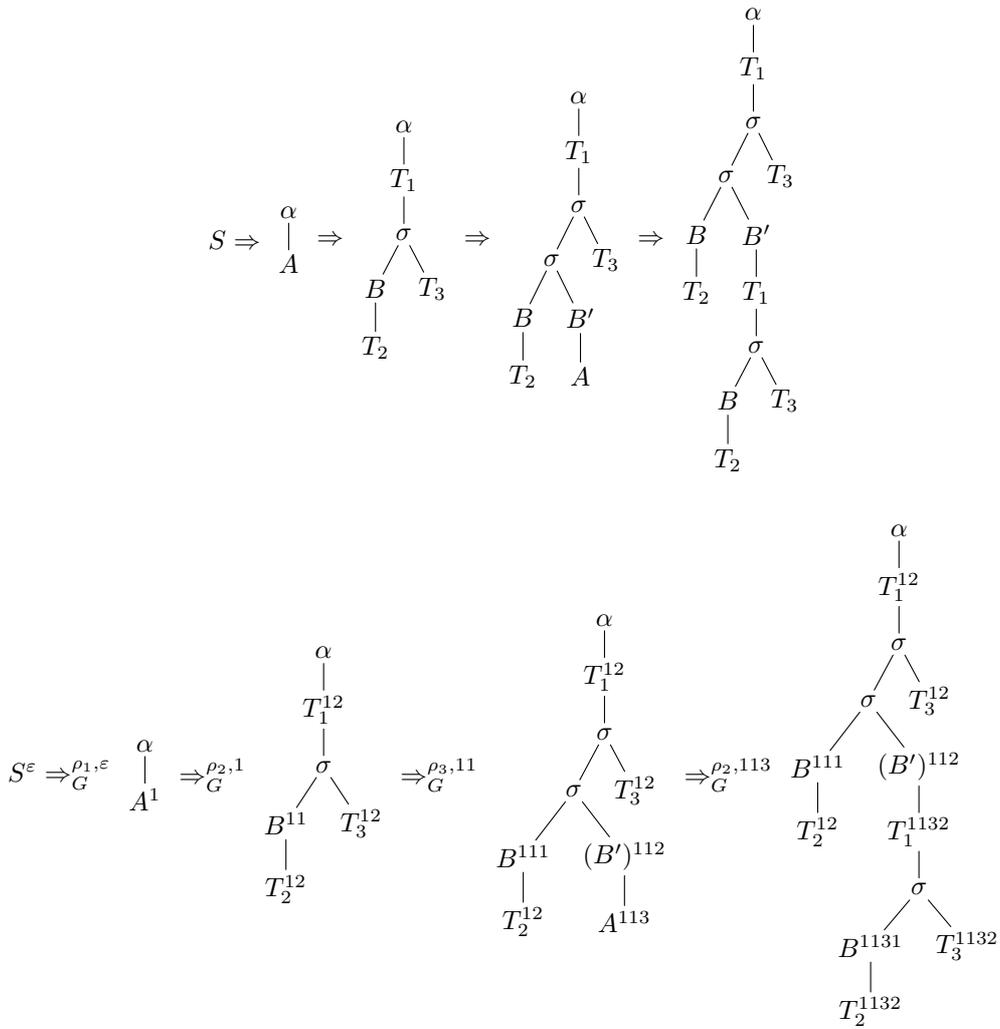}
  \caption{Derivation of the grammar~$G$ of
    Example~\protect{\ref{exa:main}}; naive in the top part and as
    formalized in the bottom part.}
  \label{fig:naivederiv}
\end{figure}

\begin{example}
  \label{exa:derivation}
  \upshape
  Let us consider the derivation tree $d = \rho_1(\rho_2(\rho_3(B, B',
  \rho_2(B, T)), T))$ of the grammar~$G$ of Examples
  \ref{exa:main}~and~\ref{exa:dtree}, where $T = (T_1, T_2, T_3)$.
  Starting with~$S$ and successively applying the rules~$\rho_1$,
  $\rho_2$, $\rho_3$, and~$\rho_2$ according to the naive approach
  yields the derivation presented in the top part of Figure~\ref{fig:naivederiv}.
  It can be checked that the final tree in this derivation
  is~$\val(d)$.  However, now we are in trouble because~$B$, $T_1$,
  $T_2$, and~$T_3$ occur twice.  With the help of the derivation steps
  as defined above and the shorthand~$C^{\ell}$ for~$\langle C,
  \ell\rangle$ with~$C\in N$ and~$\ell \in \nat^*$ we obtain the
  derivation presented in the bottom part of Figure~\ref{fig:naivederiv}.
  In its final tree the occurrences $B^{111}$~and~$B^{1131}$ of~$B$
  can be rewritten independently, and the occurrences of~$T$ are
  distinguished as~$T \otimes 12 = (T_1^{12}, T_2^{12}, T_3^{12})$ and
  $T \otimes 1132 = (T_1^{1132}, T_2^{1132}, T_3^{1132})$ and can be
  rewritten independently by $\rho_5$~or~$\rho_6$.  
  Note that $111=(1,1,1)$ and $1131=(1,1,3,1)$ are the positions of $d$ with label $B$, 
  $12=(1,2)$ and $1132=(1,1,3,2)$ are the positions of $d$ with label $T$,
  and $112=(1,1,2)$ is the unique position of $d$ with label~$B'$. \fin
\end{example}

We wish to prove that $L(G) = \{t \in T_\Sigma \mid S \otimes
\varepsilon \Rightarrow_G^* t\}$; note that~$S \otimes \varepsilon =
\langle S, \varepsilon\rangle$.  To that end, we define an infinite
MCFTG~$G^\infty$ using the properly annotated nonterminals and show
that it is equivalent to~$G$.  An infinite~MCFTG is defined as in
Definition~\ref{def:mcftg} except that~$N$, $\N$, and~$R$ are allowed
to be infinite (and similarly, in an infinite RTG $N$ and $R$ are 
allowed to be infinite). It is easy to check that all the definitions and
results for MCFTGs discussed until now are also valid for infinite
MCFTGs and infinite RTGs.  In particular, the derivation tree
grammar~$G^\infty_\der$ of~$G^\infty$ is infinite.

The infinite MCFTG is given by $G^\infty = (N^\infty, \N^\infty,
\Sigma, S^\infty, R^\infty)$ with nonterminals~$N^\infty = N \times
\nat^*$, big nonterminals $\N^\infty = \N \otimes \nat^* = \{A \otimes
\ell \mid A \in \N,\, \ell \in \nat^*\}$, initial
nonterminal~$S^\infty = S \otimes \varepsilon$, and rules~$R^\infty$
determined as follows.  If $\ell \in \nat^*$~and~$\rho = A \to (u,
\LL) \in R$ with~$\LL = \{\seq B1k\}$, then $R^\infty$~contains the
rule~$\rho \otimes \ell = A \otimes \ell \to ((u,\LL) \otimes \ell,\,
\LL \otimes \ell)$, where $\LL \otimes \ell = \{B_1 \otimes \ell1,
\dotsc, B_k \otimes \ell k\}$.  Note that $\rho$~can be reconstructed
from~$\rho \otimes \ell$.

\begin{lemma}
  \label{lem:infinite}
  $L(G^\infty) = L(G)$. 
\end{lemma}

\begin{proof}
  If $d$~is a derivation tree of~$G^\infty$, then we denote
  by~$\rem(d)$ the derivation tree of~$G$ that is obtained by removing
  all link identifiers~$\ell$ from the labels of its nodes; i.e., a
  label~$\rho \otimes \ell \in R^\infty$ is changed into~$\rho$, and
  $A \otimes \ell \in \N^\infty$~is changed into~$A$.  It is
  straightforward to show by induction on the structure of~$d$ that $d
  \in L(G^\infty_\der, A \otimes \ell)$ implies both~$\rem(d) \in
  L(G_\der, A)$ and $\val(\rem(d)) = \val(d)$.  Indeed, if $d = (\rho
  \otimes \ell)(\seq d1k)$, then $\rem(d) = \rho(\rem(d_1), \dotsc,
  \rem(d_k))$ and  
  \begin{align*}
    \val(d) 
    &= ((u, \LL) \otimes \ell)[B_i \otimes \ell i \gets \val(d_i) \mid
      1 \leq i \leq k] \\
    &= u[B_i \gets \init(B_i \otimes \ell i) \mid 1 \leq i \leq k] \,
      [B_i \otimes \ell i \gets \val(d_i) \mid 1 \leq i \leq k] \\
    &= u \bigl[B_i \gets \init(B_i \otimes \ell i) [B_i \otimes \ell i \gets
      \val(d_i) \mid 1 \leq i \leq k] \mid 1 \leq i \leq k \bigr] =
      u[B_i \gets \val(d_i) \mid 1 \leq i \leq k] \\
    &= u[B_i \gets \val(\rem(d_i)) \mid 1 \leq i \leq k] =
      \val(\rem(d)) \enspace,
  \end{align*}
  where the third equality is by Lemma~\ref{lem:comm-assoc}(4) and the
  fifth by the induction hypotheses.  Taking $A \otimes \ell = S \otimes
  \varepsilon$, we thus obtain that~$L(G^\infty) \subseteq L(G)$ by
  Theorem~\ref{thm:dtree}.  In the other direction, we consider a
  derivation tree~$d \in L(G_\der, S)$, and let~$d'$ be the tree such
  that $\pos(d') = \pos(d)$ and $d'(p) = d(p) \otimes p$ for every $p
  \in \pos(d)$; i.e., we change the label~$d(p)$ of each position~$p$
  into~$d(p) \otimes p$.  Obviously, $d' \in L(G^\infty_\der, S
  \otimes \varepsilon)$~and~$\rem(d') = d$.  Hence, by the above,
  $d'$~has the same value as~$d$, which shows that~$L(G) \subseteq
  L(G^\infty)$.
\end{proof}

\begin{lemma}
  \label{lem:occlink}
Let $d \in \DL(G^\infty_\der, S \otimes \varepsilon)$ and $A \otimes\ell \in \N^\infty$. 
Then $A \otimes\ell\in\alp_{\N^\infty}(d)$ if and only if
$\alp(A \otimes\ell)=\alp_{N\times \{\ell\}}(\val(d))$.
\end{lemma}

\begin{proof}
We first observe that for every position $p\in\pos(d)$
there exists $\alpha\in \N\cup R$ such that $d(p)=\alpha\otimes p$,
cf.\@ the proof of Lemma~\ref{lem:infinite}. 
Thus, if $A \otimes\ell$ occurs in $d$ then it occurs exactly once in $d$
and no $B \otimes\ell$ occurs in $d$ with $B\neq A$.

Let $A \otimes\ell\in\alp_{\N^\infty}(d)$. Then  
$\alp(A \otimes\ell)\subseteq\alp_{N\times \{\ell\}}(\val(d))$ by Lemma~\ref{lem:valocc}(2). 
Moreover, if $\langle C,\ell\rangle\in \alp_{N\times \{\ell\}}(\val(d))$ 
then there exists $B\in \N$ such that $C\in\alp(B)$ and $B \otimes\ell\in\alp_{\N^\infty}(d)$.
From the above observation we obtain that $B=A$ and so $\langle C,\ell\rangle\in \alp(A \otimes\ell)$. 

Now let $\alp(A \otimes\ell)=\alp_{N\times \{\ell\}}(\val(d))$. 
From the inclusion $\alp(A \otimes\ell)\subseteq\alp_{N\times \{\ell\}}(\val(d))$
we obtain, by Lemma~\ref{lem:valocc}(2) and the above observation, 
that there exists $B\in \N$ such that 
$B \otimes\ell\in\alp_{\N^\infty}(d)$ and $\alp(A\otimes\ell)\subseteq\alp(B\otimes\ell)$. 
Hence $\alp(A\otimes\ell)=\alp(B\otimes\ell)$ by the previous paragraph, and so $A=B$ by 
the second item of Definition~\ref{def:mcftg}.
\end{proof}

\begin{theorem}
  \label{thm:deriv}
  $L(G) = \{t \in T_\Sigma \mid S \otimes \varepsilon \Rightarrow_G^*
  t\}$.
\end{theorem}

\begin{proof}
  By Lemma~\ref{lem:infinite}, Theorem~\ref{thm:dtree} and
  Lemma~\ref{lem:valocc}(3), it suffices to prove the following claim:
  \begin{quote}
    For every~$t \in T_{(N \times \nat^*) \cup \Sigma}$ we have $S
    \otimes \varepsilon \Rightarrow_G^* t$~if and only if there 
    exists~$d \in \DL(G^\infty_\der, S \otimes \varepsilon)$ such
    that~$\val(d) = t$.
  \end{quote}
  (If)~The proof is by induction on the length~$n$ of a derivation~$S
  \otimes \varepsilon \Rightarrow^n_{G^\infty_\der} d$ required for~$d
  \in \DL(G^\infty_\der, S \otimes \varepsilon)$.  The claim is
  obvious for~$n = 0$; i.e., for $d = S \otimes \varepsilon$.
  Otherwise, we consider the last step of the derivation $S \otimes
  \varepsilon \Rightarrow^{n-1}_{G^\infty_\der} d'
  \Rightarrow_{G^\infty_\der} d$, and let $A \otimes \ell \to (\rho
  \otimes \ell)(B_1 \otimes \ell1, \dotsc, B_k \otimes \ell k)$ be the
  rule of~$G^\infty_\der$ that was applied in the last step, where
  $\rho = A \to (u, \LL)$ with $\LL = \{\seq B1k\}$ is the
  corresponding rule of~$G$.  Clearly,
  since $A \otimes \ell$~occurs exactly
  once in~$d'$ (as observed in the proof of Lemma~\ref{lem:occlink}),
  \[d = d'[A \otimes \ell \gets (\rho \otimes \ell)(B_1 \otimes \ell 1, \dotsc,
  B_k \otimes \ell k)] \enspace .\] 
  Since $\val((\rho \otimes \ell)(B_1 \otimes
  \ell 1, \dotsc, B_k \otimes \ell k)) = (u, \LL) \otimes \ell$,
  we obtain $\val(d) = \val(d')[A \otimes
  \ell \gets (u, \LL) \otimes \ell]$ from Lemma~\ref{lem:valsub}. 
  Hence $S \otimes
  \varepsilon \Rightarrow_G^* \val(d') \Rightarrow^{\rho, \ell}_G \val(d)$
  by the induction hypothesis, Lemma~\ref{lem:occlink} and
  the definition of~$\Rightarrow^{\rho, \ell}_G$.

  (Only if) The proof is by induction on the length~$n$ of a
  derivation~$S \otimes \varepsilon \Rightarrow_G^n t$.  It is again
  obvious for~$n = 0$.  Otherwise, we consider the last step of the 
  derivation~$S \otimes \varepsilon \Rightarrow_G^{n-1} t'
  \Rightarrow_G t$.  By the induction hypothesis there exists~$d' \in
  \DL(G^\infty_\der, S \otimes \varepsilon)$ such that~$\val(d') =
  t'$.  Moreover, by the definition of~$\Rightarrow_G$, there exist a
  rule~$\rho = A \to (u, \LL) \in R$ and a link identifier~$\ell$ such
  that $\alp(A \otimes \ell) = \alp_{N \times \{\ell\}}(t')$ and $t = t'[A \otimes
  \ell \gets (u, \LL) \otimes \ell]$.  
  Then $A \otimes\ell$ occurs in $d'$ by Lemma~\ref{lem:occlink}. 
  Defining $d$ as displayed above, we obtain from Lemma~\ref{lem:valsub} that
  $\val(d) = \val(d')[A \otimes \ell \gets (u, \LL) \otimes \ell]$; 
  i.e., $\val(d) = t$.
\end{proof}

In exactly the same way it can be proved that $L(G, A) = \{t \in
P_\Sigma(X)^{\scriptscriptstyle +} \mid \init(A \otimes \varepsilon)
\Rightarrow_G^* t\}$ for every~$A \in \N$, after extending the notion
of derivation step to forests in~$P_{(N \times \nat^*) \cup
  \Sigma}(X)^{\scriptscriptstyle +}$.  We finally mention that 
it is straightforward to prove that for every $t \in
T_{(N \times \nat^*) \cup \Sigma}$, if~$S \otimes \varepsilon
\Rightarrow_G^* t$, then (1)~$t$~is uniquely $(N \times \nat^*)$"~labeled 
and (2)~there is a unique finite subset~$\LL$ of~$\N \otimes \nat^*$ such that
the set $\{\alp(B)\mid B \in \LL\}$ 
is equal to the set $\{\alp_{N\times \{\ell\}}(t)\neq\emptyset \mid \ell\in\nat^*\}$. 
Thus, $\LL$ is the set of big nonterminals (of $G^\infty$) that can be rewritten in $t$. 
For instance, for the last tree of Figure~\ref{fig:naivederiv}
we have $\LL = \{B \otimes 111, B \otimes 1131, B' \otimes 112, T \otimes 12, T \otimes 1132\}$.

\section{Normal forms}
\label{sec:norm}
\noindent
In this section, we establish a number of normal forms for MCFTGs.  We
start in Section~\ref{sub:basicnf} with some basic normal forms.
In Section~\ref{sub:lexnf} we define the notions of 
finite ambiguity and lexicalization, and then we prove a Growing Normal Form 
that is already part of our lexicalization procedure.
Along the way we show the decidability of finite ambiguity.
Finally we establish one additional basic normal form.  
From now on, let $G = (N, \N, \Sigma, S, R)$~be the considered MCFTG.

\subsection{Basic normal forms}
\label{sub:basicnf}
\noindent
The MCFTG~$G$ is \emph{start-separated} if~$\pos_S(u) = \emptyset$ for
every rule~$A \to (u, \LL) \in R$.  In other words, the initial
nonterminal~$S$ is not allowed in the right-hand sides of the rules.
It is clear that $G$~can be transformed into an \LDTR"~equivalent
start-separated MCFTG~$G'$.  We simply take a new initial
nonterminal~$S'$, all original rules, and for every rule~$\rho = S \to
(u, \LL) \in R$ we add the rule~$\rho' = S' \to (u, \LL)$.  Then we
obviously have that~$L(G'_\der, S') = \{\rho'(\seq d1k) \mid \rho(\seq
d1k)\in L(G_\der, S)\}$, and there exist LDT"~transducers that
change~$\rho(\seq d1k)$ into~$\rho'(\seq d1k)$ and vice versa.  The
MCFTGs of Examples~\ref{exa:copy} and~\ref{exa:main} are
start-separated.

\paragraph{Convention} From now on, we assume, without loss of
generality (by Proposition~\ref{pro:LDTRcomp}), and without mentioning
it, that every MCFTG is start-separated.  Each rule of the form~$S \to
(u, \LL)$ is called an \emph{initial} rule.  We call a rule~$A \to (u,
\LL)$ \emph{terminal} if~$u \in P_\Sigma(X)^{\scriptscriptstyle +}$;
i.e., $u$~does not contain nonterminal symbols or equivalently~$\LL =
\emptyset$.  Such a rule will also be written~$A \to u$.  Note that a
rule may be both initial and terminal.  A~rule is called
\emph{proper} if it is not both initial and terminal.

\medskip
The MCFTG~$G$ is \emph{reduced} if every big nonterminal~$A \in \N
\setminus \{S\}$ is reachable and useful.  A big nonterminal~$A \in
\N$ is \emph{reachable} if~$S \hookrightarrow^*_G A$, where for
all~$B, B' \in \N$ we define~$B \hookrightarrow_G B'$ if there is a
rule~$B \to (u, \LL) \in R$ such that~$B' \in \LL$.  Moreover, $A$~is 
\emph{useful} if $L(G, A) \neq \emptyset$.  Clearly, $G$~is reduced if
and only if the RTG~$G_\der$ is reduced (in the usual, analogous
sense);  this is obvious for reachability and follows from
Theorem~\ref{thm:dtree} for usefulness.  As in the case of
context-free grammars, we may and will always assume that a given
MCFTG~$G$ is reduced, which can be achieved by removing all
nonreachable and useless big nonterminals together with the rules in
which they occur.  Since this is the same procedure for~$G_\der$, we
have that $L(G'_\der) = L(G_\der)$ for the resulting grammar~$G'$,
and hence, trivially, $G'$~is \LDTR"~equivalent to~$G$.  The MCFTGs of
Examples~\ref{exa:copy} and~\ref{exa:main} are reduced.

Let $G' = (N', \N', \Sigma, S', R')$~be another MCFTG.  We say that
$G'$~is a \emph{renaming} of~$G$ if there exists a rank-preserving
bijection~$\beta \colon \N \to \N'$ such that $S' = \beta(S)$~and~$R'
= \{\rho_\beta\mid \rho\in R\}$, where for every rule~$\rho = A \to (u, \LL) \in R$ we
let~$\rho_\beta = \beta(A) \to (u[B \gets \init(\beta(B)) \mid B \in
\LL],\, \beta(\LL))$, where $\beta(\LL)=\{\beta(B_1),\dotsc,\beta(B_k)\}$ if
$\LL=\{\seq B1k\}$.  Note that $\rho$ can easily be reconstructed from $\rho_\beta$ 
(by applying $\beta^{-1}$); i.e., the mapping $\rho\mapsto \rho_\beta$ is also a
bijection, from $R$ to $R'$. 

\begin{lemma}
  \label{lem:renaming}
  For all MCFTGs $G$~and~$G'$, if $G'$~is a renaming of~$G$, then
  $G$~and~$G'$ are \LDTR"~equivalent.
\end{lemma}

\begin{proof}
  Let $\beta$ be the required bijection.  For every
  tree~$d \in T_R$, let $M(d)$~be obtained from~$d$ by  
  changing every label~$\rho$ into~$\rho_\beta$.  In this manner we
  obtain a bijection~$M \colon T_R \to T_{R'}$.  Obviously,
  $d \in L(G_\der, A)$ if and only if $M(d) \in L(G'_\der, \beta(A))$.
  Additionally, we can easily show that~$\val(M(d)) = \val(d)$ by
  induction on the structure of~$d$.  Indeed, let $d = \rho(\seq d1k)$
  for a rule $\rho = A \to (u, \LL) \in R$ with $\LL = \{\seq
  B1k\}$ and $d_i \in L(G_\der, B_i)$ for every $i \in [k]$.  We 
  have~$\val(M(d_i)) = \val(d_i)$ for every~$i \in [k]$ by the
  induction hypotheses.  Clearly, $M(d) =
  \rho_\beta(M(d_1), \dotsc, M(d_k))$, and hence
  $\val(M(d)) = u[B_i \gets \init(\beta(B_i)) \mid 1 \leq i \leq
  k]\,[f]$, where $f$~is the 
  substitution function for~$\beta(\LL)$ such that $f(\beta(B_i)) =
  \val(M(d_i)) = \val(d_i)$ for every $i\in[k]$.  It now follows from
  Lemma~\ref{lem:comm-assoc}(4) that $\val(M(d)) = u[B_i \gets 
  \init(\beta(B_i) [f]) \mid 1 \leq i \leq k]$, which equals~$u[B_i
  \gets \val(d_i) \mid 1 \leq i \leq k] = \val(d)$.  The
  transformation~$M \colon T_R \to T_{R'}$ as well as its
  inverse $M^{-1} \colon T_{R'} \to T_R$ are tree homomorphisms 
  (even projections), and every tree
  homomorphism can be realized by an \LDTR"~transducer, which shows
  the \LDTR"~equivalence.
\end{proof}

The previous lemma shows that the actual identity of nonterminals
constituting a big nonterminal is irrelevant in MCFTGs.  We say that
the MCFTG~$G$ \emph{has disjoint big nonterminals}  if~$\alp(A) 
\cap \alp(A') = \emptyset$ for all distinct~$A, A' \in \N$.  The MCFTGs
of Examples \ref{exa:copy}~and~\ref{exa:main} indeed have disjoint big
nonterminals.  Clearly, every MCFTG~$G$ has a renaming that has
disjoint big nonterminals.  Consequently, we may always assume that a
given MCFTG~$G$ has disjoint big nonterminals.  As observed before
Example~\ref{exa:copy}, the specification of the set of links of a
rule is then no longer necessary.  Indeed we could have required
disjoint big nonterminals in Definition~\ref{def:mcftg}, but
this would have been technically inconvenient, as we will see, e.g., 
in the proof of Lemma~\ref{lem:permfree}.  

We say that the MCFTG~$G$ is \emph{free-choice} if the following
holds.  For every rule~$A \to (u, \LL) \in R$ and every~$\LL'
\subseteq \N$ that satisfies the requirement in the last item of
Definition~\ref{def:mcftg}, we require that $A \to (u, \LL')$~is also 
a rule of~$G$.  This means that the rules of~$G$ can be specified
as~$A \to u$, which stands for all possible rules~$A\to (u, \LL)$.
Obviously, if $G$~has disjoint big nonterminals, then it is
free-choice because the links are uniquely determined by
$\N$~and~$u$.  Thus, we may always assume that a given MCFTG is
free-choice.  Free-choice MCFTGs with the derivation semantics of
Section~\ref{sub:deriv} generalize the local unordered scattered
context grammars~(LUSCGs) of~\cite{ramsat99},  which are an equivalent
formulation of multiple context-free (string) grammars.

The next easy result is not a normal form result in the usual sense of
the word, but shows that the class~MCFT is closed under (simple) tree
homomorphisms; for much stronger closure properties of~MCFT we refer
to Section~\ref{sec:charact}.  Nevertheless, a special case of this
result can be used in proofs to assume that the right-hand sides of a
given MCFTG~$G$ are not only uniquely $N$"~labeled  but also uniquely
$\Sigma$"~labeled.

Let $h$~be a tree homomorphism from~$\Sigma$ to~$\Sigma'$ 
where $\Sigma'$~is a finite ranked alphabet disjoint to~$N$.  We define
the MCFTG~$G_h = (N, \N, \Sigma', S, R')$ such that 
\[ R' = \{A \to (\hat{h}(u), \LL) \mid A \to (u, \LL) \in R\}
\enspace, \] 
where $h$ is extended to a tree homomorphism from~$N\cup\Sigma$ to~$N\cup\Sigma'$
by defining $h(C)=\init(C)$ for every $C\in N$. 
We refer to Definition~\ref{def:LDTReq} for the notion of 
\LDTR "~$\hat{h}$"~equivalence. 

\newpage
\begin{lemma}
  \label{lem:cover}
  For every MCFTG~$G$ and every tree homomorphism~$h$ (as above), the
  MCFTG~$G_h$ (as defined above) is \LDTR "~$\hat{h}$"~equivalent to~$G$.
  Hence $L(G_h)=\hat{h}(L(G))$.
\end{lemma}

\begin{proof}
  The proof is similar to the one of Lemma~\ref{lem:renaming}.
  Let~$G' = G_h = (N, \N, \Sigma', S, R')$.  For every rule~$\rho = A
  \to (u, \LL) \in R$, let $\rho_h$~be the rule~$A \to (\hat{h}(u), \LL) \in
  R'$, in which the links of $\LL$ have the same order as in $\rho$.  
  For every tree~$d \in T_R$, let
  $M(d)$~be obtained from~$d$ by changing every label~$\rho$
  into~$\rho_h$.  This defines a surjection $M\colon T_R\to T_{R'}$. 
  Obviously, $d \in L(G_\der, A)$ if and only if
  $M(d) \in L(G'_\der, A)$ for every~$A \in \N$.  We now show, by
  induction on the structure of~$d$,  that~$\val(M(d)) =
  \hat{h}(\val(d))$.  Indeed, let $d = \rho(\seq d1k)$
  with~$\rho = A \to (u, \LL)$ and~$\LL = \{\seq B1k\}$, and by the
  induction hypotheses $\val(M(d_i)) =
  \hat{h}(\val(d_i))$ for every~$i \in [k]$.  Then $M(d) =
  \rho_h(M(d_1), \dotsc, M(d_k))$, and hence we have
  \begin{align*}
    \val(M(d)) 
    &= \hat{h}(u)[B_i
      \gets \hat{h}(\val(d_i)) \mid 1 \leq i \leq k] \\
    &= \hat{h}(u[B_i \gets \val(d_i) \mid 1 \leq i \leq k]) =
      \hat{h}(\val(d)) \enspace,
  \end{align*}   
  where the second equality is by Lemma~\ref{lem:comm-assoc}(3)
  applied to $\sigma_1 = \word B1k$~and~$\alp(\sigma_2) = \Sigma$.
  This shows that $\hat{h}(L(G))\subseteq L(G')$. 
  
For every rule $\rho'\in R'$, let $\rho'_h$ be a fixed rule 
$\rho\in R$ such that $\rho_h = \rho'$.
  For every tree~$d' \in T_{R'}$, let
  $M'(d')$~be obtained from~$d'$ by changing every label~$\rho'$
  into~$\rho'_h$. This defines a mapping $M'\colon T_{R'}\to T_R$.
  Obviously $M(M'(d'))=d'$ and hence, by the above, if $d' \in L(G'_\der, A)$
  then $M'(d') \in L(G_\der, A)$ and $\val(M'(d'))=\val(d')$. 
  This shows that $L(G')\subseteq \hat{h}(L(G))$. 
  
  The transformations~$M$ and~$M'$
  can be realized by projections, and hence by \LDTR"~transducers.
\end{proof}

We say that the pair~$(G,h)$ is a \emph{cover} of the MCFTG~$G_h$
if $h$~is a projection; i.e., for every~$\sigma \in \Sigma$ there
exists~$\sigma' \in \Sigma'$ such that~$h(\sigma) = \init(\sigma')$.
We define the MCFTG~$G$ to be \emph{uniquely terminal labeled} if for
every rule~$\rho \in R$:
\begin{compactenum}[\indent\upshape(1)]
\item the right-hand side~$\rhs(\rho)$ is uniquely $\Sigma$"~labeled,
  and
\item $\alp_\Sigma(\rhs(\rho)) \cap \alp_\Sigma(\rhs(\rho')) =
  \emptyset$ for every other rule~$\rho' \in R$.
\end{compactenum}
Clearly, every MCFTG~$G$ has a cover~$(G_{\text{u}}, h)$ such that
$G_{\text{u}}$~is uniquely terminal labeled.  Although the tree
languages $L(G) = \hat h(L(G_{\text{u}}))$~and~$L(G_{\text{u}})$
differ in general, this may be viewed as a normal form of~$G$.

The last basic normal form that we consider in this subsection is permutation-freeness.  
Let $\Omega$ be a ranked alphabet (such as $N\cup\Sigma$). For a tree~$t \in
T_\Omega(X)$ the string~$\yield_X(t) \in X^*$ is the sequence of
occurrences of variables in~$t$, from left to right.\footnote{The yield of $t$
with respect to $X$ is defined in the paragraph on homomorphisms in Section~\protect{\ref{sub:seqs}}.}
Clearly, if $t \in P_\Omega(X_k)$, then $\yield_X(t)$ is a permutation 
$\word x{i_1}{i_k}$ of $\word x1k$.
We say that a pattern~$t \in P_\Omega(X)$ is \emph{permutation-free}
if $\yield_X(t) = \word x1k$ for~$k = \rk(t)$, and we denote the set
of permutation-free patterns over~$\Omega$ by~$\PF_\Omega(X)$.  For~$t
\in P_\Omega(X)$ we define~$\pf(t) \in \PF_\Omega(X)$ as follows:  
if $\yield_X(t) = \word x{i_1}{i_k}$, then 
$\pf(t)$ is the unique permutation-free pattern such that 
$t = \pf(t)[x_1\gets x_{i_1},\dotsc,x_k\gets x_{i_k}]$.
For a forest $t=(\seq t 1n)$ we define 
$\yield_X^*(t) = (\yield_X(t_1),\dotsc,\yield_X(t_n))$ and
$\pf^*(t) = (\pf(t_1), \dotsc, \pf(t_n))$. 
We say that a tree homomorphism~$h$
over~$\Omega$ is permutation-free if $h(\omega)$~is permutation-free
for every~$\omega \in \Omega$.  We observe that, for such a tree
homomorphism,  $\yield_X(\hat{h}(t)) = \yield_X(t)$ for every~$t \in
T_\Omega(X)$, as can easily be shown by induction on the structure
of~$t$, and $\hat{h}(\pf(t)) = \pf(\hat{h}(t))$ for every~$t\in
P_\Omega(X)$ by Lemma~\ref{lem:treehomsub}.

The MCFTG~$G$ is \emph{permutation-free} if $\rhs(\rho) \in \PF_{N
  \cup \Sigma}(X)^{\scriptscriptstyle +}$ for every rule~$\rho \in R$.
Intuitively, permutation-free MCFTGs are easier to understand than
arbitrary MCFTGs because the application of a rule to a node of a tree
does not involve a permutation of the subtrees at the children of that
node; thus, a rule application does not affect the global structure 
of the tree.  The MCFTG~$G$ of Example~\ref{exa:main} is trivially 
permutation-free because every nonterminal of~$G$ has rank $0$~or~$1$.  

\begin{lemma}
  \label{lem:permfree}
  For every MCFTG~$G$ there is an \LDTR"~equivalent MCFTG~$G'$
  that is permutation-free.  Moreover, $\wid(G') =
  \wid(G)$, $\mu(G') = \mu(G)$, and $\lambda(G')=\lambda(G)$.
\end{lemma}

\begin{proof}
  We construct the grammar~$G' = (N', \N', \Sigma, S', R')$, in which
  $S' = \langle S, \varepsilon\rangle$ and $N'$~is the set of all
  pairs~$\langle C, \pi\rangle$ such that~$C \in N$ and $\pi$~is a
  permutation of~$\word x1{\rk(C)}$.  The rank of~$\langle C,
  \pi\rangle$ is the same as the rank of~$C$.
  The set of big nonterminals~$\N'$ consists of
  all~$(\langle A_1, \pi_1\rangle, \dotsc, \langle A_n, \pi_n\rangle)$
  with $(\seq A1n) \in \N$~and~$\langle A_i, \pi_i \rangle \in N'$ for
  every $i \in [n]$.\footnote{Note that if $G$ has disjoint big nonterminals, 
  then that is in general not the case for $G'$. Thus, this property of an MCFTG $G$
  is not preserved when information is added to the nonterminals of $G$, 
  which is the reason that we did not require it in Definition~\protect{\ref{def:mcftg}}.}
  A big nonterminal~$A' = (\langle
  A_1, \pi_1\rangle, \dotsc, \langle A_n, \pi_n\rangle)$ 
  will also be denoted by $\pair(A,\pi)$, where $A=(\seq A1n)$ and $\pi=(\seq \pi 1n)$,
  and we define $\rem(A') = A= (\seq A1n)$.
  Intuitively, if~$A$ generates~$t = (\seq t1n)$ with $t_i
  \in P_\Sigma(X_{\rk(A_i)})$ and 
  $\yield_X^*(t) = (\yield_X(t_1),\dotsc,\yield_X(t_n)) = (\seq \pi 1n)$,
  then $A'$~generates $\pf^*(t) = (\pf(t_1), \dotsc,
  \pf(t_n))$. To define the rules of~$G'$ we need the (permuting)
  tree homomorphism~$h$ over~$N' \cup \Sigma$ that is defined by $h(\langle C,
  \pi\rangle) = \langle C, \pi \rangle \pi$ for every~$\langle C,
  \pi\rangle \in N'$ and $h(\sigma) = \init(\sigma)$ for
  every~$\sigma \in \Sigma$.  For example, if~$\pi = x_3x_2x_1x_4$,
  then $h(\langle C, \pi\rangle) = \langle C, \pi\rangle(x_3, x_2,
  x_1, x_4)$; in other words, $h$~permutes the subtrees of~$\langle C,
  \pi\rangle$ according to the permutation~$\pi$.  

  Let $\rho = A \to
  (u, \LL)$ be a rule of~$G$ with $\LL = \{\seq B1k\}$.  Moreover, let
  $\seq{B'}1k$~be big nonterminals in~$\N'$ such that $\rem(B'_i) =
  B_i$ for every~$i \in [k]$, and let $u' = u[B_i \gets \init(B'_i)
  \mid 1 \leq i \leq k]$ and $\overline{\pi} = \yield_X^*(\hat{h}(u'))$.
  Then $R'$~contains the rule  
  \[ \rho_{\word{B'}1k} = \pair(A, \overline{\pi}) \to
  (\pf^*(\hat{h}(u')),\, \{\seq{B'}1k\}) \enspace. \]
  Note that this rule satisfies the requirements of
  Definition~\ref{def:mcftg} by Lemma~\ref{lem:treehom}.  
  Note also that $\rho$ can be reconstructed from $\rho_{\word{B'}1k}$. 
  This completes the construction of~$G'$. 

  To show that~$L(G) \subseteq L(G')$ we prove that for every $A \in
  \N$ and every derivation tree~$d \in L(G_\der, A)$ there exists a
  derivation tree~$d' \in L(G'_\der, \pair(A, \pi))$ such
  that $\pi = \yield_X^*(\val(d))$~and~$\val(d') = \pf^*(\val(d))$.  For
  every derivation tree~$d \in \bigcup_{B \in \N} L(G_\der, B)$, we
  let $\bign(d) = \pair(A, \yield_X^*(\val(d)))$, where $A$~is
  the type of $d$.  The proof is by induction on the structure of~$d$.  
  Simultaneously we prove that $\bign(d)$~can be defined inductively.  Let $d =
  \rho(\seq d1k)$, where $\rho$~is as shown above.  By the induction
  hypotheses, let $B'_i = \bign(d_i)=\pair(B_i,\pi_i)$ 
  such that $\pi_i = \yield_X^*(\val(d_i))$, 
  and let $d'_i \in L(G'_\der,
  B'_i)$ be such that~$\val(d'_i) = \pf^*(\val(d_i))$, for every~$i \in
  [k]$. We define~$\bign(d)$ to be the left-hand side of the
  rule~$\rho_{\word{B'}1k}$.  Moreover, we take $d' =
  \rho_{\word{B'}1k}(\seq{d'}1k)$.  Additionally, let
  $[g'_{\pf}]$~abbreviate the (simultaneous) second-order substitution
  $[B'_i \gets \pf^*(\val(d_i)) \mid 1 \leq i \leq k]$, and let
  $[g']$~and~$[g]$ abbreviate the second-order substitutions $[B'_i
  \gets \val(d_i) \mid 1 \leq i \leq k]$~and~$[B_i \gets \val(d_i)
  \mid 1 \leq i \leq k]$. Then the definition of~`$\val$' gives
  $\val(d') = \pf^*(\hat{h}(u'))[B'_i \gets \val(d'_i) \mid 1 \leq i
  \leq k] = \pf^*(\hat{h}(u'))[g'_{\pf}] = \pf^*(\hat{h}(u')[g'_{\pf}])$,
  where the last equality holds because the permutation-free tree
  homomorphism~$g'_{\pf}$ corresponding to the
  substitution~$[g'_{\pf}]$ commutes with~`$\pf$' as observed before
  this lemma.  We now show that
  \[ \hat{h}(u')[g'_{\pf}]=u'[g']=u[g]=\val(d) \enspace. \]
  The first equality holds by Lemma~\ref{lem:treehomcomp} because the
  composition of the tree homomorphisms $h$~and~$\hat{g}'_{\pf}$ is equal to
  the tree homomorphism~$g'$ corresponding to the
  substitution~$[g']$ for every symbol in $\alp_{N'\cup\Sigma}(u')$, as shown next.  
  In fact, let $B'_i = \beta \langle C, \pi
  \rangle \gamma$ with $\langle C, \pi \rangle \in N'$~and~$\beta,
  \gamma \in (N')^*$, and let $\val(d_i) = \varphi t \psi$ with~$t \in
  P_\Sigma(X_{\rk(C)})$, $\varphi, \psi \in P_\Sigma(X)^*$,
  and~$\abs{\beta} = \abs{\varphi}$.  From $\pi_i =
  \yield_X^*(\val(d_i))$, we obtain that $\pi = \yield_X(t)$. 
  Now we have~$g'_{\pf}(\langle C,
  \pi\rangle) = \pf(t)$ and therefore $\hat{g}'_{\pf}(h(\langle C, \pi
  \rangle)) = \hat{g}'_{\pf}(\langle C, \pi\rangle \pi) = t =
  g'(\langle C, \pi \rangle)$.\footnote{To be precise, if
    $\yield_X(t) = \pi = \word x{i_1}{i_m}$, then
    $\hat{g}'_{\pf}(\langle C, \pi \rangle \pi) = \pf(t)[x_1 \gets
    x_{i_1}, \dotsc, x_m \gets x_{i_m}] = t$.}
  The second equality follows easily from
  Lemma~\ref{lem:comm-assoc}(4), and  the last equality is again by
  the definition of~`$\val$'.  Hence, we have shown that $\val(d') =
  \pf^*(\val(d))$, and it remains to show that the
  permutation~$\overline{\pi}$ in the left-hand side
  of~$\rho_{\word{B'}1k}$ fulfills~$\overline{\pi} =
  \yield_X^*(\val(d))$.  By the calculation above, $\yield_X^*(\val(d)) =
  \yield_X^*(\hat{h}(u')[g'_{\pf}])$.  In addition,  $\overline{\pi} =
  \yield_X^*(\hat{h}(u'))$ by the definition of~$\rho_{\word{B'}1k}$.
  Since $g'_{\pf}$~is permutation-free, these values are the same, as
  observed before this lemma. This proves that~$L(G) \subseteq
  L(G')$. 

  It is easy to see that the above transformation from~$d$ to~$d'$ can
  be realized by an \LDTR"~transducer~$M$ with one state~$q$.  In
  fact, it should be clear from the inductive definition of~$\bign(d)$
  that the set $L_{A'} = \{d \in \bigcup_{B \in \N} L(G_\der, B) \mid
  \bign(d) = A'\}$ is a regular tree language for every~$A' \in \N'$.
  Then, for the above rule~$\rho$, the transducer~$M$ has all the
  rules
  \[ \langle q,\, \rho(y_1 \colon L_{B'_1}, \dotsc, y_k \colon
  L_{B'_k}) \rangle \to \rho_{\word{B'}1k}(\langle q, y_1 \rangle,
  \dotsc, \langle q, y_k \rangle) \enspace. \]  
  Note that $M$ is a finite-state relabeling. 

  To show that~$L(G') \subseteq L(G)$, we observe that for every
  derivation tree~$d' \in L(G'_\der)$ the derivation tree~$d \in
  L(G_\der)$, which is obtained from~$d'$ by changing every 
  label~$\rho_{\word{B'}1k}$ into~$\rho$, satisfies~$M(d) = d'$ and
  hence~$\val(d) = \val(d')$.  Since this transformation from~$d'$
  to~$d$ is a projection, it can be realized by an LDT"~transducer.
\end{proof}

\subsection{Lexical normal forms}
\label{sub:lexnf}
\noindent
We first recall the notion of finite ambiguity
from~\cite{sch90,jossch92,kuhsat12}.\footnote{It should not be
  confused with the notion of finite ambiguity
  of~\protect{\cite{golleuwot92,klilommaipri04}}.}  
We distinguish a subset~$\Delta \subseteq \Sigma$ of \emph{lexical}
symbols, which are the symbols that are preserved by the lexical 
yield mapping. The \emph{lexical yield} of a tree $t \in T_\Sigma$
is the string $\yield_\Delta(t)\in\Delta^*$, as defined in Section~\ref{sub:seqs}. 
It is the string of occurrences of lexical symbols in $t$, from left to right;
all other symbols are simply dropped. 

\begin{definition}
  \label{df:limited}
  \upshape
  The tree language~$L \subseteq T_\Sigma$ has \emph{finite
    $\Delta$"~ambiguity} if $\{t \in L \mid \yield_\Delta(t) = w\}$ is
  finite for every~$w \in \Delta^*$.  The MCFTG~$G$ has finite
  $\Delta$"~ambiguity if $L(G)$~has finite $\Delta$"~ambiguity. \fin
\end{definition}

Roughly speaking, we can say that the language~$L$ has finite
$\Delta$"~ambiguity if each~$w \in \Delta^*$ has finitely many
syntactic trees in~$L$, where $t$~is a syntactic tree of~$w$
if $w$ is its lexical yield. 
Note that $\abs{\yield_\Delta(t)} =
\abs{\pos_\Delta(t)}$; thus, $L$~has finite $\Delta$"~ambiguity if and
only if $\{t \in L \mid \abs{\pos_\Delta(t)} = n\}$ is finite for
every~$n \in \nat_0$.  Note also that if $\Sigma^{(0)} \cup \Sigma^{(1)}
\subseteq \Delta$~or~$\Sigma \setminus \Sigma^{(0)} \subseteq \Delta$,
then every tree language~$L \subseteq T_\Sigma$ has finite
$\Delta$"~ambiguity.

\begin{example}
  \label{exa:finamb}
  \upshape
  For the MCFTG~$G$ of Example~\ref{exa:main} we consider the
  set~$\Delta = \Sigma \setminus \{\sigma, \gamma\} = \{\alpha, \beta,
  \tau, \nu\}$ of lexical symbols.  It should be clear from
  Example~\ref{exa:main} that in each tree of~$L(G)$ the number of
  occurrences of~$\gamma$ coincides with the number of occurrences
  of~$\beta$.  Since $\Delta \cup \{\gamma\} = \Sigma^{(0)} \cup
  \Sigma^{(1)}$, this implies that $L(G)$~as well as~$G$ have finite
  $\Delta$"~ambiguity.  Similarly, the number of occurrences of~$\nu$
  in a tree of~$L(G)$ coincides with the number of occurrences
  of~$\tau$, and the number of occurrences of~$\beta$ is half the
  number of occurrences of~$\alpha$.  Hence $G$~also has finite
  $\{\alpha, \tau\}$"~ambiguity, but for convenience we will continue
  to use the lexical symbols~$\Delta$ in examples. \fin
\end{example}

In this contribution, we want to lexicalize~MCFTGs, which means that
for each MCFTG~$G$ that has finite $\Delta$"~ambiguity, we want to
construct an equivalent MCFTG~$G'$ such that each proper
rule\footnote{Recall from the beginning of
  Section~\protect{\ref{sub:basicnf}} that a rule is proper if it is
  not both initial and terminal.} contains at least one lexical
symbol.  Let us formalize our lexicalization property.

\begin{definition}
  \label{def:lexicalized}
  \upshape
  The forest~$t$ is \emph{$\Delta$"~lexicalized} if $\pos_\Delta(t)
  \neq \emptyset$.  The rule~$A\to (u, \LL)$ is $\Delta$"~lexicalized
  if $u$~is $\Delta$"~lexicalized.  The MCFTG~$G$ is
  $\Delta$"~lexicalized if all its proper rules are
  $\Delta$"~lexicalized.  A~forest or rule is \emph{$\Delta$"~free} if
  it is not $\Delta$"~lexicalized.  The rule~$A \to (u, \LL)$ is
  \emph{doubly} $\Delta$"~lexicalized if $\abs{\pos_\Delta(u)} \geq
  2$, and it is \emph{singly} $\Delta$"~lexicalized if
  $\abs{\pos_\Delta(u)} = 1$. \fin
\end{definition}

Clearly, for every derivation tree~$d$, the value~$\val(d)$ is
$\Delta$"~free if and only if all rules that occur in~$d$ are
$\Delta$"~free by Lemma~\ref{lem:valocc}(1).  For the grammar~$G$ of
Example~\ref{exa:main} with $\Delta = \{\alpha, \beta, \tau, \nu\}$ as
in Example~\ref{exa:finamb},  the rules
\[ \rho_1 = S \to \alpha(A) \quad \rho_5 = (T_1(x_1), T_2, T_3) \to
(\alpha(T_1(\beta(x_1))), \alpha(T_2), \gamma(T_3)) \quad \rho_6 =
(T_1(x_1), T_2, T_3) \to (x_1, \tau, \nu) \]
are $\Delta$"~lexicalized ($\rho_1$~singly and both
$\rho_5$~and~$\rho_6$ doubly), whereas rule~$\rho_4 = B(x_1) \to x_1$
is not even $\Sigma$"~lexicalized.  

Thus, for each MCFTG~$G$ that has finite $\Delta$"~ambiguity, we want to
construct an equivalent MCFTG~$G'$ that is $\Delta$"~lexicalized.
This notion of lexicalization is also called strong
lexicalization~\protect{\cite{sch90,jossch92,kuhsat12}} because it
requires strong equivalence of $G$~and~$G'$; i.e.,~$L(G') = L(G)$.
Weak lexicalization~\protect{\cite{jossch92}} just requires weak
equivalence of $G$~and~$G'$; i.e., $\yield_\Delta(L(G')) =
\yield_\Delta(L(G))$. Clearly, with slight adaptations, 
these definitions can be applied 
to any type of context-free-like grammar that has terminal 
(ranked or unranked) alphabet~$\Sigma$. 
In the literature only two cases are considered: 
$\Delta=\Sigma$ for unranked alphabets and 
$\Delta=\Sigma^{(0)}\setminus\{e\}$ for ranked alphabets. 
It seems to be quite natural and relevant to consider arbitrary $\Delta$.
  
It should be intuitively clear
(and will be shown below) that an MCFTG that does not have finite
$\Delta$"~ambiguity cannot be lexicalized (with respect to $\Delta$).  Thus, we will
prove that an MCFTG can be lexicalized (with respect to $\Delta$) if and only if it has
finite $\Delta$"~ambiguity.  Moreover, we will prove that this
property is decidable.

To lexicalize an MCFTG of finite ambiguity, we need an auxiliary
normal form (stated in Theorem~\ref{thm:dec-growing}).  It generalizes
the Growing Normal Form of \cite{sta09,staott07} for
spCFTGs.  In the remainder of this section the MCFTG $G = (N, \N,
\Sigma, S, R)$ is not assumed to have finite $\Delta$"~ambiguity
unless this is explicitly mentioned.  We only assume that $G$~is
start-separated and reduced.  A rule~$\rho$ is \emph{monic} if
$\abs{\LL(\rho)} = 1$; i.e., $\LL(\rho)$~is a singleton or
equivalently $\rho$~has rank~1 in~$G_\der$.

\begin{definition}
  \label{def:grow}
  \upshape
  The MCFTG~$G$ is \emph{$\Delta$"~growing} if all its non-initial
  terminal rules are doubly $\Delta$"~lexicalized, and all its monic
  rules are $\Delta$"~lexicalized.  It is \emph{almost
    $\Delta$"~growing} if all its non-initial terminal rules and all
  its monic rules are $\Delta$"~lexicalized. \fin
\end{definition}

The application of a proper rule of a $\Delta$"~growing MCFTG
increases the sum of the number of occurrences of lexical symbols and the number of
occurrences of big nonterminals.  In this section we will prove that for every
MCFTG~$G$ of finite $\Delta$"~ambiguity there is an equivalent
$\Delta$"~growing MCFTG (see Theorem~\ref{thm:dec-growing}).
The instance of this result for spCFTGs~and~$\Delta = \Sigma$ is due
to~\cite[Proposition~2]{staott07} and fully proved in~\cite{sta09}.
Note that if $G$~is almost $\Sigma$"~growing, then \emph{all} its
terminal rules are $\Sigma$"~lexicalized.  Note also that every
$\Delta$"~growing MCFTG is almost $\Delta$"~growing, and that every
$\Delta$"~lexicalized MCFTG is almost $\Delta$"~growing.  The
grammar~$G$ of Example~\ref{exa:main} with~$\Delta = \{\alpha, \beta,
\tau, \nu\}$ as in Example~\ref{exa:finamb} is \emph{not} almost
$\Delta$"~growing because of rule~$\rho_4 = B(x_1) \to x_1$.

If the MCFTG~$G$ is almost $\Delta$"~growing, then all its rules
satisfy the requirements for a $\Delta$"~growing grammar except the
non-initial terminal rules, which might be singly
$\Delta$"~lexicalized.  The application of such a rule does not change
the sum of the number of occurrences of lexical symbols and the number of 
occurrences of big nonterminals because a big nonterminal is replaced by 
a lexical symbol. This leads to the following lemma.

\begin{lemma}
  \label{lem:grow-amb}
  If $G$~is almost $\Delta$"~growing, then $G$~has finite
  $\Delta$"~ambiguity and
  \begin{align}
    \label{eq:grow-amb:1} 
    \abs{\pos(d)} \leq 2 \cdot (\abs{\pos_\Delta(\val(d))} +
    \abs{\pos_\N(d)}) + 1 \leq 2 \cdot \abs{\pos_{N \cup
        \Delta}(\val(d))} + 1 \tag{$\dagger$}
  \end{align}
  for every derivation tree~$d$ of $G$; i.e., for every $d\in \bigcup_{A \in \N}
  \DL(G_\der, A)$.
\end{lemma}

\begin{proof}
  We begin with~($\dagger$).  Let $R_{\text{it}}$~be the set of all
  initial terminal rules.  The first inequality is clearly fulfilled
  for~$d \in R_{\text{it}}$, and it suffices to show that
  $\abs{\pos(d)} + 1 \leq 2 \cdot (\abs{\pos_\Delta(\val(d))} +
  \abs{\pos_\N(d)})$ for the remaining derivation trees $d \notin
  R_{\text{it}}$.  For every such tree~$d$ we have
  \[ \abs{\pos(d)} + 1 \leq 2 \cdot \Bigl(\abs{\pos_\N(d)} +
  \abs{\pos_{R^{(0)}}(d)} + \abs{\pos_{R^{(1)}}(d)} \Bigr)
  \enspace, \]
  where $R^{(0)}$~and~$R^{(1)}$ are the sets of terminal and monic
  rules, respectively (see Section~\ref{sub:trees}).  Since $G$~is
  almost $\Delta$"~growing and $\pos_{R_{\text{it}}}(d) = \emptyset$,
  we obtain
  \[ \abs{\pos_{R^{(0)}}(d)} + \abs{\pos_{R^{(1)}}(d)} \leq
  \sum_{p \in \pos_R(d)} \abs{\pos_\Delta(\rhs(d(p)))} =
  \abs{\pos_\Delta(\val(d))} \enspace, \] where the last equality
  holds by Lemma~\ref{lem:valocc}(1).  The second inequality
  in~$(\dagger)$ follows from the first because $\abs{\pos_\N(d)} \leq
  \abs{\pos_N(\val(d))}$ by Lemma~\ref{lem:valocc}(2).

  For the first part of the statement, we consider the set~$L_w = \{t
  \in L(G) \mid \yield_\Delta(t) = w\}$ for some~$w \in \Delta^*$.
  For every derivation tree~$d \in L(G_\der)$ we have~$\pos_\N(d) =
  \emptyset$, and consequently we obtain $\abs{\pos_\Delta(\val(d))} +
  \abs{\pos_\N(d)} = \abs{\yield_\Delta(\val(d))}$.  Hence
  $\abs{\pos(d)} \leq 2 \cdot \abs{w} + 1$ if $\val(d)\in
  L_w$, utilizing~($\dagger$).  This shows that $D_w = \{d \in
  L(G_\der) \mid \val(d) \in L_w\}$~is finite, and so~$L_w$ is finite
  because $L_w = \val(D_w)$ by Theorem~\ref{thm:dtree}.
\end{proof}

The previous result also shows that if $G$~does not have finite
$\Delta$"~ambiguity, then there is no $\Delta$"~lexicalized MCFTG
equivalent to~$G$, as we observed above.  

Our first goal (in proving Theorem~\ref{thm:dec-growing}) is to make sure that all
the non-initial terminal rules are $\Delta$"~lexicalized; i.e., contain a
lexical symbol.  However, for later use, we start by proving a more
general lemma that will allow us to remove every non-initial terminal
rule of which the right-hand side has a certain property~$\F$ subject
to certain requirements.  In particular, the value of a derivation
tree~$d$ has property~$\F$ if and only if $d$~only contains rules of a
corresponding subset~$F \subseteq R$ of rules.   Additionally, each big
nonterminal can only generate finitely many forests with property~$\F$. 
An example of such a property is $\Sigma$"~freeness.
The next construction generalizes the removal of epsilon-rules~$A \to
\varepsilon$ from a context-free grammar~\cite{hopmotull01}.

\newpage
\begin{lemma}
  \label{lem:terminal-removal}
  Let $\F \subseteq P_\Sigma(X)^{\scriptscriptstyle +}$ and $F
  \subseteq R$.  If
  \begin{compactenum}[\indent\upshape(1)]
  \item $L(G,A) \cap \F$~is finite for every $A \in \N$, and
  \item $\val(d) \in \F$~if and only if~$d \in T_F$, for every $d \in
    \bigcup_{A \in \N} L(G_\der, A)$,
  \end{compactenum}
  then there is an \LDTR"~equivalent MCFTG~$G'= (N, \N, \Sigma, S,
  R')$ such that $\rhs(\rho) \notin \F$~for every non-initial terminal
  rule~$\rho \in R'$.
\end{lemma}

\begin{proof}
  For the effectiveness of the constructions in this proof, we assume
  that $\F$~is a decidable subset of~$P_\Sigma(X)^+$, and that the
  elements of~$L(G, A) \cap \F$ are effectively given for every~$A \in
  \N$.  For~$A \in \N$, let~$\F_A = L(G, A) \cap \F$, which is finite
  by~(1).  Moreover,  $\F_A = \val(L(G_\der, A) \cap T_F)$ by~(2) and
  Theorem~\ref{thm:dtree}.  For every $A \in \N$~and~$t \in \F_A$, let
  $L_{\langle A, t\rangle} = \{d \in L(G_\der, A) \cap T_F \mid
  \val(d) = t\}$.  By Lemma~\ref{lem:easy-dtree} applied with~$\N' =
  \emptyset$, the tree language~$L_{\langle A, t\rangle}$ is regular.

  We now construct the MCFTG~$G' = (N, \N, \Sigma, S, R')$. The
  rule~$\rho_{S, t} = S \to t$ is in~$R'$ for every~$t \in \F_S$.
  Moreover, for every rule~$\rho = A \to (u, \LL)$ of~$G$ and every
  substitution function~$f$ for~$\LL$ such that $f(B) \in
  \F_B \cup \{\init(B)\}$ for every~$B \in \LL$, the set~$R'$
  contains the rule
  \[ \rho_f = A \to (u[f],\, \{B \in \LL
  \mid f(B) = \init(B)\}) \enspace, \]
  provided that~$u[f] \notin \F$.  The linear order on~$\LL(\rho_f)$
  is inherited from the one on~$\LL$.  To be precise, let $\LL =
  \{\seq B1k\}$~and~$\Phi = \{i \in [k] \mid f(B_i) \in \F_{B_i}\}$.
  Moreover, let~$[k] \setminus \Phi = \{\seq i1n\}$ with~$i_1 < \dotsb
  < i_n$.  Then $\LL(\rho_f) = \{\seq B{i_1}{i_n}\}$.  This ends the
  construction of~$G'$, so no other rules are in~$R'$.

  First, we prove that for every derivation tree~$d \in L(G_\der, A)
  \setminus T_F$ a derivation tree~$d' \in L(G'_\der, A)$
  with~$\val(d') = \val(d)$ exists.  This shows~$L(G) \subseteq
  L(G')$ because $L(G) = \val(L(G_\der, S) \setminus T_F) \cup \F_S$.
  The proof proceeds by induction on the structure of~$d$.  Let $d =
  \rho(\seq d1k)$ for some $k \in \nat_0$, rule~$\rho = A \to (u, \LL)
  \in R$ with~$\LL = \{\seq B1k\}$, and $d_i \in L(G_\der, B_i)$ for
  every~$i \in [k]$.  Let $\Phi = \{i \in [k] \mid d_i \in T_F\}$, and
  let $f$~be the substitution function for~$\LL$ such that $f(B_i) =
  \val(d_i)$ if~$i \in \Phi$ and $f(B_i) = \init(B_i)$~otherwise.
  Note that $f(B_i) \in \F_{B_i}$ for every~$i \in \Phi$
  by~(2). Since~$d \notin T_F$ we have~$u[f] \notin \F$.  In fact, if
  $u[f] \in \F \subseteq P_\Sigma(X)^{\scriptscriptstyle +}$, then
  $f(B_i) \neq \init(B_i)$ for all~$i \in [k]$ by
  Lemma~\ref{lem:treehom}(2), which yields that~$u[f] = u[B_i \gets
  \val(d_i) \mid 1 \leq i \leq k] = \val(d)$ is in~$\F$ and thus that
  $d \in T_F$ by~(2).  Consequently, $\rho_f \in R'$.  Now let $[k]
  \setminus \Phi = \{\seq i1n\}$ with~$i_1 < \dotsb < i_n$.  By the
  induction hypothesis, there exists a derivation tree~$d'_{i_j} \in
  L(G'_\der, B_{i_j})$ with~$\val(d'_{i_j}) = \val(d_{i_j})$ for
  every~$j \in [n]$.   We now take $d' = \rho_f(\seq{d'}{i_1}{i_n}) 
  \in L(G'_\der, A)$ and prove that~$\val(d') = \val(d)$.  Let
  $[g]$~abbreviate~$[B_i \gets \val(d_i) \mid i \in
  \{i_1,\dots,i_n\}]$.   Then $\val(d') = u[f] [g]$.  By
  Lemma~\ref{lem:comm-assoc}(4) this implies that $\val(d') = u[B_i
  \gets f(B_i) [g] \mid 1 \leq i \leq k]$. 
  Clearly, $f(B_i)[g] = \val(d_i)$~for every~$i \in [k]$, which shows
  that $\val(d') = \val(d)$.

  It should be clear that the transformation from~$d$ to~$d'$, as
  defined above, can be realized by an \LDTR"~transducer~$M$ from~$R$
  to~$R'$.  It has one state~$q$, and for its look-ahead it uses the
  regular tree languages~$L_{\langle A, t\rangle}$, defined above
  for $A \in \N$~and~$t \in \F_A$ in addition to the regular tree
  language~$L_0 = T_R \setminus T_F$.   All subtrees in~$T_F$ are
  deleted by~$M$.  The translation of derivation trees~$d = \rho(\seq
  d1k) \in L(G_\der, A) \setminus T_F$ (as discussed above) is
  realized by the rules $\langle q,\, \rho(y_1 \colon L_{b_1}, \dotsc,
  y_k \colon L_{b_k}) \colon L_0 \rangle \to \rho_f(\langle q,
  y_{i_1}\rangle, \dotsc, \langle q,y_{i_n} \rangle)$ such that $b_i
  \in \{0\} \cup \{\langle B_i, t_i\rangle \mid t_i \in \F_{B_i}\}$
  for all $i \in [k]$, where $f(B_i) = \init(B_i)$ if~$b_i = 0$ and
  $f(B_i) = t_i$ if~$b_i = \langle B_i, t_i\rangle$, and $\{i \in [k]
  \mid b_i = 0\} = \{\seq i1n\}$ with~$i_1 < \dotsb < i_n$.  The
  translation of derivation trees~$d\in L(G_\der) \cap T_F$ is
  realized by the rules $\langle q,\, \rho(\seq y1k) \colon L_{\langle
    S, t\rangle}\rangle \to \rho_{S, t}$ with $t\in \F_S$.

  Second, we show that~$L(G') \subseteq L(G)$.  For every $A \in
  \N$~and~$t \in \F_A$, let $d_{A, t}$~be a fixed derivation tree
  in~$L_{\langle A, t\rangle}$, which can be constructed from the
  regular tree grammar that generates~$L_{\langle A, t\rangle}$.
  Since~$\F_S \subseteq L(G)$, it suffices to prove 
  that for every derivation tree~$d' \in L(G'_\der, A)$ of which the
  root is labeled with a rule~$\rho_f$, a derivation tree~$d \in
  L(G_\der, A) \setminus T_F$ can be constructed such that~$\val(d) =
  \val(d')$.  The proof proceeds by induction on the structure
  of~$d'$.  Let $d' = \rho_f(\seq{d'}{i_1}{i_n})$ with the same
  notation as in the construction of~$G'$.   By the induction
  hypotheses, there are derivation trees~$\seq d{i_1}{i_n}$ of~$G$
  such that $d_{i_j} \notin T_F$~and~$\val(d_{i_j}) = \val(d'_{i_j})$
  for every~$j \in [n]$.  We now take~$d = \rho(\seq d1k)$, where $d_i
  = d_{B_i, f(B_i)}$ for every~$i \in \Phi = [k] \setminus \{\seq
  i1n\}$.  Thus $d_i \in T_F$~and~$\val(d_i) = f(B_i)$ for every~$i
  \in \Phi$.  Then~$d \notin T_F$ because if we suppose~$d \in T_F$,
  then $\seq d1k \in T_F$, which yields~$\Phi = [k]$ and the equality
  \[ u[f] = u[B_i \gets f(B_i) \mid 1 \leq i \leq k] = u[B_i \gets
  \val(d_i) \mid 1 \leq i \leq k] = \val(d) \enspace, \]  which 
  in turn yields the statement~$u[f] \in \F$, contradicting the
  fact that $\rho_f\in R'$.  It is easy to check that the
  \LDTR"~transducer~$M$, in the proof of~$L(G) \subseteq L(G')$,
  transforms~$d$ into~$d'$.  Hence~$\val(d) = \val(d')$.

  The transformation from~$d'$ to~$d$, as defined above, can easily be
  realized by an LDT"~transducer~$M'$ with one state~$q$.  For every
  rule~$\rho'$ of~$G'$, fix either $\rho$~and~$f$ with~$\rho' =
  \rho_f$ or $S$~and~$t$ with~$\rho' = \rho_{S,t}$ (there may be more
  than one such choice).  In the first case, $M'$~has the
  rule~$\langle q,\, \rho'(\seq y1n) \rangle \to \rho(\seq t1k)$, 
  where $t_i = d_{B_i, f(B_i)}$~for every~$i \in \Phi$ and $t_{i_j} =
  \langle q, y_j \rangle$ for every~$j \in [n]$.  In the second case, 
  it has the rule~$\langle q,\, \rho' \rangle \to d_{S,t}$.  This ends
  the proof that $G$~and~$G'$ are \LDTR"~equivalent.
\end{proof}

In the next lemma we show how Lemma~\ref{lem:terminal-removal} can be
used to remove $\Delta$"~free non-initial terminal rules.

\begin{lemma}
  \label{lem:d-epsilon-free}
  Let $F \subseteq R$ be the set of $\Delta$"~free rules.  If
  $\val(L(G_\der, A) \cap T_F)$ is finite for every~$A \in \N$, then
  there is an \LDTR"~equivalent MCFTG~$G'$ such that all its
  non-initial terminal rules are $\Delta$"~lexicalized.  Moreover, if
  $G$~is almost $\Sigma$"~growing, then so is~$G'$.  
\end{lemma}

\begin{proof}
  For the purpose of effectiveness, we assume that the elements
  of~$\val(L(G_\der, A) \cap T_F)$ are effectively given for every~$A
  \in \N$.  Let $\F$~be the set of $\Delta$"~free forests
  in~$P_\Sigma(X)^{\scriptscriptstyle +}$.  As observed before, for
  every derivation tree~$d$, the value~$\val(d)$ is $\Delta$"~free if
  and only if all rules that occur in~$d$ are $\Delta$"~free.  Thus,
  $\F$~and~$F$ satisfy requirement~(2) of
  Lemma~\ref{lem:terminal-removal}.  Hence, for every $A \in \N$ the
  set~$\F_A$, given by~$\F_A = L(G,A) \cap \F = \val(L(G_\der, A) \cap
  T_F)$, is finite and its elements are effectively given.  Thus,
  $\F$~also satisfies requirement~(1) of
  Lemma~\ref{lem:terminal-removal}.

  Let $G'$~be the \LDTR"~equivalent MCFTG as constructed in the proof
  of Lemma~\ref{lem:terminal-removal}.  Then all non-initial terminal
  rules of~$G'$ are $\Delta$"~lexicalized.  Assume now that $G$~is
  almost $\Sigma$"~growing.  Since all non-initial terminal rules
  of~$G$ are $\Sigma$"~lexicalized, the elements of~$L(G, A)$, and
  hence of~$\F_A$, are $\Sigma$"~lexicalized (by
  Theorem~\ref{thm:dtree} and Lemma~\ref{lem:valocc}(1)).  Now
  consider a rule~$\rho$ of~$G$ and a substitution function~$f$
  for~$\LL(\rho)$ such that $f(B) \in \F_{B} \cup \{\init(B)\}$ for
  every $B \in \LL(\rho)$.  If there is at least one~$B \in \LL$ such
  that $f(B) \in \F_{B}$,  then the rule~$\rho_f$ of~$G'$ is
  $\Sigma$"~lexicalized by Lemma~\ref{lem:treehom}(2).  Otherwise, we
  obviously have~$\rho_f = \rho$ and $\rho$~satisfies the requirements
  by assumption.  Hence $G'$~is almost $\Sigma$"~growing.
\end{proof}

We now remove the $\Sigma$"~free terminal rules from~$G$. 

\begin{lemma}
  \label{lem:epsilon-free}
  For every MCFTG~$G$ there is an \LDTR"~equivalent MCFTG~$G'$ of
  which all terminal rules are $\Sigma$"~lexicalized.
\end{lemma}

\begin{proof}
  As in the previous lemma, let $F$~be the set of $\Sigma$"~free rules
  in~$R$, and let $\F$~be the set of $\Sigma$"~free forests
  in~$P_\Sigma(X)^{\scriptscriptstyle +}$.  Then $\val(L(G_\der, A)
  \cap T_F) = L(G, A) \cap \F$ as demonstrated in the proof of
  Lemma~\ref{lem:d-epsilon-free}.  Clearly, a forest~$t \in
  P_\Sigma(X)^{\scriptscriptstyle +}$ is $\Sigma$"~free if and only if
  $t \in x_1^{\scriptscriptstyle +}$; i.e., $t$~is of the form~$(x_1,
  \dotsc, x_1)$.  Such a forest~$t$ can only be generated by a big
  nonterminal of rank~$(1, \dotsc, 1)$.  Hence, $L(G,A) \cap \F$~is either empty or
  equal to~$\{x_1^k\}$ with~$k = \abs{A}$.  Moreover,
  $\val(L(G_\der, A) \cap T_F)$ can be computed because it is empty if
  and only if the regular tree language~$L(G_\der, A) \cap T_F$ is
  empty.  By Lemma~\ref{lem:d-epsilon-free} there is an
  \LDTR"~equivalent MCFTG~$G'$, of which all non-initial terminal
  rules are $\Sigma$"~lexicalized.  Obviously, the initial terminal
  rules of an MCFTG are also $\Sigma$"~lexicalized.
\end{proof}

\begin{example}
  \label{exa:epsilon-free}
  \upshape
  In the MCFTG~$G$ of Example~\ref{exa:main}, the rules $\rho_4 =
  B(x_1) \to x_1$~and~$\rho'_4 = B'(x_1) \to x_1$ are the only
  $\Sigma$"~free rules.  The construction in the proof of
  Lemma~\ref{lem:terminal-removal} asks us to apply these
  rules in all possible ways to the right-hand sides of the other
  rules.  Thus, we change the set~$R$ of rules by removing rules
  $\rho_4$~and~$\rho'_4$ and adding the following rules:
  \begin{align*}
    A 
    &\to T_1(\sigma(T_2, T_3)) \\
    B(x_1) 
    &\to \sigma(x_1, B'(A)) \qquad 
    & B(x_1)
    &\to \sigma(B(x_1), A) \qquad
    & B(x_1)
    &\to \sigma(x_1, A) \\
    B'(x_1) 
    &\to \sigma(x_1, B'(A)) \qquad
    & B'(x_1)
    &\to \sigma(B(x_1), A) \qquad
    & B'(x_1)
    &\to \sigma(x_1, A) \enspace.
  \end{align*}
  In the resulting MCFTG~$G'$, which we will call~$G$ again, all
  terminal rules are $\Sigma$"~lexicalized.  In fact, $G$~is now both
  $\Sigma$"~lexicalized and $\Sigma$"~growing, and all its terminal
  rules are $\Delta$"~lexicalized for~$\Delta = \{\alpha, \beta, \tau,
  \nu\}$ as in Example~\ref{exa:finamb}. \fin
\end{example}

Our second goal is to make sure that all monic rules (i.e., rules
whose right-hand side contains exactly one big nonterminal) are
$\Delta$"~lexicalized.  In the next construction we remove
$\Delta$"~free monic rules thereby generalizing the removal of chain
rules~$A \to B$ from a context-free grammar~\cite{hopmotull01}.

\newpage
\begin{lemma}
  \label{lem:d-chain-free}
  Suppose that all non-initial terminal rules of~$G$ are
  $\Delta$"~lexicalized.  Let $F \subseteq R$~be the set of
  $\Delta$"~free monic rules.  If $\val(\DL(G_\der, A) \cap T_{\N \cup
    F})$~is finite for every~$A \in \N$, then there is an
  \LDTR"~equivalent almost $\Delta$"~growing MCFTG~$G'$.
\end{lemma}

\begin{proof}
  Let $\F_A = \val(\DL(G_\der, A) \cap T_{\N \cup F})$ for every~$A
  \in \N$.  Again, for the purpose of effectiveness, we assume that
  the elements of~$\F_A$ are effectively given.  Note that~$\init(A)
  \in \F_A$.  Every forest~$t \in \F_A$ is of the form~$\val(d)$ 
  with~$d \in \DL(G_\der, A) \cap T_{\N \cup F}$, and every such
  derivation tree~$d$ is of the form~$d = wB$ with $w \in F^*$~and~$B
  \in \N$.  Hence $t$~is $\Delta$"~free because all rules that occur
  in~$d$ are $\Delta$"~free.  Moreover, by Lemma~\ref{lem:valocc}(2),
  $t$~is uniquely $N$"~labeled and $\alp_N(t) = \alp(B)$.  In other
  words, the big nonterminal~$B$ occurs exactly once in~$t$, and no
  other nonterminals occur in~$t$.  We will denote~$B$ by~$B_t$.  Note
  that, since $G$~is start-separated, if~$B_t = S$ then~$A = S$
  because~$w = \varepsilon$.  For every~$t\in\F_A$, let $d_{A, t} \in
  T_{\N \cup F}$~be a particular derivation tree of~$G$ of type~$A$
  such that~$\val(d_{A, t}) = t$.  Such a derivation tree can be
  computed by Lemma~\ref{lem:easy-dtree} applied with $\N' = \N$.

  We construct the MCFTG~$G' = (N, \N, \Sigma, S, R')$ such that for
  every big nonterminal~$A \in \N$, tree~$t \in \F_A$, and rule~$\rho
  = B_t \to (u, \LL) \in R \setminus F$, the rule~$\rho_{A,t} = A \to
  (t[B_t \gets u], \LL)$ is in~$R'$, where the links in~$\LL$ have the
  same order as in the rule~$\rho$.  Since~$\rho \notin F$, it is
  straightforward to check that $\rho_{A, t}$~satisfies the
  requirements for~$G'$ to be almost $\Delta$"~growing: (i)~If
  $\rho$~is $\Delta$"~lexicalized, then so is~$\rho_{A, t}$ because
  $u$~is substituted for~$B_t$.  (ii)~If $\rho_{A, t}$~is monic, then
  $\rho$~is monic and hence $\Delta$"~lexicalized because~$\rho \notin
  F$.  (iii)~If $\rho$~is initial (i.e.,~$B_t = S$), 
  then~$\rho_{A, t}$~is initial (because $A = S$); thus,
  if $\rho_{A,t}$~is
  non-initial terminal, then $\rho$~is non-initial terminal and hence
  $\Delta$"~lexicalized by assumption on~$G$.

  To show the correctness of~$G'$, we first prove that for every
  derivation tree~$d \in L(G_\der, A)$ there is a derivation tree~$d'
  \in L(G'_\der, A)$ with~$\val(d') = \val(d)$.  Clearly, $d$~has
  the unique form~$d = w \rho(\seq d1k)$ such that~$w\in F^*$, $\rho
  \notin F$, and~$\seq d1k \in T_R$.   Let $\rho = B \to (u, \LL)$
  with~$\LL = \{\seq B1k\}$, and let~$t = \val(wB) \in \F_A$.  By the
  induction hypothesis there is a derivation tree~$d'_i \in L(G'_\der,
  B_i)$ with~$\val(d'_i) = \val(d_i)$ for every $i \in [k]$.  We take
  $d' = \rho_{A,t}(\seq{d'}1k)$.  Then 
  \begin{align*}
    \val(d') 
    &= t[B \gets u] [B_i \gets \val(d_i) \mid 1 \leq i \leq k] 
      = t[B \gets u [B_i \gets \val(d_i) \mid 1 \leq i \leq k]\,] \\
    &= t[B \gets \val(\rho(\seq d1k))] = \val(w\rho(\seq d1k)) =
      \val(d) \enspace,
  \end{align*}
  where the second equality holds by Lemma~\ref{lem:comm-assoc}(4) and
  the penultimate equality holds by Lemma~\ref{lem:valsub}.  This
  shows that~$L(G) \subseteq L(G')$.

  The \LDTR"~transducer~$M$ that transforms~$d$ into~$d'$, as above,
  uses the tree languages
  \[ L_{A,t} = \{wd \in L(G_\der, A) \mid w \in F^*,\, d \in L(G_\der,
  B_t),\, d(\varepsilon) \notin F,\, \val(wB_t) = t\} \]
  as look-ahead, where $A \in \N$~and~$t \in \F_A$.  It is easy to see
  that $L_{A, t}$~is regular.  An RTG that generates~$L_{A,t}$ can be
  obtained from the grammar for the regular tree language~$L_{\langle
    A, t\rangle}$ in the proof of Lemma~\ref{lem:easy-dtree} as
  follows.  First, add the nonterminals and rules of~$G_\der$.
  Second, replace every rule~$\langle B, \init(B) \rangle \to B$ by
  all rules~$\langle B, \init(B) \rangle\to \rho(\seq B1k)$, where $B
  \to \rho(\seq B1k)$ is a rule of~$G_\der$ and~$\rho \notin F$.  The
  transducer~$M$ has initial state~$q_0$ and the states~$q_{A,t}$ for
  every $A \in \N$~and~$t \in \F_A$.  For every rule~$\rho \in R
  \setminus F$, the transducer~$M$ has the rule~$\langle q_0,
  \rho(\seq y1k) \rangle \to \rho(\langle q_0, y_1 \rangle, \dotsc,
  \langle q_0, y_k \rangle)$ and all the rules~$\langle q_{A, t},
  \rho(\seq y1k) \rangle \to \rho_{A, t}(\langle q_0, y_1 \rangle,
  \dotsc, \langle q_0, y_k \rangle)$.  Moreover, for every rule~$\rho
  \in F$, the transducer~$M$ has all rules $\langle q_0,
  \rho(y_1) \colon L_{A,t} \rangle \to \langle q_{A,t}, y_1 \rangle$
  and $\langle q_{A,t}, \rho(y_1) \rangle \to \langle q_{A,t}, y_1
  \rangle$.  It should be clear that $M$~indeed transforms~$d$
  into~$d'$.

  Next, we prove that for every derivation tree~$d' \in L(G'_\der, A)$
  there is a derivation tree~$d \in L(G_\der, A)$ with
  $\val(d) = \val(d')$.  The proof is by induction on~$d'$, so let $d'
  = \rho_{A,t}(\seq{d'}1k)$ with~$\rho$, $A$, and~$t$ as in the
  construction of~$G'$.  By the induction hypothesis, there is a
  derivation tree~$d_i$ of~$G$ such that~$\val(d_i) = \val(d'_i)$ for
  every~$i \in [k]$.  We now take~$d = d_{A,t}[B_t \gets \rho(\seq
  d1k)]$, where the derivation tree~$d_{A,t}$ was defined at the end
  of the first paragraph of this proof.  Since $d_{A,t}$~is of the
  form~$wB_t$ with~$w \in F^*$, and hence~$d = w\rho(\seq d1k)$, it
  should be clear that the construction in the proof of~$L(G)
  \subseteq L(G')$ (i.e., the \LDTR"~transducer~$M$) transforms~$d$
  into~$d'$, which implies that~$\val(d) = \val(d')$.

  The transformation from~$d'$ to~$d$, as defined above, can easily be
  realized by an LDT-transducer~$M'$ with one state~$q$.  For every
  rule~$\rho'$ of~$G'$, fix~$\rho$, $A$, and~$t$ such that $\rho' =
  \rho_{A,t}$.  Then $M'$~has the rule~$\langle q,\, \rho'(\seq y1k)
  \rangle \to d_{A,t}[B_t \gets \rho(\langle q, y_1 \rangle, \dotsc,
  \langle q, y_k \rangle)]$.  We finally observe that the
  transformation from~$d$ to~$d'$ can also be realized by an
  LDT"~transducer (without look-ahead), but the above transducer~$M$
  is easier to understand.
\end{proof}

\newpage
\begin{lemma}
  \label{lem:chain-free}
  For every MCFTG~$G$ there is an \LDTR"~equivalent almost
  $\Sigma$"~growing MCFTG~$G'$. 
\end{lemma}

\begin{proof}
  By Lemma~\ref{lem:epsilon-free} and Proposition~\ref{pro:LDTRcomp},
  we may assume that all terminal rules of~$G$ are
  $\Sigma$"~lexicalized.  Let~$F \subseteq R$ be the set of
  $\Sigma$"~free monic rules.  The statement holds using
  Lemma~\ref{lem:d-chain-free} if we prove that $\F_A =
  \val(\DL(G_\der,A) \cap T_{\N \cup F})$ is finite and that its
  elements can be computed for every~$A \in \N$.  For every big
  nonterminal~$A$, let $\M_A$~be the set of all $\Sigma$"~free
  forests~$t$ in~$P_{N\cup\Sigma}(X)^{\scriptscriptstyle +}$ such
  that~$\rk(t) = \rk(A)$, $t$~is uniquely $N$"~labeled, and~$\alp_N(t)
  = \alp(B)$ for some~$B \in \N$.  Clearly $\M_A$~is finite because
  $\abs{\pos_{N \cup \Sigma}(t)} = \abs{\pos_N(t)} \leq \mu(G)$ and
  $\abs{\pos_X(t)} \leq \mu(G) \cdot \wid(G)$.  As argued in the
  beginning of the proof of Lemma~\ref{lem:d-chain-free}, $\F_A
  \subseteq \M_A$.  Consequently, $\F_A$~is finite and its elements
  can be computed by a standard iteration because the sets~$\F_A$
  with~$A \in \N$ are the smallest sets of forests such that
  (i)~$\init(A) \in \F_A$ and (ii)~if $A \to (u, \{B\}) \in F$ and $t
  \in \F_B$, then~$u[B \gets t] \in \F_A$.
\end{proof}

Let $G$~be an almost $\Sigma$"~growing MCFTG.  Then, for every
forest~$t$, there are only finitely many derivation trees~$d$ such
that $\val(d) = t$ by inequality~($\dagger$) of Lemma~\ref{lem:grow-amb}.
This implies that the finiteness problem is decidable for $L(G)$~and~$L(G, A)$.  
In fact, $L(G)$~is finite if and only if $L(G_\der)$~is finite,  which is
decidable because $G_\der$~is an RTG.  Moreover, if $L(G)$~is finite,
then the elements of~$L(G)$ can be computed because the elements
of~$L(G_\der)$ can be computed and~$L(G) = \val(L(G_\der))$.  Similar
statements hold for~$L(G, A)$. Thus, by Lemma~\ref{lem:chain-free},
the finiteness problem is decidable for MCFTGs.

We now show that if $G$~is almost
$\Sigma$"~growing, then the requirements of
Lemmas~\ref{lem:terminal-removal} and~\ref{lem:d-chain-free} are
fulfilled.

\begin{lemma}
  \label{lem:finamb}
  Let $G$~be almost $\Sigma$"~growing.  Moreover, let $F$~be the set
  of all $\Delta$"~free rules and $F' \subseteq F$~be the set of all
  $\Delta$"~free monic rules.  Finally, let $\F_A = \val(L(G_\der, A)
  \cap T_F)$ and $\F'_A = \val(\DL(G_\der, A) \cap T_{\N \cup F'})$
  for every $A \in \N$.
  \begin{compactenum}[\upshape (1)]
  \item It is decidable for~$A \in \N$ whether or not~$\F_A$
    (respectively,~$\F'_A$) is finite, and if so, its elements can be
    computed.
  \item If $G$~has finite $\Delta$"~ambiguity, then $\F_A$~and~$\F'_A$
    are finite for every~$A \in \N$.
  \end{compactenum}
\end{lemma}

\begin{proof} 
  For~(1) we observe that since $G$~is almost $\Sigma$"~growing,
  inequality~($\dagger$) of Lemma~\ref{lem:grow-amb} implies that $\F_A$~is finite if and
  only if~$L(G_\der, A) \cap T_F$ is finite.  The latter is a regular
  tree language, and it is decidable whether or not it is finite.
  Moreover, if so, its elements, and thus also the elements of~$\F_A$,
  can be computed.  The same argument holds for~$\F'_A$.
  
  For~(2) we assume that $G$~has finite $\Delta$"~ambiguity and
  that $\F_A$~is infinite.  Since we may assume that
  $G$~and~$G_\der$ are reduced, there is a derivation tree~$d_0
  \in \DL(G_\der, S)$ with $\abs{\pos_\N(d_0)} = \abs{\pos_A(d_0)} =
  1$.  Let $D_0 = \{d_0[A \gets d] \mid d \in L(G_\der, A) \cap T_F\} 
  \subseteq L(G_\der)$.  Since $\F_A$~is infinite, also $L(G_\der, A)
  \cap T_F$~is infinite, and thus $D_0$~is infinite by
  Lemma~\ref{lem:treehom}.  Since $G$~is almost $\Sigma$"~growing, the
  set~$L_0 = \val(D_0)$~is an infinite subset of~$L(G)$.  Now, for 
  every derivation tree~$d' \in T_{\N \cup R}$, let $\pr_\Delta(d') =
  \sum_{p \in \pos_R(d')} \abs{\pos_\Delta(\rhs(d'(p)))}$.
  Lemma~\ref{lem:valocc}(1) and Lemma~\ref{lem:treehom} yield
  $\abs{\pos_\Delta(\val(d_0[A \gets d]))} = \pr_\Delta(d_0[A\gets d])
  = \pr_\Delta(d_0) + \pr_\Delta(d) = \pr_\Delta(d_0)$ for every~$d
  \in L(G_\der, A) \cap T_F$, where the last equality uses~$d \in
  T_F$.  Consequently, $\abs{\pos_\Delta(t)} \leq \pr_\Delta(d_0)$ for
  every tree~$t$ in the infinite set~$L_0$,  which contradicts the
  finite $\Delta$"~ambiguity of~$L(G)$.

  A similar proof works for~$\F'_A$. Since $\DL(G_\der, A) \cap T_{\N
    \cup F'}$~is infinite, there exists~$B \in \N$ such
  that~$\DL(G_\der, A) \cap T_{\{B\} \cup F'}$ is infinite.  Since
  $G_\der$~is reduced, there exists a derivation tree~$d_1 \in
  L(G_\der, B)$.  Now let $D'_0 = \{d_0[A \gets d[B \gets d_1]\,] \mid
  d \in \DL(G_\der, A) \cap T_{\{B\} \cup F'}\} \subseteq L(G_\der)$.
  By similar arguments as above, we then obtain that
  $\abs{\pos_\Delta(t)} \leq \pr_\Delta(d_0) + \pr_\Delta(d_1)$ for
  every tree~$t$ in the infinite set~$L'_0 = \val(D'_0) \subseteq
  L(G)$,  which again contradicts the finite $\Delta$"~ambiguity
  of~$L(G)$.
\end{proof} 

Now we are able to turn~$G$ into an equivalent almost
$\Delta$"~growing MCFTG, provided that it has finite
$\Delta$"~ambiguity.

\begin{lemma}
  \label{lem:dec-almost}
  It is decidable whether or not the MCFTG~$G$ has finite
  $\Delta$"~ambiguity, and if so, there is an \LDTR"~equivalent almost
  $\Delta$"~growing MCFTG~$G'$.
\end{lemma}

\begin{proof}
  By Lemma~\ref{lem:chain-free} we may assume that $G$~is almost
  $\Sigma$"~growing.  By Lemma~\ref{lem:finamb} it is decidable
  whether $\F_A$~is finite for every~$A \in \N$, and if not, then
  $G$~does not have finite $\Delta$"~ambiguity.  If they are, then we
  may assume by Lemma~\ref{lem:d-epsilon-free} that all non-initial
  terminal rules of~$G$ are $\Delta$"~lexicalized.   Again by
  Lemma~\ref{lem:finamb}, it is decidable whether $\F'_A$~is finite
  for every~$A \in \N$, and if not, then $G$~does not have finite
  $\Delta$"~ambiguity.  If they are, then we may assume by
  Lemma~\ref{lem:d-chain-free} that $G$~is almost $\Delta$"~growing.
  Finally, in this case $G$~has finite $\Delta$"~ambiguity by
  Lemma~\ref{lem:grow-amb}.
\end{proof}

\begin{example}
  \label{exa:d-chain-free}
  \upshape
  The MCFTG~$G$ of Example~\ref{exa:epsilon-free} is already
  $\Sigma$"~growing. Moreover, all its terminal rules are
  $\Delta$"~lexicalized for~$\Delta = \{\alpha, \beta, \tau, \nu\}$.
  Let us turn~$G$ into an almost $\Delta$"~growing grammar by
  Lemma~\ref{lem:d-chain-free}.  We omit parentheses around the
  arguments of unary terminals.  The set~$F$ of $\Delta$"~free monic
  rules of~$G$ consists of the rules~$A \to T_1(\sigma(T_2, T_3))$,
  $B(x_1) \to \sigma(x_1, A)$, and~$B'(x_1) \to \sigma(x_1, A)$.
  Next, for each big nonterminal~$A' \in \N$ we compute the sets
  $\F_{A'} = \val(\DL(G_\der, A') \cap T_{\N\cup F})$ and obtain
  \begin{align*}
    \F_T 
    &= \{\init(T)\} \qquad 
    & \F_A 
    &= \{\init(A),\, T_1(\sigma(T_2, T_3))\} \qquad
    & \F_B 
    &= \{\init(B),\, \sigma(x_1, A),\, \sigma(x_1, T_1(\sigma(T_2,
      T_3))) \}  \\
    && \F_S 
    &= \{ \init(S) \} 
    & \F_{B'} 
    &= \{\init(B'),\, \sigma(x_1, A),\, \sigma(x_1, T_1(\sigma(T_2,
      T_3))) \} \enspace,
  \end{align*}
  where~$T = (T_1, T_2, T_3)$, which are all finite.  The construction
  in the proof of Lemma~\ref{lem:d-chain-free} asks us to apply
  \begin{compactitem}
  \item the rules $\rho_5 = (T_1(x_1), T_2, T_3) \to (\alpha T_1(\beta
    x_1), \alpha T_2, \gamma T_3)$~and~$\rho_6 = (T_1(x_1), T_2, T_3)
    \to (x_1, \tau, \nu)$ for~$T$ to~$T_1(\sigma(T_2,T_3)) \in 
    \F_A$ and $\sigma(x_1, T_1(\sigma(T_2,T_3))) \in \F_B \cap
    \F_{B'}$, and
  \item the rule~$\rho_2 = A \to T_1(\sigma(B(T_2), T_3))$ for~$A$
    to~$\sigma(x_1, A) \in \F_B \cap \F_{B'}$.
  \end{compactitem}
  Consequently, we change the set of rules of~$G$ by removing the
  above three $\Delta$"~free monic rules and adding the following
  5~rules, and the 3~additional rules that make $B'$~an alias of~$B$:
  \begin{align*}
    A 
    &\to \alpha T_1(\beta \sigma(\alpha T_2, \gamma T_3)) \quad
    & A 
    &\to \sigma(\tau, \nu) \\
    B(x_1) 
    &\to \sigma(x_1, \alpha T_1(\beta \sigma(\alpha T_2, \gamma T_3)))
      \quad
    & B(x_1) 
    &\to \sigma(x_1, \sigma(\tau, \nu)) \quad
    & B(x_1) 
    &\to \sigma(x_1, T_1(\sigma(B(T_2), T_3))) \enspace. 
  \end{align*}
  The resulting grammar~$G'$, which we will again call~$G$, now has
  the following rules (and the rules required to make~$B'$ an alias
  of~$B$):
  \begin{align*}
    A 
    &\to \alpha T_1(\beta \sigma(\alpha T_2, \gamma T_3)) \quad
    & A 
    &\to \sigma(\tau, \nu) \quad
    & A 
    &\to T_1(\sigma(B(T_2), T_3)) \\
    B(x_1) 
    &\to \sigma(x_1, \alpha T_1(\beta \sigma(\alpha T_2, \gamma T_3)))
      \quad
    & B(x_1) 
    &\to \sigma(x_1, \sigma(\tau, \nu)) \quad
    & B(x_1) 
    &\to \sigma(x_1, T_1(\sigma(B(T_2), T_3))) \\
    B(x_1) &\to \sigma(B(x_1), B'(A)) \quad
    & B(x_1) 
    &\to \sigma(x_1, B'(A)) \qquad 
    & B(x_1)
    &\to \sigma(B(x_1), A) \\
    T &\to
    (\alpha T_1(\beta x_1),\, \alpha T_2,\, \gamma T_3) \qquad
    & S 
    &\to \alpha A
    & T 
    &\to (x_1, \tau, \nu)
  \end{align*}
  with $T = (T(x_1), T_2, T_3)$.  This MCFTG~$G$ is not only almost
  $\Delta$"~growing, but even $\Delta$"~growing.  It is also almost
  $\{\alpha, \tau\}$"~growing, which proves that $L(G)$~has finite
  $\{\alpha,\tau\}$"~ambiguity by Lemma~\ref{lem:grow-amb} (as
  observed in Example~\ref{exa:finamb}).  The only rules of~$G$
  (without rules with left-hand side~$B'$) that are not
  $\Delta$"~lexicalized are
  \begin{align*}
    A &\to T_1(\sigma(B(T_2), T_3)) \quad
    & B(x_1) &\to \sigma(B(x_1), A) \\
    B(x_1) &\to \sigma(x_1, T_1(\sigma(B(T_2), T_3))) \quad 
    & B(x_1) &\to \sigma(B(x_1), B'(A)) \quad
    & B(x_1) &\to \sigma(x_1, B'(A)) \enspace.
  \end{align*}
  It is easy to lexicalize this grammar.  The first non-lexicalized
  rule~$\rho_2 = A \to T_1(\sigma(B(T_2), T_3))$ can be replaced by the
  two lexicalized rules $A \to \alpha T_1(\beta(\sigma(B(\alpha T_2),
  \gamma T_3)))$~and~$A \to \sigma(B(\tau), \nu)$ that are obtained
  from~$\rho_2$ by applying the two rules for~$T$ to its right-hand
  side.  By Lemma~\ref{lem:comm-assoc}(4) this process
  preserves~$L(G)$, and it should be clear that the resulting grammar
  is \LDTR"~equivalent to~$G$.  Now all four rules for~$A$ are
  lexicalized.  The remaining non-lexicalized rule in the first column
  can be replaced by two lexicalized rules in the same way.  Finally,
  the same process can be used for all the remaining non-lexicalized
  rules by applying the four lexicalized rules for~$A$ to their
  right-hand sides; this does, however, not preserve 
  \LDTR"~equivalence.\footnote{\label{foo:fc1} The resulting MCFTG is $\X$"~equivalent to $G$
  for the class $\X$ of tree transductions realized by \emph{finite-copying} 
  deterministic top-down tree transducers
  with regular look-ahead.}   \fin
\end{example}

It remains to construct an equivalent $\Delta$"~growing MCFTG, which
is the main result of this section.

\begin{theorem}
  \label{thm:dec-growing}
  It is decidable whether or not the MCFTG~$G$ has finite
  $\Delta$"~ambiguity, and if so, there is an \LDTR"~equivalent
  $\Delta$"~growing MCFTG~$G'$.  Moreover, $\wid(G') =
  \wid(G)$~and~$\mu(G') = \mu(G)$.
\end{theorem}

\begin{proof}
  By Lemma~\ref{lem:dec-almost} it suffices to show that if $G$~is
  almost $\Delta$"~growing, then there is an \LDTR"~equivalent
  $\Delta$"~growing MCFTG~$G'$.  Consequently, it remains to remove
  all non-initial terminal rules that are singly $\Delta$"~lexicalized,
  using the construction in the proof of
  Lemma~\ref{lem:terminal-removal}.  Let~$\F = \{t \in
  P_\Sigma(X)^{\scriptscriptstyle +} \mid \abs{\pos_\Delta(t)} = 1\}$,
  and let $F$~be the set of all (terminal) rules~$A \to u \in R$ such
  that~$u \in \F$.  Note that~$T_F = F$.  Since $G$~is almost
  $\Delta$"~growing, $\val(d) \in \F$~if and only if~$d \in T_F$, for
  every~$d \in L(G_\der, A)$.  In fact, since all non-initial terminal
  rules are $\Delta$"~lexicalized, $\pos_\Delta(\val(d')) \neq
  \emptyset$ for every~$d' \in L(G_\der, B)$ with~$B \in \N \setminus
  \{S\}$.  Hence, if $d = \rho(\seq d1k)$ with~$k \geq 1$, then 
  either~$k \geq 2$ and both $\val(d_1)$~and~$\val(d_2)$ contribute a
  lexical position to~$\val(d)$, or~$k = 1$ and both $\val(d_1)$~and
  the right-hand side of~$\rho$ contribute a lexical position
  to~$\val(d)$ because the monic rule~$\rho$ is $\Delta$"~lexicalized.
  Thus, $\F$~satisfies requirement~(2) of
  Lemma~\ref{lem:terminal-removal}.  Additionally, $\F$~satisfies
  requirement~(1) of Lemma~\ref{lem:terminal-removal} because $L(G, A)
  \cap \F = \val(L(G_\der, A) \cap T_F) = \{u \mid A \to u \in F\}$. 
  Let $\F_A = \{u \mid A \to u \in F\}$ for every $A \in \N$, and let
  $G'$~be the \LDTR"~equivalent MCFTG as constructed in the proof of
  Lemma~\ref{lem:terminal-removal}.  If~$\rho = A \to (u, \LL)$ is a
  rule of~$G$, and $f$~is a substitution function for~$\LL$ such that
  $f(B) \in \F_B \cup \{\init(B)\}$ for every~$B \in \LL$, then the
  new rule~$\rho_f = A \to (u[f],\, \{B \in \LL \mid f(B) =
  \init(B)\})$ is either equal to the old rule~$\rho$ (because~$f(B) =
  \init(B)$ for all~$B \in \LL$) or is $\Delta$"~lexicalized
  (because~$f(B) \in \F$ for some~$B \in \LL$).  This implies that
  $G'$~is almost $\Delta$"~growing.  Moreover, $\rho_f$~is a rule
  of~$G'$ only if $u[f] \notin \F$, so $G'$~does not have non-initial
  terminal rules that are singly $\Delta$"~lexicalized, and hence is
  $\Delta$"~growing.

  We finally observe that $G'$~has the same ranked alphabet~$N$ of
  nonterminals and the same set~$\N$ of big nonterminals as~$G$, as
  one can easily check from the constructions in Lemmas
  \ref{lem:terminal-removal}~and~\ref{lem:d-chain-free}.  That implies
  that $\wid(G') = \wid(G)$~and~$\mu(G') = \mu(G)$.
\end{proof}

\begin{example}
  \label{exa:d-growing}
  \upshape
  We have seen that the new grammar~$G$ in
  Example~\ref{exa:d-chain-free} is almost $\{\alpha,
  \tau\}$"~growing.  However, it is not $\{\alpha, \tau\}$"~growing
  because the right-hand side of each terminal rule has exactly one
  lexical position (always labeled~$\tau$).  Let $F$~be the set of all
  terminal rules of~$G$; i.e., 
  \[ F = \bigl\{ A \to \sigma(\tau, \nu),\; B(x_1) \to \sigma(x_1,
  \sigma(\tau, \nu)),\; B'(x_1) \to \sigma(x_1, \sigma(\tau, \nu)),\;
  (T_1(x_1), T_2, T_3) \to (x_1, \tau, \nu) \bigr\} \enspace. \]
  In the construction in the proof of Theorem~\ref{thm:dec-growing} we
  apply the rules of~$F$ in all possible ways to the right-hand sides
  of the other rules of~$G$ (and then remove the rules~$F$).  As an
  example, the rule~$B(x_1) \to \sigma(x_1, B'(A))$ is replaced by
  itself and the following three additional  $\{\alpha,
  \tau\}$"~growing rules
  \[ B(x_1) \to \sigma \bigl(x_1, \underbrace{\sigma(A, \sigma(\tau,
    \nu))}_{B'(A)} \bigr) \quad B(x_1) \to \sigma \bigl(x_1,
  B'(\underbrace{\sigma(\tau, \nu)}_A) \bigr) \quad \text{and} \quad
  B(x_1) \to \sigma \bigl(x_1,
  \underbrace{\sigma(\underbrace{\sigma(\tau, \nu)}_A, \sigma(\tau,
    \nu))}_{B'(\sigma(\tau, \nu))} \bigr) \enspace, \]
  in which we marked the substitutions. \fin
\end{example}

Since every MCFTG has finite $\Sigma$"~ambiguity, we obtain the
following result from Theorem~\ref{thm:dec-growing}.
It generalizes the corresponding result of~\cite{staott07,sta09} for spCFTGs,
which is the special case~$\mu(G) = 1$. 

\begin{corollary}
  \label{cor:growing}
  For every MCFTG~$G$ there is an \LDTR"~equivalent $\Sigma$"~growing
  MCFTG~$G'$.  Moreover, $\wid(G') = \wid(G)$~and~$\mu(G') =
  \mu(G)$. 
\end{corollary} 

At the end of this section we consider
an additional basic normal form for MCFTGs that generalizes one that
is familiar from multiple context-free grammars (viz.\@ condition~(N3)
of~\cite[Lemma~2.2]{sekmatfujkas91}), and will be needed in
Section~\ref{sub:mctag}.  We say that the MCFTG~$G$ is
\emph{nonerasing} if $u_i \neq x_1$~for every rule~$(\seq A1n) \to
((\seq u1n), \LL)$ and every~$i \in [n]$.  Note that in a grammar~$G$,
the tree~$u_i$ can only be equal to~$x_1$ if~$\rk(A_i) = 1$.

\begin{lemma}
  \label{lem:nonerasing}
  For every MCFTG~$G$ there is an \LDTR"~equivalent nonerasing
  MCFTG~$G'$.  If the grammar $G$~is $\Delta$"~lexicalized, then so is~$G'$.
  Moreover, $\wid(G') = \wid(G)$~and~$\mu(G') = \mu(G)$.
\end{lemma}

\begin{proof}
For a sequence~$w =(\seq a1n)$ we denote, in this proof only,
$[n]$ by $\num(w)$, and $a_j$ by $w|_j$ for every~$j \in \num(w)$.  
For every~$\Psi \subseteq \num(w)$, we
denote by~$w|_\Psi$ the ``scattered subsequence'' 
$(\seq a{j_1}{j_m})$ of~$w$, in which $\Psi =\{\seq j1m\}$
and $1 \leq j_1 < \dotsb < j_m \leq n$.  
Intuitively, $w|_\Psi$~is obtained 
from $w$ by selecting the $j$-th element of $w$ for every $j\in\Psi$. 

  By Lemma~\ref{lem:epsilon-free} we may assume that all terminal
  rules of~$G = (N, \N, \Sigma, S, R)$ are $\Sigma$"~lexicalized.
  Moreover, we can assume that $G$~has disjoint big nonterminals, as
  observed after Lemma~\ref{lem:renaming}.  The set $\N'$ of big nonterminals of
  the new grammar~$G' = (N, \N', \Sigma, S, R')$ consists of 
  all~$A|_\Psi$ such that $A \in \N$, $\Psi \subseteq \num(A)$, $\Psi\neq\emptyset$,
  and $\rk(A|_j) = 1$ for every $j\in\num(A)\setminus\Psi$.
  Intuitively, $\Psi$~selects those nonterminals of~$A$
  that do not generate~$x_1$.  Since all terminal rules of~$G$ are
  $\Sigma$"~lexicalized, it is not possible that all nonterminals
  of~$A$ generate~$x_1$.  Note that $S = S|_{\{1\}}$ and that
  for every~$A' \in \N'$ there are a unique~$A \in \N$ and a
  unique~$\Psi \subseteq \num(A)$ such that $A' = A|_\Psi$
  because $G$~has disjoint big nonterminals. Note also that $\num(A)=\num(u)$
  for every rule $A \to (u, \LL)$ of~$G$. 

  Let $\rho = A \to (u, \LL)$~be a rule of~$G$ with~$\LL = \{\seq B1k\} \subseteq \N$,
  and let $\seq \Psi1k \subseteq \nat$ such that~$B_i|_{\Psi_i} \in \N'$ 
  for every~$i \in [k]$.  Finally, let
  $u' = u[B_i|_j \gets x_1 \mid i \in 
  [k],\, j\notin \Psi_i]$, and let~$\Psi = \{j \in \num(A) \mid u'|_j \neq
  x_1\}$.  Then $R'$~contains the rule~$\rho_{\seq \Psi1k} =
  A|_\Psi \to (u'|_\Psi, \LL')$ with
  $\LL' = \{B_1|_{\Psi_1}, \dotsc, B_k|_{\Psi_k}\}$ 
  provided that~$\Psi \neq \emptyset$.  This concludes the
  definition of~$G'$.

  For every derivation tree~$d \in L(G_\der, A)$ we define $\Psi(d) =
  \{j\in\num(A)\mid \val(d)|_j\neq x_1\}$.  Then, as already observed before, we
  have~$A|_{\Psi(d)} \in \N'$.  It is straightforward to
  verify that if~$d = \rho(\seq d1k)$, where $\rho$~is the rule of the
  previous paragraph, then the left-hand side of the
  rule~$\rho_{\Psi(d_1), \dotsc, \Psi(d_k)}$
  is~$A|_{\Psi(d)}$ because $\val(d)|_j = x_1$ if and only
  if~$u|_j = wx_1$ with~$w \in \{B_i|_\ell \mid i \in [k],\,
  \ell\in\num(B_i),\,\val(d_i)|_\ell = x_1\}^*$.  

  For every derivation tree~$d \in L(G_\der,
  A)$ there exists a derivation tree~$d' \in L(G'_\der,
  A|_{\Psi(d)})$ such that $\val(d') =
  \val(d)|_{\Psi(d)}$.  In fact, let $d = \rho(\seq d1k)$,
  and let $d'_i \in L(G'_\der, B_i|_{\Psi(d_i)})$ be a
  derivation tree such that $\val(d'_i) =
  \val(d_i)|_{\Psi(d_i)}$ for every~$i \in [k]$, which
  exist by the induction hypotheses.  By Lemma~\ref{lem:comm-assoc}(2)
  we have $\val(d') = \val(d)|_{\Psi(d)}$ for $d' =
  \rho_{\Psi(d_1), \dotsc, \Psi(d_k)}(\seq{d'}1k)$.  This shows
  that~$L(G) \subseteq L(G')$.  Clearly, $L_\Psi = \{d \in L(G_\der,
  A) \mid \Psi(d) = \Psi\}$ is a regular tree language for
  every~$\Psi$.  Thus, $d'$~can be computed from~$d$ by the one-state
  \LDTR"~transducer~$M$ with the rules
  \[ \langle q,\rho(y_1 \colon L_{\Psi_1}, \dotsc, y_k \colon
  L_{\Psi_k}) \rangle \to \rho_{\seq \Psi1k}(\langle q, y_1\rangle,
  \dotsc, \langle q,y_k\rangle) \enspace. \] 

  Vice versa, for every derivation tree~$d' \in L(G'_\der,
  A|_\Psi)$ there exists a derivation tree~$d \in L(G_\der,
  A)$ such that $M(d) = d'$~and~$\Psi = \Psi(d)$, where $A$~is
  uniquely determined by~$A|_\Psi$ because $G$~has
  disjoint big nonterminals.  In fact, let $d' = \rho'(\seq{d'}1k)$
  with $d'_i \in L(G'_\der, B_i|_{\Psi_i})$.  Then there
  exists a rule~$\rho$ as above such that $\rho' = \rho_{\seq
    \Psi1k}$.  Clearly, if~$d_i \in L(G_\der, B_i)$ such that $M(d_i)
  = d'_i$~and~$\Psi_i = \Psi(d_i)$, then $M(d) = d'$~and~$\Psi =
  \Psi(d)$ for~$d = \rho(\seq d1k)$.  Thus~$L(G') \subseteq L(G)$, and
  $d$~can be computed by an LDT"~transducer.
\end{proof}

\section{Lexicalization}
\label{sec:lex}
\noindent In this section, in Lemma~\ref{lem:main}, we present the
main lexicalization step, in which we lexicalize all non-monic
non-terminal rules.  It generalizes the transformation of a
context-free grammar into Operator Normal Form
(see~\cite[Theorem~1.2]{grahar72} and~\cite[Theorem~3.5]{aubebo97}).
We assume that $G$~is $\Delta$"~growing (see
Theorem~\ref{thm:dec-growing}).  Thus, all non-initial terminal rules
are doubly $\Delta$"~lexicalized and all monic rules are
$\Delta$"~lexicalized.  In the following we will simply write
`lexicalized' to mean `$\Delta$"~lexicalized'.

For a derivation tree~$d \in L(G_\der)$ and a position~$r \in \pos(d)$
such that $d(r)$~is a non-lexicalized rule of rank at least~$2$, we 
say that the ``source'' of~$r$ is the first position~$q$ in a
pre-order traversal of the second direct subtree of $r$ (i.e., the subtree 
at~$r2$) such that $d(q)$~is a doubly lexicalized rule.  Clearly, 
since every terminal rule at the leaves of~$d$ is doubly lexicalized,
such a position exists and can be found by only exploring the first
children of each visited node; i.e., $q = r21^m$ for some $m \in
\nat_0$.  The basic idea of the lexicalization construction is to
remove one lexical symbol~$\delta$ from the source~$q$ and
transport it to the ``target''~$r$.  Then $d(q)$~is still
lexicalized, and $d(r)$~has become lexicalized.  Note that different
targets have different sources, which is a simple fact that is well known
to be useful (cf.\@ \cite[Section~3]{pottho93} and~\cite[page~346]{hoopas97}).  The
transportation of~$\delta$ from the source node~$q$ to the target
node~$r$ is the task of the non-lexicalized or singly lexicalized
rules at the positions along the path from~$q$ to~$r$.  The required
relabeling of the derivation tree can be realized deterministically by an
\LDTR"~transducer that uses its look-ahead at~$r$ to determine the
node label~$d(q)$.  From the rewriting point of view
(Section~\ref{sub:deriv}), it is a guess-and-verify process.  We
guess~$\delta$ at position~$r$ and verify it at position~$q$.

\begin{example}
  \label{exa:newmain}
  \upshape
  As before, let~$\Delta = \{\alpha, \beta, \tau, \nu\}$.  Since the
  resulting grammar~$G$ in Example~\ref{exa:d-chain-free} can be
  lexicalized by simple substitution of rules (as discussed in 
  Example~\ref{exa:d-chain-free}), we consider another
  $\Delta$"~growing grammar, which is similar to the original grammar 
  of Example~\ref{exa:main}, but has an additional non-lexicalized
  rule~$A \to B(\gamma(A))$.  Moreover, we replace the rule~$\rho_4 =
  B(x_1) \to x_1$ by the two doubly lexicalized rules $B(x_1) \to
  \sigma(x_1, \alpha T_1(\beta \sigma(\alpha T_2, \gamma T_3)))$ and
  $B(x_1) \to \sigma(x_1, \sigma(\tau, \nu))$, which are taken from
  Example~\ref{exa:d-chain-free}.  The (big) nonterminal~$B'$ remains
  an alias of~$B$.  The resulting $\Delta$"~growing MCFTG, which we
  again call~$G$, has the following rules (renamed with respect to
  Example~\ref{exa:main}):  
  \begin{alignat*}{7}
    \rho_1 \colon && S &\to \alpha A
    & \rho_2 \colon && A &\to T_1(\sigma(B(T_2), T_3)) 
    & \rho_3 \colon && A &\to B(\gamma A) \\
    \rho_4 \colon && B(x_1) &\to \sigma(B(x_1), B'(A)) 
    & \rho_5 \colon && B(x_1) &\to \sigma(x_1, \alpha T_1(\beta
      \sigma(\alpha T_2, \gamma T_3)))
    & \rho_6 \colon && B(x_1) &\to \sigma(x_1, \sigma(\tau, \nu)) \\
    \rho'_4 \colon && B'(x_1) &\to \sigma(B(x_1), B'(A)) \quad 
    & \rho'_5 \colon && B'(x_1) &\to \sigma(x_1, \alpha T_1(\beta
      \sigma(\alpha T_2, \gamma T_3))) \quad
    & \rho'_6 \colon && B'(x_1) &\to \sigma(x_1, \sigma(\tau, \nu)) \\ 
    && && \rho_7 \colon && T &\to (\alpha T_1(\beta
      x_1),\, \alpha T_2,\,\gamma T_3) 
    & \rho_8 \colon && T &\to (x_1, \tau, \nu) 
  \end{alignat*}
  with $T = (T_1(x_1), T_2, T_3)$.  Rule~$\rho_1$ is singly
  lexicalized, whereas rules~$\rho_2$, $\rho_3$, $\rho_4$,
  and~$\rho'_4$ are non-lexicalized.  The remaining rules are doubly
  lexicalized.  We will remove the lexical symbol $\beta$~or~$\tau$
  from each doubly lexicalized rule that labels a source and transport
  it to the target.  For our derivation trees, we need to fix the
  order of the big nonterminals in the rules, so we let
  \[ \LL(\rho_2) = \{B, (T_1, T_2, T_3)\} \qquad \LL(\rho_3) = \{A,
  B\} \qquad \text{and} \qquad \LL(\rho_4) = \LL(\rho'_4) = \{B, B',
  A\} \enspace. \]
  Figure~\ref{fig:transport} shows a derivation tree of~$L(G_\der)$
  together with arrows indicating sources, corresponding targets, and
  transported lexical elements.  A transportation of~$\beta$ is marked
  by a dashed arrow, whereas a transport of~$\tau$ is marked by a
  dotted arrow. \fin
\end{example} 

\begin{figure}[t]
  \centering
  \includegraphics{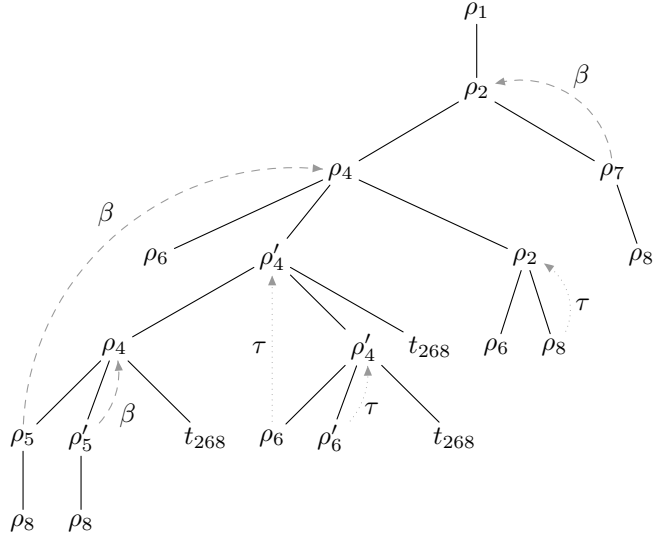}
  \caption{Derivation tree of~$L(G_\der)$ for the MCFTG~$G$ of
    Example~\protect{\ref{exa:newmain}} with indicated sources,
    targets, and transported lexical elements, where $t_{268} =
    \rho_2(\rho_6, \rho_8)$ with $\tau$ transported from 
    $\rho_8$ to $\rho_2$.}
  \label{fig:transport}
\end{figure}

We need some more terminology.  Let $\Omega$~be a ranked alphabet
(such as~$N \cup \Sigma$) and let $X_\infty=X\setminus\{\hole\}$;
i.e., $X_\infty=\{x_1,x_2,\dotsc\}$. 
For a finite subset~$Z$ of~$X_\infty$, if
$Z = \{\seq x{i_1}{i_n}\}$ with $n \in \nat_0$ and $i_1 <i_2 < \dotsb
< i_n$, then we define $\sequ(Z) = \word x{i_1}{i_n} \in X_\infty^*$, the
sequence of variables in~$Z$ with increasing indices.  A tree~$t$ in~$T_\Omega(X)$ 
is \emph{linear} if each variable occurs at most once
in it; i.e., $\abs{\pos_x(t)}\leq 1$ for every $x\in X$. 
For a linear tree~$t \in T_\Omega(X)$, we denote by~$\var(t)$
the set of variables $x_i$ that occur in~$t$; i.e.,~$\var(t) = \alp_{X_\infty}(t)$.
If $\sequ(\var(t)) = \word x{i_1}{i_n}$, then we define~$\ren(t) =
t[x_{i_j} \gets x_j \mid 1 \leq j \leq n]$, the \emph{renumbering}
of~$t$, which is a pattern in~$P_\Omega(X_n)$ if $\hole$ does not occur in $t$.  
Note that $t =
\ren(t)[x_j \gets x_{i_j} \mid 1 \leq j \leq n]$.  As an example,
if~$t = \sigma(x_4, \sigma(x_2, x_5))$ then~$\var(t) = \{x_2, x_4,
x_5\}$, $\sequ(\var(t)) = x_2x_4x_5$, and $\ren(t) = \sigma(x_2,
\sigma(x_1, x_3))$.  We will use the easy fact that if $h$~is a tree
homomorphism over~$\Omega$ and $t \in T_\Omega(X)$~is linear, then
$\hat{h}(t)$~is linear and $\var(\hat{h}(t)) = \var(t)$ by
Lemma~\ref{lem:treehom}(1), and $\ren(\hat{h}(t)) = \hat{h}(\ren(t))$
by Lemma~\ref{lem:treehomsub}.

To define contexts, we use the special variable~$\hole$.  
A \emph{context} is a tree~$t$ with exactly one
occurrence of~$\SBox$; i.e., $\abs{\pos_{\SBox}(t)} = 1$.  For a linear
context~$t \in T_\Omega(X)$ we define $\renh(t) =
\ren(t)[\SBox \gets x_{n+1}]$, where $n = \abs{\var(t)}$.  Note that
$\renh(t)$~is a pattern in~$P_\Omega(X_{n+1})$.  The above easy fact 
also holds for contexts: $\hat{h}(t)$~is a linear context
and, by Lemma~\ref{lem:treehomsub} again, $\renh(\hat{h}(t)) =
\hat{h}(\renh(t))$.

For a tree~$t \in T_\Omega(X_k)$ and a position~$p \in
\pos_\Omega(t)$, there exist a unique context~$c \in T_\Omega(X_k \cup
\{\SBox\})$ and a unique tree~$u \in T_\Omega(X_k)$
such that~$\pos_{\SBox}(c) = \{p\}$ and~$t = c[\SBox \gets u]$.  The context~$c$ is
called the \emph{$p$"~context} of~$t$ and denoted by~$t|^p$, and the
tree~$u$ is called the \emph{subtree} of~$t$ at~$p$ and denoted
by~$t|_p$.  If~$p \in \pos_\omega(t)$ with~$\rk(\omega) = m$, then $t
= t|^p[\SBox \gets \omega(\seq{t|}{p1}{pm})]$.  Let $h$~be a tree
homomorphism over~$\Omega$.  By Lemma~\ref{lem:treehomsub},
$\hat{h}(c[\SBox \gets u]) = \hat{h}(c)[\SBox \gets \hat{h}(u)]$.
Thus, if $\pos_{\SBox}(\hat{h}(t|^p)) = \{\hat p\}$, then
$\hat{h}(t|^p) = \hat{h}(t)|^{\hat p}$ and $\hat{h}(t|_p) =
\hat{h}(t)|_{\hat p}$.  Moreover,  if $p \in
\pos_\omega(t)$~and~$h(\omega) = \init(\omega)$, then $\hat p \in
\pos_\omega(\hat{h}(t))$ and $\hat{h}(t|_{pi}) = \hat{h}(t)|_{\hat pi}$ for
every~$i \in [m]$. 

\begin{lemma}
  \label{lem:main}
  For every $\Delta$"~growing MCFTG~$G$ there is an \LDTR"~equivalent
  $\Delta$"~lexicalized MCFTG~$G'$.
\end{lemma}

\begin{proof}
  Let $G = (N, \N, \Sigma, S, R)$ be a $\Delta$"~growing MCFTG.  We
  can assume that all its terminal rules are doubly lexicalized
  because initial terminal rules can be removed from~$G$ and added
  after lexicalization.  Moreover, for technical convenience, we
  assume that there is a subset~$\Deltadl$ of~$\Delta$ such
  that (1)~for every doubly lexicalized rule~$A \to (u, \LL)$ there is
  a lexical symbol~$\delta \in \Deltadl$ that occurs exactly
  once in~$u$, and (2)~for every singly lexicalized rule~$A \to (u,
  \LL)$, the lexical symbol that occurs in~$u$ is not an element
  of~$\Deltadl$.  This can be assumed because we could even
  assume that $G$~is uniquely terminal labeled as defined after
  Lemma~\ref{lem:cover}.  In fact, as observed there, $G$~has a
  cover~$(G_{\text u}, h)$ such that $G_\textrm{u} = (N, \N,
  \Sigma_{\text u}, S, R_{\text u})$~is uniquely terminal labeled.  If
  we let~$\Delta_{\text u} = \{\sigma \in \Sigma_{\text u} \mid
  h(\sigma) \in \Delta\}$, then $G_{\text u}$~is $\Delta_{\text
    u}$"~growing.  Let $G'_{\text u}$~be a $\Delta_{\text
    u}$"~lexicalized MCFTG that is \LDTR"~equivalent to~$G_{\text u}$,
  and let $G' = (G'_{\text u})_h$; i.e., $G'$~is the unique MCFTG such
  that $(G'_{\text u}, h)$~is a cover of~$G'$.  Then $G'$~is
  $\Delta$"~lexicalized.  Moreover, $G$~is \LDTR "~$\hat h$"~equivalent
  to~$G_{\text u}$ and $G'$~is \LDTR "~$\hat h$"~equivalent
  to~$G'_{\text u}$, 
  by Lemma~\ref{lem:cover}.  Consequently, we can conclude that
  $G$~and~$G'$ are \LDTR"~equivalent.  This shows that we could even
  assume that $G$~is uniquely terminal labeled.  However we do not do so,
  because we wish to illustrate the construction in this proof on the
  grammar~$G$ of Example~\ref{exa:newmain}, for which
  $\Deltadl = \{\beta, \tau\}$.

  For every doubly lexicalized rule~$\rho = A \to (u, \LL)$ of~$G$,
  let $\lex(\rho) \in \Deltadl$ be a fixed lexical symbol
  that occurs exactly once in~$u$.  In the grammar~$G'$ to be
  constructed, this symbol will possibly be removed from~$u$, leaving
  a rule that is still lexicalized.  

  We let 
  \[ N_\new = \{\langle C, \delta, i, Z\rangle \mid C \in N,\, \delta \in
  \Deltadl,\, 0 \leq i \leq \rk(\delta),\, Z \subseteq X_{\rk(C)}\} \]
  be a set of new nonterminals such that $\rk(\langle C, \delta, 0,
  Z\rangle) = \abs{Z} + 1$ and $\rk(\langle C, \delta, i, Z\rangle) =
  \abs{Z}$ for every~$i \in [\rk(\delta)]$.  The grammar~$G'$ will 
  have the set of nonterminals~$N'=N \cup N_\new$.

  Let us provide some intuition for these new nonterminals.  We first 
  observe that for every derivation tree~$d \in L(G_\der, A)$ there is
  a natural label-preserving bijection~$\tau_d$ between the sets 
  $\pos_\Sigma(\val(d))$ and~$\bigcup_{q \in \pos(d)}
  \bigl(\{q\} \times \pos_\Sigma(\rhs(d(q))) \bigr)$; i.e., between the set of
  terminal positions of~$\val(d)$ and the disjoint union of the sets of terminal positions of 
  the right-hand sides of the rules that occur in~$d$, cf. Lemma~\ref{lem:valocc}(1).  For 
  positions~$q\in\pos(d)$ and $p\in\pos_\Sigma(\rhs(d(q)))$, let $\tau_d(q, p)$~be the corresponding
  position in~$\pos_\Sigma(\val(d))$.  Since $\tau_d$~is only needed in 
  this paragraph, we do not give its
  straightforward, but tedious, definition.  The existence of~$\tau_d$
  should be intuitively clear, and can be proved by induction on the
  structure of~$d$; the induction step is based on the fact that for a
  tree homomorphism~$h$ over~$N \cup \Sigma$ and a forest~$u$, there
  is a natural label-preserving bijection between the sets
  $\pos_\Sigma(\hat{h}(u))$ and~$\bigcup_{q \in \pos(u)}
  \bigl(\{q\} \times \pos_\Sigma(h(u(q))) \bigr)$, cf.\@
  Lemma~\ref{lem:treehom}(2).  Now, roughly speaking, the intuition for the
  new nonterminals is the following.  Consider a derivation tree~$d
  \in L(G_\der, A)$, and let $q$~be the shortest position of~$d$ 
  of the form $1^m$ for some $m\in\nat_0$
  such that $\rho = d(q)$~is doubly lexicalized.
  Thus, $q$ is a potential source for a target that has $d$ 
  as its second direct subtree (in some other derivation tree).
  Let $\lex(\rho) = \delta$, and let $p \in \pos_\Sigma(\rhs(\rho))$
  be the unique $\delta$"~labeled position of the right-hand side of
  rule~$\rho$.  Moreover, suppose that the corresponding $\delta$"~labeled 
  position~$\tau_d(q, p)$ of~$\val(d)$ belongs to the $j$"~th tree~$t$ of the
  forest~$\val(d)$ with $1 \leq j \leq \abs{A}$; i.e., $\tau_d(q, p) =
  \#^{j-1}p'$ with $p' \in \pos(t)$.  Let the nonterminal~$C$ be the
  $j$"~th element of the big nonterminal~$A$.  Thus, $C$~(as part
  of~$A$) generates (in~$G$) the terminal tree~$t$.  Then $\langle C,
  \delta, 0, Z_0\rangle$ generates (in~$G'$) the $p'$"~context of~$t$
  (with $\SBox$ at position~$p'$), 
  and $\langle C, \delta, i, Z_i\rangle$ generates the subtree of~$t$ at~$p'i$ 
  for every~$i \in [\rk(\delta)]$.  The sets $Z_0$~and~$Z_i$ consist of
  the variables that occur in that context and that subtree, so $Z_0 =
  \var(t|^{p'})$~and~$Z_i = \var(t|_{p'i})$.  To be more precise,
  $\langle C, \delta, 0, Z_0\rangle$~generates~$\renh(t|^{p'})$ and
  $\langle C, \delta, i, Z_i\rangle$~generates~$\ren(t|_{p'i})$. 

  We now continue the formal proof.  For a nonterminal~$C \in N$, we
  say that the triple~$(C, \delta, Z)$ is a \emph{skeleton} of~$C$
  if $\delta \in \Deltadl$~and~$Z = (Z_0,\seq Z1m)$, where $m =
  \rk(\delta)$ and $\{Z_0,\seq Z1m\}$ is a partition of~$X_{\rk(C)}$.\footnote{Recall from
the beginning of Section~\protect{\ref{sec:prelim}} that we allow the empty set to be an element of a partition.
Thus, we allow $Z_i=\emptyset$.}  
  For such a skeleton, we will denote by~$\tree(C, \delta, Z)$ the tree 
\[\langle C,\delta,0,Z_0\rangle \sequ(Z_0) 
\;\delta(\langle C,\delta,1,Z_1\rangle \sequ(Z_1),\dots,\langle C,\delta,m,Z_m\rangle \sequ(Z_m))\]
  of which we observe (for clearness sake) that it looks as follows:
  
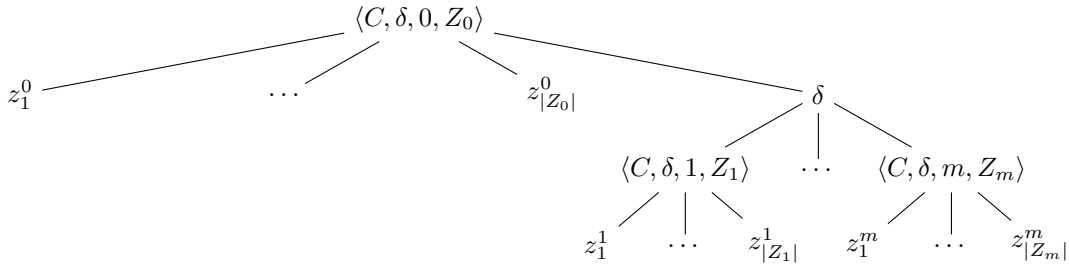
\begin{center}
    \begin{tikzpicture}[level/.style={sibling distance=35mm/#1,
        level distance=10mm}]
      \node {$\langle C, \delta, 0, Z_0\rangle$} child {node {$z^0_1$}
      } child {node {$\dots$} } child {node {$z^0_{\abs{Z_0}}$} }
      child {node {$\delta$} child {node
          {$\langle C, \delta, 1, Z_1\rangle$} child {node {$z^1_1$} }
          child {node {$\dots$} } child {node {$z^1_{\abs{Z_1}}$} } }
        child {node {$\dots$} } child {node
          {$\langle C, \delta, m, Z_m \rangle$} child {node {$z^m_1$}
          } child {node {$\dots$} } child {node {$z^m_{\abs{Z_m}}$} }
        } };
    \end{tikzpicture}
  \end{center}
  where $\sequ(Z_i) = \word{z^i}1{\abs{Z_i}}$ for every $0 \leq i
  \leq m$.  Note that $\tree(C, \delta, Z) \in P_{N_ \new \cup
    \{\delta\}}(X_{\rk(C)})$.  Moreover, we will denote 
  $\yield_{N_\new}(\tree(C, \delta, Z))$ by~$\sequ(C,\delta, Z)$; i.e., 
  $\sequ(C,\delta, Z)$ is the sequence 
  \[\langle C, \delta, 0, Z_0 \rangle \langle
  C, \delta, 1, Z_1 \rangle \dotsm \langle C, \delta, m, Z_m \rangle \enspace.\] 
  Obviously, the skeleton $(C, \delta, Z)$ can be reconstructed from $\sequ(C,\delta, Z)$,
  and hence from $\tree(C, \delta, Z)$. 

To motivate $\tree(C, \delta, Z)$ and $\sequ(C,\delta, Z)$, we observe that 
for every pattern $t \in P_{N'\cup\Sigma}(X_{\rk(C)})$ 
and every $\delta$-labeled position $p'$ of $t$ (i.e., $p' \in \pos_\delta(t)$), 
the pattern $t$ can be decomposed as
  \[t = \tree(C, \delta, Z)[\sequ(C, \delta, Z) \gets (t_0,\seq t1m)] \enspace, \] 
where $t_0 = \renh(t|^{p'})$ is the renumbered $p'$"~context and $Z_0 =
  \var(t|^{p'})$ is the set of its variables before renumbering, and moreover, for
  every~$i \in [m]$, $t_i = \ren(t|_{p'i})$ is the renumbered subtree at $p'i$
  and~$Z_i = \var(t|_{p'i})$ is the set of its variables before renumbering.
Intuitively, $\tree(C, \delta, Z)$ can be viewed as the ``skeleton'' of this decomposition,
which was our reason to call $(C, \delta, Z)$ a skeleton of $C$.  

  We let $\N_\new$~be the new set of big nonterminals of the form
  $\beta \cdot \sequ(C, \delta, Z) \cdot \gamma$, where $\beta C
  \gamma \in \N$ with $C \in N$ and $\beta, \gamma \in N^*$, and
  $(C, \delta, Z)$ is a skeleton of $C$.  We now construct the new MCFTG
  $G' = (N', \N', \Sigma, S, R')$ 
  with $N'=N\cup N_\new$ and $\N'=\N\cup \N_\new$.  To define the
  set~$R'$ of rules of~$G'$, we first define an auxiliary MCFTG 
  $G_+ = (N', \N', \Sigma, S, R\cup R_+)$ where $R_+$ 
  is a set of new rules
  that, intuitively, realize the transport of a lexical symbol from a
  source to a target (but not yet its arrival at the target). 

  For every doubly lexicalized rule~$\rho = \word A1n \to ((\seq u1n),
  \LL)$ of~$G$ with $\LL = \{\seq B1k\}$ (in that order), 
  $A_i \in N$, and $u_i \in P_{N \cup \Sigma}(X_{\rk(A_i)})$, 
  we define a skeleton~$\skel(\rho)$ and a new rule~$\overline{\rho}$ in $R_+$
  as follows.  Let $\delta = \lex(\rho)$ and let $\#^{j-1}p$~be the unique
  $\delta$"~labeled position of~$(\seq u1n)$, so~$j\in[n]$ and~$\pos_\delta(u_j)
  = \{p\}$.  Moreover, let $u = u_j$, 
  $\rk(\delta) = m$, and $Z = (Z_0,\seq Z1m)$ with $Z_0 =
  \var(u|^p)$ and $Z_i = \var(u|_{pi})$ for every~$i \in [m]$.  Then
  we define~$\skel(\rho) = (A_j, \delta, Z)$.  Note that $u \in P_{N \cup
    \Sigma}(X_{\rk(A_j)})$~and hence $(A_j, \delta, Z)$~is a skeleton
  of~$A_j$.  Additionally, we define the rule
  \begin{align*}
    \overline{\rho} = {} 
    & \phantom{{}\to{}} \word A1{j-1} \cdot \sequ(A_j, \delta, Z) \cdot
      \word A{j+1}n  \\
    &{} \to ((\seq u1{j-1}, v_0, \seq v1m, \seq u{j+1}n),\, \LL
      ) \enspace,
  \end{align*} 
  where $v_0 = \renh(u|^p)$~and~$v_i = \ren(u|_{pi})$ for every~$i \in
  [m]$ (and~$\LL = \{\seq B1k\}$, in the same order).  Clearly,
  $\overline{\rho}$~is lexicalized because $\abs{\pos_\Delta((\seq
    u1n))} \geq 2$ and $\abs{\pos_\Delta((\seq v0m))} =
  \abs{\pos_\Delta(u)} - 1$.

  For every non-lexicalized or singly lexicalized rule $\rho = \word
  A1n \to ((\seq u1n), \LL)$ of~$G$ with $\LL = \{\seq B1k\}$
  and $k \geq 1$, and for every skeleton~$(C, \delta, W)$ 
  such that $C \in \alp(B_1)$, we define a skeleton~$\skel(\rho, (C,
  \delta, W))$ and a new rule~$\rho_{C, \delta, W}$ in $R_+$ as follows.  Let
  $j \in [n]$~be the unique integer such that~$C \in \alp_N(u_j)$, and
  let $u' = u_j[C \gets \tree(C, \delta, W)]$.  Moreover, let
  $\rk(\delta) = m$, $\pos_\delta(u') = \{p\}$,\footnote{Note that by
    our second assumption on~$G$, the symbol~$\delta$ does not occur
    in~$u_j$ because~$\delta \in \Deltadl$ and $\rho$~is
    non-lexicalized or singly lexicalized.} 
  and $Z = (Z_0,\seq Z1m)$ with
  $Z_0 = \var(u'|^p)$ and $Z_i = \var(u'|_{pi})$ for every~$i \in
  [m]$.  Then we define $\skel(\rho, (C, \delta, W)) = (A_j, \delta,
  Z)$.  Let $B_1 = \beta C \gamma$ for some~$\beta, \gamma \in N^*$,
  which are unique because $B_1$~is repetition-free.  Then we define
  the rule
  \begin{align*}
    \rho_{C, \delta, W} = {}
    &\phantom{{}\to{}} \word A1{j-1} \cdot \sequ(A_j, \delta, Z) \cdot
      \word A{j+1}n  \\
    &{} \to ((\seq u1{j-1}, v'_0, \seq{v'}1m, \seq u{j+1}n),\, \LL')
      \enspace,
  \end{align*} 
  where $v'_0 = \renh(u'|^p)$ and $v'_i = \ren(u'|_{pi})$ for every~$i
  \in [m]$.  Additionally, $\LL' = \{B'_1, \seq B2k\}$ with $B'_1 =
  \beta \cdot \sequ(C, \delta, W) \cdot \gamma$.  Note that
  $\rho_{C, \delta, W}$~is non-lexicalized or singly lexicalized,
  respectively, because $\abs{\pos_\Delta((v'_0, \seq{v'}1m))} =
  \abs{\pos_\Delta(u')} - 1 = \abs{\pos_\Delta(u_j)}$.

  These are all the rules of $R_+$. Thus, $G_+$~is the grammar 
  obtained from~$G$ by adding all
  the above new rules $\overline{\rho}$~and~$\rho_{C, \delta, W}$
  to~$R$.  It is straightforward to check that from the
  rule~$\overline{\rho}$ the original rule~$\rho$ can be
  reconstructed, and similarly, from~$\rho_{C, \delta, W}$ we can
  reconstruct both $\rho$~and~$(C, \delta, W)$.  Note that all
  terminal and all monic rules of~$G_+$ are lexicalized.  

  We now define the set~$R'$ of rules of~$G'$.  First, $R'$~contains all
  lexicalized rules of~$G_+$.  Second, we define rules that 
  realize the arrival of a lexical symbol $\delta'$ at a target. 
  Let $\rho = A \to (u,
  \LL)$ be a non-lexicalized rule of~$G_+$ with $\LL =
  \{\seq B1k\}$, where~$k \geq 2$, $B_1 \in \N \cup \N_\new$, and~$B_i \in
  \N$ for~$2 \leq i \leq k$.  For every skeleton~$(C', \delta', Z)$
  such that $C' \in \alp(B_2)$, we define the new rule~$\langle \rho 
  \rangle_{C', \delta', Z}$ in~$R'$ as follows.  Let $B_2
  = \beta C' \gamma$ with~$C' \in N$ and~$\beta, \gamma \in N^*$, which
  are again unique because $B_2$~is repetition-free.  Then $\langle
  \rho \rangle_{C', \delta', Z} = A \to (u', \LL')$, where $u' = u[C'
  \gets \tree(C', \delta', Z)]$ and $\LL' = \{B_1, B'_2, \seq B3k\}$
  with $B'_2 = \beta \cdot \sequ(C', \delta', Z) \cdot \gamma$.
  Clearly, $\langle \rho \rangle_{C', \delta', Z}$ is lexicalized
  because $\delta'$~occurs in its right-hand side.  It is easy to check
  that from the rule~$\langle \rho \rangle_{C', \delta', Z}$ we can
  reconstruct both $\rho$~and~$(C', \delta', Z)$.  Thus, $R'$~consists
  of:
  \begin{compactitem}
  \item all lexicalized rules~$\rho$ of~$G$, 
  \item all rules~$\overline{\rho}$, where $\rho$~is a doubly
    lexicalized rule of~$G$,  
  \item all rules~$\rho_{C, \delta, W}$, where $\rho$~is a singly
    lexicalized rule of~$G$, and
  \item all rules $\langle \rho \rangle_{C', \delta', Z}$~and~$\langle
    \rho_{C, \delta, W} \rangle_{C', \delta', Z}$, where $\rho$~is a
    non-lexicalized rule of~$G$.
  \end{compactitem}
  This ends the construction of~$G'$.  It remains to show that
  $G$~and~$G'$ are \LDTR"~equivalent.  We first show how to transform
  the derivation trees of~$G$ into those of~$G'$.  We start by
  defining a skeleton for every derivation tree of~$G$.

  For every derivation tree~$d \in L(G_\der, A)$ we define a
  skeleton~$\skel(d) = (C, \delta, Z)$ with~$C \in \alp(A)$ (and
  $\delta = \lex(\rho)$ for the label~$\rho$ of the shortest position
  of~$d$ of the form $1^m$ such that $\rho$~is doubly
  lexicalized).  The definition is by induction on the structure of~$d
  = \rho(\seq d1k)$.  If $\rho$~is a doubly lexicalized rule (in
  particular if~$k = 0$), then we define~$\skel(d) = \skel(\rho)$ as
  defined above.  Otherwise $\rho$~is not doubly lexicalized (and
  so~$k \geq 1$); let $\rho = A \to (u, \LL)$ with~$\LL = \{\seq
  B1k\}$.  By the induction hypothesis we have~$\skel(d_1) = (C,
  \delta, W)$, where~$C \in \alp(B_1)$.  Then we define~$\skel(d) =
  \skel(\rho, (C, \delta, W))$ as defined above.  Clearly, for every
  skeleton~$(C, \delta, Z)$, the set of derivation trees
  \[ L_{C, \delta, Z} = \{d \in \bigcup_{A \in \N} L(G_\der, A) \mid
  \skel(d) = (C, \delta, Z)\} \]
  is a regular tree language, which can be recognized by a
  deterministic bottom-up finite tree automaton using all skeletons
  as states.
 
  For every derivation tree~$d \in L(G_\der, A)$ we define two
  derivation trees $\dtr_1(d)$~and~$\dtr_2(d)$ of~$G'$ with $\dtr_1(d)
  \in L(G'_\der, A)$~and~$\dtr_2(d) \in L(G'_\der, \beta \cdot
  \sequ(C, \delta, Z) \cdot \gamma)$, where $\skel(d) = (C, \delta,
  Z)$~and~$A = \beta C \gamma$ with~$\beta, \gamma\in N^*$.  These two
  derivation trees are relabelings of~$d$.  They are defined by
  induction on the structure of~$d = \rho(\seq d1k)$.
  \begin{compactitem}
  \item If $\rho$~is a doubly lexicalized rule (in particular, if~$k =
    0$), then we define
    \begin{align*}
      \dtr_1(d) 
      &= \rho(\dtr_1(d_1), \dotsc, \dtr_1(d_k)) \\
      \dtr_2(d) 
      &= \overline{\rho}(\dtr_1(d_1), \dotsc, \dtr_1(d_k)) \enspace.
    \end{align*}
  \item Now let $\rho = A \to (u, \LL)$~be a rule with~$\LL = \{\seq
    B1k\}$ that is not doubly lexicalized (and hence~$k \geq 1$).
    Moreover, let~$\skel(d_1) = (C, \delta, W)$, where~$C \in
    \alp(B_1)$.
    \begin{compactitem}
    \item If $\rho$~is singly lexicalized, then we define
      \begin{align*}
        \dtr_1(d) 
        &= \rho(\dtr_1(d_1), \dotsc, \dtr_1(d_k)) \\
        \dtr_2(d) 
        &= \rho_{C, \delta, W}(\dtr_2(d_1), \dtr_1(d_2), \dotsc,
          \dtr_1(d_k)) \enspace.
      \end{align*}
    \item If $\rho$~is non-lexicalized, and thus~$k \geq 2$, then we
      let~$\skel(d_2) = (C', \delta', Z)$, where $C' \in \alp(B_2)$,
      and we define
      \begin{align*}
        \dtr_1(d) 
        &= \langle \rho \rangle_{C', \delta', Z}(\dtr_1(d_1),
          \dtr_2(d_2), \dtr_1(d_3), \dotsc, \dtr_1(d_k)) \\
        \dtr_2(d) 
        &= \langle \rho_{C, \delta, W} \rangle_{C', \delta',
          Z}(\dtr_2(d_1), \dtr_2(d_2), \dtr_1(d_3), \dotsc,
          \dtr_1(d_k)) \enspace.
      \end{align*}
    \end{compactitem}
  \end{compactitem}

  Clearly, there is an \LDTR"~transducer~$M$ that transforms~$d \in
  L(G_\der)$ into~$\dtr_1(d) \in L(G'_\der)$.  It has states
  $q_1$~and~$q_2$ with initial state~$q_1$, and it uses the regular
  tree languages~$L_{C, \delta, Z}$ as look-ahead.  It has the
  following rules, corresponding directly to the above definitions,
  where $\langle q_1, \word yik\rangle$ abbreviates $\langle q_1, y_i
  \rangle, \dotsc, \langle q_1, y_k \rangle$ for~$i \in [k]$: 
  \begin{compactitem}
  \item for every doubly lexicalized rule~$\rho$
    \begin{align*}
      \langle q_1,\, \rho(\seq y1k) \rangle 
      &\to \rho(\langle q_1, \word y1k \rangle) \\
      \langle q_2,\, \rho(\seq y1k) \rangle 
      &\to \overline{\rho}(\langle q_1, \word y1k \rangle)
    \end{align*}
  \item for every singly lexicalized rule~$\rho$ and every
    skeleton~$(C, \delta, W)$
    \begin{align*}
      \langle q_1,\, \rho(\seq y1k) \rangle 
      &\to \rho(\langle q_1, \word y1k \rangle) \\
      \langle q_2,\, \rho(y_1 \colon L_{C, \delta, W}, \seq y2k)
        \rangle 
      &\to \rho_{C, \delta, W}(\langle q_2, y_1\rangle, \langle q_1,
        \word y2k \rangle)
    \end{align*}
  \item and for every non-lexicalized rule~$\rho$ and all skeletons
    $(C', \delta', Z)$~and~$(C, \delta, W)$
    \begin{align*}
      \langle q_1,\, \rho(y_1, y_2 \colon L_{C', \delta', Z}, \seq
      y3k) \rangle 
      &\to \langle \rho \rangle_{C', \delta', Z}(\langle q_1, y_1
        \rangle, \langle q_2, y_2 \rangle, \langle q_1, \word y3k
        \rangle) \\
      \langle q_2,\, \rho(y_1 \colon L_{C, \delta, W}, y_2 \colon
      L_{C', \delta', Z}, \seq y3k) \rangle 
      &\to \langle \rho_{C, \delta, W} \rangle_{C', \delta',
        Z}(\langle q_2, y_1 \rangle, \langle q_2, y_2 \rangle, \langle
        q_1, \word y3k\rangle) \enspace.
    \end{align*}
  \end{compactitem}
  We will prove below that $d$~and~$\dtr_1(d)$ have the same value.
  However, to express the relationship between
  $\val(d)$~and~$\val(\dtr_2(d))$, we need the following definition.
  Let $A \in \N$ be a big nonterminal and $(C, \delta, Z)$~be a
  skeleton such that~$A = \beta C \gamma$ for some~$\beta,
  \gamma \in N^*$.  Moreover, let $s$~and~$s'$ be forests
  in~$P_\Sigma(X)^{\scriptscriptstyle +}$ such that~$\rk(s) = \rk(A)$
  and $s = \zeta t \eta$ for some~$\zeta, \eta \in
  P_\Sigma(X)^*$ with $\abs{\zeta} = \abs{\beta}$~and~$t \in
  P_\Sigma(X_{\rk(C)})$.  We note that~$\beta$, $\gamma$, $\zeta$,
  $t$, and~$\eta$ are unique given~$A$, $C$, and~$s$.  We say that
  \emph{$s'$~decomposes~$s$ for $A$~and~$(C, \delta, Z)$} if there
  exists a position~$p' \in \pos_\delta(t)$ such that $s' = \zeta 
  \cdot (t_0,\seq t1m) \cdot \eta$, where~$m = \rk(\delta)$, $t_0 =
  \renh(t|^{p'})$, $Z_0 = \var(t|^{p'})$, and $t_i =
  \ren(t|_{p'i})$~and~$Z_i = \var(t|_{p'i})$ for every~$i \in [m]$.

  We now prove by induction on the structure of~$d \in L(G_\der, A)$ that  
  \begin{compactenum}[\indent (i)]
  \item $\val(\dtr_1(d)) = \val(d)$ and
  \item $\val(\dtr_2(d))$~decomposes~$\val(d)$ for $A$~and~$\skel(d)$.
  \end{compactenum}
  Let~$d = \rho(\seq d1k)$ and suppose that (i)~and~(ii) hold for
  $\seq d1k$.
  \begin{compactitem}
  \item We first consider the case where $\rho$~is \emph{doubly
      lexicalized}.  Since (i)~is obvious from the definition
    of~`$\val$' and by the induction hypotheses, it remains to
    prove~(ii).  Let $\rho$~be as in the definition
    of~$\overline{\rho}$, and let us adopt the terminology in that
    definition.  Abbreviating~$[B_i \gets \val(d_i) \mid 1 \leq i \leq
    k]$ by~$[f]$, we obtain that
    \begin{align*}
      \val(d) 
      &= (\seq u1{j-1}, u, \seq u{j+1}n)[f] = \zeta t \eta
      \\
      \val(\dtr_2(d)) 
      &= (\seq u1{j-1}, v_0,\seq v1m, \seq u{j+1}n)[f] = \zeta \cdot
        (t_0,\seq t1m) \cdot \eta \enspace,
    \end{align*}
    where the first equality in the second line uses the induction
    hypotheses and where we define $\zeta = (\seq u1{j-1})[f]$, $t = u[f]$, $\eta =
    (\seq u{j+1}n)[f]$, and~$t_i = v_i[f]$ for every~$0 \leq i \leq
    m$.  We know that~$v_0 = \renh(u|^p)$, $Z_0 = \var(u|^p)$, and
    $v_i = \ren(u|_{pi})$~and~$Z_i = \var(u|_{pi})$ for every~$i \in
    [m]$.  It remains to show that a position~$p' \in \pos_\delta(t)$
    exists with $t_0 = \renh(t|^{p'})$, $Z_0 = \var(t|^{p'})$, and
    $t_i = \ren(t|_{p'i})$~and~$Z_i = \var(t|_{p'i})$ for every~$i \in
    [m]$.  We select the unique position~$p' \in
    \pos_{\SBox}((u|^p)[f])$.  Then, using the easy facts that are
    stated before this lemma (for the tree homomorphism corresponding
    to~$[f]$), we obtain that~$(u|^p)[f] = u[f]|^{p'} = t|^{p'}$ with~$p' \in
    \pos_\delta(t)$, and~$(u|_{pi})[f] = u[f]|_{p'i} = t|_{p'i}$, and so 
    \begin{align*}
      t_0 
      &= v_0[f] = \renh(u|^p)[f] = \renh(u|^p[f]) = \renh(t|^{p'}) \\
      Z_0
      &= \var(u|^p) = \var(u|^p[f]) = \var(t|^{p'}) \enspace,  
    \end{align*}
    and similarly for $t_i = v_i[f]$~and~$Z_i$ for every~$i \in
    [m]$. 
  \item Next we consider the case where $\rho$~is
    \emph{non-lexicalized}, and we prove~(i).  Let $\rho$~be as in the
    definition of~$\langle \rho \rangle_{C', \delta', Z}$
    with~$\skel(d_2) = (C', \delta', Z)$, where $C' \in \alp(B_2)$, and
    let us adopt the terminology found there.  By definition, 
    $\dtr_1(d) = \langle \rho \rangle_{C', \delta', Z}(\dtr_1(d_1), 
    \dtr_2(d_2), \dtr_1(d_3), \dotsc, \dtr_1(d_k))$.  Hence, 
    \[ \val(\dtr_1(d)) = u[C' \gets \tree(C', \delta', Z)]
    [B_1B'_2\word B3k \gets \val(d_1) \val(\dtr_2(d_2)) \val(d_3)
    \dotsm \val(d_k)] \enspace. \]
    We know that $B_2 = \beta C' \gamma$~and~$B'_2 = \beta \cdot
    \sequ(C', \delta', Z) \cdot \gamma$.  Let $\val(d_2) = \zeta 
    t \eta$ with~$\abs{\zeta} = \abs{\beta}$.  By~(ii)
    for~$d_2$, there exists~$p' \in \pos_\delta(t)$ such that
    $\val(\dtr_2(d_2)) = \zeta \cdot (t_0,\seq t1m) \cdot \eta$, where $m
    = \rk(\delta)$, $t_0 = \renh(t|^{p'})$, $Z_0 = \var(t|^{p'})$, and
    $t_i = \ren(t|_{p'i})$~and~$Z_i = \var(t|_{p'i})$ for every~$i \in
    [m]$.  By Lemmas
    \ref{lem:comm-assoc}(2)~and~\ref{lem:comm-assoc}(4), we now obtain
    that
    \[ \val(\dtr_1(d)) = u \bigl[C' \gets \tree(C', \delta',
    Z)[\sequ(C', \delta', Z) \gets (t_0,\seq t1m)] \bigr]\,[g]
    \enspace, \]
    where~$[g] = [B_1 \cdot \beta \gamma \cdot \word B3k \gets
    \val(d_1) \cdot \zeta \eta \cdot \val(d_3) \dotsm
    \val(d_k)]$.  As observed earlier (in the paragraph after the definition of 
    `$\tree$'~and~`$\sequ$'),  
    \[ \tree(C', \delta', Z)[\sequ(C', \delta', Z) \gets (t_0,\seq t1m)] =
    t \] 
    and so, again by Lemmas
    \ref{lem:comm-assoc}(2)~and~\ref{lem:comm-assoc}(4), 
    \[ \val(\dtr_1(d)) = u[C' \gets t]\,[g] = u[B_1 \cdot \beta C'
    \gamma \cdot \word B3k \gets \val(d_1) \cdot \zeta t 
    \eta \cdot \val(d_3) \dotsm \val(d_k)] \enspace, \] 
    which equals~$u[B_i \gets \val(d_i) \mid 1 \leq i \leq k] =
    \val(d)$.
  \item Next we consider the case where the rule~$\rho$ is
    \emph{singly lexicalized}.  Again, (i)~is obvious, so it remains
    to prove~(ii).  Let $\rho$~be as in the definition of~$\rho_{C,
      \delta, W}$, and let us adopt the terminology there.  Note that
    $\rho = A \to ((\seq u1n), \LL)$ with $\LL = \{\seq
    B1k\}$~and~$B_1 = \beta C \gamma$.  Consider the auxiliary new
    rule~$\rho' = A \to ((\seq u1{j-1}, u', \seq u{j+1}n), \LL')$ with
    $\LL' = \{B'_1, \seq B2k\}$ and $B'_1 = \beta \cdot \sequ(C,
    \delta, W) \cdot \gamma$.  This rule~$\rho'$ is analogous to the
    rule~$\langle \rho \rangle_{C, \delta, W}$, except that $C$~occurs
    in~$B_1$ instead of~$B_2$ (and $\rho$~is singly lexicalized
    instead of non-lexicalized).  However, we can prove~$\val(d') =
    \val(d)$ exactly as in the previous case, where $d' =
    \rho'(\dtr_2(d_1), \dtr_1(d_2), \dotsc, \dtr_1(d_k))$. 
    Also, the rule~$\rho_{C, \delta, W}$ is analogous to the
    rule~$\overline{\rho'}$, if we define $\lex(\rho') = \delta$.  In
    the first (doubly lexicalized) case we have shown that the value
    of~$\overline{\rho}(\dtr_1(d_1), \dotsc, \dtr_1(d_k))$ decomposes
    the value of~$d = \rho(\seq d1k)$ for $A$~and~$\skel(d)$ under the
    assumption that $\dtr_1(d_i)$~and~$d_i$ have the same value.  In
    exactly the same way we can prove here that the value of~$\rho_{C,
      \delta, W}(\dtr_2(d_1), \dtr_1(d_2), \dotsc, \dtr_1(d_k))$
    decomposes the value of~$d' = \rho'(\dtr_2(d_1), \dtr_1(d_2),
    \dotsc, \dtr_1(d_k))$ for $A$~and~$\skel(d')$.  In other words,
    $\val(\dtr_2(d))$~decomposes~$\val(d)$ for $A$~and~$\skel(d')$.
    Since $\skel(d') = \skel(\rho') = \skel(\rho, (C, \delta, W))$, which
    in turn equals~$\skel(d)$, this proves~(ii).
  \item It remains to prove~(ii) in the case where the rule~$\rho$ is
    \emph{non-lexicalized}.  We now apply the argument that we used to
    prove~(i) to the rule~$\rho_{C, \delta, W}$ instead of~$\rho$.
    For~$\rho_{C, \delta, W}$ we obtain from the previous case (even
    though $\rho$~is a non-lexicalized rather than a singly
    lexicalized rule) that the value of
    \[ \rho_{C, \delta, W}(\dtr_2(d_1), \dtr_1(d_2), \dotsc,
    \dtr_1(d_k)) \] 
    decomposes~$\val(d)$ for $A$~and~$\skel(d)$.  From the argument
    for~(i) we obtain that the value of
    \[ \langle\rho_{C, \delta, W} \rangle_{C', \delta',
      Z}(\dtr_2(d_1), \dtr_2(d_2), \dtr_1(d_3), \dotsc,
    \dtr_1(d_k)) \] equals the value of~$\rho_{C, \delta,
      W}(\dtr_2(d_1), \dtr_1(d_2), \dotsc, \dtr_1(d_k))$,
    hence $\val(\dtr_2(d))$~decomposes~$\val(d)$ for
    $A$~and~$\skel(d)$.
\end{compactitem}
This concludes the proof that~$L(G) \subseteq L(G')$.  To prove
the converse~$L(G') \subseteq L(G)$, it is straightforward to check
that, vice versa, (i)~for every derivation tree~$d' \in L(G'_\der, A)$
there is a derivation tree~$d \in L(G_\der, A)$ with~$\dtr_1(d) = d'$,
and (ii)~for every derivation tree~$d' \in L(G'_\der, \beta \cdot
\sequ(C, \delta, Z) \cdot \gamma)$ there is a derivation tree~$d \in
L(G_\der, \beta C \gamma)$ with $\dtr_2(d) = d'$~and~$\skel(d) = (C,
\delta, Z)$.  To be precise, in both cases $d$~can be obtained
from~$d'$ by changing every label~$\overline{\rho}$, $\rho_{C, \delta,
  W}$, $\langle \rho \rangle_{C', \delta', Z}$, and $\langle \rho_{C,
  \delta, W} \rangle_{C', \delta', Z}$ into just~$\rho$.  Thus, it is
obvious that $d$~can be computed from~$d'$ by an LDT"~transducer.
Hence $G$~and~$G'$ are \LDTR"~equivalent.

Finally, we present a procedure that directly constructs the reduced
version of~$G'$ provided that $G$~is reduced.  For a rule~$\rho = A
\to (u, \LL)$ with~$\LL = \{\seq B1k\}$, we define~$\bign_i(\rho) = 
B_i$ for $i \in [k]$ and $\bign_0(\rho) = A$.
\begin{compactitem}
\item Construct the set~$\text{Target} \subseteq \N$ of
  all~$\bign_2(\rho)$, where $\rho$~is a non-lexicalized rule.
\item Construct the directed graph~$g$ with set~$\N$ of nodes and with 
  edges~$\bign_0(\rho) \to \bign_1(\rho)$ for all non-lexicalized and 
  singly lexicalized rules~$\rho$, and let $g_\red$~be obtained
  from~$g$ by removing all nodes (and all incident edges) that are not
  reachable from a node in~$\text{Target}$.
\item Let $\Dec$~be a variable set of skeletons, which is
  initialized to~$\emptyset$.
\item Compute all rules~$\overline{\rho}$ such that $\bign_0(\rho)$~is
  a node of~$g_\red$, and add~$\skel(\rho)$ to~$\Dec$. 
\item Repeat the following subitem until $\Dec$~does not change
  any more:
  \begin{compactitem}
  \item compute all rules~$\rho_{C, \delta, W}$ such that $(C, \delta,
    W)$~is in~$\Dec$ and the edge~$\bign_0(\rho) \to \bign_1(\rho)$ is
    in~$g_\red$, and add~$\skel(\rho, (C, \delta, W))$ to~$\Dec$.
  \end{compactitem}
\item Finally, compute all rules~$\langle \rho \rangle_{C', \delta', Z}$
  such that $(C', \delta', Z)$ is in~$\Dec$, for the rules~$\rho$
  obtained so far.
\end{compactitem}
We leave the correctness of this procedure to the reader.
\end{proof}

\begin{figure}[t]
  \centering
  \includegraphics{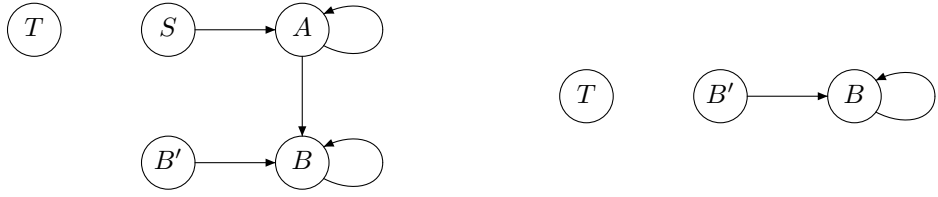}
  \caption{The graphs $g$~[left] and $g_\red$~[right] constructed in
    Example~\protect{\ref{exa:lex}}.}
  \label{fig:reach}
\end{figure}

\newpage
\begin{example}
  \label{exa:lex}
  \upshape
  Let us lexicalize the new grammar~$G$ of Example~\ref{exa:newmain},
  according to the construction in the proof of Lemma~\ref{lem:main}.
  We immediately construct the reduced version of~$G'$ with the
  procedure presented at the end of the proof of that lemma.  Note
  that $G$~satisfies the assumptions mentioned in the beginning of the
  proof for~$\Deltadl = \{\beta, \tau\}$.  For the doubly lexicalized
  rules~$\rho$ of~$G$; i.e., for the rules
  \begin{alignat*}{5}
    \rho_5 \colon && B(x_1) &\to \sigma(x_1, \alpha
    T_1(\underline{\beta} \sigma(\alpha T_2, \gamma T_3))) \quad
    & \rho_6 \colon && B(x_1) &\to \sigma(x_1,
    \sigma(\underline{\tau}, \nu)) \\
    \rho_7 \colon && T &\to (\alpha T_1(\underline{\beta}
      x_1),\, \alpha T_2,\,\gamma T_3) \quad
    & \rho_8 \colon && T &\to (x_1, \underline{\tau}, \nu) 
  \end{alignat*}
  (and the rules $\rho_5'$ and $\rho_6'$) 
  we define~$\lex(\rho) = \beta$ if $\beta$~occurs in~$\rho$, and
  $\lex(\rho) = \tau$ otherwise.  We marked the lexical element in the
  rules by underlining it.  We obtain that $\text{Target} = \{T,
  B, B'\}$, where~$T = (T_1, T_2, T_3)$.  The graphs $g$~and~$g_\red$
  are displayed in Figure~\ref{fig:reach}.  Since all doubly
  lexicalized rules~$\rho$ have their left-hand side in~$g_\red$, we
  construct the new rule~$\overline{\rho}$ for each of them.  We will
  use the following abbreviations for the new nonterminals
  \begin{align*}
    B_{\beta, 0} 
    &= \langle B, \beta, 0, \{x_1\} \rangle 
    & B_{\beta, 1} 
    &= \langle B, \beta, 1, \emptyset \rangle
    & B_\tau 
    &= \langle B, \tau, 0, \{x_1\} \rangle
      \tag*{of rank~$2$, $0$, and $2$, resp.} \\
    T_{1, \beta, 0} 
    &= \langle T_1, \beta, 0, \emptyset \rangle 
    & T_{1, \beta, 1} 
    &= \langle T_1, \beta, 1, \{x_1\} \rangle
    & T_{2, \tau} 
    &= \langle T_2, \tau, 0, \emptyset \rangle
      \tag*{all of rank~$1$.}
  \end{align*}
  Then we obtain the new rules
  $\overline{\rho}_5$~and~$\overline{\rho}_6$ on the first line and
  the rules $\overline{\rho}_7$~and~$\overline{\rho}_8$ on the second
  line:
  \begin{alignat*}{5}
    (B_{\beta,0}(x_1, x_2), B_{\beta,1})
    &\to (\sigma(x_1, \alpha T_1(x_2)),\, \sigma(\alpha T_2, \gamma
    T_3)) \quad
    & B_\tau(x_1, x_2)
    &\to \sigma(x_1, \sigma(x_2, \nu)) \\ 
    (T_{1, \beta, 0}(x_1), T_{1, \beta, 1}(x_1), T_2, T_3)
    &\to (\alpha T_1(x_1),\, x_1,\, \alpha T_2,\, \gamma T_3) \quad
    & (T_1(x_1), T_{2, \tau}(x_1), T_3)
    &\to (x_1, x_1, \nu) \enspace.  
  \end{alignat*}
  The construction of the first new rule is illustrated in the top box of 
  Figure~\ref{fig:cut}.  The rules
  $\overline{\rho}'_5$~and~$\overline{\rho}'_6$ are obtained from
  $\overline{\rho}_5$~and~$\overline{\rho}_6$ by changing every~$B$
  into~$B'$.  Let~$Z_\beta = (\{x_1\}, \emptyset)$, $Z_\tau =
  (\{x_1\})$, $Z'_\beta = (\emptyset, \{x_1\})$, and $Z'_\tau =
  (\emptyset)$. Then 
  \[ \Dec = \{(B, \beta, Z_\beta),\, (B', \beta, Z_\beta),\, (B, \tau,
  Z_\tau),\, (B', \tau, Z_\tau),\, (T_1, \beta, Z'_\beta),\, (T_2,
  \tau, Z'_\tau) \} \enspace. \] 

  The only non-lexicalized or singly lexicalized rules~$\rho$
  with~$\bign_0(\rho) \to \bign_1(\rho)$ in~$g_\red$ are
  the rules $\rho_4 = B(x_1) \to \sigma(B(x_1), B'(A))$ and the
  corresponding rule~$\rho'_4$ with left-hand side~$B'(x_1)$.  Since
  its first link is the nonterminal~$B$, we construct the new
  rules~$\rho_{C, \delta, W}$ for $(C, \delta, W) \in \{(B, \beta,
  Z_\beta), (B, \tau, Z_\tau)\} \subseteq \Dec$ and $\rho \in \{\rho_4,
  \rho'_4\}$.  For the right-hand side~$u$ of~$\rho_4$ (and~$\rho'_4$)
  we get
  \begin{align*} 
    u[B \gets \tree(B, \beta, Z_\beta)]
    &= u[B \gets B_{\beta, 0}(x_1, \beta(B_{\beta, 1}))] =
      \sigma(B_{\beta, 0}(x_1, \beta(B_{\beta, 1})), B'(A)) \\  
    u[B \gets \tree(B, \tau, Z_\tau)]
    &= u[B \gets B_\tau(x_1, \tau)] = \sigma(B_\tau(x_1, \tau), B'(A))
      \enspace,
  \end{align*}
  and consequently we obtain the rules
  \begin{alignat*}{5}
    (\rho_4)_{B, \beta, Z_\beta} = {}
    && (B_{\beta,0} (x_1, x_2), B_{\beta,1}) 
    &\to (\sigma(B_{\beta, 0}(x_1, x_2), B'(A)),\, B_{\beta, 1}) \\
    (\rho_4)_{B, \tau, Z_\tau} = {}
    && B_\tau(x_1, x_2) 
    &\to \sigma(B_\tau(x_1, x_2), B'(A)) \enspace,
  \end{alignat*}
  and similar rules for~$\rho'_4$.  
  The construction of the first rule is illustrated in the bottom box of 
  Figure~\ref{fig:cut}. Clearly, the set~$\Dec$ does not
  change, so we do not have to repeat this step.  In the final step we
  lexicalize the non-lexicalized (old and new) rules 
  by substituting $\tree(C',\delta',Z)$ for a nonterminal~$C'$ of the second link of each rule.  
  From $\rho_2 = A \to T_1(\sigma(B(T_2), T_3))$ we obtain
  the following two new rules, by substituting 
  $\tree(T_1, \beta, Z'_\beta)=T_{1, \beta, 0}(\beta \,T_{1, \beta, 1} (x_1))$ and 
  $\tree(T_2, \tau, Z'_\tau)=T_{2, \tau}(\tau)$ for $T_1$ and $T_2$ respectively:
 \begin{alignat*}{5}
    \langle \rho_2 \rangle_{T_1, \beta, Z'_\beta} = {}
    && A
    &\to T_{1, \beta, 0}(\beta \,T_{1, \beta, 1} (\sigma(B(T_2), T_3))) \\
    \langle \rho_2 \rangle_{T_2, \tau, Z'_\tau} = {} 
    && A
    &\to T_1(\sigma(B(T_{2, \tau}(\tau)), T_3)) 
    \enspace.
  \end{alignat*}
Moreover, from $\rho_3 = A \to B(\gamma A)$ and 
$\rho_4 = B(x_1) \to \sigma(B(x_1), B'(A))$
we obtain the new rules
  \begin{alignat*}{5}
    \langle \rho_3 \rangle_{B, \beta, Z_\beta} = {}
    && A
    &\to B_{\beta, 0}(\gamma A, \beta B_{\beta, 1}) \\
    \langle \rho_3 \rangle_{B, \tau, Z_\tau} = {}
    && A
    &\to B_\tau(\gamma A, \tau) \\
    \langle \rho_4 \rangle_{B', \beta, Z_\beta} = {}
    && B(x_1)
    &\to \sigma(B(x_1), B'_{\beta, 0}(A, \beta B'_{\beta, 1})) \\
    \langle \rho_4 \rangle_{B', \tau, Z_\tau} = {}
    && B(x_1)
    &\to \sigma(B(x_1), B'_\tau(A, \tau)) 
  \end{alignat*} 
  and from the rules $(\rho_4)_{B, \beta, Z_\beta}$ 
  and $(\rho_4)_{B, \tau, Z_\tau}$ we obtain 
  \begin{alignat*}{5}
    \langle(\rho_4)_{B, \beta, Z_\beta} \rangle_{B', \beta, Z_\beta} =
    {} 
    && (B_{\beta, 0}(x_1, x_2), B_{\beta, 1})
    &\to (\sigma(B_{\beta, 0}(x_1, x_2), B'_{\beta, 0}(A,
    \beta B'_{\beta, 1})),\, B_{\beta, 1}) \\ 
    \langle(\rho_4)_{B, \beta, Z_\beta} \rangle_{B', \tau, Z_\tau} =
    {} 
    && (B_{\beta, 0}(x_1, x_2), B_{\beta, 1})
    &\to (\sigma(B_{\beta, 0}(x_1, x_2), B'_\tau(A, \tau)),\,
    B_{\beta, 1}) \\ 
    \langle(\rho_4)_{B, \tau, Z_\tau} \rangle_{B', \beta, Z_\beta} =
    {} 
    && B_\tau(x_1, x_2) 
    &\to \sigma(B_\tau(x_1, x_2), B'_{\beta, 0}(A, \beta B'_{\beta,
      1})) \\ 
    \langle(\rho_4)_{B, \tau, Z_\tau} \rangle_{B', \tau, Z_\tau} = {}
    && B_\tau(x_1, x_2)
    &\to \sigma(B_\tau(x_1, x_2), B'_\tau(A, \tau))
    \end{alignat*}
    and similar rules for~$\rho'_4$.  The (reduced) lexicalized
    grammar~$G'$ has the rules  
    \begin{compactitem}
    \item $\rho_1$, $\rho_5$, $\rho_6$, $\rho_7$, $\rho_8$,
      $\overline{\rho}_5$, $\overline{\rho}_6$, $\overline{\rho}_7$,
      $\overline{\rho}_8$,
    \item $\langle \rho_2 \rangle_{T_1, \beta, Z'_\beta}$, $\langle
      \rho_2 \rangle_{T_2, \tau, Z'_\tau}$, $\langle \rho_3
      \rangle_{B, \beta, Z_\beta}$, $\langle \rho_3 \rangle_{B, \tau,
        Z_\tau}$,  $\langle \rho_4 \rangle_{B', \beta, Z_\beta}$,
      $\langle \rho_4 \rangle_{B', \tau, Z_\tau}$, 
    \item $\langle(\rho_4)_{B, \beta, Z_\beta} \rangle_{B', \beta,
        Z_\beta}$, $\langle(\rho_4)_{B, \beta, Z_\beta} \rangle_{B',
        \tau, Z_\tau}$, $\langle(\rho_4)_{B, \tau, Z_\tau}
      \rangle_{B', \beta, Z_\beta}$, $\langle(\rho_4)_{B, \tau,
        Z_\tau} \rangle_{B', \tau, Z_\tau}$, 
\end{compactitem}
and the corresponding rules for~$\rho'_4$, $\rho'_5$, and~$\rho'_6$.
Note that in all these rules, as in the grammar~$G$ of
Example~\ref{exa:main},  there is only one possibility for the set of
links~$\LL$.  Note also that the left-hand sides of the primed rules
are aliases of the left-hand sides of the nonprimed ones.  We finally
observe that rules $\langle \rho_2 \rangle_{T_1, \beta,
  Z'_\beta}$~and~$\overline{\rho}_7$ can be replaced by one rule~$A
\to \alpha T_1(\beta \sigma(B(\alpha T_2), \gamma T_3))$, and
similarly $\langle \rho_2 \rangle_{T_2, \tau,
  Z'_\tau}$ and $\overline{\rho}_8$ can be replaced by~$A \to
\sigma(B(\tau), \nu)$.  In fact, these rules could have been obtained
directly in the beginning as observed in
Example~\ref{exa:d-chain-free}.  After this replacement, and
disregarding the primed rules for aliases,  the resulting lexicalized
grammar has 17~rules.

Consider in the derivation tree~$d$ of Figure~\ref{fig:transport} the
path from the root to the left-most leaf with label~$\rho_8$.  The
sequence of node labels along this path is~$(\rho_1, \rho_2, \rho_4,
\rho'_4, \rho_4, \rho_5, \rho_8)$.  In the derivation tree~$\dtr_1(d)$
of~$G'_2$ these nodes are relabeled to~$(\rho_1, \langle \rho_2
\rangle_{T_1, \beta, Z'_\beta}, \langle \rho_4 \rangle_{B', \beta,
  Z_\beta}, \langle(\rho'_4)_{B, \beta, Z_\beta} \rangle_{B', \tau,
  Z_\tau}, \langle(\rho_4)_{B, \beta, Z_\beta} \rangle_{B', \beta,
  Z_\beta}, \overline{\rho}_5, \rho_8)$. \fin
\end{example}

\begin{figure}[t]
  \centering
  \includegraphics{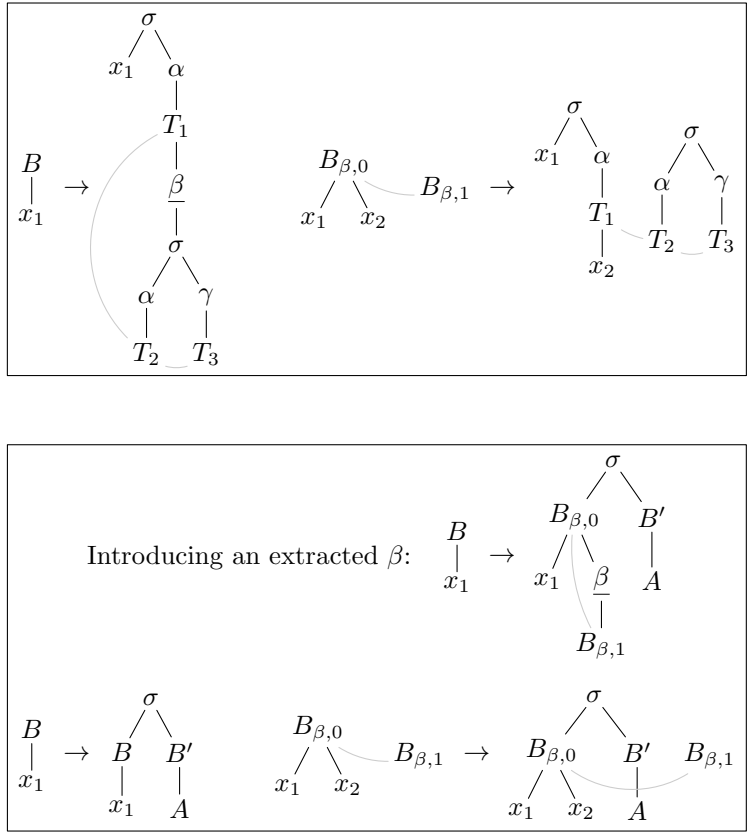}
  \caption{Illustration of the construction of the
    rule~$\overline{\rho}_5$ extracting the underlined~$\beta$~[top
    box] by splitting the right-hand side into the parts above and
    below the extracted symbol.  In the construction of the
    rule~$(\rho_4)_{B, \beta, Z_\beta}$ [bottom box] we first introduce
    the lexical element~$\beta$ (replacing~$B$) and the corresponding
    nonterminals [top rule] and then extract it again to obtain the
    final rule displayed at the bottom right.}
  \label{fig:cut}
\end{figure}

\newpage
We now state the main theorem of this paper.

\begin{theorem}
  \label{thm:main}
  It is decidable for the MCFTG~$G$ whether or not $G$~has finite
  $\Delta$"~ambiguity, and if so, there is a $\Delta$"~lexicalized
  MCFTG~$G'$ that is \LDTR"~equivalent to~$G$.  Moreover, $G'$ can be
  chosen such that $\wid(G') = \wid(G) + 1$~and~$\mu(G') = \mu(G)
  + \mr_\Delta$.\footnote{Recall that $\mr_\Delta$~is the maximal rank of the symbols
  in~$\Delta$.}
\end{theorem}

\begin{proof}
  The first statement is immediate from Theorem~\ref{thm:dec-growing}
  and Lemma~\ref{lem:main}.  Since Theorem~\ref{thm:dec-growing}
  preserves $\wid(G)$~and~$\mu(G)$, it suffices to check that the
  construction in the proof of Lemma~\ref{lem:main} satisfies the
  second statement.
\end{proof}

Note that if~$\Delta \subseteq \Sigma^{(0)}$, then $G'$~has the same
multiplicity as~$G$.  Thus, as a corollary we obtain (a more specific
version of) the main result of~\cite{maleng12}.  

\begin{corollary}
  \label{cor:main}
  If we have~$\Delta \subseteq \Sigma^{(0)}$, then
  Theorem~\ref{thm:main} holds for~spCFTG instead of~MCFTG.
\end{corollary}

Since every MCFTG has finite $(\Sigma^{(0)}\cup\Sigma^{(1)})$-ambiguity, 
we also obtain the following special case of Theorem~\ref{thm:main}.

\begin{corollary}
For every MCFTG $G$ there is an \LDTR"~equivalent $\Sigma$"~lexicalized MCFTG $G'$
such that $\wid(G')=\wid(G)+1$ and $\mu(G')=\mu(G)+1$.
\end{corollary}

It should be clear that Theorems
\ref{thm:dec-growing}~and~\ref{thm:main} can be combined.  If $G$~has
finite $\Delta$"~ambiguity, then there is an \LDTR"~equivalent
$\Delta$"~growing $\Delta$"~lexicalized MCFTG.  Since every
$\Delta$"~lexicalized MCFTG is almost $\Delta$"~growing, it suffices
to apply once more the construction in the proof of
Theorem~\ref{thm:dec-growing} to the $\Delta$"~lexicalized MCFTG~$G'$
of Theorem~\ref{thm:main}.

It should even be clear that, by combining rules in a standard way, we
can now ensure that every rule contains at least $n$~lexical symbols
for any~$n \in \nat$.  This will be used in Section~\ref{sub:monadic}.  
Unfortunately, such a combination of rules cannot be realized 
by an \LDTR"~transducer.\footnote{\label{foo:fc2} It \emph{can} be realized by a \emph{finite-copying} 
  deterministic top-down tree transducer
  with regular look-ahead.}  
For every~$n \geq 1$, let
us say that a rule~$A \to (u, \LL)$ of an
MCFTG~$G$ is \emph{$n$"~$\Delta$"~lexicalized} if
$\abs{\pos_\Delta(u)} \geq n$, and that $G$~is
$n$"~$\Delta$"~lexicalized if all its proper rules are
$n$"~$\Delta$"~lexicalized.

\begin{lemma}
  \label{lem:nlex}
  For every $\Delta$"~lexicalized MCFTG~$G$ and every~$n \geq 1$ there
  is an $n$"~$\Delta$"~lexicalized MCFTG~$G'$
  such that $\wid(G') = \wid(G)$~and~$\mu(G') = \mu(G)$.
\end{lemma}

\begin{proof}
  The proof is by induction on~$n$.  For the induction step, let
  $G$~be an $n$"~$\Delta$"~lexicalized MCFTG.  We may assume that all
  non-initial terminal rules of~$G$ are
  $(n+1)$"~$\Delta$"~lexicalized because otherwise we can apply once
  more the construction in the proof of Theorem~\ref{thm:dec-growing}
  for~$\F = \{t \in P_\Sigma(X)^{\scriptscriptstyle +} \mid
  n = \abs{\pos_\Delta(t)} \}$.  Moreover, we may assume that every big
  nonterminal~$A \neq S$ has an alias~$\bar{A}$ such that
  $A$~and~$\bar{A}$ do not occur together in any right-hand side of a
  rule.  This can be achieved by introducing a new symbol~$\bar{C}$
  for every nonterminal~$C$, and letting~$\bar{A} = (\seq{\bar A}1n)$
  be an alias of~$A = (\seq A1n)$. 

  Now let $G = (N, \N, \Sigma, S, R)$.  We construct~$G' = (N, \N,
  \Sigma, S, R')$, where $R'$~is defined as follows.  Let $\rho = A
  \to (u, \LL)$~be a rule in~$R$ with $\LL = \{\seq B1k\}$~and~$k \geq
  1$, and let~$\rho' = B_1 \to (u', \LL')$ be a rule in~$R$ with
  left-hand side~$B_1$ and~$\LL' = \{\seq{B'}1{\ell}\}$.  Let $u'' =
  u'[B'_i \gets \init(\bar{B}'_i) \mid 1 \leq i \leq \ell]$.  Then
  $R'$~contains the rule~$\langle \rho, \rho' \rangle = A \to (u[A_1
  \gets u''], \LL'')$, where $\LL'' = \{\seq{\bar B'}1{\ell}, \seq B2k\}$.
  Moreover, $R'$~contains all terminal rules of~$R$.
  Obviously, $G'$ is $(n+1)$"~$\Delta$"~lexicalized. 

  It is straightforward to prove that the derivation trees of $G'$ are obtained from those of $G$
  by the value-preserving mapping $M$ such that if $d=\rho(\rho'(d_1',\dotsc,d_\ell'),d_2,\dotsc,d_k)$  
  then $$M(d)=\langle \rho,\rho'\rangle(M(d_1'),\dotsc,M(d_\ell'),M(d_2),\dotsc,M(d_k)),$$
  and if $d=\rho$ where $\rho$ is a terminal rule then $M(d)=d$. 
  Vice versa, the derivation trees of $G$ are obtained from those of $G'$
  by the value-preserving tree homomorphism $M'$ such that 
  $$M'(\langle \rho,\rho'\rangle)=\rho(\rho'(x_1,\dotsc,x_\ell),x_{\ell+1},\dotsc,x_{\ell+k-1})$$
  and $M'(\rho)=\rho$ for every terminal rule $\rho$. 
  That proves that $L(G')=L(G)$. 
\end{proof}

\section{MCFTG and MC"~TAG}
\label{sec:monadic}
\noindent 
In this section we show that MC"~TAGs have (``almost'') 
the same tree generating power as MCFTGs. It is shown in~\cite{keprog11} that 
non-strict tree adjoining grammars (nsTAGs) have the same 
tree generating power as monadic spCFTGs, where 
an spCFTG~$G$ is monadic if~$\wid(G) \leq 1$; i.e., all
its nonterminals have rank~$1$~or~$0$. 
In the first subsection we prove that MCFTGs have the
same tree generating power as non-strict set-local multi-component
tree adjoining grammars (nsMC"~TAGs), generalizing the result
of~\cite{keprog11}. To avoid the introduction of the formal machinery 
that is needed to define nsMC"~TAGs in the usual way, we define 
them to be ``footed'' MCFTGs, similar to the footed spCFTGs from~\cite{keprog11}. 
As shown in~\cite[Section~4]{keprog11} for nsTAGs, 
the translation from one definition to the other is straightforward.   
In the second subsection we prove that MCFTGs have the
same tree generating power as (strict) set-local
multi-component tree adjoining grammars~(MC"~TAGs), where we define MC"~TAGs 
as a special type of footed MCFTGs. 
The last result implies that MC"~TAGs can be (strongly) lexicalized.
It also implies, as shown in the third subsection, that MCFTGs have 
the same tree generating power as monadic MCFTGs
(i.e., MCFTGs of width at most~1), which is essentially the same result as
in~\cite[Theorem~3]{aludan11}.\footnote{It is shown in~\cite[Theorem~3]{aludan11}
that multi-parameter STTs (streaming tree transducers) have the same power as one-parameter STTs.
Multi-parameter STTs are closely related to finite-copying macro tree transducers
(cf.~\cite[Section~4.2]{aludan11}), and hence to MCFTGs as will be shown in Section~\ref{sec:charact}. 
The number of parameters of the STT corresponds to the width of the MCFTG.} 
These results can be viewed as additional normal forms for MCFTGs. 

Roughly speaking, the transformation of an MCFTG into an MC-TAG will be realized by decomposing 
each tree $u_i$ in the right-hand side of a rule $A\to (u,\LL)$ with $A=(\seq A 1n)$ and $u=(\seq u 1n)$ 
into a bounded number of parts, to replace $u_i$ in $u$ by the sequence of these parts,    
and to replace $A_i$ in $A$ by a corresponding sequence of new nonterminals that 
simultaneously generate these parts.  
This is similar to the construction in the proof of Lemma~\ref{lem:main} where, however,
just one $u_i$ was decomposed into parts. 

\subsection{Footed MCFTGs}
\label{sub:mctag}
\noindent 
Tree adjoining grammars~(TAGs) are closely 
related to ``footed''  (simple) context-free tree grammars as shown
in~\cite[Section~4]{keprog11}.
An spCFTG is footed if for every rule $A(\seq x1k)\to u$ with $k\geq 1$ there is 
a node of $u$ with exactly $k$ children, which are labeled $\seq x1k$ from left to right. 
In other words, the arguments of $A$ are passed in the same order to one node of $u$.  
In this section we generalize this notion to MCFTGs and prove that 
for every MCFTG there is an equivalent footed MCFTG. 

\begin{definition}
  \label{def:footed}
  \upshape
  Let $G = (N, \N, \Sigma, S, R)$~be an MCFTG.  A pattern~$t \in P_{N
    \cup \Sigma}(X_k)$ with $k \in \nat_0$ is~\emph{footed} if 
  either $k=0$, or $k\geq 1$ and
  there exists a position~$p \in \pos_{N \cup \Sigma}(t)$, called the
  \emph{foot node} of~$t$, such that $\rk(t(p)) = k$~and $t(pi) = x_i$
  for every~$i \in [k]$.  A rule~$\rho=A \to ((\seq u1n), \LL)\in R$ is footed
  if $u_j$~is footed for every~$j \in [n]$.  The MCFTG~$G$ is footed
  if every rule~$\rho \in R$ is footed. \fin
\end{definition}

Note that, by definition and for technical convenience, 
every tree $t \in T_{N\cup\Sigma}=P_{N \cup \Sigma}(X_0)$ is footed.
The foot node of a footed pattern~$t \in P_{N \cup \Sigma}(X_k)$ with
$k \geq 1$ is obviously unique.  If~$p$~is the foot node of~$t$, then
$t|_p = \init(t(p))$.  It is straightforward to show, for a footed MCFTG $G$, that if~$(\seq
t1n) \in L(G, A)$, then $t_j$~is footed for every~$j \in [n]$. Assuming that $G$ is reduced,
this implies that $\wid(G)\leq \mr_\Sigma$.
Moreover, $G$ is permutation-free and nonerasing (cf.\@
Lemmas~\ref{lem:permfree} and~\ref{lem:nonerasing}).    

Based on the close relationship between non-strict TAGs and 
footed context-free tree grammars as shown in~\cite[Section~4]{keprog11}, we
define a \emph{non-strict tree adjoining grammar} (in short,~nsTAG) to
be a footed~spCFTG, and similarly we define a \emph{non-strict (set-local)
  multi-component~TAG} (in short, nsMC"~TAG) to be a footed~MCFTG.
This definition will be motivated after we have proved that for every
MCFTG there is an equivalent footed~MCFTG, which shows that MCFTGs
and  nsMC"~TAGs have the same tree generating power. 

It is shown in~\cite[Proposition~3]{keprog11} that every monadic
nonerasing~spCFTG can be transformed into an equivalent
footed~spCFTG.  However, the proof of that proposition is not entirely
correct, which can be seen from the following example.  Consider the
spCFTG~$G$ with rules~$S \to A(e)$, $A(x_1) \to \sigma(A(x_1))$,
and~$A(x_1) \to \tau(a, x_1, b)$.  Clearly, the last rule is not
footed.  In the proof of~\cite[Proposition~3]{keprog11} this grammar
is transformed into the equivalent spCFTG~$G'$ with rules~$S \to
A(e)$, $S \to A'(T_1, e, T_3)$, $A(x_1) \to \sigma(A(x_1))$, $A(x_1)
\to \sigma(A'(T_1, x_1, T_3))$, $A'(x_1, x_2, x_3) \to \tau(x_1, x_2,
x_3)$, $T_1 \to a$, and~$T_3 \to b$.  However, the rule~$A(x_1) \to
\sigma(A'(T_1, x_1, T_3))$ is not footed, which is due to the fact that
the foot node of the right-hand side~$\sigma(A(x_1))$ of the second
rule of~$G$ has a nonterminal label.  The solution to this problem is
to introduce the nonterminals $T_1$~and~$T_3$ in the first step of
each derivation rather than in the last step.  Thus, the footed
spCFTG~$G''$ with rules~$S \to A'(T_1, e, T_3)$, $A'(x_1, x_2, x_3)
\to \sigma(A'(x_1, x_2, x_3))$, $A'(x_1, x_2, x_3) \to \tau(x_1, x_2,
x_3)$, $T_1 \to a$, and~$T_3 \to b$ is equivalent to~$G$.  It is not
difficult to repair the proof of~\cite[Proposition~3]{keprog11}, 
but the construction becomes more
complicated.  We generalize that construction in the proof of the next
theorem (without preserving the multiplicity, however). 
Since MRTGs are trivially footed, we restrict ourselves to MCFTGs $G$ 
with $\wid(G)\geq 1$. 

\begin{theorem}
  \label{thm:footed}
  For every MCFTG~$G$ with $\wid(G)\geq 1$ 
  there is an \LDTR"~equivalent footed MCFTG~$G'$ such that 
  $\mu(G') \leq \mu(G) \cdot \mr_\Sigma \cdot (2\cdot\wid(G)-1)$,
  where~$\Sigma$ is the terminal alphabet of $G$.  
  Moreover, if $G$~is $\Delta$"~lexicalized, then so is~$G'$.  
\end{theorem}

\begin{proof}
The basic idea of this proof is that, for any ranked alphabet~$\Omega$,
every tree~$u\in T_\Omega(X)$ with $u\notin X$ and $\pos_X(u)\neq \emptyset$ can be decomposed 
into at most $\mr_\Omega \cdot (2k-1)$ footed patterns, where $k=\abs{\pos_X(u)}$.
This can be understood as follows. Clearly, there are a unique $m\geq 1$,
a unique footed pattern $u_\varepsilon\in P_\Omega(X_m)$, and unique trees $\seq u 1m\in T_\Omega(X)$ 
such that $u=u_\varepsilon[x_i\gets u_i\mid 1\leq i\leq m]$ and $\abs{\pos_X(u_i)} < \abs{\pos_X(u)}$
for every $i\in[m]$ with $u_i\notin X$. In fact, the foot node of $u_\varepsilon$ is the position $p$ 
which, in $u$, is the 
least common $\Omega$-labeled ancestor of the nodes in $\pos_X(u)$; i.e., the longest position
such that $u(p)\in\Omega$ and $\abs{\pos_X(u|_p)}=\abs{\pos_X(u)}$. Note that the requirement
$u(p)\in\Omega$ is only needed when $\abs{\pos_X(u)}=1$. 
Thus, we have decomposed $u$ as $u_\varepsilon[x_i\gets u_i\mid 1\leq i\leq m]$
where $u_\varepsilon$ is a footed pattern. For every $i\in[m]$ with $u_i\notin X$,
either $u_i\in T_\Omega$ and so $u_i$ is a footed pattern of rank~0, or 
$\pos_X(u_i)\neq \emptyset$ in which case $u_i$ can be decomposed further. It should also be clear that, 
in this inductive process, there are at most $2k-1$ such foot node positions $p$. The factor 
$\mr_\Omega$ is due to the footed patterns of rank~0. 
As an example, consider the ranked alphabet 
$\Omega=\{\tau^{(3)},\sigma^{(2)},\beta^{(1)},a^{(0)},b^{(0)}\}$ 
and the tree $u=\sigma(a,\sigma(v,w))$ with 
$v=\sigma(a,\sigma(a,\tau(x_1,a,\beta(\beta(x_2)))))$ and $w=\sigma(x_3,b)$.
For readability, let us use the notation $t_0[t_1,\dotsc,t_n]$ for $t_0[x_i\gets t_i\mid 1\leq i\leq n]$. 
Then we obtain the decomposition
$u=u_\varepsilon[u_1[x_1,u_{12},u_{13}[x_2]],u_2[x_3,u_{22}]]$,
illustrated in Figure~\ref{fig:monadic2}, of $u$ with the footed patterns $u_\varepsilon=\sigma(a,\sigma(x_1,x_2))$, 
$u_1=\sigma(a,\sigma(a,\tau(x_1,x_2,x_3)))$, $u_{12}=a$, $u_{13}=\beta(\beta(x_1))$,  
$u_2=\sigma(x_1,x_2)$, and $u_{22}=b$. 
Using new symbols $C^m_p$ of rank $m$, with $p\in\nat^*$, we can also express this as
$u=K[\gamma]$ where $K$ is the tree 
$C^2_\varepsilon(C^3_1(x_1,C^0_{12},C^1_{13}(x_2)),C^2_2(x_3,C^0_{22}))$,
which can be viewed as the skeleton of the decomposition, and 
$\gamma$ is the second-order substitution such that $\gamma(C^m_p)=u_p$. 
A~formal version of this decomposition
is formulated below and applied to (a variant of) the trees 
in the right-hand sides of the rules of $G$. 
We note here that this decomposition is closely related to the one used 
in~\cite[Section~6]{lohmansch12} to turn a ``straight-line'' spCFTG into a monadic one.

\begin{figure}[t]
  \centering
  \includegraphics{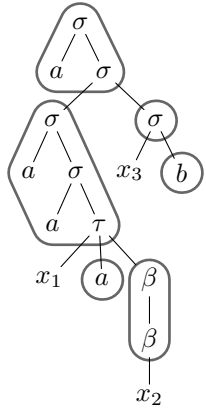}
  \caption{Decomposition into footed patterns.}
  \label{fig:monadic2}
\end{figure}

Let $G = (N, \N, \Sigma, S, R)$ be an MCFTG with $\wid(G)\geq 1$.
Then $\mr_\Sigma \cdot (2\cdot\wid(G)-1)\geq 1$, because
$\mr_\Sigma\geq 1$ by Definition~\ref{def:mcftg}. 
By Lemmas~\ref{lem:permfree} and~\ref{lem:nonerasing} we may assume that $G$ is 
permutation-free and nonerasing.\footnote{First apply Lemma~\protect{\ref{lem:nonerasing}} and 
then Lemma~\protect{\ref{lem:permfree}}. It is easy to check that
Lemma~\protect{\ref{lem:permfree}}  
preserves the nonerasing and $\Delta$"~lexicalized properties.}
This means that if $(\seq A 1n)\to ((\seq u 1n),\LL)$ is a rule in~$R$, 
then the pattern $u_i$ is in $\PF_{N \cup \Sigma}(X_{\rk(A_i)})\setminus X$ 
for every $i\in[n]$.\footnote{Recall 
from Lemma~\protect{\ref{lem:permfree}} that $\PF_\Omega(X)$ denotes the set of permutation-free patterns  
over the ranked alphabet $\Omega$. The requirement that $u_i\notin X$ is only relevant 
when $\rk(A_i)=1$, meaning that $u_i\neq x_1$.}

We define~$G' = (N', \N', \Sigma, S', R')$.  
The set~$N'$ of nonterminals consists of all
  triples~$\langle C, m, p \rangle$ with~$C \in N$,
  $0\leq m\leq \mr_\Sigma$, and~$p \in \nat^*$ such that $\abs{p} \leq \wid(G)$.
  The rank of~$\langle C, m, p \rangle$ is~$m$.  The initial nonterminal
  is~$S' = \langle S, 0, \varepsilon \rangle$. 
  For every nonterminal~$C \in N$, a \emph{skeleton} of~$C$ is a permutation-free
  pattern~$K \in \PF_{N'}(X_{\rk(C)})\setminus X$ such that\footnote{We usually do not denote 
  trees with a capital, but $k$ is already used for natural numbers.}
  \begin{compactenum}[\indent \quad (1)]
  \item for every $p \in \pos_{N'}(K)$ there exists $0\leq m\leq \mr_\Sigma$
    such that $K(p) = \langle C, m, p \rangle$, and
  \item for every $p\in\nat^*$ and $i\in\nat$, if $pi \in \pos_{N'}(K)$ 
    then $\abs{\pos_X(K|_{pi})} < \abs{\pos_X(K|_p)}$. 
  \end{compactenum}
  For such a skeleton~$K$, we define~$\sequ(K) = \yield_{N'}(K)$,
  which is an element of~$(N')^{\scriptscriptstyle +}$.\footnote{Recall the definition
    of $\yield_{N'}$ from the paragraph on homomorphisms in Section~\protect{\ref{sub:seqs}}.} 
There are only finitely many skeletons of~$C$. 
In fact, it is easy to show that $\abs{\pos_{N'}(K)}\leq \mr_\Sigma \cdot (2k-1)$ 
for every skeleton $K$ of $C$, if $k=\rk(C)\geq 1$. 
Additionally, if~$\rk(C) = 0$, then the only skeleton of~$C$ 
is~$\langle C, 0, \varepsilon \rangle$. 
Note that $K$~can be reconstructed from~$\sequ(K)$ because $K$~is permutation-free. 
In the example above, the tree $K$ is a skeleton of $C$, provided that $C^m_p$ 
denotes $\langle C, m, p \rangle$, and 
$\sequ(K) = (C^2_\varepsilon,C^3_1,C^0_{12},C^1_{13},C^2_2,C^0_{22})$. 

We will apply the above basic idea to a pattern 
$u\in \PF_{N' \cup \Sigma}(X_{\rk(C)})\setminus X$. 
This leads to a decomposition of $u$
that can be represented by a skeleton~$K$ of $C$ and 
a substitution function~$\gamma$ such that $u=K[\gamma]$.
This is formalized as follows. 
Let $K$~be
  a skeleton of~$C \in N$.  A~substitution function~$\gamma$
  for~$\alp_{N'}(K)$ is \emph{footed} if, for every~$C' \in
  \alp_{N'}(K)$, the pattern~$\gamma(C') \in P_{N' \cup \Sigma}(X)$ is
  footed. We say that the pair~$\langle K, \gamma \rangle$ is
  a \emph{footed $C$"~decomposition} of the tree~$K[\gamma]$.

  \paragraph{Basic fact} Every pattern~$u$ as above
  has a footed $C$"~decomposition~$\dec_C(u)$.\footnote{The decomposition is even unique,
  but that will not be needed.}  More
  precisely, for every~$C \in N$ and
  every~$u \in \PF_{N' \cup \Sigma}(X_{\rk(C)})\setminus X$ 
  there is a pair~$\dec_C(u) = \langle K, \gamma
  \rangle$ such that $K$~is a skeleton of~$C$,
  $\gamma$~is a footed substitution function for~$\alp_{N'}(K)$,
  and~$K[\gamma] = u$.

  \paragraph{Proof of the basic fact} 
To prove this by induction, we prove it for arbitrary $u \in T_{N' \cup \Sigma}(X)$
and we allow~$K$ to be an element of $T_{N'}(X)$ such that $\yield_X(K)=\yield_X(u)$. 
Obviously, if~$K[\gamma] = u$ and $u$ is a $k$-ary permutation-free pattern $\neq x_1$, 
then so is $K$. 

If $u=x\in X$, then~$\dec_C(u) = \langle K, \gamma \rangle$ with 
$K = x$~and~$\gamma$ is the empty function.
If $u\in T_{N' \cup \Sigma}$, then~$\dec_C(u) = \langle K, \gamma \rangle$ with 
$K = \langle C, 0, \varepsilon \rangle$ and $\gamma(\langle C, 0, \varepsilon \rangle)=u$.
Now suppose that~$u\notin X$ and~$\pos_X(u)\neq \emptyset$.  
We proceed by induction on~$\abs{\pos_X(u)}$. 
Let the footed pattern $u_\varepsilon$ in $P_{N' \cup \Sigma}(X_m)$ 
and the trees $\seq u 1m$ in $T_{N' \cup \Sigma}(X)$ be 
as in the basic idea above, and let, by the induction hypotheses or by the previous two basic cases,
$\dec_C(u_i) = \langle K_i, \gamma_i \rangle$ for every $i\in[m]$. 
Then~$\dec_C(u) = \langle K, \gamma \rangle$, where $K$ and $\gamma$ are defined as follows. 
For every $i\in[m]$ let $K'_i$~be obtained from~$K_i$ by changing every 
label~$\langle C, m', p \rangle$ into~$\langle C, m', ip \rangle$.
Then $K$~is the tree $K=\langle C,m,\varepsilon\rangle(\seq {K'} 1m)$. 
Moreover, the substitution function~$\gamma$ is defined by 
$\gamma(\langle C, m, \varepsilon \rangle) = u_\varepsilon$ and
$\gamma(\langle C, m', ip \rangle) = \gamma_i(\langle
    C, m', p \rangle)$ for every~$i\in[m]$ and every~$\langle C, m', ip
    \rangle \in \alp_{N'}(K'_i)$.
  It is straightforward to verify that $K$~and~$\gamma$ satisfy the
  requirements, which completes the proof of the basic fact.

\medskip
We define the set $\N'$ of big nonterminals to consist of all sequences $\sequ(K_1)
  \dotsm \sequ(K_n)$ for which there exists $(\seq A1n) \in \N$ such that $K_j$~is a
  skeleton of~$A_j$ for every~$j \in [n]$.  
  A \emph{skeleton function} for~$A \in \N$ is a substitution
  function~$\skel$ for~$\alp(A)$ that assigns a
  skeleton~$\skel(C)$ of~$C$ to every nonterminal~$C \in \alp(A)$. 
  The string homomorphism~$h_{\skel}$ from~$\alp(A)$ to~$N'$ is defined
  by $h_{\skel}(C) = \sequ(\skel(C))$ for every~$C \in \alp(A)$.  Note
  that $\N'$~is the set of all~$h^*_{\skel}(A)$, where $A \in
  \N$~and~$\skel$ is a skeleton function for~$A$.

  We now define the set $R'$ of rules. 
  Let~$\rho = A \to (u, \LL)$ be a rule in~$R$ such that~$A = (\seq
  A1n)$, $u = (\seq u1n)$, and~$\LL = \{\seq B1k\}$.  Moreover, let
  $\overline{\skel} = (\seq \skel1k)$, where $\skel_i$~is a
  skeleton function for~$B_i$ for every~$i \in [k]$.  
  Intuitively, $\overline{\skel}$ guesses 
  for every nonterminal $C$ that occurs in $\seq B1k$
  the skeleton of a footed $C$-decomposition of the tree generated by $C$. 
  Let $f$~be the
  substitution function for~$\alp_N(u)$ such that~$f = \bigcup_{i \in
    [k]} \skel_i$; i.e.,~$f(C) = \skel_i(C)$ if~$C \in \alp(B_i)$.
  It should be clear that $u_j[f] \in \PF_{N' \cup \Sigma}(X_{\rk(A_j)})\setminus X$
  for every~$j \in [n]$. 
  For every~$j \in [n]$, let $u'_j = u_j[f]$, let~$\dec_{A_j}(u'_j) =
  \langle K_j, \gamma_j \rangle$ 
(the footed $A_j$"~decomposition of $u'_j$ 
according to the above basic fact), and let~$v'_j =
  \gamma_j^*(\sequ(K_j))$.\footnote{Thus, if $\sequ(K_j)=(\seq {C'}1\ell)$ 
  with $\seq {C'}1\ell\in N'$,
  then $v'_j=(\gamma_j(C'_1),\dotsc,\gamma_j(C'_\ell))$.}
  Then $R'$~contains the rule
  \[ \langle \rho, \overline{\skel} \rangle = \sequ(K_1) \dotsm
  \sequ(K_n) \to (\word{v'}1n, \LL') \]
  with $\LL' = \{h^*_{\skel_1}(B_1), \dotsc, h^*_{\skel_k}(B_k)\}$.
  We also define the skeleton function~$\skel_{\rho,
    \overline{\skel}}$ for~$A$ by~$\skel_{\rho,
    \overline{\skel}}(A_j) = K_j$ for every~$j \in [n]$.  
  Intuitively, $K_j$ is the skeleton of a footed $A_j$-decomposition 
  of the tree generated by $A_j$, resulting from the skeletons
  guessed by $\overline{\skel}$.
  Note that the left-hand side of the rule~$\langle \rho, \overline{\skel}
  \rangle$ is~$h^*_{\skel_{\rho, \overline{\skel}}}(A)$.  This
  concludes the definition of~$G'$. It should be clear that $G'$ is footed. 
  Moreover, since the right-hand sides of the rules
  $\rho$~and~$\langle \rho, \overline{\skel} \rangle$ contain the
  same terminal symbols, $G'$~is $\Delta$"~lexicalized 
  if $G$~is $\Delta$"~lexicalized.  
  It remains to prove the correctness of~$G'$.

  For every derivation tree~$d \in L(G_\der, A)$ we define 
  a skeleton function~$\skel_d$ for~$A$ and 
  a derivation tree~$q(d) \in L(G'_\der, h^*_{\skel_d}(A))$ 
  inductively as follows. If~$d = \rho(\seq d1k)$
  with~the rule $\rho$ as above, then we define $\skel_d = \skel_{\rho,
    \overline{\skel}}$~and~$q(d) = \langle \rho, \overline{\skel}
  \rangle(q(d_1), \dotsc, q(d_k))$, where~$\overline{\skel} = (\seq
  \skel{d_1}{d_k})$.  We now claim the following.  

\smallskip
  \emph{Claim:} For every~$A = (\seq A1n)\in\N$ and every~$d \in L(G_\der, A)$, 
  if $\val(d) = (\seq t1n)$
  then $K_j[h^*_{\skel_d}(A) \gets \val(q(d))] = t_j$, 
  where $K_j=\skel_d(A_j)$, for every~$j \in [n]$.
  
\smallskip
\emph{Proof of Claim:}  Assume that~$d = \rho(\seq d1k)$ as above,
  and that the claim holds for $d_i$ for every~$i \in [k]$.
  Let $g$~be the substitution function for~$\alp_N(u)$ such that
  $g(C)$~is the $m$"~th element of~$\val(d_i)$ if $C$~is
  the $m$"~th element of~$B_i$.  So, $\val(d) = u[B_i \gets \val(d_i)
  \mid 1 \leq i \leq k] = u[g]$, and hence~$u_j[g] = t_j$ for every~$j
  \in [n]$.  We write~$[g']$ for the substitution~$[h^*_{\skel_{d_i
      }}(B_i) \gets \val(q(d_i)) \mid 1 \leq i \leq k]$.
  Consequently,~$\val(q(d)) = u'[g']$, where~$u' = \word{v'}1n$.
  We first show that~$u_j[f] [g'] = u_j[g]$ for every~$j \in [n]$.  By
  Lemma~\ref{lem:comm-assoc}(4) it suffices to show that $f(C)[g'] =
  g(C)$ for every~$C \in \alp_N(u)$.  For every~$C \in
  \alp(B_i)$ we obtain that  
  $f(C)[g'] = \skel_{d_i}(C)[h^*_{\skel_{d_i}}(B_i) \gets \val(q(d_i))]$,
  which equals $g(C)$~by the induction hypotheses.  Now let~$j \in
  [n]$. Then
  \[ K_j[h^*_{\skel_d}(A) \gets \val(q(d))] = K_j \bigl[\sequ(K_j)
  \gets v'_j[g'] \bigr] = K_j \bigl[\sequ(K_j) \gets
  \gamma_j^*(\sequ(K_j))[g'] \bigr] \enspace. \]
  By Lemma~\ref{lem:comm-assoc}(4) this equals~$K_j[\gamma_j][g']$.
  Since $\dec_{A_j}(u'_j) = \langle K_j, \gamma_j \rangle$, we obtain that
  \[ K_j[\gamma_j] [g'] = u'_j[g'] = u_j[f][g'] = u_j[g] = t_j
  \enspace. \]
  \emph{This proves the claim.} Note that it provides a footed $A_j$-decomposition of $t_j$
  (in fact, the unique one). 

\smallskip
  In the case where~$A = S$ we obtain that~$\skel_d(S) = \langle S,
  0, \varepsilon \rangle$.  Thus,
  $\val(q(d)) = \val(d)$~by the claim. Hence~$L(G) \subseteq L(G')$.   
  Clearly, for every skeleton function~$\skel$, the set $L_{\skel}$ of all
  derivation trees~$d$ with~$\skel_d = \skel$ is a
  regular tree language, which
  can be computed by a deterministic bottom-up finite tree
  automaton that uses all skeleton functions as states.  
  The \LDTR"~transducer~$M$ that computes~$q(d)$
  from~$d$ has one state~$q$, and it has the rules  
  \[ \langle q, \rho(y_1 \colon L_{\skel_1}, \dotsc, y_k
  \colon L_{\skel_k}) \rangle \to \langle \rho,
  \overline{\skel} \rangle (\langle q, y_1 \rangle, \dotsc,
  \langle q, y_k\rangle) \enspace, \] 
  where $\overline{\skel} = (\seq \skel 1k)$. 
  In the other direction, every derivation
  tree~$d' \in L(G'_\der)$ can be turned into a derivation tree~$d = M'(d')$
  in~$L(G_\der)$ by changing every label~$\langle \rho,
  \overline{\skel} \rangle$ into just~$\rho$, and it is straightforward to
  show that~$q(d) = d'$.  This shows that
  $L(G') \subseteq L(G)$, and hence the correctness of~$G'$.
\end{proof}

\begin{example}
  \label{exa:monadic}
  \upshape
  Let $\Sigma = \{\tau^{(3)}, \ell^{(1)}, r^{(1)}, a^{(0)}, b^{(0)},
  e^{(0)}\}$.  Intuitively $\ell$~stands for a left parenthesis and
  $r$~for a right parenthesis.  We consider the footed spCFTG~$G_1 =
  (N_1, \Sigma, S, R_1)$ with the set of nonterminals~$N_1 = \{S, A, A'\}$, of
  which $A$~has rank~$1$ and $A'$~is an alias of~$A$, and the rules
  \[ S \to \ell A(A'(re)) \qquad \qquad 
  A(x_1) \to \ell A(A'(rx_1)) \qquad \text{and} \qquad 
  A(x_1) \to \ell \tau(a, b, r x_1) \enspace, \]
  where we have omitted the rules with left-hand side~$A'(x_1)$.  Let
  $\Delta = \{a,b\}$.  Since $G_1$~is $\Delta$"~growing, it has finite
  $\Delta$"~ambiguity.  However, as we will show in
  Remark~\ref{rem:nstagnonlex},  there is no $\Delta$"~lexicalized
  footed spCFTG~$G$ with~$L(G) = L(G_1)$.  
  The basic reason for this is that the set 
  $\{\yield_{\{\ell,r\}}(t)\mid t\in L(G_1)\}\subseteq\{\ell,r\}^*$
  consists of all balanced strings of parentheses $\ell$ and $r$. 
  In fact, $G_1$~is a
  straightforward variant of the TAG of~\cite{kuhsat12}, for which
  there is no (strongly) equivalent $\Delta$"~lexicalized TAG.
  Note that we defined nsTAGs to be footed spCFTGs. 
  We will also show in Remark~\ref{rem:nstagnonlex} that there is no 
  $\Delta$"~lexicalized spCFTG~$G$ with $\wid(G)\leq 1$ that is 
  equivalent to~$G_1$.

  From Corollary~\ref{cor:main}, we obtain a $\Delta$"~lexicalized
  spCFTG~$G_2$ with~$\wid(G_2) = 2$ that is equivalent to~$G_1$.  It
  has the new nonterminals $B = \langle A, b, 0, X_1 \rangle$~and~$B'
  = \langle A', b, 0, X_1 \rangle$, where $\rk(B) = 2$~and $B'$~is an 
  alias of~$B$.  For the sake of readability we interchanged the two
  arguments of $B$ (and those of~$B'$), and similarly we used~$B$ instead of~$B'$
  in the first two rules, so that $A'$~has become superfluous.  Its
  rules are
  \begin{alignat*}{5}
    \rho_1 \colon \;S 
    &\to \ell A(B(b, re)) 
    & \rho_2 \colon \;A(x_1) 
    & \to \ell A(B(b, rx_1))
    & \rho_3 \colon \;A(x_1)
    & \to \ell \tau(a, b, rx_1) \\
    && \qquad \rho_4 \colon \;B(x_1, x_2) 
    & \to \ell B(x_1, B'(b, rx_2)) \qquad
    & \rho_5 \colon \;B(x_1, x_2)
    & \to \ell \tau(a, x_1, rx_2) \enspace,
  \end{alignat*}
  plus the rules $\rho_4'$ and $\rho_5'$ for the alias $B'$ of $B$. 
  Clearly, the tree~$B(b, x_1)$ generates the same terminal trees
  as~$A(x_1)$.  More precisely, if~$A(x_1)$ generates the tree~$\ell^n 
  \tau(a, b, wx_1)$, where $n \in \nat$~and~$w \in \Sigma^*$, then $B(x_1,
  x_2)$~generates~$\ell^n \tau(a, x_1, w x_2)$.

  Rules $\rho_4$ and $\rho_5$ are not footed. 
  We now turn~$G_2$ into an equivalent $\Delta$"~lexicalized footed
  MCFTG~$G_2'$ using the construction
  in the proof of Theorem~\ref{thm:footed}.  
  For rule $\rho_5=B(x_1, x_2)\to u_5$ and~$\overline{\skel} = \varepsilon$, we obtain 
  the footed $B$-decomposition $\dec_B(u_5)=\langle K_5,\gamma_5\rangle$ such that 
  $K_5= B_0(B_1,x_1,B_3(x_2))$, where $B_0=\langle B,3,\varepsilon\rangle$, 
  $B_1=\langle B,0,1\rangle$, and $B_3=\langle B,1,3\rangle$, and $\gamma_5$ is defined as follows:
  $\gamma_5(B_0)= \ell\tau(x_1,x_2,x_3)$, $\gamma_5(B_1)= a$, and $\gamma_5(B_3)= rx_1$. 
  The resulting rule $\tilde{\rho}_5=\langle \rho_5,\varepsilon\rangle$ is 
  \[\tilde{\rho}_5 \colon \quad (B_0(x_1,x_2,x_3), B_1, B_3(x_1))
    \to (\ell \tau(x_1,x_2,x_3), a, rx_1) \]
  with left-hand side $\sequ(K_5) = (B_0,B_1,B_3)$, 
  and the corresponding skeleton function for~$B$ is~$\skel_5 = \skel_{\rho_5, \varepsilon}$
  such that~$\skel_5(B)=K_5$. 
  The construction of this rule is illustrated in the first part of Figure~\ref{fig:decomp}. 
  Of course we obtain similar primed results for $B'$. 
  Taking $\overline{\skel} = (\skel_5,\skel_5')$ and 
  substituting $K_5$ for $B$ and $K_5'$ for $B'$ in the right-hand side $u_4=\ell B(x_1, B'(b, rx_2))$
  of rule $\rho_4$, we obtain 
  $u_4'= \ell B_0(B_1,x_1,B_3(B_0'(B_1',b,B_3'(rx_2))))$
  which has the footed $B$-decomposition $\dec_B(u_4')=\langle K_4,\gamma_4\rangle$ where $K_4=K_5$, 
  $\gamma_4(B_0)= \ell B_0(x_1,x_2,x_3)$, $\gamma_4(B_1)= B_1$, and 
  $\gamma_4(B_3)= B_3(B_0'(B_1',b,B_3'(rx_1)))$. 
  The resulting rule $\tilde{\rho}_4=\langle \rho_4,(\skel_5,\skel_5')\rangle$ is 
  \[\tilde{\rho}_4 \colon \quad (B_0(x_1,x_2,x_3), B_1, B_3(x_1))
    \to (\ell B_0(x_1,x_2,x_3),\,B_1,\,B_3(B_0'(B_1',b,B_3'(rx_1)))) \enspace.\]
  Since the skeleton function $\skel_{\rho_4,(\skel_5,\skel_5')}$ for $B$ is again $\skel_5$, 
  these are all the necessary rules of $G_2'$ with left-hand side $(B_0,B_1,B_3)$, and similarly for 
  $(B_0',B_1',B_3')$. 
  The decomposition $\dec_A(u_3)=\langle K_3,\gamma_3\rangle$ of $u_3=\ell \tau(a, b, rx_1)$
  is simply $K_3=\langle A,1,\varepsilon\rangle(x_1)$ and $\gamma_3(\langle A,1,\varepsilon\rangle)=u_3$. 
  Thus, identifying $\langle A,1,\varepsilon\rangle$ with~$A$, 
  grammar $G_2'$ has the rule $\tilde{\rho}_3=\rho_3$.
  Substituting $K_3$ for $A$ and $K_5$ for $B$ in the right-hand side $u_2$ of $\rho_2$ 
  we obtain the tree $u_2'= \ell A(B_0(B_1, b, B_3(rx_1)))$ which, just as $u_3$, decomposes into itself.
  Thus, $G_2'$ has the rule $\tilde{\rho}_2 = A(x_1)\to u_2'$. The construction of this rule is 
  illustrated in the second part of Figure~\ref{fig:decomp}. Finally, by a similar process 
  (identifying $\langle S,0,\varepsilon\rangle$ with $S$), 
  we obtain the rule $\tilde{\rho}_1 = S\to u_2'[x_1\gets e]$.
  Summarizing, $G'_2$~has the
  nonterminals~$\{S, A, B_0, B'_0, B_1, B'_1, B_3, B'_3\}$ and the big
  nonterminals~$\{S, A, (B_0, B_1, B_3), (B'_0, B'_1, B'_3)\}$. Its
  rules (apart from those for the alias $(B'_0, B'_1, B'_3)$) are 
   \begin{align*}
    \tilde{\rho}_1 \colon
    && S
    & \to \ell A(B_0(B_1, b, B_3(re))) \\
    \tilde{\rho}_2 \colon
    && A(x_1)
    & \to \ell A(B_0(B_1, b, B_3(rx_1))) \\
    \tilde{\rho}_3 \colon
    && A(x_1)
    & \to \ell \tau(a, b, rx_1) \\ 
    \tilde{\rho}_4 \colon
    && (B_0(x_1,x_2,x_3), B_1, B_3(x_1))
    &\to (\ell B_0(x_1,x_2,x_3),\, B_1,\, B_3(B'_0(B'_1, b,
    B'_3(rx_1)))) \\
    \tilde{\rho}_5 \colon
    && (B_0(x_1,x_2,x_3)), B_1, B_3(x_1))
    & \to (\ell \tau(x_1,x_2,x_3), a, rx_1) \enspace.
  \end{align*}
  To see that $L(G_2')=L(G_1)$ we observe that  
  the tree~$K_5=B_0(B_1, x_1, B_3(x_2))$
  generates the same terminal trees as~$B(x_1, x_2)$ (as formalized in 
  the Claim in the proof of Theorem~\ref{thm:footed}), and hence
  $B_0(B_1, b, B_3(x_1))$ generates the same terminal trees
  as~$A(x_1)$. \fin
\end{example} 

\begin{figure}[t]
  \centering
  \includegraphics{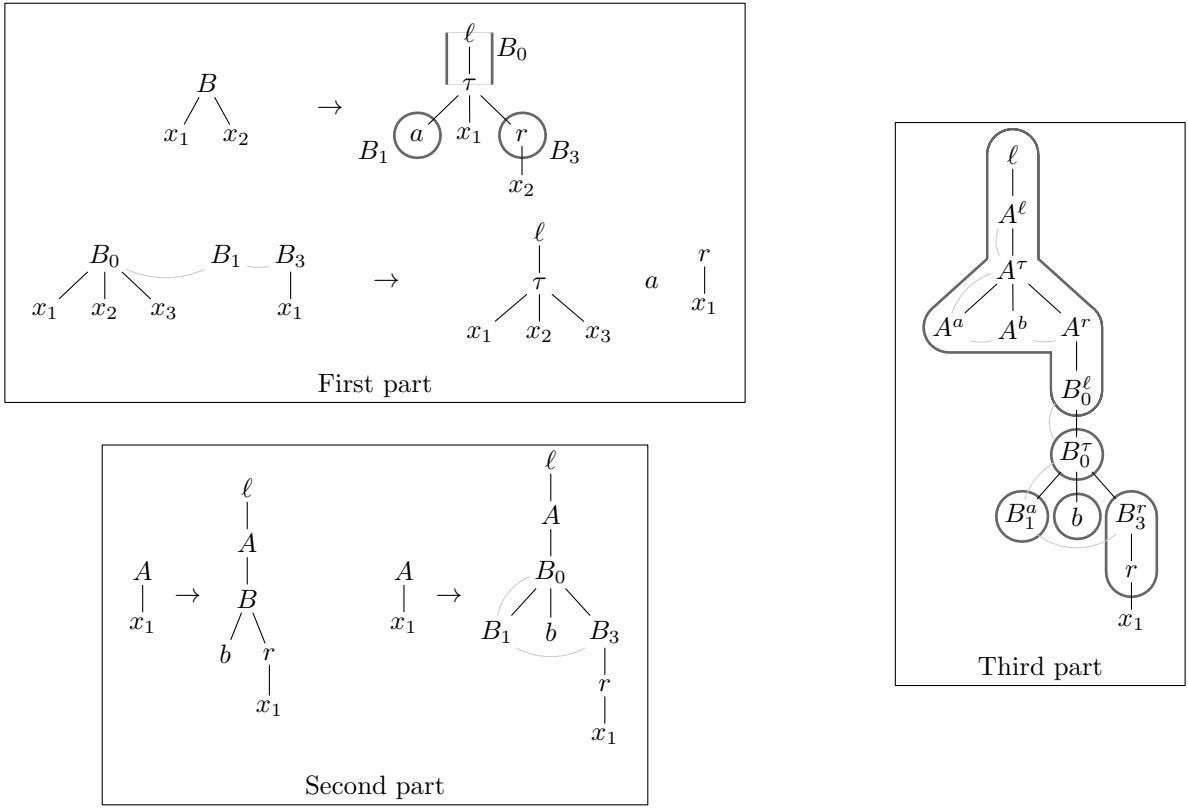}
  \caption{First part: Illustration of the footed
    decomposition~$\langle K_5,\gamma_5\rangle$ of the right-hand side
    of rule~$\rho_5$, with the resulting rule~$\tilde{\rho}_5$.  
    Second part: Substitution of the skeleton~$K_5$ of~$B$ into
    rule~$\rho_2$.  Third part: Adjoining $A$"~decomposition of
    Example~\protect{\ref{exa:foottotag1}}.} 
  \label{fig:monadic}
  \label{fig:decomp}
\end{figure}

\begin{example}
  \label{exa:footed}
  \upshape
  As another, very simple example we again consider the spCFTG~$G$
  with the following rules
  \[ S \to A(e) \qquad A(x_1) \to \sigma(A(x_1)) \qquad \text{and}
  \qquad A(x_1) \to \tau(a, x_1, b) \enspace, \] which was also
  discussed before Theorem~\ref{thm:footed}. The only skeleton of $A$ needed by the
  equivalent footed MCFTG~$G'$ is $A_0(A_1,x_1,A_3)$ where $A_0=\langle A,3,\varepsilon\rangle$,
  $A_1=\langle A,0,1\rangle$, and $A_3=\langle A,0,3\rangle$.
  Its big nonterminals are $S'=\langle S,0,\varepsilon\rangle$~and~$(A_0, A_1, A_3)$,
  and its rules are
  \begin{align*}
    S' 
    &\to A_0(A_1, e, A_3) \\
    (A_0(x_1, x_2, x_3), A_1, A_3)
    &\to (\sigma(A_0(x_1, x_2,
    x_3)), A_1, A_3)  \\
    (A_0(x_1, x_2, x_3), A_1, A_3) 
    &\to (\tau(x_1, x_2,
    x_3), a, b) \enspace. 
  \end{align*}
  Note that $G'$ is not an spCFTG. \fin
\end{example}

Let us now discuss set-local multi-component tree adjoining
grammars~(MC"~TAGs).  In the beginning of this subsection we have
defined a \emph{non-strict} MC"~TAG~(nsMC"~TAG) to be a footed MCFTG.
To convince the reader familiar with~TAGs we add some more
terminology, which should make this clear.  Let~$A \to (u, \LL)$ be a
rule with $A =(\seq A1n)$~and~$u = (\seq u1n)$.  If the rule is
initial (i.e.,~$A = S$), then the right-hand side~$u$ together with
the set~$\LL$ of links is called an \emph{initial tree}, and
otherwise it is called an \emph{auxiliary forest}.  Application of
the rule consists of adjunctions and substitutions.  The
replacement of the nonterminal~$A_j$ by~$u_j$ is called an
\emph{adjunction} if $\rk(A_j) > 0$ and a \emph{substitution}
if~$\rk(A_j) = 0$.  An occurrence of a nonterminal~$C \in N$ in~$u$
with~$\rk(C) > 0$ has an obligatory adjunction~(OA) constraint,
whereas an occurrence of a terminal~$\sigma \in \Sigma$ in~$u$
with~$\rk(\sigma) > 0$ has a null adjunction~(NA) constraint.  In the
same manner we handle obligatory and null substitution (OS~and~NS)
constraints.  Each big nonterminal~$B \in \LL$ can be viewed as a
selective adjunction/substitution~(SA/SS) constraint, which restricts
the auxiliary forests that can be adjoined/substituted for~$B$ to the
right-hand sides of the rules with left-hand side~$B$.

In the literature, MC"~TAGs~are usually free-choice, which means that
the set~$\LL$ of links can be dropped from the rules (see
Section~\ref{sub:basicnf}).  By Lemma~\ref{lem:renaming} this is no
restriction on footed~MCFTGs.  An~MCFTG is said to be
\emph{tree-local} (as opposed to `set-local') if for every rule as above and every~$B \in \LL$ 
there exists~$j \in [n]$ such that~$\alp(B) \subseteq \alp_N(u_j)$.
It can easily be proved that tree-local~MCFTGs have the same power as
spCFTGs, and similarly that tree-local nsMC"~TAGs have the same power
as nsTAGs.  

The first statement of Theorem~\ref{thm:footed} shows that
nsMC"~TAGs have the same tree generating power as~MCFTGs.  The second
statement shows together with Theorem~\ref{thm:main} that nsMC"~TAGs
can be (strongly) lexicalized.

\begin{corollary}
  \label{cor:mctaglex}
  For every finitely $\Delta$"~ambiguous nsMC"~TAG~$G$ there is an \LDTR"~equivalent
  $\Delta$"~lexicalized nsMC"~TAG~$G'$ such that 
  $\mu(G') \leq (\mu(G) + \mr_\Delta) \cdot \mr_\Sigma \cdot (2 \cdot \wid(G) + 1)$, 
  where $\Sigma$~is the terminal alphabet of~$G$.  
\end{corollary}

\begin{remark}
  \label{rem:nstagnonlex}
  \upshape
  In Example~\ref{exa:monadic}, the finitely $\Delta$"~ambiguous
  spCFTG~$G_1$ is footed and hence an~nsTAG.  
  Similarly, the $\Delta$"~lexicalized MCFTG~$G'_2$
  equivalent to~$G_1$ is footed and hence an nsMC"~TAG.  We now prove
  that there does not exist a $\Delta$"~lexicalized nsTAG equivalent
  to~$G_1$.  In other words, as opposed to nsMC"~TAGs, nsTAGs~cannot
  be strongly lexicalized.  The proof is a straightforward variant of
  the one in~\cite{kuhsat12}, and we present it here for completeness'
  sake.  

  To obtain a contradiction, let~$G = (N, \Sigma, S, R)$ be a
  reduced $\Delta$"~lexicalized nsTAG equivalent to~$G_1$.  Note that
  $G$~is a footed~spCFTG, and recall from the observations after
  Definition~\ref{def:footed} that every tree in~$L(G, A)$ is footed
  for every nonterminal~$A$.  Hence the nonterminals of~$G$
  have rank~$0$, $1$, or~$3$.  This implies that $G$~is
  \emph{right-footed}; i.e., for every rule~$A(\seq x1k) \to u \in R$ 
  of~$G$ with~$k \geq 1$, the right-hand side~$u$ is of the form~$v
  \word x1k$ with~$v \in (N \cup \Sigma)^{\scriptscriptstyle +}$.  In fact,
  if $u$~is not of that form, then it is of the form~$v \omega(u_1,
  u_2, u_3)$ with $v \in (N\cup\Sigma)^*$~and~$\omega \in N^{(3)} \cup
  \{\tau\}$ such that the foot node of~$u$ occurs in $u_1$~or~$u_2$; 
  i.e., either $u_1$~or~$u_2$ is of the form~$v_1 \omega'(\seq x1k)
  v_2$ with $v_1, v_2 \in (N \cup \Sigma)^*$~and~$\omega' \in N \cup
  \Sigma$.  But then $A$ generates terminal trees of the form~$w
  \tau(t_1, t_2, t_3)$ with~$w \in \Sigma^*$ such that either
  $t_1$~or~$t_2$ is of the form~$w_1 \gamma(\seq x1k) w_2$ with $w_1,
  w_2 \in \Sigma^*$~and~$\gamma \in \{\tau, \ell, r\}$.  This
  contradicts the form of the trees in~$L(G_1)$, in which the first
  and second arguments of~$\tau$ are always $a$~and~$b$,
  respectively.  Consequently $A$~cannot be reachable, contradicting
  the fact that $G$~is reduced.  Now it is easy to see that every right-footed
  spCFTG~$G$ can be viewed as an ordinary context-free grammar
  generating~$L(G)$ viewed as a string language.  We just replace
  every rule~$A(\seq x1k) \to v \word x1k$ by the rule~$A \to
  v$.\footnote{This generalizes the fact that every regular tree
    grammar is a context-free grammar (see
    Section~\protect{\ref{sub:trees}}).} Thus, it now remains to show
  that there is no $\{a, b\}$"~lexicalized context-free grammar~$G$
  such that~$L(G) = L(G_1)$, where $G_1$~is the context-free grammar
  with rules~$S \to \ell A A r e$, $A \to \ell AAr$, and~$A \to \ell \tau
  abr$.  Here `$\{a,b\}$"~lexicalized' means that $a$~or~$b$ occurs
  in every right-hand side of a rule of~$G$.  For a string~$w \in
  \Sigma^*$, let~$c(w) = \#_\ell(w) - \#_r(w)$, where $\#_\ell(w)$~is
  the number of occurrences of~$\ell$ in~$w$, and similarly
  for~$\#_r(w)$.  Since the ``parentheses'' $\ell$~and~$r$ are
  balanced in every string in~$L(G) = L(G_1)$, it follows
  from~\cite[Lemma~4]{knu67} that for every nonterminal~$A$ of~$G$
  there is a number~$c(A) \in \nat_0$ such that~$c(w) = c(A)$ for
  every~$w \in L(G, A)$.  For every~$v = \word v1k\in(N\cup\Sigma)^*$ 
  with~$\seq v1k \in N \cup \Sigma$, we let $c(v) = \sum_{i = 1}^k c(v_i)$.  
  Now consider a derivation~$S \Rightarrow_G v_1 \alpha v_2
  \Rightarrow_G^* w_1 \alpha w_2$ such that~$\alpha \in \{a,b\}$,
  $v_1, v_2 \in (N \cup \Sigma)^*$, $w_1, w_2 \in \Sigma^*$, and~$v_i
  \Rightarrow_G^* w_i$ for~$i \in \{1,2\}$.  Consequently,~$w_1 \alpha
  w_2 \in L(G)$. Thus $c(w_1)\in\nat_0$, due to the balancing of $\ell$ and $r$. 
  By the above, $c(w_1) = c(v_1)$.  Since $G$~is
  $\{a,b\}$"~lexicalized and has only finitely many initial rules,
  this shows that there is a number~$\kappa \in \nat_0$ 
  with the following property: for
  every string~$w \in L(G)$ there exist $\alpha \in \{a,
  b\}$~and~$w_1, w_2 \in \Sigma^*$ such that $w = w_1 \alpha
  w_2$~and~$c(w_1) \leq \kappa$.  This is a contradiction
  because it is easy to see that this does not hold for~$w = t_\kappa
  e \in L(G_1)$, where $t_0 = \ell \tau abr$~and~$t_{n + 1} = \ell t_n
  t_nr$ for every~$n \in \nat_0$.  

  This shows that nsTAGs cannot be strongly lexicalized. It also shows 
  that context-free grammars cannot be $\Delta$"~lexicalized.
  They can of course be $\Sigma$"~lexicalized.
  
  The spCFTG $G_1$ of Example~\ref{exa:monadic} is also monadic; 
  more precisely, it has width $\wid(G_1)=1$. 
  We finally prove that, as observed in Example~\ref{exa:monadic},
  there is no $\Delta$"~lexicalized monadic spCFTG equivalent to~$G_1$.  
  Let $G$~be such a grammar. By Lemma~\ref{lem:nonerasing} 
  we may assume that $G$~is nonerasing.  It can then be shown as
  above that $G$~is right-footed. However, in this case we must have 
  $k = 1$~and~$\omega = \tau$; moreover, either $u_1$~or~$u_2$ contains~$x_1$ and
  hence generates a tree that contains some~$\gamma \in \{\tau, \ell,
  r\}$ because $G$~is nonerasing.  The remainder of the proof is the
  same as above. This shows that to lexicalize an MCFTG $G$, either the width $\wid(G)$
  or the multiplicity $\mu(G)$ must increase. 
  \fin
\end{remark}

We now define strict MC"~TAGs as follows.  A \emph{(strict set-local)
  multi-component tree adjoining grammar} (in short,~MC"~TAG) is a
footed MCFTG~$G = (N, \N, \Sigma, S, R)$ for which there exists an
equivalence relation~$\equiv$ on~$N \cup \Sigma$ such that
\begin{compactenum}[\indent (1)]
\item  for all~$\sigma, \tau \in \Sigma$, if $\sigma \neq \tau$ and
  $\sigma \equiv \tau$, then $\rk(\sigma) \neq \rk(\tau)$; 
\item for every~$C \in N$ there exists~$\sigma \in \Sigma$ such
  that~$C \equiv \sigma$; and
\item for every rule~$(\seq A1n) \to ((\seq u1n), \LL)$ in~$R$ and
  every~$j \in [n]$,
  \begin{compactenum}[(a)]
  \item $u_j(\varepsilon) \equiv A_j$ and 
  \item if~$\rk(A_j) \geq 1$, then~$u_j(p) \equiv A_j$, where $p$~is the
    foot node of~$u_j$.
  \end{compactenum}
\end{compactenum}
The first requirement means that distinct equivalent terminal symbols
can be viewed as the same ``final'' symbol with different ranks.  In
this way, $\Sigma$~can be viewed as corresponding to a ``final''
alphabet, in which each symbol can have a finite number of different 
ranks, as for example in derivation trees of context-free grammars.
The second requirement means that each nonterminal~$C$ that is
equivalent to terminal~$\sigma$ can be viewed as the same final symbol
as~$\sigma$ together with some information that is relevant to 
SA~constraints.  The third requirement means that the root and foot
node of~$u_j$ are equivalent to~$A_j$; i.e., represent the same final 
symbol as~$A_j$.  Thus, intuitively, adjunction always replaces a
final symbol by a tree with that same final symbol as root label and
foot node label.  We define a \emph{tree adjoining grammar}~(in
short,~TAG) to be an MC"~TAG of multiplicity~$1$; i.e., a
footed~spCFTG that satisfies the requirements above.

\begin{example}
  \label{exa:tag}
  \upshape
A simple example of a TAG~$\overline{G}_1$ is obtained from
    the spCFTG~$G_1$ in Example~\ref{exa:monadic} by adding a terminal
    symbol~$\gamma$ of rank~$1$.  The rules of~$\overline{G}_1$ are
    \[ S \to \gamma \ell A(A'(re)) \qquad \qquad 
    A(x_1) \to \gamma \ell A(A'(r\gamma x_1)) \qquad \text{and} \qquad
    A(x_1) \to \gamma \ell \tau(a, b, r\gamma x_1) \enspace, \]
    where $A'$~is an alias of~$A$.  The equivalence relation~$\equiv$
    is the smallest one such that~$S \equiv A \equiv A' \equiv
    \gamma$.  It clearly satisfies the above three requirements.  This
    TAG is closely related to the one in~\cite{kuhsat12}.  It can be proved
    in exactly the same way as in Remark~\ref{rem:nstagnonlex} that
    there is no $\{a,b\}$"~lexicalized nsTAG equivalent
    to~$\overline{G}_1$, which slightly generalizes the result
    of~\cite{kuhsat12}.\footnote{The language class~TAL generated by
      TAGs is properly included in the language class~nsTAL, which is
      generated by~nsTAGs.  The tree language~$L = \{\ell^n r^n
      e \mid n \in \nat\}$, which is root consistent (cf.\@
      Corollary~\protect{\ref{cor:charmctal}} in the next subsection), is a witness for the
      properness. It is generated by an nsTAG with rules $S\to A(e)$, 
      $A(x_1)\to \ell A(rx_1)$, and $A(x_1)\to \ell rx_1$. 
      For the sake of a contradiction, let $G = (N,
      \Sigma, S, R)$ be a TAG such that $L(G) = L$.  
      Clearly, $\wid(G)\leq 1$ and $G$ must be right-footed 
      (cf.\@ Remark~\protect{\ref{rem:nstagnonlex}}). 
      For any unary nonterminal~$A \in N^{(1)}$ we have $L(G, A)
      \subseteq \{\ell^k x_1 \mid k \in \nat\}$ or
      $L(G, A) \subseteq \{r^k x_1 \mid k \in \nat\}$ due to the
      condition that the root label and foot node
      label must coincide.  However, since $G$ can be viewed 
      as an ordinary context-free grammar generating the string language $L$,  
      these languages $L(G,A)$ must be finite, due to pumping. 
      Hence we can transform $G$ into an equivalent right-linear context-free grammar, 
      which is a contradiction because $L$ is not regular.}  
      Thus, TAGs~cannot be strongly lexicalized by~nsTAGs. 

The MCFTG~$G'$ of Example~\ref{exa:footed} is an MC"~TAG. 
    The equivalence relation~$\equiv$
    is the smallest one such that~$S' \equiv A_0 \equiv \sigma \equiv \tau$, 
    $A_1 \equiv a$, and~$A_3 \equiv b$. Note that $\rk(\sigma) \neq \rk(\tau)$. 
 \fin
\end{example}

Let MC"~TAL denote the class of tree languages generated by MC"~TAGs. 
In the next subsection we prove that MCFT~and~MC"~TAL are almost the
same class of tree languages.

\subsection{MC"~TAL almost equals MCFT}
\label{sub:mcftismctal}
\noindent By definition, we have $\text{MC"~TAL} \subseteq
\text{MCFT}$.  In the other direction, the inclusion~$\text{MCFT}
\subseteq \text{MC"~TAL}$ does not hold because a tree language
from~MC"~TAL cannot contain two trees of which the roots are labeled
with two different symbols of the same rank.  In this subsection we
show that this is indeed the only necessary restriction.  
To prove that every language $L \in \text{MCFT}$ satisfying this restriction 
is in $\text{MC"~TAL}$, we begin with the case where the root of each tree~$t
\in L$ is labeled by the same symbol~$\sigma_0$.   In this case we
will construct an MC"~TAG of a special type, which we define next.
We first need some more terminology.

Let $G = (N, \N, \Sigma, S, R)$~be an MCFTG.  Recall from
Definition~\ref{def:footed} that a pattern~$t \in P_{N \cup
  \Sigma}(X_k)$ with $k \in \nat_0$ is footed if either $k=0$, 
or $k\geq 1$ and there is a
position~$p \in \pos_{N \cup \Sigma}(t)$, called the foot node
of~$t$, with $\rk(t(p)) = k$~and~$t(pi) = x_i$ for every~$i \in  
[k]$.  
Given a footed pattern $t \in P_{N \cup \Sigma}(X_k)$ 
with $k\geq 1$, we define~$\rlab(t) = t(\varepsilon)$
and~$\flab(t) = t(p)$ where~$p$ is the (unique) foot node of $t$.
Thus, $\rlab(t)$~and~$\flab(t)$ are the
labels of the root and the foot node of~$t$, respectively.
In the case where $k=0$ we define~$\rlab(t) = t(\varepsilon)$
and, for technical convenience, also~$\flab(t) = t(\varepsilon)$.
Thus, in this case $\rlab(t)$ is also the
label of the root of~$t$ and~$\flab(t)=\rlab(t)$. 
For~$k \geq 1$ we define the \emph{spine} of~$t$ to be the set of all
ancestors of its foot node (including the foot node itself), whereas
for~$k = 0$ the spine of~$t$ is defined to be the empty set.  

An \emph{adjoining} MCFTG is a footed MCFTG~$G$ for which there is a
mapping~$\varphi \colon N \cup \Sigma \to \Sigma$ such that 
\begin{compactenum}[\indent \quad (1)]
\item $\varphi(\sigma) = \sigma$ for every~$\sigma \in \Sigma$, and
\item $\varphi(\rlab(u_j)) = \varphi(\flab(u_j)) = \varphi(A_j)$ for
  every rule~$(\seq A1n) \to ((\seq u1n), \LL)\in R$ and every~$j \in [n]$.
\end{compactenum}
This implies that $\varphi$~is rank-preserving for nonterminals of
rank at least~$1$ (assuming that such a nonterminal generates at least
one terminal tree).  Obviously, every adjoining MCFTG is an MC"~TAG
with respect to the equivalence relation~$\equiv$ that is the kernel
of~$\varphi$; i.e.,~$\alpha \equiv \beta$ if~$\varphi(\alpha) =
\varphi(\beta)$.  By~(1) above, $\equiv$~is the identity on~$\Sigma$.
Vice versa, if $G$~is an MC"~TAG with respect to an equivalence
relation that is the identity on~$\Sigma$, then $G$~is an adjoining
MCFTG (as can easily be checked).

We now prove that for every footed~MCFTG~$G$ that generates a tree
language in which all trees have the same root label~$\sigma_0$, there is an
equivalent adjoining~MCFTG, which is also lexicalized if
$G$~is lexicalized.  In fact, the next lemma proves a slightly more
general fact, which will be needed to prove the theorem following the
lemma. The proof of the lemma is very similar to the one of Theorem~\ref{thm:footed},
with a further decomposition of the trees in the right-hand sides of the rules. 

\begin{lemma}
  \label{lem:foottotag}
  Let $G = (N, \N, \Sigma, S, R)$ be a footed~MCFTG and let $\sigma_0 \in
  \Sigma$.  Then there is an adjoining MCFTG~$G^{\sigma_0}$ such that
  $L(G^{\sigma_0}) = \{t \in L(G) \mid t(\varepsilon) = \sigma_0\}$ 
  and~$\mu(G^{\sigma_0}) =
  \mu(G) \cdot \abs{\Sigma} \cdot \mr_\Sigma$.
  Moreover, if~$G$~is $\Delta$"~lexicalized, then so is~$G^{\sigma_0}$.
\end{lemma}

\begin{proof}
  The basic idea of this proof is that, for any alphabet~$\Omega$,
  every string~$w \in \Omega^{\scriptscriptstyle +}$ can be decomposed
  as~$w = \word w1n$ such that $1 \leq n \leq \abs{\Omega}$, $w_i \in
  \Omega^{\scriptscriptstyle +}$, and the first and last symbol
  of~$w_i$ are the same.  We quickly prove this by induction
  on~$\abs{\Omega}$.  Let $a$~be the first symbol of~$w$, and let
  $w_1$~be the longest prefix of~$w$ that ends on~$a$.  Then~$w = w_1
  w'$ with~$w' \in (\Omega \setminus \{a\})^*$.  If~$w' =
  \varepsilon$, then we are ready.  Otherwise we apply the induction
  hypothesis. This decomposition is of course not unique. For example,
  the proof gives $abab = aba \cdot b$, but another 
  decomposition is $abab = a \cdot bab$. 

  Let $G = (N, \N, \Sigma, S, R)$ be a footed~MCFTG, and let~$\sigma_0
  \in \Sigma$.  We define~$G^{\sigma_0} = (N', \N', \Sigma,
  S^{\sigma_0}, R')$, where~$N'$, $\N'$, and~$R'$ do not depend
  on~$\sigma_0$.  The set~$N'$ of nonterminals consists of all
  $4$"~tuples~$\langle C, \sigma, m, p \rangle$ with~$C \in N$,
  $\sigma \in \Sigma$, $m \in \{0,\rk(\sigma)\}$, and~$p \in \nat^*$ 
  such that $\abs{p} < \abs{\Sigma}$.
  The rank of~$\langle C, \sigma, m, p \rangle$ is $m$.  The initial nonterminal
  is~$S^{\sigma_0} = \langle S, \sigma_0, 0, \varepsilon \rangle$. 
  Let~$\varphi \colon N' \cup \Sigma \to \Sigma$ be defined
  by~$\varphi(\langle C, \sigma, m, p \rangle) = \varphi(\sigma) =
  \sigma$.  We will define $\N'$~and~$R'$ in such a way that
  $G^{\sigma_0}$~is an adjoining MCFTG with respect to~$\varphi$.

  For every nonterminal~$C \in N$, a \emph{skeleton} of~$C$ is a footed
  pattern~$K \in P_{N'}(X_{\rk(C)})$ such that
  \begin{compactenum}[\indent \quad (1)] 
  \item for every $p \in \pos_{N'}(K)$ there exist $\sigma\in\Sigma$ 
  and $m \in \{0,\rk(\sigma)\}$ such that 
    $K(p) = \langle C, \sigma, m, p \rangle$,
  \item every subtree of~$K$ in~$T_{N'}$ is in~$(N')^{(0)}$, and
  \item $\varphi(K(p)) \neq \varphi(K(p'))$ for every two distinct
    positions $p, p' \in \pos_{N'}(K)$ on the spine of~$K$.
  \end{compactenum}
  For such a skeleton~$K$, we define~$\sequ(K) = \yield_{N'}(K)$,
  which is an element of~$(N')^{\scriptscriptstyle +}$. We note
  that there are only finitely many skeletons of~$C$. In fact, 
  $\abs{\pos_{N'}(K)}\leq \abs{\Sigma}\cdot \mr_\Sigma$ 
  for every skeleton $K$ of~$C$, if $\rk(C)\geq 1$. 
  Additionally, if~$\rk(C) = 0$, then every skeleton of~$C$ is of the
  form~$\langle C, \sigma, 0, \varepsilon \rangle$ with~$\sigma \in
  \Sigma$.  We finally note that $K$~can be reconstructed
  from~$\sequ(K)$ because $K$~is footed.

We will apply the above basic idea to the sequence of $\varphi$-images of the labels 
of the nodes on the spine of a footed pattern~$u$; i.e., to the sequence 
$(\varphi(u(p_1)),\dotsc,\varphi(u(p_n)))$ where $\seq p1n$ are the positions 
on the spine of $u$, in the order of increasing length. 
This leads to a decomposition of $u$
that can be represented by a skeleton~$K$ and 
a substitution function~$\gamma$ such that $u=K[\gamma]$.
Formally, let $K \in P_{N'}(X)$~be
  a skeleton of~$C \in N$.  A~substitution function~$\gamma$
  for~$\alp_{N'}(K)$ is \emph{adjoining} if, for every~$C' \in
  \alp_{N'}(K)$, the pattern~$\gamma(C') \in P_{N' \cup \Sigma}(X)$ is
  footed and $\varphi(\rlab(\gamma(C'))) = \varphi(\flab(\gamma(C'))) =
  \varphi(C')$.  We say that the pair~$\langle K, \gamma \rangle$ is
  an \emph{adjoining $C$"~decomposition} of the tree~$K[\gamma]$.

  \paragraph{Basic fact} Every footed pattern~$u$ has an adjoining 
  $C$"~decomposition~$\dec_C(u)$.  More precisely, 
  for every~$C \in N$ and every 
  footed pattern~$u \in P_{N' \cup \Sigma}(X_{\rk(C)})$ 
  there is a pair~$\dec_C(u) = \langle K, \gamma
  \rangle$ such that $K$~is a skeleton of~$C$,
  $\gamma$~is an adjoining substitution function for~$\alp_{N'}(K)$,
  and~$K[\gamma] = u$.

  \paragraph{Proof of the basic fact} Let~$\sigma =
  \varphi(\rlab(u))$.  First suppose that~$\rk(u) = 0$.
  Then~$\dec_C(u) = \langle K, \gamma \rangle$ with $K = \langle C,
  \sigma, 0, \varepsilon \rangle$~and~$\gamma(\langle C, \sigma, 0,
  \varepsilon \rangle) = u$.  Now suppose that~$\rk(u) \geq 1$.  We
  use induction on the cardinality of the spine of~$u$.  Let $q$~be the
  longest position on the spine of~$u$ such that~$\varphi(u(q)) = 
  \sigma$, and let $\rk(\sigma) = m$.  If $q$~is the foot node of~$u$, 
  then $\dec_C(u) = \langle
  K, \gamma \rangle$ with $K = \init(\langle C, \sigma, m, \varepsilon
  \rangle)$~and~$\gamma(\langle C, \sigma, m, \varepsilon \rangle) =
  u$.  Otherwise, let~$i \in \nat$ be the unique integer such that
  $qi$~is a position on the spine of~$u$.  Let~$u' = u|_{qi}$, and
  let~$\dec_C(u') = \langle K', \gamma' 
  \rangle$ by the induction hypothesis.  Then~$\dec_C(u) = \langle K,
  \gamma \rangle$, where $K$~and~$\gamma$ are defined as follows.  Let
  $K''$~be obtained from~$K'$ by changing every label~$\langle C,
  \sigma', m', p \rangle$ into~$\langle C, \sigma', m', ip \rangle$.
  Then $K$~is the tree
  
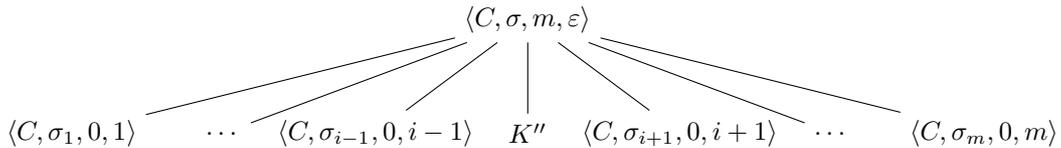
\begin{center}
    \begin{tikzpicture}[level/.style={sibling distance=20mm/#1,
        level distance=15mm}]
      \node {$\langle C, \sigma, m, \varepsilon \rangle$} 
      child {node {$\langle C, \sigma_1, 0, 1 \rangle$} } 
      child {node {$\dots$} } 
      child {node {$\langle C, \sigma_{i-1}, 0, i-1 \rangle$} }
      child {node {$K''$} }
      child {node {$\langle C, \sigma_{i+1}, 0, i+1 \rangle$} } 
      child {node {$\dots$} } 
      child {node {$\langle C, \sigma_m, 0, m \rangle$} };
    \end{tikzpicture}
  \end{center}
  where $\sigma_j = \varphi(u(qj))$ for every~$j
  \in [m] \setminus \{i\}$.  Moreover, the substitution
  function~$\gamma$ is defined by: 
  \begin{compactitem}
  \item $\gamma(\langle C, \sigma, m, \varepsilon \rangle) =
    (u|^q)[\SBox \gets \init(\sigma)]$,
  \item $\gamma(\langle C, \sigma_j, 0, j \rangle) = u|_{qj}$ for
    every~$j \in [m] \setminus \{i\}$, and
  \item $\gamma(\langle C, \sigma', m', ip \rangle) = \gamma'(\langle
    C, \sigma', m', p \rangle)$ for every~$\langle C, \sigma', m', ip
    \rangle \in \alp_{N'}(K'')$.
  \end{compactitem}
  It is straightforward to verify that $K$~and~$\gamma$ satisfy the
  requirements, which completes the proof of the basic fact.

\medskip
The definition of the set $\N'$ of big nonterminals and the set $R'$ of rules
is exactly the same as in the proof of Theorem~\ref{thm:footed}.\footnote{Except that 
in the construction of the rule $\langle \rho,\overline{\delta}\rangle$ 
it should be clear that $u_j[f]$ is a footed pattern in~$P_{N' \cup \Sigma}(X_{\rk(C)})$.
Moreover, the decomposition $\dec_{A_j}(u_j[f])$ is of course an 
adjoining $A_j$-decomposition of $u_j[f]$.}
It should be clear that $G^{\sigma_0}$~is adjoining with respect to~$\varphi$.
The correctness of~$G^{\sigma_0}$ is also proved in the same way as 
in the proof of Theorem~\ref{thm:footed}.
The Claim and its proof are exactly the same. 
In the case where~$A = S$ we obtain in the claim that~$\skel_d(S) = \langle S,
  \sigma, 0, \varepsilon \rangle$ with~$\sigma \in \Sigma$, and hence
  $\val(q(d)) = \val(d)$.   Since $G^{\sigma_0}$~is
  adjoining, it is easy to see that~$\sigma = \val(q(d))(\varepsilon)$;
  i.e., the root symbol of~$\val(d)$.  Hence~$\{t \in L(G)
  \mid t(\varepsilon) = \sigma_0\} \subseteq L(G^{\sigma_0})$.  As in the proof of
  Theorem~\ref{thm:footed} there is an \LDTR"~transducer~$M$ that
  computes~$q(d)$ from~$d$, and every derivation
  tree~$d' \in L(G^{\sigma_0}_\der, \langle S, \sigma, 0, \varepsilon
  \rangle)$ can be turned into a derivation tree~$d = M'(d') \in L(G_\der)$ 
  such that~$q(d) = d'$ by changing every label~$\langle \rho,
  \overline{\skel} \rangle$ into~$\rho$. 
  Taking~$\sigma = \sigma_0$ this shows that
  $L(G^{\sigma_0}) \subseteq \{t \in L(G) \mid t(\varepsilon) =
  \sigma_0\}$, and hence the correctness of~$G^{\sigma_0}$.
\end{proof}

\begin{example}
  \label{exa:foottotag1}
  \upshape
    Let us consider the MCFTG~$G'_2$ of
    Example~\ref{exa:monadic}.  As already observed in
    Remark~\ref{rem:nstagnonlex}, $G'_2$~is footed and hence
    an~nsMC"~TAG. Here we illustrate the proof of
    Lemma~\ref{lem:foottotag} by constructing the adjoining
    MCFTG~$G^{\ell}$ for~$G = G'_2$; note that $G^{\ell}$~is
    equivalent to~$G'_2$ because~$t(\varepsilon) = \ell$ for every~$t
    \in L(G'_2)$.  We recall that $G'_2$~has the following rules 
    (where we replace $\tilde{\rho}_i$ by $\rho_i$, for convenience): 
   \begin{align*}
    \rho_1 \colon 
    && S
    & \to \ell A(B_0(B_1, b, B_3(re))) \\
    \rho_2 \colon
    && A(x_1)
    & \to \ell A(B_0(B_1, b, B_3(rx_1))) \\
    \rho_3 \colon
    && A(x_1)
    & \to \ell \tau(a, b, rx_1) \\  
    \rho_4 \colon
    && (B_0(x_1,x_2,x_3), B_1, B_3(x_1))
    &\to (\ell B_0(x_1,x_2,x_3),\, B_1,\, B_3(B'_0(B'_1, b,
    B'_3(rx_1)))) \\
    \rho_5 \colon
    && (B_0(x_1), B_1, B_3(x_1))
    & \to (\ell \tau(x_1,x_2,x_3), a, rx_1) \enspace,
  \end{align*}
    plus the rules $\rho_4'$ and $\rho_5'$ for the alias $(B_0',B_1',B_3')$ of $(B_0,B_1,B_3)$.
    For rule~$\rho_5$
    and~$\overline{\skel} = \varepsilon$, we obtain the skeleton
    function~$\skel_5 = \skel_{\rho_5, \varepsilon}$
    for~$(B_0, B_1, B_3)$ such that
    \[ \skel_5(B_0) = B_0^\ell(B_0^\tau(x_1,x_2,x_3))
    \qquad \skel_5(B_1) = B_1^a
    \qquad \text{and} \qquad \skel_5(B_3) = B_3^r(x_1) \enspace, \] 
where 
$B_0^\ell = \langle B_0, \ell, 1, \varepsilon \rangle$, 
$B_0^\tau = \langle B_0, \tau, 1, 1 \rangle$,
$B_1^a = \langle B_1, a, 0, \varepsilon \rangle$, and 
$B_3^r = \langle B_3, r, 1, \varepsilon \rangle$. 
The resulting rule $\tilde{\rho}_5=\langle \rho_5,\varepsilon\rangle$ is 
  \[\tilde{\rho}_5 \colon \quad (B_0^\ell(x_1),B_0^\tau(x_1,x_2,x_3), B_1^a, B_3^r(x_1))
    \to (\ell x_1, \,\tau(x_1,x_2,x_3), \,a, \,rx_1) \enspace.\]
Substituting $\skel_5(B_i)$ for $B_i$ (and $\skel'_5(B_i')$ for $B_i'$) 
in the right-hand side $u_4$ of rule $\rho_4$, we obtain the forest 
$u_4'= (\ell B_0^\ell(B_0^\tau(x_1,x_2,x_3)), \,B_1^a, 
\,B_3^r(B_0'^\ell(B_0'^\tau(B_1'^a,b,B_3'^r(rx_1)))))$ and 
from that the following rule $\tilde{\rho}_4 = \langle \rho_4, (\skel_5, \skel'_5) \rangle$: 
    \begin{align*}
      \tilde{\rho}_4 \colon \quad 
      & \phantom{{}\to{}}
        (B_0^\ell(x_1),B_0^\tau(x_1,x_2,x_3), B_1^a, B_3^r(x_1)) \\*
      &\to (\ell B_0^\ell(x_1), \,B_0^\tau(x_1,x_2,x_3), \,B_1^a, 
            \,B_3^r(B_0'^\ell(B_0'^\tau(B_1'^a,b,B_3'^r(rx_1))))) \enspace,
    \end{align*}
and the skeleton function~$\skel_{\rho_4, (\skel_5, \skel'_5)} = \skel_5$
for~$(B_0, B_1, B_3)$. Thus, these are all the new rules obtained from $\rho_4$ and $\rho_5$. 
    We now turn to rules $\rho_3$~and~$\rho_2$.  The only skeleton
    needed for~$A$ is the tree
    \[ K = \skel_{\rho_3, \varepsilon}(A) = A^\ell(A^\tau(A^a, A^b,
    A^r(x_1))) \enspace, \] 
    where~$A^\ell = \langle A, \ell, 1, \varepsilon \rangle$, $A^\tau
    = \langle A, \tau, 1, 1 \rangle$, $A^a = \langle A, a, 0, 11
    \rangle$, $A^b = \langle A, b, 0, 12 \rangle$, and~$A^r = \langle
    A, r, 1, 13 \rangle$.  The resulting rule~$\tilde{\rho}_3 = \langle
    \rho_3, \varepsilon \rangle$ is
    \[ \tilde{\rho}_3 \colon \quad (A^\ell(x_1), A^\tau(x_1, x_2, x_3), A^a,
    A^b, A^r(x_1)) \to (\ell x_1,\, \tau(x_1, x_2, x_3),\, a,\, b,\,
    rx_1) \enspace. \]
    Substituting~$K$ for~$A$ and $\skel_5(B_i)$ for $B_i$ in the right-hand side~$u_2=\ell
    A(B_0(B_1, b, B_3(rx_1)))$ of~$\rho_2$, we
    obtain the tree
    \[ u_2' = \ell A^\ell(A^\tau(A^a, A^b, A^r(B_0^\ell(B_0^\tau(B_1^a, b,
    B_3^r(rx_1)))))) \enspace. \]
    It has the adjoining $A$"~decomposition
    $\dec_A(u_2') = \langle K, \gamma \rangle$ such that $\gamma(A^\ell)
    = \ell A^\ell(A^\tau(A^a, A^b, A^r(B_0^\ell(x_1))))$, $\gamma(A^\tau) = 
    B_0^\tau(x_1, x_2, x_3)$, $\gamma(A^a) = B_1^a$, $\gamma(A^b) = b$,
    and~$\gamma(A^r) = B_3^r(r x_1)$, which is illustrated in
    the third part of Figure~\ref{fig:decomp}.  The resulting rule~$\tilde{\rho}_2 =
    \langle \rho_2, (\skel_{\rho_3, \varepsilon}, \skel_5) \rangle$
    is 
    \begin{align*}
      \tilde{\rho}_2 \colon \quad 
      & \phantom{{}\to{}}
        (A^\ell(x_1), A^\tau(x_1, x_2, x_3), A^a, A^b, A^r(x_1)) \\
      &\to (\ell A^\ell(A^\tau(A^a, A^b, A^r(B_0^\ell(x_1)))),\, B_0^\tau(x_1,
        x_2, x_3),\, B_1^a,\, b,\, B_3^r(r x_1)) \enspace.
    \end{align*}
    Finally, we consider rule~$\rho_1$.  The only skeleton needed
    for~$S$ is~$S^\ell = \langle S, \ell, 0, \varepsilon \rangle$,
    which is the initial nonterminal of~$G^\ell$.  Substituting~$K$
    for~$A$ and $\skel_5(B_i)$ for $B_i$ in the right-hand side~$\ell A(B_0(B_1, b,
    B_3(re)))$ of~$\rho_1$, we obtain the tree~$u_2'[x_1 \gets e]$ and
    the new rule
    \[ \tilde{\rho}_1 \colon \quad S^\ell \to \ell A^\ell(A^\tau(A^a, A^b,
    A^r(B_0^\ell(B_0^\tau(B_1^a, b, B_3^r(re)))))) \enspace, \]
    where~$\tilde{\rho}_1 = \langle \rho_1, (\skel_{\rho_3, \varepsilon},
    \skel_5) \rangle$.  Thus, $G^\ell$~has the rules~$\{\tilde{\rho}_1,
    \tilde{\rho}_2, \tilde{\rho}_3, \tilde{\rho}_4, \tilde{\rho}_5, 
    \tilde{\rho}'_4, \tilde{\rho}'_5\}$.  Clearly, the tree~$K$
    generates the same terminal trees as~$A(x_1)$ and the tree $\skel_5(B_i)$
    generates the same terminal trees as~$\init(B_i)$ for every $i\in[3]$.  
    It is easy to
    check that $G^\ell$~is an $\{a, b\}$"~lexicalized MC"~TAG with
    respect to the smallest equivalence~$\equiv$ such that~$C^x \equiv x$
    for every $C\in\{S,A,B_0,B_0',B_1,B_1',B_3,B_3'\}$ and 
    every~$x \in \{\ell, \tau, a, b, r\}$. 

We finally mention that, in Example~\ref{exa:monadic}, the first rule of the grammar $G_2$ 
could be changed into 
the rule $S\to \ell B(b,B'(b,re))$, because $B(b,x_1)$ generates the same terminal trees as $A(x_1)$.
This makes the nonterminal $A$ superfluous. We have not done this, 
for the sake of illustration of our constructions. As a result of this change, 
the three rules $\tilde{\rho}_1,\tilde{\rho}_2, \tilde{\rho}_3$ of $G^\ell$ can be changed into 
the one rule 
$S^\ell \to \ell B_0^\ell(B_0^\tau(B_1^a, b, B_3^r(B_0'^\ell(B_0'^\tau(B_1'^a, b, B_3'^r(re))))))$.
\fin
\end{example}

\begin{example}
  \label{exa:foottotag2}
  \upshape
  As another, similar example, let us consider the $\{a,
  b\}$"~lexicalized MCFTG~$G$ obtained from~$G_2'$ by changing in its
  rules every~$\ell$ into~$\gamma\ell$ and every~$r$ (except the one
  in~$\rho_1$) into~$r\gamma$, where $\gamma$~has rank~$1$.  Thus,
  $G$~has the rules
  \begin{align*}
    \rho_1 \colon
    && S
    & \to \gamma \ell A(B_0(B_1, b, B_3(re)))
    & \rho_3 \colon
    && A(x_1)
    &\to \gamma\ell\tau(a, b, r \gamma x_1) \\
    \rho_2 \colon
    && A(x_1) 
    &\to \gamma \ell A(B_0(B_1, b, B_3(r\gamma x_1)))
    & \rho_5 \colon
    && B
    & \to \bigl(\gamma\ell \tau(x_1,x_2,x_3), a, r\gamma x_1 \bigr) \\
    \rho_4 \colon 
    && B 
    & \rlap{$\displaystyle {} \to \bigl(\gamma \ell B_0(x_1,x_2,x_3),\,
      B_1,\, B_3(B'_0(B'_1, b, B'_3(r \gamma x_1))) \bigr) \enspace,$}
  \end{align*}
  where $B = (B_0(x_1,x_2,x_3), B_1, B_3(x_1))$.  Clearly, $G$~is
  equivalent to the TAG~$\overline{G}_1$ of Example~\ref{exa:tag},
  for which there is no equivalent $\{a, b\}$"~lexicalized nsTAG. 

  Since $\rho_2$~and~$\rho_3$ are MC"~TAG rules with respect to~$A
  \equiv \gamma$, they do not have to be changed.  It is not difficult
  to see that the only skeleton function needed for~$(B_0, B_1,
  B_3)$ is~$\skel_5$ with~$\skel_5(B_0) =
  B_0^\gamma(B_0^\ell(B_0^\tau(x_1)))$, $\skel_5(B_1) = B_1^a$,
  and~$\skel_5(B_3) = B_3^r(B_3^\gamma(x_1))$, where~$B_0^\gamma =
  \langle B_0, \gamma, 1, \varepsilon \rangle$, $B_0^\ell = \langle
  B_0, \ell, 1, 1 \rangle$, $B_0^\tau = \langle B_0, \tau, 1, 11
  \rangle$, and similarly for~$B_3$, and~$B_1^a = \langle B_1, a, 0,
  \varepsilon \rangle$.  Given these skeletons, it is
  straightforward to construct the following rules for~$G^\gamma$: 

  \begin{alignat*}{5}
    \tilde{\rho}_1 \colon
    && \; S^\gamma
    & \to \gamma \ell A(B_0^\gamma(B_0^\ell(B_0^\tau(B_1^a, b,
    B_3^r(B_3^\gamma(re))))))
    & \tilde{\rho}_3 \colon
    && \; A(x_1)
    & \to \gamma \ell \tau(a, b, r\gamma x_1) \\
    \tilde{\rho}_2 \colon
    && \; A(x_1) 
    & \to \gamma \ell A(B_0^\gamma(B_0^\ell(B_0^\tau(B_1^a, b,
    B_3^r(B_3^\gamma(r \gamma x_1))))))
    & \hspace{0.25cm} \tilde{\rho}_5 \colon
    && \; \bar{B}
    & \to \bigl(\gamma x_1, \ell x_1, \tau(x_1,x_2,x_3), a, r x_1,
    \gamma x_1 \bigr) \\
    \rlap{$\displaystyle \hspace{-1.7em} \tilde{\rho}_4 \colon \;
      \bar{B} \to \bigl(\gamma \ell B_0^\gamma(x_1),\,
      B_0^\ell(x_1),\, B_0^\tau(x_1,x_2,x_3),\,  B_1^a,\,
      B_3^r(B_3^\gamma(B_0'^\gamma(B_0'^\ell(B_0'^\tau(B_1'^a,
      b, B_3'^r(B_3'^\gamma(rx_1))))))),\, \gamma x_1 \bigr)$}
  \end{alignat*}   
  where~$\bar{B} = (B_0^\gamma(x_1), B_0^\ell(x_1),
  B_0^\tau(x_1,x_2,x_3), B_1^a, B_3^r(x_1), B_3^\gamma(x_1))$.
  Clearly, $G^\gamma$~is an $\{a, b\}$"~lexicalized MC"~TAG equivalent
  to the TAG~$\overline{G}_1$. \fin
\end{example}

Let us say that a tree language~$L$ is \emph{root consistent}
if~$\rk(t_1(\varepsilon)) \neq \rk(t_2(\varepsilon))$
for all~$t_1, t_2 \in L$
such that $t_1(\varepsilon) \neq t_2(\varepsilon)$.  It should
be clear that every tree language in MC"~TAL is root consistent.

\begin{theorem} 
  \label{thm:mctal}
  For every MCFTG~$G$ such that $L(G)$~is root consistent, there is an
  \LDTR"~equivalent MC"~TAG~$G'$ such that 
  \[\mu(G') \leq
  \begin{cases} 
  \mu(G)  & \text{if } \wid(G) = 0 \\
  \mu(G) \cdot \abs \Sigma \cdot \mr^2_\Sigma \cdot (2 \cdot \wid(G) - 1)
  & \text{if } \wid(G) \geq 1 \enspace,
  \end{cases} \]
  where $\Sigma$~is the terminal
  alphabet of~$G$.  Moreover, if~$G$~is $\Delta$"~lexicalized, then so is~$G'$.
\end{theorem}

\begin{proof}
  With the help of Theorem~\ref{thm:footed}, we may assume that~$G =
  (N, \N, \Sigma, S, R)$ is a footed MCFTG.  The set~$\Omega =
  \{t(\varepsilon) \mid t \in L(G)\}$ can be
  computed by deciding the emptiness of~$L(G^\sigma)$ for
  every~$\sigma \in \Sigma$, where $G^\sigma$~is the MCFTG of
  Lemma~\ref{lem:foottotag}.  Now let $\sigma_0$~be an arbitrary
  element of~$\Omega$, and construct the adjoining MCFTG~$G^{\sigma_0}
  = (N', \N', \Sigma, S^{\sigma_0}, R')$ as in the proof of
  Lemma~\ref{lem:foottotag}.  From~$G^{\sigma_0}$ we construct~$G'$ by
  identifying all nonterminals~$\langle S, \sigma, 0, \varepsilon
  \rangle$ such that~$\sigma \in \Omega$ and taking the resulting
  nonterminal~$S'$ to be the initial nonterminal of~$G'$.  Since
  $G^{\sigma_0}$~is adjoining, it is straightforward to check that
  $G'$~is an MC"~TAG with respect to the smallest equivalence~$\equiv$
  such that~$\sigma_1 \equiv \sigma_2 \equiv S'$ for all~$\sigma_1,
  \sigma_2 \in \Omega$ and~$\langle C, \sigma, b, p \rangle \equiv
  \sigma$ for all~$\langle C, \sigma, b, p \rangle \in N'$.  It is
  easy to modify the \LDTR"~transducers $M$~and~$M'$ in the proof of
  Lemma~\ref{lem:foottotag} such that they show the \LDTR"~equivalence
  of $G$~and~$G'$. We finally note that if $\wid(G)=0$, then 
  $\mu(G^{\sigma_0})= \mu(G)$ by the proof of 
  Lemma~\ref{lem:foottotag}.
\end{proof}

We now can characterize MCFT~and~MC"~TAL in terms of each other in a
very simple way.  

\begin{corollary}
  \label{cor:charmctal}
  Let $\#$~be a new symbol of rank~$1$.  Then
  \[ \textup{MC"~TAL} = \{L \in \textup{MCFT} \mid L \text{ is root
    consistent}\} \qquad \text{and} \qquad 
  \textup{MCFT} = \{ L \mid
  \#(L) \in \textup{MC"~TAL}\} \enspace.\]  
\end{corollary}

\begin{proof}
  The first equality is immediate from Theorem~\ref{thm:mctal} and the
  fact that every tree language in MC"~TAL is root consistent.  It is
  easy to see that if~$L \in \text{MCFT}$, then $\#(L) \in
  \text{MCFT}$.  This also holds in the other direction because MCFT
  is closed under tree homomorphisms by Lemma~\ref{lem:cover}.  The
  second equality now follows from Theorem~\ref{thm:mctal} because
  $\#(L)$~is root consistent.
\end{proof}

As observed in the Introduction this corollary 
settles a problem stated in~\cite[Section~4.5]{wei88}, which 
can be reformulated as ``it would be interesting to investigate whether 
MC"~TAL is properly included in MCFT''. 
By the first statement of Corollary~\ref{cor:charmctal} that is indeed the case; 
i.e., MCFTGs are slightly more powerful than MC"~TAGs.
However, by the second statement they have the same power 
provided that MC"~TAGs are allowed to make use of a root-marker.
Another obvious way to ``force'' equality of MCFT and MC"~TAL is to allow 
MCFTGs, and hence MC"~TAGs, to use several initial nonterminals instead of just one.
It is clear that this does not change the class MCFT. 
Thus, the proper inclusion of MC"~TAL in MCFT is due to minor technicalities. 
For that reason we feel justified to state that MCFTGs and MC"~TAGs
have the same tree generating power. 

As another corollary we obtain from Theorems
\ref{thm:mctal}~and~\ref{thm:main} that MC"~TAGs
can be (strongly) lexicalized.  Thus,
although TAGs cannot be strongly lexicalized, as proved
in~\cite{kuhsat12} (cf.\@ Remark~\ref{rem:nstagnonlex}), MC"~TAGs
can.  This was illustrated in Example~\ref{exa:foottotag2}. 

\begin{theorem}
  \label{thm:lexmctag}
  For every finitely $\Delta$"~ambiguous MC"~TAG~$G$ there is an \LDTR"~equivalent
  $\Delta$"~lexicalized MC"~TAG~$G'$ such that 
  $\mu(G') \leq
  (\mu(G) + \mr_\Delta) \cdot \abs \Sigma \cdot \mr^2_\Sigma \cdot (2 \cdot \wid(G) + 1),$
  where $\Sigma$~is the terminal alphabet of~$G$.
\end{theorem}

\subsection{Monadic MCFTGs}
\label{sub:monadic}
\noindent 
We say that an MCFTG $G$ is \emph{monadic} if $\wid(G) \leq 1$.
For instance, the grammars of Examples~\ref{exa:copy}, \ref{exa:main},  
and~\ref{exa:monadic} are monadic.
As observed in the beginning of this section, it is shown in~\cite{keprog11} that 
nsTAGs have the same tree generating power as monadic spCFTGs. 
Similarly, on the basis of Theorem~\ref{thm:mctal}, we can now prove that 
MCFTGs have the same tree generating power as monadic MCFTGs.
The construction in the proof is the same as in~\cite{gebost15}. 

\begin{theorem}
  \label{thm:monadic}
  For every MCFTG~$G$ with $\wid(G) \geq 2$ there is an \LDTR"~equivalent 
  monadic MCFTG~$G'$ such that
  $\mu(G') \leq \mu(G) \cdot \abs \Sigma \cdot \mr^2_\Sigma \cdot (2 \cdot \wid(G) - 1)$, 
  where $\Sigma$~is the terminal
  alphabet of~$G$.  Moreover, if~$\Delta \subseteq \Sigma^{(0)}$ and
  $G$~is $\Delta$"~lexicalized, then $G'$~is
  $\Delta$"~lexicalized. 
  \end{theorem}
  
\begin{proof} 
It should be clear from Lemma~\ref{lem:foottotag} and the proof of Corollary~\ref{cor:charmctal} 
that we may assume that $G = (N,\N,\Sigma,S,R)$ is an adjoining MCFTG with respect to
a mapping $\varphi \colon N\cup \Sigma \to \Sigma$, 
as defined in Section~\ref{sub:mcftismctal}.\footnote{Otherwise, we replace every initial rule 
$S\to (u,\LL)$ by $S\to (\#(u),\LL)$ and after the construction remove $\#$ by Lemma~\protect{\ref{lem:cover}}.} 
We define the monadic $G' = (N,\N,\Sigma,S,R')$ such that every nonterminal $C\in N$ 
with $\rk(C)\geq 2$ in $G$ now has rank $\rk'(C) = 1$ in~$G'$, and 
$\rk'(C)=\rk(C)$ for the nonterminals with $\rk(C)\leq 1$.
The idea of the proof is that every occurrence of a nonterminal $C(\seq x 1m)$ of rank $m\geq 1$ 
is replaced by $C(\sigma(\seq x 1m))$ where $\sigma = \varphi(C)$, such that in $G'$ the nonterminal 
$C$ does not generate the foot node of the tree generated by $C$ in $G$.   
Thus, for a footed pattern $t\in P_{N \cup \Sigma}(X)$ of rank at least 1, 
let $\cut(t)$ denote the unique pattern of rank~1 
such that $t = \cut(t)[x_1\gets \init(\flab(t))]$. For instance, 
$\cut(\sigma(a,\tau(x_1,x_2))) = \sigma(a,x_1)$. Moreover, for simplicity, let $\cut(t)=t$
for every tree $t\in T_{N \cup \Sigma}$. 
Now let $\rho = A \to ((\seq u 1n),\LL)$ be a rule in $R$ with $A=(\seq A 1n)$, and let $f$ be the 
substitution function for~$N$ such that $f(C) = C(\init(\varphi(C)))$  
if $\rk(C)\geq 1$ and $f(C)=C$ if $\rk(C)=0$, for every $C\in N$. 
Then $R'$ contains the rule $\rho' = A \to ((\seq {u'} 1n),\LL)$
where $u'_j = \cut(u_j[f])$ for every $j\in[n]$.
It can be shown that $L(G',A) = \{(\cut(t_1),\dotsc,\cut(t_n))\mid (\seq t 1n)\in L(G,A)\}$
and so $L(G')=L(G)$. The formal proof, together with the proof of \LDTR"~equivalence, 
is left to the reader.  

  If $G'$~is $\Delta$"~lexicalized and~$\Delta \subseteq \Sigma^{(0)}$, 
  then $G$~is $\Delta$"~lexicalized. In fact, 
  the right-hand sides of~$\rho$ and~$\rho'$ contain the same elements of~$\Delta$ 
  because the only symbols that are removed or added have rank at least~$2$.
  We also observe that, for unrestricted $\Delta\subseteq \Sigma$,
  if $G$~is $n$"~$\Delta$"~lexicalized for~$n > \mu(G)$, 
  as defined before Lemma~\ref{lem:nlex}, then $G'$~is $(n -
   \mu(G))$"~$\Delta$"~lexicalized. In fact, in the definition of~$\rho'$
  we have that for every $j\in[n]$, 
   $\abs{\pos_\Delta(u_j[f])} \geq \abs{\pos_\Delta(u_j)}$ and 
   $\abs{\pos_\Delta(u'_j)}\geq \abs{\pos_\Delta(u_j[f])}-1$. 
\end{proof}

For unrestricted $\Delta$ this theorem also holds except that $G'$ 
is just equivalent to $G$, not necessarily \LDTR"~equivalent. 
This follows from Lemma~\ref{lem:nlex} and the last paragraph of 
the proof of Theorem~\ref{thm:monadic}. Thus,
for every $\Delta$"~lexicalized MCFTG~$G$ with $\wid(G) \geq 2$ there is an equivalent 
$\Delta$"~lexicalized monadic MCFTG~$G'$ such that
$\mu(G') \leq \mu(G) \cdot \abs \Sigma \cdot \mr^2_\Sigma \cdot (2 \cdot \wid(G) - 1)$.

\begin{example}
\upshape
We consider the MCFTG $G = (N,\N,\Sigma,S,R)$ with $N=\{S,A^{(2)},B^{(2)}\}$, 
$\N=\{S,(A,B)\}$, $\Sigma=\{\sigma^{(2)},\tau^{(2)},a^{(0)},b^{(0)},e^{(0)}\}$, 
and the rules 
  \begin{align*}
    S 
    &\to  A(a,B(e,b)) \\
    (A(x_1, x_2),B(x_1,x_2))
    &\to (\sigma(a,A(x_1,x_2)),\,B(\tau(x_1,x_2),b))  \\
    (A(x_1, x_2),B(x_1,x_2))
    &\to (\sigma(x_1,x_2),\,\tau(x_1,x_2)) \enspace. 
  \end{align*}
It generates the tree language $L(G)=\{(\sigma a)^n \tau^n e b^n \mid n\geq 1\}$. 
Note that we here use string notation. Thus, e.g., $(\sigma a)^2 \tau^2 e b^2$ is the tree
$\sigma a \sigma a \tau \tau e b b$ which can be written as the term 
$\sigma(a,\sigma(a,\tau(\tau(e,b),b)))$. 
Obviously, $G$ is an adjoining MCFTG with $\varphi(S)=\varphi(A)=\sigma$ and $\varphi(B)=\tau$. 
The equivalent monadic grammar $G'$ as constructed in the proof of Theorem~\ref{thm:monadic} 
has the rules 
  \begin{align*}
    S 
    &\to  A(\sigma(a,B(\tau(e,b)))) \\
    (A(x_1),B(x_1))
    &\to (\sigma(a,A(x_1)),\,B(\tau(x_1,b)))  \\
    (A(x_1),B(x_1))
    &\to (x_1,\,x_1) \enspace. 
  \end{align*}
Note that $G'$ is not footed. 
\fin
\end{example}

As observed in the Introduction, Theorem~\ref{thm:monadic} does not hold for spCFTGs; 
i.e., spCFTGs do not have the same tree generating power as monadic spCFTGs. In fact, 
it is shown in~\cite[Theorem~6.5]{engrozslu80} (see also~\cite[Lemma~24]{lohmansch12}) that 
spCFTGs (and arbitrary context-free tree grammars) give rise to a strict hierarchy 
with respect to $\wid(G)$. It is shown in~\cite[Theorem~10]{lohmansch12}) that every 
``straight-line'' spCFTG can be transformed into an equivalent monadic one in polynomial time; 
the construction is similar to the one for Theorem~\ref{thm:monadic} 
(in particular to the one in the proof of Theorem~\ref{thm:footed}). 

We finally observe that some tree languages in MCFT cannot be generated by an
MCFTG that is both monadic and footed.  An example is the
language~$L = \{(ca)^n (da)^n e \mid n \in \nat_0\}$ that is generated
by the spCFTG with rules~$S \to A(e)$, $A(x_1) \to c(a, A(d(a,
x_1)))$, and~$A(x_1) \to x_1$.  If $G$~is a monadic footed MCFTG
with~$L(G) = L$, then $G$~must be an MRTG because there is
no terminal symbol of rank~1.\footnote{We already observed below
  Definition~\protect{\ref{def:footed}} that every tree of the
  forest~$(\seq t1n) \in L(G, (\seq A1n))$ is footed.  Suppose that
  $\rk(A_j) = 1$ for some $j \in [n]$, then the corresponding
  tree~$t_j \in P_{\Sigma}(X_1)$ only contains terminal symbols (and
  the variable~$x_1$).  The foot node label of~$t_j$ must have
  rank~$1$, but the ranked alphabet~$\Sigma$ does not contain a unary  
  symbol.  Hence no unary nonterminal can be useful.}  It follows from
Theorem~\ref{thm:charact} in Section~\ref{sec:charact}
and~\cite[p.~277]{rou70} that all tree languages in~MRT have regular
``path languages''.  However, the intersection of the path language
of~$L$ with~$c^* d^* e$ is~$\{c^n d^n e \mid n \in \nat_0\}$, which is
not regular. Thus, $L$ is not in MRT (see also the last paragraph of
Section~\ref{sec:charact}).

\section{Multiple context-free grammars}
\label{sec:mcfg}
\noindent
In this section we define the multiple context-free (string) grammars (MCFG)
of~\cite{sekmatfujkas91,shaweijos87}. We first prove that MCFGs can be lexicalized.
Then we prove that every tree language in MCFT 
can be generated by an MCFG, which is possible because we defined $T_\Sigma$
as a subset of $\Sigma^*$. Using this we prove that MCFTGs have the same 
string generating power as MCFGs, by taking the yields of the generated tree languages. 
Moreover, we show that MCFTGs can be parsed in polynomial time.

\subsection{String generating power of MCFTGs}
\label{sub:stringgen}
\noindent
To avoid the formalities involved in defining MCFGs in the classical way,
we define them as a special case of~MCFTGs. 
We introduce a special symbol~$\e$ of rank~$0$
and we identify, as usual, the strings over a finite (unranked) alphabet~$\Sigma$
with the trees  over the ``monadic'' ranked alphabet~$\Sigma \cup
\{\e\}$, where every symbol in~$\Sigma$ has rank~$1$.  Thus, $w \in
\Sigma^*$~is identified with~$w \e \in T_{\Sigma \cup \{\e\}}$.  

A \emph{multiple context-free grammar} (in short,~MCFG) is an MCFTG~$G
= (N \cup \{S\}, \N, \Sigma \cup \{\e\}, S, R)$ such that~$S \notin
N$, every nonterminal in~$N$ has rank~$1$, $\e \notin \Sigma$, and
every terminal in~$\Sigma$ has rank~$1$.  We also require (without
loss of generality) that $G$~is start-separated; i.e., that $S$~does
not occur in the right-hand sides of rules.  With the above
identification we have~$L(G) \subseteq \Sigma^*$, and for every~$A \in
\N \setminus \{S\}$ we have $L(G, A) \subseteq
P_\Sigma(X_1)^{\scriptscriptstyle +}$~and~$P_\Sigma(X_1) = \Sigma^*
x_1$.  Note that every rule of~$G$ is either of the form~$S \to (u \e,
\LL)$ with~$u \in (N \cup \Sigma)^*$ or of the form~$(\seq A1n) \to
((u_1 x_1, \dotsc, u_nx_1), \LL)$ where $\seq A1n \in N$~and~$\seq u1n
\in (N \cup \Sigma)^*$.  For a uniquely $N$"~labeled tree~$t = vCw
\e$ (or~$vCw x_1$) with $v, w \in (N \cup \Sigma)^*$~and~$C \in N$,
the rewriting of~$C$ by~$u x_1$ with~$u \in (N \cup \Sigma)^*$ results
in the tree~$t[C \gets u x_1]$, which equals~$vuw \e$ (or~$vuw x_1$);
thus, it is the usual rewriting of a nonterminal in a sentential form
of a context-free grammar.  It is straightforward to see that this
definition of~MCFG is equivalent to the classical notion of multiple
context-free grammar~\cite{sekmatfujkas91,shaweijos87}, taking
into account the information-lossless condition~(f3)
of~\cite[Lemma~2.2]{sekmatfujkas91}.  The class of languages generated
by~MCFGs will be denoted by~MCF.  

Through the above identification of strings with monadic trees,
MCFTGs can also generate strings directly as opposed to taking yields of the generated trees. 
In the next lemma we show that every MCFTG that generates strings in this way, has an equivalent MCFG. 

\begin{lemma}
  \label{lem:strmcftg}
  For every MCFTG~$G$ with terminal alphabet~$\Sigma \cup \{\e\}$,
  where every symbol in~$\Sigma$ has rank~$1$, there is an
  \LDTR"~equivalent MCFG~$G'$.  Moreover,~$\mu(G') = \mu(G)$.
\end{lemma}

\begin{proof}
  Due to the specific form of the terminal alphabet, it should be
  clear that reachable and useful big nonterminals cannot contain
  nonterminals of rank strictly larger than~$1$.  Consequently, we may
  assume that~$G$ is monadic without the help of
  Theorem~\ref{thm:monadic}.  We transform~$G$ into an MCFG~$G'$ with
  the same big nonterminals and the same nonterminals, which all have
  rank~$1$ in~$G'$ except for the initial nonterminal~$S$ of rank~$0$.
  Additionally, in the right-hand side of every initial rule we
  replace every occurrence of a nullary nonterminal~$C$ by~$C(\e)$,
  and in the right-hand side of every non-initial rule we replace
  every occurrence of a nullary nonterminal~$C$ by~$C(x_1)$ and every
  occurrence of~$\e$ by~$x_1$.
\end{proof}

Let strMCFT~denote the class of all string languages
generated by~MCFTGs, where strings over~$\Sigma$ are viewed as
monadic trees over~$\Sigma \cup \{\e\}$ as explained above.

\begin{corollary}
  \label{cor:strmcftg}
  $\textup{strMCFT} = \textup{MCF}$. 
\end{corollary}

Another consequence of Lemma~\ref{lem:strmcftg} is that MCFGs can be
lexicalized, as stated in~\cite[Section~4.4]{yos06} for the case 
$\Delta = \Sigma$. This should be contrasted to the fact that 
context-free grammars cannot be $\Delta$"~lexicalized for every $\Delta$, 
as shown in Remark~\ref{rem:nstagnonlex}. 

\begin{corollary}
  \label{cor:ambmcf}
  For every finitely $\Delta$"~ambiguous MCFG~$G$ there is a
  $\Delta$"~lexicalized MCFG $G'$ that is \LDTR"~equivalent to $G$.  Moreover,
  $\mu(G') = \mu(G) + 1$.
\end{corollary}

\begin{proof}
  By Theorem~\ref{thm:main} there is an \LDTR"~equivalent
  $\Delta$"~lexicalized MCFTG~$G'$ such that~$\wid(G') =
  2$ and $\mu(G') = \mu(G) + 1$.  Next we apply
  Lemma~\ref{lem:strmcftg}.
\end{proof}

\begin{example}
  \label{exa:ambmcf}
  \upshape
  Consider the context-free grammar~$G$ with rules~$S \to \ell AAr$,
  $A \to \ell AAr$, and~$A \to \ell \tau abr$ 
(cf.\@ Example~\ref{exa:monadic} and Remark~\ref{rem:nstagnonlex}).  Obviously, we may
  view~$G$ as an MCFG of multiplicity~$1$ with an alias~$A'$ of~$A$.
  Its terminal alphabet is~$\Sigma \cup \{\e\}$ with~$\Sigma =
  \{\tau^{(1)}, \ell^{(1)}, r^{(1)}, a^{(1)}, b^{(1)}\}$, and its
  rules for $S$~and~$A$ are
  \[ S \to \ell A(A(r \e)) \qquad A(x_1) \to \ell A(A'(rx_1)) \qquad
  \text{and} \qquad A(x_1) \to \ell \tau abr x_1 \enspace. \]
  Let~$\Delta = \{a, b\}$.  Since $G$~is $\Delta$"~growing, it has
  finite $\Delta$"~ambiguity.  Applying a slightly simplified version
  of the proof of Corollary~\ref{cor:ambmcf}, we obtain a
  $\Delta$"~lexicalized MCFG~$G'$ of multiplicity~$2$ such that~$L(G')
  = L(G)$. It has the big nonterminals~$\{S, A, (B,C), (B',C')\}$,
  where $(B', C')$~is an alias of~$(B, C)$, and the following rules,
  in which we omit $\e$~and~$x_1$ (and all the parentheses in trees) for readability:    
  \begin{align*}
    S 
    &\to \ell ABbCr 
    & A
    &\to \ell ABbCr
    & A
    &\to \ell\tau abr
    & (B, C)
    &\to (\ell B,\, CB'bC'r)
    & (B, C)
    &\to (\ell\tau a,\, r) \enspace.
  \end{align*}
  Clearly, the string~$BbC$ generates the same terminal strings as~$A$. \fin
\end{example}

In the next theorem, we show that every tree language that is generated by
an~MCFTG can also be generated by an~MCFG provided that we change the
ranks of the terminal symbols.  In this theorem we (temporarily)
identify each tree~$t$ over the ranked alphabet~$\Sigma$ (which is defined
as a string over the unranked alphabet~$\Sigma$) 
with the tree~$t\e$ over the ranked alphabet~$\Sigma \cup
\{\e\}$, in which every symbol of~$\Sigma$ has rank~$1$.
As an example, the tree $\sigma(a,b)=\sigma ab$ 
is identified with the tree $\sigma ab\e=\sigma(a(b(\e)))$.
The idea behind the proof is essentially the same as the one
of~\cite[Theorem~15]{engman02}.  In the case of an~spCFTG, the
resulting~MCFG is well-nested
(see~\cite{kan09,kan09b,mon10,rodkuhsat10}).

\begin{theorem}
  \label{thm:mcft-to-mcf}
  For every MCFTG~$G$ there is an \LDTR"~equivalent MCFG~$G'$.  If
  $G$~is $\Delta$"~lexicalized, then so is~$G'$.  Moreover,~$\mu(G') =
  \mu(G) \cdot (\wid(G) + 1)$ and $\lambda(G')=\lambda(G)$. 
  If $G$ is footed (i.e., is an nsMC"~TAG) then $\mu(G') = 2\cdot\mu(G)$. 
\end{theorem}

\begin{proof}
  By Lemma~\ref{lem:permfree} we may assume that $G = (N, \N, \Sigma,
  S, R)$~is permutation-free.  We will define the MCFG~$G' = (N' \cup
  \{S'\}, \N' \cup \{S'\}, \Sigma \cup \{\e\}, S', R')$, where $S'$~is
  a new nonterminal and all the symbols in~$\Sigma$ now have rank~$1$.
  First of all, we let~$N' = \{\langle C, i \rangle \mid C \in N,\, 0
  \leq i \leq \rk(C)\}$.
  For every~$C \in N^{(k)}$ the intuition behind this is that $\langle C, i
  \rangle(x_1)$~generates the string~$w_i x_1$, when $C(\seq
  x1k)$~generates (as part of a big nonterminal) the terminal
  tree~$w_0 x_1 w_1 \dotsm x_k w_k \in \PF_\Sigma(X_k)$ with~$\seq w1k
  \in \Sigma^*$.  For every~$C \in N^{(k)}$, let its expansion be $\exp(C)
  = (\langle C, 0 \rangle, \langle C, 1 \rangle, \dotsc, \langle C, k
  \rangle) \in (N')^{\scriptscriptstyle +}$, and for every~$A = (\seq
  A1n) \in \N$, let~$\exp(A) = \exp^*(A) = \exp(A_1) \dotsm \exp(A_n) \in
  (N')^{\scriptscriptstyle +}$ be the concatenation of the expansions
  of its nonterminals.  Then we define~$\N' = \{\exp(A) \mid A \in
  \N\}$.

  In the remainder of this proof we need the following two bijections $\pi$~and~$\lambda$. 
  The right-hand side forest~$u = (u_1x_1, \dotsc, u_nx_1) \in 
  P_{N' \cup \Sigma}(X_1)^{\scriptscriptstyle +}$ of a possible non-initial rule of~$G'$ 
  is in one-to-one correspondence with the string~$\pi(u) =
  u_1x_1 \dotsm u_nx_1 \in (N' \cup \Sigma \cup X_1)^*$ that ends
  on~$x_1$ and with the sequence~$\lambda(u) = (\seq u1n)$ of
  strings~$\seq u1n \in (N' \cup \Sigma)^*$.  
  For the definition of the rules of~$G'$ we need the
  expansion of the right-hand side forests of the rules of~$G$.
  For every~$t \in T_{N \cup \Sigma}(X)$ we define $\exp(t) =
  \pi^{-1}(\exp'(t) \cdot x_1) \in P_{N' \cup
    \Sigma}(X_1)^{\scriptscriptstyle +}$, where $\pi$~is the bijection
  defined above and where $\exp'(t) \in (N' \cup \Sigma \cup
  X_1)^*$ is defined inductively as follows:
  \[ \exp'(t) =
  \begin{cases}
    x_1 & \text{if } t \in X \\
    \sigma \cdot \exp'(t_1) \dotsm \exp'(t_k) & \text{if } t =
    \sigma(\seq t1k) \text{ with } \sigma \in \Sigma \\
    \langle C, 0 \rangle \cdot \exp'(t_1) \cdot \langle C, 1 \rangle
    \dotsm \exp'(t_k) \cdot \langle C, k \rangle & \text{if } t =
    C(\seq t1k) \text{ with } C \in N \enspace.
  \end{cases}
  \]
  We note that~$\exp'(t) = t[x \gets x_1 \mid x \in X]$ 
  if~$t \in T_\Sigma(X)$.  Given~$t = (\seq t1n) \in T_{N \cup
    \Sigma}(X)^{\scriptscriptstyle +}$ we let~$\exp(t) = \exp^*(t) = \exp(t_1)
  \dotsm \exp(t_n)$ be the concatenation of the expansions of its elements. 

  Now, if~$\rho = A \to (u, \{\seq B1k\}) \in R$, then $R'$~contains the non-initial rule 
  \[\rho_{\exp} = \exp(A) \to
  (\exp(u), \{\exp(B_1), \dotsc, \exp(B_k)\}) \enspace.\]  
  Clearly, the
  rule~$\rho$ can be reconstructed from~$\rho_{\exp}$.   
  Finally we define the initial rules of $G'$. 
  If $\rho_{\exp}=$ $\langle S,0\rangle \to (vx_1,\LL)$ is a rule in $R'$ 
  as constructed above for $A=S$,
  then $R'$ contains the additional rule $\rho'_{\exp}=S' \to (v\e,\LL)$.\footnote{The rule
$\rho_{\exp}=$ $\langle S,0\rangle \to (vx_1,\LL)$ is then superfluous, but we keep it to simplify
the correctness proof.}
  At this point, we completed the construction of~$G'$.  To prove its
  correctness we need the following claim.

  \paragraph{Claim} Given a tree~$t \in T_{N \cup \Sigma}(X)$, a
  repetition-free sequence~$(\seq A1n) \in N^n$, and permutation-free
  patterns~$\seq s1n \in \PF_\Sigma(X)$ such that~$\rk(A_i) = \rk(s_i)$
  for every~$i \in [n]$, we have
  \[ \exp(t[(\seq A1n) \gets (\seq s1n)]) = \exp(t)[\exp(A_1) \dotsm
  \exp(A_n) \gets \exp(s_1) \dotsm \exp(s_n)] \enspace. \]

  \paragraph{Proof of claim} It can first be shown that $\exp'(t[(\seq
  A1n) \gets (\seq s1n)]) = h^*(\exp'(t))$, where $h$~is the string
  homomorphism over~$N' \cup \Sigma \cup X_1$ such that the string
  $h(\langle A_i, j \rangle)$ is the $(j+1)$"~th element of the
  sequence~$\lambda(\exp(s_i))$, with the bijection $\lambda$
  defined above, for every $i \in [n]$~and~$\langle A_i,
  j \rangle \in N'$.  Moreover, $h$~is the identity for the remaining
  elements of~$N' \cup \Sigma \cup X_1$.   The straightforward proof
  is left to the reader; it is by induction on the structure of~$t$
  using the obvious fact that $\exp'(s[x_i \gets t_i \mid 1 \leq i
  \leq k] = g^*(s)$ for all $s \in T_\Sigma(X_k)$~and~$\seq t1k \in T_{N
    \cup \Sigma}(X)$, where $g$~is the string homomorphism
  over~$\Sigma \cup X_k$ such that~$g(x_i) = \exp'(t_i)$ for~$i \in
  [k]$ and~$g(\sigma) = \sigma$ for~$\sigma \in \Sigma$.  Thus, the
  left-hand side of the equation is~$\pi^{-1}(h^*(\exp'(t)) \cdot x_1)$.
  We now observe that this is equal to~$\pi^{-1}(h^*(\pi(\exp(t))))$,
  which clearly equals the right-hand side of the equation.   
  \emph{This proves the claim.}

\medskip
  For every derivation tree~$d \in L(G_\der, A)$ there is a derivation
  tree~$d' \in L(G'_\der, \exp(A))$ such that $\val(d') =
  \exp(\val(d))$.  In fact, $d'$~is obtained from~$d$ by changing
  every label~$\rho$ simply into~$\rho_{\exp}$.  Let~$d = \rho(\seq
  d1k)$, and let~$d'_i$ be such that~$\val(d'_i) = \exp(\val(d_i))$.
  Now consider~$d' = \rho_{\exp}(\seq{d'}1k)$. Then we have
  \begin{align*}
    \val(d')
    &= \exp(u)[\exp(B_1) \dotsm \exp(B_k) \gets \exp(\val(d_1))
      \dotsm \exp(\val(d_k))] \text{ and }\\
    \val(d) 
    &= u[\word B1k \gets \val(d_1) \dotsm \val(d_k)] \enspace.
  \end{align*}
  Thus~$\val(d') = \exp(\val(d))$ by the above claim.  Hence if~$d
  \in L(G_\der)$, then $d'\in L(G'_\der,\langle S,0\rangle)$
  and hence $d''\in L(G'_\der)$ where $d''$ is obtained from $d'$ 
  by priming the label of its root. Moreover,
  $\val(d')=\exp(\val(d))=\val(d)x_1$. Hence, by Lemma~\ref{lem:treehomsub},
  $\val(d'')=\val(d')[x_1 \gets \e]=\val(d)\e = \val(d)$.
  This shows that~$L(G) \subseteq
  L(G')$.  Clearly, there is a two-state LDT"~transducer that
  transforms~$d$ into~$d''$. In fact, it is a finite-state relabeling. 
  Since, obviously, every
  derivation tree in~$L(G'_\der)$ is of the form~$d''$ with~$d
  \in L(G_\der)$, it also follows that~$L(G') \subseteq L(G)$.
  Clearly, there is a one-state LDT"~transducer that
  transforms~$d''$ into~$d$ by changing
  every~$\rho_{\exp}$ and~$\rho'_{\exp}$ into~$\rho$.
  
  If $G$ is footed, then the nonterminals $\langle C,1\rangle,\dotsc,\langle C,k-1\rangle$,
  where $k=\rk(C)$, are superfluous because they always generate $x_1$. 
  Thus, in this case it suffices to define $\exp(C)=(\langle C,0\rangle,\langle C,k\rangle)$
  and adapt the construction accordingly. The resulting construction is similar to the one 
  described in~\cite[Section~4.5.1]{wei88} where it is shown that the yield of 
  a tree language in MC"~TAL is in MCF. 
\end{proof}

\begin{example}
\upshape
  We consider the permutation-free MCFTG~$G = (N, \N, \Sigma, S, R)$ with
  $N = \{S, A^{(2)}, B^{(0)}\}$, $\,\N =
  \{S, (A, B)\}$, $\,\Sigma = \{\sigma^{(2)}, \alpha^{(0)},
  \beta^{(0)}, \gamma^{(0)}\}$, and the following three rules:
  \begin{align*}
    S 
    &\to \sigma(A(\alpha, \beta), B) \\
    \bigl(A(x_1, x_2), B \bigr) 
    &\to \bigl(\sigma(\alpha, A(\sigma(\beta, x_1), \sigma(\gamma,
      x_2))),\, \sigma(B, \beta) \bigr)
    & \bigl(A(x_1, x_2), B \bigr)
    &\to \bigl(\sigma(x_1, x_2),\, \gamma \bigr) \enspace. 
  \end{align*}
  Clearly, $L(G, (A, B))$~consists of all
  forests~$\bigl((\sigma \alpha)^n \sigma(\sigma \beta)^n x_1 (\sigma
  \gamma)^n x_2,\, \sigma^n \gamma \beta^n)$ with~$n \in \nat_0$.
  Consequently, $L(G) = \{\sigma(\sigma \alpha)^n \sigma(\sigma
  \beta)^n \alpha(\sigma \gamma)^n \beta \sigma^n \gamma \beta^n \mid
  n \in \nat_0\}$.  The MCFG~$G'$ constructed in the proof of
  Theorem~\ref{thm:mcft-to-mcf} has the following four rules
  (in which we omit all parentheses in trees):
  \begin{align*}
    S'
    &\to \sigma \langle A, 0 \rangle \alpha \langle A, 1 \rangle \beta
      \langle A, 2 \rangle \langle B, 0 \rangle \e \\
    \langle S, 0 \rangle x_1
    &\to \sigma \langle A, 0 \rangle \alpha \langle A, 1 \rangle \beta
      \langle A, 2 \rangle \langle B, 0 \rangle x_1 \\
    (\langle A, 0 \rangle x_1, \langle A, 1 \rangle x_1,
    \langle A, 2 \rangle x_1, \langle B, 0 \rangle x_1)
    &\to (\sigma \alpha \langle A, 0 \rangle \sigma \beta x_1,
      \,\langle A, 1 \rangle \sigma \gamma x_1, \,\langle A, 2
      \rangle x_1, \,\sigma \langle B, 0 \rangle \beta x_1) \\
    (\langle A, 0 \rangle x_1, \langle A, 1 \rangle x_1,
    \langle A, 2 \rangle x_1, \langle B, 0 \rangle x_1)
    &\to (\sigma x_1, \,x_1, \,x_1, \,\gamma x_1) \enspace.
    \end{align*}
    Clearly,
    $L(G', (\langle A, 0 \rangle, \langle A, 1 \rangle, \langle A, 2
    \rangle, \langle B, 0 \rangle)) = \bigl\{\bigl((\sigma \alpha)^n
    \sigma(\sigma \beta)^n x_1,\, (\sigma \gamma)^n x_1,\, x_1,\,
    \sigma^n \gamma \beta^n x_1 \bigr) \mid n \in\nat_0 \bigr\}$  
    and hence~$L(G') = L(G)$. The second rule of $G'$ is of course superfluous.

    For the next lemma and corollary we note that 
    the MCFG $G''$ that is obtained from~$G'$ by removing~$\sigma$ 
    and thus has the rules
  \begin{align*}
    S'
    &\to \langle A, 0 \rangle \alpha \langle A, 1 \rangle \beta
      \langle A, 2 \rangle \langle B, 0 \rangle \e \\
    (\langle A, 0 \rangle x_1, \langle A, 1 \rangle x_1,
    \langle A, 2 \rangle x_1, \langle B, 0 \rangle x_1)
    &\to (\alpha \langle A, 0 \rangle \beta x_1,
      \,\langle A, 1 \rangle \gamma x_1, \,\langle A, 2
      \rangle x_1, \,\langle B, 0 \rangle \beta x_1) \\
    (\langle A, 0 \rangle x_1, \langle A, 1 \rangle x_1,
    \langle A, 2 \rangle x_1, \langle B, 0 \rangle x_1)
    &\to (x_1, \,x_1, \,x_1, \,\gamma x_1) \enspace, 
    \end{align*}
generates the string language 
$\yield(L(G)) = \{\alpha^n \beta^n \alpha \gamma^n \beta
    \gamma \beta^n \mid n \in \nat_0\}$. \fin
\end{example}

Theorem~\ref{thm:mcft-to-mcf} suggests that we do not need~MCFTGs at
all, because MCFGs~can generate the ``same'' languages.  However,
the~MCFTG is a way of guaranteeing that all intermediate results
during the generation process are trees, which supports the structured
generation of the trees.  

It follows from
Theorem~\ref{thm:mcft-to-mcf} that known properties of MCF~languages
(see, e.g., \cite{shaweijos87,sekmatfujkas91,kall10}) 
also hold for MCFT~tree languages.  Thus $\text{MCFT} \subseteq
\text{LOG}(\text{CFL})$; i.e.,  the recognition problem for an
MCFT~tree language is log"~space reducible to that of a context-free
string language.  Also, every tree language generated by an MCFTG~$G$
can be parsed in polynomial time by first parsing the given tree
according to the MCFG~$G'$ of Theorem~\ref{thm:mcft-to-mcf} in
polynomial time and then transforming the resulting derivation tree
of~$G'$ by the corresponding \LDTR"~transducer into one of~$G$ in
linear time.  This will be discussed in more detail in
Section~\ref{sub:parsing}.  Additionally, every MCFT~tree language
is semi-linear.

Next we show that MCFGs~generate exactly the yield languages of the
tree languages generated by~MCFTGs. 
We recall that the yield of a tree $t\in T_\Sigma$ is defined 
as~$\yield(t) = \yield_{\Sigma^{(0)} \setminus \{e\}}(t)$,
where~$e$ is a special symbol~$e$ of rank~$0$ that 
satisfies~$\yield(e) = \varepsilon$. 
For a class~$\X$ of tree
languages, let $\text{y}\X$ be the class of all languages~$\yield(L)$
with~$L \in \X$.
Thus, we will show that $\text{yMCFT} = \text{MCF}$. 
In fact, this is already a consequence of (the second equation of) 
Corollary~\ref{cor:charmctal} in Section~\ref{sub:mcftismctal},
which implies that $\text{yMCFT}=\text{yMC"~TAL}$, 
and the equation $\text{yMC"~TAL} = \text{MCF}$ 
which was shown in~\cite{wei88}.\footnote{The equality 
$\text{yMCFT} = \text{MCF}$ is also stated
in~\protect{\cite[Theorem~1]{boukalsal12}}.}
We additionally prove \LDTR "~$\yield$"~equivalence, for which we
refer to Definition~\ref{def:LDTReq}. 
In the first half of the next lemma we consider, more generally,
a subset $\Delta\subseteq \Sigma$ of lexical symbols and we prove 
\LDTR "~$\yield_\Delta$"~equivalence (where $\yield_\Delta$ is defined 
in the paragraph on homomorphisms in Section~\ref{sub:seqs}). This
general case will be used in the proof of Theorem~\ref{thm:mcftgparse}.

\begin{lemma}
  \label{lem:ymcft}
  Let~$\Delta \subseteq \Sigma$.
  \begin{compactenum}[(1)]
  \item For every MCFTG~$G$ there is an MCFG~$G''$ that is
    \LDTR"~$\yield_\Delta$"~equivalent to~$G$. 
  \item For every MCFG~$G$ there is an MRTG~$G_1$ such that $G$~is
    \LDTR"~$\yield$"~equivalent to~$G_1$.
  \end{compactenum}
\end{lemma}
\begin{proof}
  It is straightforward to generalize the well-known proofs for~RTGs
  and context-free grammars (see, e.g., \cite[Theorem~3.28]{eng75b}). 
    To prove statement~(1), 
    let~$G = (N, \N, \Sigma, S, R)$ be an MCFTG, and let~$G' =
    (N', \N', \Sigma \cup\{\e\}, S', R')$ be the \LDTR"~equivalent
    MCFG that exists by Theorem~\ref{thm:mcft-to-mcf}.\footnote{For
      the purpose of this proof, there is no need to reconsider the
      construction of~$G'$ in its proof.}  Clearly, the
    mapping~$\yield_\Delta$ is a tree homomorphism over the monadic
    ranked alphabet~$\Sigma \cup \{\e\}$.  To be precise, let $h$~be
    the tree homomorphism from~$\Sigma \cup \{\e\}$ to~$\Delta \cup
    \{\e\}$ such that~$h(\alpha) = x_1$ if~$\alpha \in \Sigma
    \setminus \Delta$ and $h(\alpha) = \init(\alpha)$~otherwise.  Then
    $\hat{h}(t\e) = \yield_\Delta(t)\e$ for every~$t \in T_\Sigma$,
    and so $\hat{h} = \yield_\Delta$.  Now let $G''$~be the
    grammar~$G'_h$ as defined before Lemma~\ref{lem:cover}.  Clearly,
    $G''$~is again an MCFG and \LDTR "~$\hat{h}$"~equivalent to~$G'$ by
    that lemma.  Since $G'$~and~$G$ are \LDTR"~equivalent, it follows that
    $G''$~is \LDTR"~$\yield_\Delta$"~equivalent to~$G$. 

To prove statement~(2), let~$G = (N, \N, \Sigma \cup \{\e\}, S, R)$ be an MCFG. We
    construct the MRTG $G_1 = (N, \N, \Sigma \cup \{e, c\}, S, R_1)$,
    where $c$~is a new terminal symbol of rank~$2$, and all symbols of
    $\Sigma \cup \{e\}$~and~$N$ have rank~$0$.  The new set~$R_1$ of
    rules is obtained by replacing each rule~$\rho = A \to ((\seq
    u1n), \LL)$ of~$G$ by the rule $\rho' = A \to ((\seq{u'}1n),
    \LL)$ of~$G_1$, where $\seq{u'}1n$ are defined as follows.  For~$u
    \in (N \cup \Sigma)^*\{\e, x_1\}$, if~$u \in \{\e, x_1\}$,
    then~$u' = e$, and if~$u = \gamma v$ with~$\gamma \in N \cup
    \Sigma$, then~$u' = c(\gamma, v')$.  Note that $\rho$~can be
    reconstructed from~$\rho'$.  It should be clear that $G$
    is \LDTR"~$\yield$"~equivalent to $G_1$ because the derivation trees
    of~$G_1$ are the primed versions of the derivation trees of~$G$.
\end{proof}

Recall that if $G'$~is \LDTR"~$\yield$"~equivalent to~$G$, then~$L(G')
= \yield(L(G))$.  Thus, we immediately obtain from
Lemma~\ref{lem:ymcft} (with $\Delta=\Sigma^{(0)} \setminus \{e\}$) 
that MCFGs~generate the yield languages of the
tree languages generated by~MCFTGs.  

\begin{corollary}
  \label{cor:ymcft}
  $\textup{yMCFT} = \textup{MCF} = \textup{yMRT}$.
\end{corollary}

Thus, $\text{strMCFT} = \text{yMCFT}$ by Corollary~\ref{cor:strmcftg}.
This is quite unusual for a class of tree languages as already
observed at the end of~\cite[Section~4]{engman02}.  For instance, the
monadic tree languages generated by RTGs are the regular string languages, whereas the yield
languages are the context-free string languages.  

The proof of~$\text{MRT} \subseteq \text{MCF}$, and hence of~$\text{MCF} = \text{yMRT}$, 
is also straightforward (cf.\@ Example~\ref{exa:copy}). 
For an MRTG $G = (N, \N, \Sigma, S, R)$, we construct
the MCFG $G' = (N \cup \{S'\}, \N \cup \{S'\}, \Sigma, S', R')$, where
the set~$R'$ consists of 
all rules~$A \to ((u_1 x_1, \dotsc, u_nx_1), \LL)$ 
such that~$A \to ((\seq u1n), \LL)\in R$ and 
all rules~$S' \to (v \e, \LL)$ such that~$S \to (v, \LL)\in R$. 
Then~$L(G') = L(G)$ because this
construction is a special case of the construction in the
proof of Theorem~\ref{thm:mcft-to-mcf} if, in that proof, $\langle C, 0
\rangle$~is identified with~$C$ for every~$C \in N$. 
Note that in the constructions that prove $\text{MCF} = \text{yMRT}$
the multiplicity of the grammars is preserved. 

As observed before Theorem~\ref{thm:mcft-to-mcf}, the above proofs
show that $\text{CFT}_{\text{sp}} \subseteq
\text{MCF}_{\text{wn}}$~and~$\text{yCFT}_{\text{sp}} \subseteq
\text{MCF}_{\text{wn}}$, where $\text{MCF}_{\text{wn}}$~denotes the
class of languages generated by well-nested MCFGs.  It is, in fact, 
not difficult to prove that $\text{yCFT}_{\text{sp}} =
\text{MCF}_{\text{wn}}$ as stated in~\cite{kan09b}.  The multiplicity
of the well-nested MCFG equals one plus the width of the spCFTG.
It is proved in~\cite{kansal10} that
$\text{MCF}_{\text{wn}}$ is properly included in~$\text{MCF}$.

\subsection{Parsing of MCFTGs}
\label{sub:parsing}
\noindent
In the remainder of this section we consider the parsing problem
for~MCFTGs.  We start by showing the well-known fact (cf., e.g.,
\cite{sekmatfujkas91,kall10}) that every MCFG~$G$ can be parsed in
polynomial time in the sense that 
given a string~$w$ as input, the
parsing algorithm outputs an RTG~$H_w$ that generates all derivation
trees of~$G$ with value~$w$. In fact, the usual CYK~parsing algorithm
for~MCFGs constructs the RTG~$H_w$ in such a way that all its nonterminals are useful.  
Clearly, $w \in L(G)$~if and only if~$L(H_w) \neq \emptyset$, which can be
tested in linear time.  Moreover, a derivation tree with value~$w$ can
be computed from $H_w$ in linear time provided that~$L(H_w) \neq \emptyset$. 
In the next lemma we also state the degree of the polynomial, 
as taken from~\cite{sekmatfujkas91}.\footnote{In~\protect{\cite{sekmatfujkas91}}
a recognition algorithm is presented for MCFGs in a certain normal form.
In~\protect{\cite[Section~7]{kall10}} a parsing algorithm is presented for all MCFGs, 
with the RTG defined as a chart with back-pointers, but the degree of the polynomial 
is not analyzed.} 
It involves both the multiplicity $\mu(G)$ and the rule-width $\lambda(G)$ of $G$.\footnote{As 
defined after Definition~\protect{\ref{def:mcftg}}, the rule-width of $G$ is
$\lambda(G)=\max\{\abs{\LL(\rho)}\mid \rho\in R\}$ where $R$ is the set of rules of~$G$.}
It should be noted that, as shown in~\cite{sat92} (see also~\cite{bjobereri16,kajnaksekkas94}), 
the uniform membership problem for MCFGs is NP-hard, even when $\mu(G)$ or $\lambda(G)$ is fixed 
(except of course for $\mu(G)=1$ and for the trivial case $\lambda(G)=0$).

\begin{lemma}
  \label{lem:cyk}
  For every MCFG~$G$ with terminal alphabet~$\Sigma \cup \{\e\}$ there
  is a polynomial time algorithm that, on input~$w \in \Sigma^*$,
  outputs an RTG~$H_w$ such that~$L(H_w) = \{d \in L(G_\der) \mid
  \val(d) = w\}$. The degree of the polynomial is $\mu(G)\cdot(\lambda(G)+1)$.
\end{lemma}

\begin{proof}
Let $G = (N \cup \{S\}, \N, \Sigma \cup \{\e\}, S, R)$~and~$w \in
\Sigma^*$. Moreover, let $w=\word \sigma 1n$ with $n\in\nat_0$ and $\seq \sigma 1n\in\Sigma$. 
We define the set of positions of $w$ by $\pos(w)=\{0,1,\dots,n\}$. 
Intuitively, position~0 is just before $\sigma_1$ and position $i$ is 
just after $\sigma_i$ for every $i\in[n]$. For positions $i,j\in\pos(w)$ with $i\leq j$
we let $w[i,j]=\word\sigma{i+1}j$ be the substring of $w$ between positions $i$ and $j$. 
Note that $w[i,i]=\varepsilon$ for every $i\in\pos(w)$. 

The construction of $H_w$ is similar to the usual ``triple construction'' 
for proving that the intersection of a context-free language 
with a regular language is again context-free (in this case the regular language $\{w\}$). 
We construct the RTG~$H_w = (N_w, R, S_w, R_w)$, in which
$N_w$~is the set of all sequences $(\langle \ell_1,A_1, r_1 \rangle,
\dotsc, \langle \ell_m,A_m, r_m \rangle)$ such that $(\seq A1m) \in \N$ and
$0\leq \ell_i\leq r_i\leq n$ for all~$j \in [m]$. Moreover, $S_w=\langle 0,S,n\rangle$. 
The idea of the proof is that $(\langle \ell_1,A_1, r_1 \rangle,
\dotsc, \langle \ell_m,A_m, r_m \rangle)$ generates all derivation trees 
$d\in L(G_\der,(\seq A 1m))$ such that $\val(d)= (w[\ell_1,r_1],\dotsc,w[\ell_m,r_m])$. 

We now define the set~$R_w$ of rules of $H_w$. 
Let $\rho = A\to (u,\LL)$ be a rule in~$R$ with $A=(\seq A1m)$, $\LL = \{\seq B1k\}$, and
$u=(u_1x_1,\dotsc,u_mx_1)$ if $A\neq S$ and $u=u_1\e$ otherwise 
(with $A_j\in N$ and $u_j\in\Sigma^*$ for every $j\in[m]$). 
Moreover, let $\ell_1,r_1,\dots,\ell_m,r_m\in\pos(w)$ and 
let $\ell$ and $r$ be mappings from $\alp_N(u)$ to $\pos(w)$ such that
\begin{compactenum}[\indent (a)]
\item $\ell_i\leq r_i$ for every $i\in[m]$ and 
$\ell(C)\leq r(C)$ for every $C\in\alp_N(u)$,
\item for every $j\in[m]$, if $u_j=v_0C_1v_1\cdots C_pv_p$ with $p\in\nat_0$, 
$v_0,v_i\in\Sigma^*$, and $C_i\in N$ for every $i\in[p]$, then 
\begin{compactenum}[(1)]
\item $v_0=w[\ell_j,\ell_j+\abs{v_0}]$ and $v_i=w[r(C_i),r(C_i)+\abs{v_i}]$ for every $i\in[p]$, 
\item $\ell(C_1)=\ell_j+\abs{v_0}$ and $\ell(C_{i+1})=r(C_i)+\abs{v_i}$ for every $i\in[p-1]$,
\item $r_j=l_j+\abs{v_0}$ if $p=0$ and $r_j=r(C_p)+\abs{v_p}$ otherwise.
\end{compactenum}
\end{compactenum}
Then the set $R_w$ contains the rule~$(\langle \ell_1, A_1, r_1 \rangle,
\dotsc, \langle \ell_m, A_m, r_m \rangle) \to \rho(\hat{h}(B_1),\dotsc,\hat{h}(B_k))$, 
where~$h$ is the string homomorphism from $\alp_N(u)$ to $\pos(w)\times N\times \pos(w)$ 
such that $h(C)=\langle \ell(C), C, r(C) \rangle$ for every $C\in\alp_N(u)$.
Note that $\alp_N(u)=\bigcup_{i=1}^k\alp(B_i)$. 

The above proof idea can easily be shown by induction on the structure of $d$.
Thus, $S_w$ generates all derivation trees in $d\in L(G_\der)$ such that $\val(d)=w[0,n]=w$. 
Before constructing the rules of $R_w$, the set $\{i\in\pos(w)\mid v=w[i,i+\abs{v}]\}$ 
can be computed for every string $v\in\Sigma^*$ that occurs in a rule of $R$. 
Since $G$ is fixed, this can be done in linear time and 
takes care of the conditions in (1) above. 
When constructing a rule in $R_w$ corresponding to the rule $\rho\in R$ as above, 
it clearly suffices to choose $\seq \ell 1m$ and the mapping $r$, because $\seq r 1m$
are determined by~(3) above and the mapping $\ell$ is determined by~(2) above. 
Since each rule of $R_w$ can be constructed in constant time, 
constructing the rules corresponding to $\rho$ takes time $O(n^q)$ 
where $q=m+\sum_{i=1}^k\abs{B_i}$ is the number of possible choices of $\seq \ell 1m$ and $r$. 
Thus, the algorithm runs in time $O(n^k)$ 
where $k=\mu(G)+\lambda(G)\cdot \mu(G)=\mu(G)\cdot(\lambda(G)+1)$. 

We note that the set $N_w$ can be constructed in quadratic time.\footnote{In the trivial case 
where $\lambda(G)=0$ (and hence $\mu(G)=1$) we can take $N_w=\{S_w\}$.}
In fact, it should be clear that
$H_w$ can be constructed in such a way that only useful nonterminals occur in its rules.
Such a construction corresponds directly to a CYK parsing algorithm. 
\end{proof}

We now generalize this result to~MCFTGs.  Let $G$~be an MCFTG with
terminal alphabet~$\Sigma$, and let~$\Delta \subseteq \Sigma^{(0)}
\setminus \{e\}$ be a set of lexical symbols.  We can use the
MCFTG~$G$ to specify the MCF string language~$\yield_\Delta(L(G))$
together with a set of ``syntactic trees'',  where every tree~$t$
in~$L(G)$ is viewed as a syntactic tree for the
string~$\yield_\Delta(t)$.  In such a case, the parsing problem
for~$G$ amounts to finding the syntactic trees for a given string
over~$\Delta$.

\begin{theorem}
  \label{thm:mcftgparse}
  For every MCFTG~$G$ with terminal alphabet~$\Sigma$ and
  every~$\Delta \subseteq \Sigma$, there is a polynomial time
  algorithm that, on input~$w \in \Delta^*$, outputs an RTG~$H_w$ and
  an MCFTG~$G_w$ such that
  \[ L(H_w) = \{d \in L(G_\der) \mid \yield_\Delta(\val(d)) = w\}
  \qquad \text{and} \qquad L(G_w) = \{t \in L(G) \mid \yield_\Delta(t)
  = w\} \enspace. \]
  The degree of the polynomial is $\mu(G)\cdot(\wid(G)+1)\cdot(\lambda(G)+1)$.
  If $G$ is footed (i.e., is an nsMC"~TAG) then the degree is $2\cdot\mu(G)\cdot(\lambda(G)+1)$. 
\end{theorem}

\begin{proof}
  Let $G'$~be the \LDTR"~$\yield_\Delta$"~equivalent MCFG that exists
  by Lemma~\ref{lem:ymcft}(1), and let $M$~be the \LDTR"~transducer
  from~$G$ to~$G'$. It can easily be verified that $\mu(G')=\mu(G)\cdot(\wid(G)+1)$
  and $\lambda(G')=\lambda(G)$, and that $M$ is a (composition of) finite-state relabeling(s). 
  Moreover, let~$w \in \Delta^*$.  By
  Lemma~\ref{lem:cyk} we can construct an RTG~$H'_w$ such
  that~$L(H'_w) = \{d \in L(G'_\der) \mid \val(d) = w\}$, in the required polynomial time.  
  Then, by Proposition~\ref{pro:LDTRinv} and using a product construction with
  the RTG~$G_\der$, we construct in linear time an RTG~$H_w$ such that
  \[ L(H_w) = M^{-1}(L(H'_w)) \cap L(G_\der) = \{d \in L(G_\der) \mid
  \val(M(d)) = w\} \enspace, \]
  which satisfies the requirement because~$\val(M(d)) =
  \yield_\Delta(\val(d))$.  It remains to construct~$G_w$ from
  $G$~and~$H_w$, which we achieve in linear time by an easy product construction.
  Let $G = (N, \N, \Sigma, S, R)$~be the MCFTG and~$H_w = (N_w, R,
  S_w, R_w)$ be the constructed RTG.  
  We construct $G_w = (N', \N', \Sigma, S', R')$ such that $N' =
  N \times N_w$~and~$\N'$ consists of all~$(\langle A_1, C \rangle,
  \dotsc, \langle A_n, C \rangle)$ with $(\seq A1n) \in \N$~and~$C
  \in N_w$.  For $A = (\seq A1n)$, we denote~$(\langle A_1, C
  \rangle, \dotsc, \langle A_n, C \rangle)$ by~$A \otimes C$.  The
  initial nonterminal of~$G_w$ is~$S' = S \otimes S_w = \langle S,
  S_w \rangle$.  If~$A \to (u, \LL)$ is a rule in~$R$
  with~$\LL = \{\seq B1k\}$ and $C_0 \to \rho(\seq C1k)$ is a rule
  in~$R_w$, then~$R'$ contains the rule~$A \otimes C_0 \to
  (u', \LL')$, in which $u' = u[B_i \gets \init(B_i \otimes C_i) \mid
  1 \leq i \leq k]$ and $\LL' = \{B_1 \otimes C_1, \dotsc, B_k
  \otimes C_k\}$.
  It is easy to show that~$L(G_w, A \otimes C) =
  \val(L(G_\der, A) \cap L(H_w, C))$ for every big nonterminal $
  A \otimes C \in \N'$.  Hence~$L(G_w) = \val(L(H_w))$, which shows
  that $G_w$~satisfies the requirement.
\end{proof}

For $\Delta = \Sigma$ this theorem shows that MCFTGs can be parsed as
tree grammars in polynomial time. For every input tree $t\in T_\Sigma$ the parsing algorithm 
produces as output an RTG $H_t$ such
that $L(H_t) = \{d \in L(G_\der) \mid \val(d) = t\}$.  The algorithm
can easily be extended to test in linear time whether or not $t \in
L(G)$ by testing whether $L(H_t)$ is nonempty.  Additionally, if
$L(H_t) \neq \emptyset$, then it can also compute in linear time an element
of $L(H_t)$; i.e., a derivation tree $d \in L(G_\der)$ such
that $\val(d) = t$.  

For~$\Delta \subseteq \Sigma^{(0)} \setminus
\{e\}$ we are in the situation described before the theorem. For every input string $w\in\Delta^*$ 
the parsing algorithm outputs an MCFTG~$G_w$ such that $L(G_w)$~is the set of all
syntactic trees~$t \in L(G)$ with~$\yield_\Delta(t) = w$.  Using $H_w$ as in the
previous case, the algorithm can be extended to test in linear time whether~$w \in
\yield_\Delta(L(G))$, and if so compute a derivation tree~$d \in
L(G_\der)$ such that~$\yield_\Delta(\val(d)) = w$.  Moreover, it can
then compute~$t = \val(d)$ in linear time; i.e., a syntactic tree~$t
\in L(G)$ with~$\yield_\Delta(t) = w$.

We note that, by the proof of Theorem~\ref{thm:mcftgparse}, these
parsing algorithms are directly based on a parsing algorithm
for~MCFGs; i.e., any algorithm that satisfies Lemma~\ref{lem:cyk}.  If
such a parsing algorithm for the \LDTR"~$\yield_\Delta$"~equivalent
MCFG~$G'$ does not output an RTG $H_w'$ for \emph{all} derivation trees~$d'$
with value~$w$, but outputs just \emph{one} such derivation tree~$d'$,
then there is no need to construct $H_w$ and $G_w$ because 
the above derivation tree~$d \in L(G_\der)$ and syntactic tree~$t
\in L(G)$ can be obtained in linear time as $d = M'(d')$~and~$t =
\val(d)$,  where $M'$~is the \LDTR"~transducer from~$G'$ to~$G$.

\section{Characterization}
\label{sec:charact}
\noindent
In this section we prove that MCFT~is equal to the
class~$\text{DMT}_{\text{fc}}(\text{RT})$ of images of the regular
tree languages under (total) deterministic finite-copying macro tree
transducers, and hence equal to the class~$\text{DMSOT}(\text{RT})$ of
images of the regular tree languages under (total) deterministic MSO
tree transducers.\footnote{\label{foo:dom} Since the domain of a macro tree
  transduction is a regular tree
  language~\protect{\cite[Theorem~7.4]{engvog85}},  the
  class~$\text{DMT}_{\text{fc}}(\text{RT})$ does not depend on the totality of the
  transducers. The same is true for MSO tree transductions and 
  the class~$\text{DMSOT}(\text{RT})$.} 
After proving this result we discuss a number of consequences, 
in particular several alternative characterizations of MCFT. 
As opposed to the usual notation in the
literature~\cite{engvog85,fulvog98,engman99,engman00},  we use~$Y$ as the set
of input variables and $X$~as the set of output variables (or
parameters) for macro tree transducers.  We only consider \emph{total
  deterministic} macro tree transducers that are \emph{simple} (i.e.,
linear and nondeleting) \emph{in the parameters}; this is indicated by
`D'~and~`sp', respectively.

A \emph{macro tree transducer} (in short,
$\text{DMT}_{\text{sp}}$"~transducer) is a system~$M = (Q, \Omega,
\Sigma, q_0, R)$, where $Q$~is a finite ranked alphabet of
\emph{states}, $\Omega$~and~$\Sigma$ are finite ranked alphabets of
\emph{input} and \emph{output symbols}, respectively, with~$Q \cap
\Sigma = \emptyset$, $q_0 \in Q^{(0)}$~is the \emph{initial state}, 
and $R$~is a finite set of \emph{rules}.  For every $q \in
Q^{(m)}$~and~$\omega \in \Omega^{(k)}$ with~$m, k \in \nat_0$ there is
exactly one rule of the form 
\[ \langle q, \omega(\seq y1k) \rangle(\seq x1m) \to \zeta \] 
in~$R$ such that~$\zeta \in P_{(Q \times Y_k) \cup \Sigma}(X_m)$,
where every element~$\langle q', y_i \rangle$ of~$Q \times Y_k$ has
the same rank as~$q'$.  We denote~$\zeta$ by~$\rhs_M(q,
\omega)$.

For every input tree~$s \in T_\Omega$ and every state~$q \in Q$, the
\emph{$q$"~translation} of~$s$ by~$M$, denoted by~$M_q(s)$, is a tree
in~$P_\Sigma(X_{\rk(q)})$ defined inductively as follows.  Let~$s =
\omega(\seq s1k)$ and consider the above rule.  Then~$M_q(s) =
\zeta[\langle q', y_i \rangle \gets M_{q'}(s_i) \mid q' \in Q, 1 \leq
i \leq k]$.  As in the case of \LDTR"~transducers, we define~$M(s) =
M_{q_0}(s)$ and call it the \emph{translation} of~$s$ by~$M$. 
Since $q_0$ has rank~0, $M(s)$ is a tree in $T_\Sigma$. The
\emph{tree transduction realized by~$M$}, also denoted by~$M$, is the
total function~$M = \{(s, M(s)) \mid s \in T_\Omega\}$ from~$T_\Omega$
to~$T_\Sigma$.  A $\text{DMT}_{\text{sp}}$"~transducer is a (total
deterministic) \emph{top-down tree transducer} (in short,
DT"~transducer) if all its states have rank~$0$. 

Finite-copying macro tree transducers were introduced
in~\cite{engman99}.  To define them, we need the well-known notion of
``state sequence'' (cf.~\cite[Definition~3.1.8]{engrozslu80}).
Let~$(\seq q1n) \in Q^*$ with $n \in \nat_0$ and $\seq q1n \in Q$, and
let~$\omega \in \Omega^{(k)}$ for some~$k \in \nat_0$.  For~$i \in
[k]$ we define~$\sts_{\omega, i}(\seq q1n) \in Q^*$ to be the sequence
of states
\[ \sts_{\omega, i}(\seq q1n) = \pi^*_i(\rhs_M(q_1, \omega) \dotsm
\rhs_M(q_n, \omega)) \enspace, \]
where $\pi_i$~is the string homomorphism from~$(Q \times Y_k) \cup
\Sigma \cup X$ to~$Q$ such that~$\pi_i(\langle q', y_i \rangle) = q'$
for every~$q' \in Q$ and $\pi_i(\alpha) = \varepsilon$~for every 
$\alpha \in \Sigma \cup X$.  For $s \in T_\Omega$~and~$p \in \pos(s)$, we define
the \emph{state sequence} of~$M$ at~$p$, denoted by~$\sts(s, p)$,
inductively as follows: (i)~$\sts(s, \varepsilon) = q_0$ and
(ii)~if~$\sts(s, p) = (\seq q1n)$ and~$s(p) = \omega \in \Omega^{(k)}$, 
then $\sts(s, pi) = \sts_{\omega, i}(\seq q1n)$ for every~$i \in [k]$.  
The \emph{set of state sequences} of~$M$, denoted by~$\sts(M)$, is defined 
by~$\sts(M) = \{\sts(s, p) \mid s \in T_\Omega,\, p \in \pos(s)\}$.
Note that it is the smallest subset~$S$ of~$Q^*$ such that (i)~$q_0
\in S$ and (ii)~if~$\overline{q} \in S$, then~$\sts_{\omega,
  i}(\overline{q}) \in S$ for all~$k \in \nat_0$, $\omega \in
\Omega^{(k)}$, and~$i \in [k]$.  We say that the
$\text{DMT}_{\text{sp}}$"~transducer~$M$ is \emph{finite-copying} (in
short, $\text{DMT}_{\text{fc}}$"~transducer) if $\sts(M)$~is finite;
it is \emph{$m$"~copying} for~$m \in \nat$, if the state sequences
in~$\sts(M)$ have length at most~$m$.
A~$\text{DT}_{\text{fc}}$"~transducer is a finite-copying $\text{DT}$"~transducer.

For a notion of $\X$"~transducer, we denote by~$\X$ the class of
transductions realized by $\X$"~transducers.  For a class~$\X$ of
transductions, we denote by~$\X(\text{RT})$ the class of all tree
languages~$M(L)$, where $M \in \X$~and~$L \in \text{RT}$ is a regular
tree language.  

The finite-copying macro tree transducers
of~\cite{engman99} are not necessarily simple; i.e., linear and
nondeleting in the parameters.  However, it follows from the results
of~\cite[Section~6]{engman99} that adding the feature of regular
look-ahead, which we do not need here, to the above finite-copying
macro tree transducers yields the same expressive power as
in~\cite{engman99}.  In particular, our notion of state sequence
corresponds to the one in Definition~6.8 and Lemma~6.9
of~\cite{engman99}.  Since regular look-ahead can be simulated by a 
relabeling of the input tree (see~\cite{eng77}), the
class~$\text{DMT}_{\text{fc}}(\text{RT})$, which we are interested in
here, coincides with the one in~\cite{engman99}
(denoted~$\text{MTT}_{\text{fc}}(\text{REGT})$ there).  Let us finally
note that it is decidable whether or not a macro tree transducer is
finite-copying~\cite[Lemma~4.10]{engman03a},  and if so, its set of
state sequences can be computed by iteration.

The inclusion~$\text{MCFT} \subseteq
\text{DMT}_{\text{fc}}(\text{RT})$ is a direct consequence of the next
lemma and Theorem~\ref{thm:dtree}.  The lemma shows that `$\val$' can be realized 
by a $\text{DMT}_{\text{fc}}$"~transducer. In its proof we use 
the following additional terminology.  For~$\bar{q} = (\seq q1n) \in
Q^{\scriptscriptstyle +}$ with $n \in \nat$~and~$\seq q1n \in Q$, we
define the $\bar{q}$"~translation of~$s \in T_\Omega$ by~$M_{\bar{q}}(s) =
(M_{q_1}(s), \dotsc, M_{q_n}(s))$.

\begin{lemma}
  \label{lem:valismtt}
  For every MCFTG~$G$ there is a
  $\text{DMT}_{\text{fc}}$"~transducer~$M$ such that~$M(d) = \val(d)$
  for every~$d\in L(G_\der)$.  If $G$~is an MRTG, then $M$~is a
  $\text{DT}_{\text{fc}}$"~transducer.
\end{lemma}

\begin{proof}
  Let~$G = (N, \N, \Sigma, S, R)$ be an MCFTG.  Since the result is
  obvious if~$L(G) = \emptyset$, we may assume that~$\Sigma^{(0)} \neq
  \emptyset$.  We construct the macro tree transducer~$M = (N, R,
  \Sigma, S, R_M)$.  Thus, $M$ uses the nonterminals of~$G$ with the
  same rank as states, of which $S$ is the initial state.  Moreover, the
  input alphabet is~$R$ and the output alphabet is~$\Sigma$.  If
  $\rho = (\seq A1n) \to ((\seq u1n), \LL)$ is a rule in~$R$ such
  that~$\LL = \{\seq B1k\}$ with $\seq B1k \in \N$, then
  $R_M$~contains the following rule for every~$j \in [n]$: 
  \[ \langle A_j, \rho(\seq y1k) \rangle(\seq x1{\rk(A_j)}) \to u_j[C
  \gets \init(\langle C, y_i \rangle) \mid C \in \alp(B_i),\, 1 \leq i \leq
  k] \enspace. \] 
  Moreover, it has the (dummy) rule~$\langle C, \rho(\seq y1k)
  \rangle(\seq x1m) \to t_m$ for every~$C \in N \setminus \{\seq A1n\}$ of
  rank~$m$, where $t_m$~is an arbitrary element
  of~$P_\Sigma(X_m)$.\footnote{If~$\Sigma^{(k)} \neq \emptyset$ for
    some~$k \geq 2$, then~$P_\Sigma(X_m) \neq \emptyset$ for all~$m$
    (recall that~$\Sigma^{(0)} \neq \emptyset$).  If~$\Sigma =
    \Sigma^{(0)} \cup \Sigma^{(1)}$, then~$N=N^{(0)} \cup N^{(1)}$
    (because $G$~is reduced) and we only need $t_0 \in
    \Sigma^{(0)}$~and~$t_1 = x_1$.}

  Clearly, $M$~is simple in the parameters because~$u_j \in P_{N \cup
    \Sigma}(X_{\rk(A_j)})$.  Let~$d \in L(G_\der)$.  We claim
  that~$\sts(d, p)$, the state sequence of~$M$ at a node~$p$ of~$d$,
  is a permutation of the left-hand side of the rule~$d(p)$ of~$G$.
  This is obvious for the root of~$d$ with state sequence~$S$, and if
  it holds for~$p$, then it holds for~$pi$ for every~$i \in [k]$ by
  the definition of the above rules of~$M$.  Hence $M$~is
  finite-copying on~$L(G_\der)$.  It is, in fact, finite-copying
  everywhere because the state sequence becomes empty due to the dummy
  rules as soon as there is a type error in the input tree (which means that 
  the input tree is not a derivation tree of $G$).  

  We now claim that~$M_A(d) = \val(d)$ for every~$A \in \N$ and every
  derivation tree~$d \in L(G_\der, A)$, where $M_A(d)$~is defined just
  before this lemma. The proof is by induction on the structure of $d$.
  Let~$d = \rho(\seq d1k)$.  For the above
  rule~$\rho$ of~$G$, let $A = (\seq A1n)$~and~$u = (\seq u1n)$.  Then
  $\val(d) = u[B_i \gets \val(d_i) \mid 1 \leq i \leq k]$.   From the
  definition of the rules of~$M$ we obtain that
  \[ M_A(d) = u[C \gets M_C(d_i) \mid C \in \alp(B_i),\, 1 \leq i \leq
  k] = u[B_i \gets M_{B_i}(d_i) \mid 1 \leq i \leq k] \enspace. \] 
  By the induction hypotheses, $M_{B_i}(d_i) = \val(d_i)$~for every~$i
  \in [k]$.  Consequently,~$M_A(d) = \val(d)$.  In particular,
  if~$d\in L(G_\der)$, then~$M(d) = M_S(d) = \val(d)$.
\end{proof}

For the converse inclusion we need a normal form
for~$\text{DMT}_{\text{fc}}$"~transducers from~\cite{engman99}, which
is based on the same result for~$\text{DT}_{\text{fc}}$"~transducers
in~\cite{vug96}.  The $\text{DMT}_{\text{fc}}$"~transducer~$M$ is
\emph{repetition-free} if all its state sequences in~$\sts(M)$ are
repetition-free.

\begin{proposition}
  \label{pro:mtt-repfree}
  For every $\text{DMT}_{\text{fc}}$"~transducer~$M$ there is a
  repetition-free $\text{DMT}_{\text{fc}}$"~transducer~$M'$ that
  realizes the same tree transduction as~$M$.  Moreover, if~$M$ is a
  $\text{DT}_{\text{fc}}$"~transducer, then so is~$M'$.  
\end{proposition}

\begin{proof}
  It is proved in~\cite[Lemma~6.10]{engman99} that there is a
  single-use restricted $\text{DMT}_{\text{sp}}$"~transducer~$M'$ that
  realizes the same tree transduction as~$M$.  It is in fact proved
  for macro tree transducers with regular look-ahead, but the
  construction preserves the absence of look-ahead.  Moreover, in the
  proof of~\cite[Theorem~6.12]{engman99} it is shown that single-use
  restricted $\text{DMT}_{\text{sp}}$"~transducers are finite-copying
  and repetition-free.  The construction
  in~\cite[Lemma~6.10]{engman99} preserves
  $\text{DT}_{\text{fc}}$"~transducers, but for them the result was
  already proved in~\cite[Lemma~5.3]{vug96}.
\end{proof}

\begin{lemma}
  \label{lem:mttrt-mcft}
  $\textup{DMT}_{\textup{fc}}(\textup{RT}) \subseteq \textup{MCFT}$ and
  $\textup{DT}_{\textup{fc}}(\textup{RT}) \subseteq \textup{MRT}$. 
\end{lemma}

\begin{proof}
  Let $M = (Q, \Omega, \Sigma, q_0, R_M)$ be a
  $\text{DMT}_{\text{fc}}$"~transducer, of which we assume, by
  Proposition~\ref{pro:mtt-repfree}, that it is repetition-free.
  Moreover, let~$G = (N, \Omega, S, R)$ be an RTG.
  We can assume that in each of its rules $C \to \omega(\seq C1k)$, with $C, \seq C1k \in
  N$~and~$\omega \in \Omega^{(k)}$, the sequence~$(\seq C1k)$ is repetition-free (cf.\@
  Section~\ref{sub:seqs}).  We will construct an MCFTG~$G' = (N', \N',
  \Sigma, S', R')$ such that~$L(G') = M(L(G))$.  The MCFTG~$G'$ will
  simulate both $M$~and~$G$.  Thus, we define~$N' = Q \times N$, where
  every~$\langle q, C \rangle \in N'$ has the same rank as~$q$, and
  $S' = \langle q_0, S \rangle$.  For every nonempty state
  sequence~$\bar{q} = (\seq q1n) \in Q^{\scriptscriptstyle +}$ and
  nonterminal~$C \in N$, we abbreviate the sequence~$(\langle q_1, C
  \rangle, \dotsc, \langle q_n, C \rangle) \in
  (N')^{\scriptscriptstyle +}$ by~$\bar{q} \otimes C$.
  Then we define $\N' = \{\bar{q} \otimes C \mid \bar{q} \in \sts(M) 
  \setminus \{\varepsilon\},\, C \in N\}$, so in other words, the big
  nonterminals of~$G'$ are of the form~$(\langle q_1, C \rangle,
  \dotsc, \langle q_n, C \rangle)$, where $(\seq q1n)$~is a nonempty
  state sequence of~$M$, and $C$~is a nonterminal of~$G$.  It remains
  to define the rules of~$G'$.  Let~$\rho = C \to \omega(\seq C1k)$ be
  a rule of~$G$, and let~$\bar{q} = (\seq q1n)$ be a nonempty state
  sequence of~$M$.  Then $R'$~contains the rule 
  \[ \rho_{\bar{q}} = (\langle q_1, C \rangle, \dotsc, \langle q_n, C
  \rangle) \to ((\seq u1n), \LL) \]
  with left-hand side $\bar{q} \otimes C$, 
  where $u_j = \rhs_M(q_j, \omega)[\langle q, y_i \rangle \gets
  \init(\langle q, C_i \rangle) \mid q \in Q,\, 1 \leq i \leq k]$ 
  for every~$j \in [n]$
  and $\LL = \{\sts_{\omega, i}(\bar{q}) \otimes C_i \mid i \in [k]\}
  \cap \N'$ .  Note that~$(\seq u1n)$ is uniquely
  $N'$"~labeled because $(\seq C1k)$~is repetition-free and every
  state sequence~$\sts_{\omega, i}(\bar{q})$ is repetition-free.  The 
  correctness of~$G'$ is a direct consequence of the following claim.

  \paragraph{Claim} For every nonempty state sequence~$\bar{q} \in
  \sts(M) \setminus \{\varepsilon\}$, nonterminal~$C \in N$, and
  forest~$t \in P_\Sigma(X)^{\scriptscriptstyle +}$ we have $t \in L(G',
  \bar{q} \otimes C)$ if and only if there exists~$s \in L(G,
  C)$ such that~$M_{\bar{q}}(s) = t$.\footnote{In other words,
  $L(G',\bar{q} \otimes C \rangle) = M_{\bar{q}}(L(G,C))$. 
  Recall the definition of~$M_{\bar{q}}(s)$ just before
  Lemma~\protect{\ref{lem:valismtt}}.}

  \paragraph{Proof of sufficiency} We have to show that
  $M_{\bar{q}}(s) \in L(G', \bar{q} \otimes C)$ for every~$s
  \in L(G, C)$.  The proof is by induction on the structure of~$s$.
  Let $s = \omega(\seq s1k)$.  Then there is a rule~$\rho = C \to
  \omega(\seq C1k)$ of~$G$ such that~$s_i \in L(G, C_i)$ for every~$i
  \in [k]$.  Let~$\bar{q}_i = \sts_{\omega, i}(\bar{q})$ for every~$i
  \in [k]$.  By the induction hypotheses, $M_{\bar{q}_i}(s_i) \in
  L(G', \bar{q}_i \otimes C_i)$ provided that~$\bar{q}_i \neq
  \varepsilon$.  Let $\rho_{\bar{q}}$~be the rule in~$R'$ as defined
  above.  Then the least fixed point semantics of~$G'$ implies that
  $L(G', \bar{q} \otimes C)$ contains the forest
  \[ (\seq u1n)[\bar{q}_i \otimes C_i \gets M_{\bar{q}_i}(s_i)
  \mid i\in[k],\, \bar{q}_i \neq \varepsilon] \enspace, \]
  which equals~$M_{\bar{q}}(s)$.

  \paragraph{Proof of necessity} The proof is similar and proceeds by
  induction on the structure of a derivation tree~$d \in L(G'_\der,
  \bar{q} \otimes C)$ with~$\val(d) = t$.  Let~$d =
  \rho_{\bar{q}}(\seq d1k)$.  Then 
  \[ t = (\seq u1n)[\bar{q}_i \otimes C_i \gets \val(d_i)
  \mid i\in[k],\, \bar{q}_i \neq \varepsilon] \enspace. \]
  By the induction hypotheses, there exist trees~$s_i \in L(G, C_i)$
  such that~$M_{\bar{q}_i}(s_i) = \val(d_i)$ for every~$i \in [k]$
  with~$\bar{q}_i \neq \varepsilon$.  Since we assume that $G$~is
  reduced, there also exist trees~$s_i \in L(G, C_i)$ for every~$i \in
  [k]$ with~$\bar{q}_i = \varepsilon$.  Consequently, $s \in L(G,
  C)$~and~$M_{\bar{q}}(s) = t$ for~$s = \omega(\seq s1k)$.
\end{proof}

From Lemmas \ref{lem:valismtt}~and~\ref{lem:mttrt-mcft}
we obtain our characterization result, of which the second part  
was proved in~\cite[Proposition~4.8]{rao97}.

\begin{theorem}
  \label{thm:charact}
  $\textup{MCFT} = \textup{DMT}_{\textup{fc}}(\textup{RT})$ and $\textup{MRT} =
  \textup{DT}_{\textup{fc}}(\textup{RT})$. 
\end{theorem}

We observe that the multiplicity of the MCFTG corresponds to 
the ``copying number'' of the corresponding
$\text{DMT}_{\text{fc}}$"~transducer.  For every~$m \in \nat$, let
$m$"~MCFT be the class of tree languages generated by MCFTGs~$G$
with~$\mu(G) \leq m$, and let $\text{DMT}_{\text{fc}(m)}$~be the class
of transductions realized by~$m$"~copying
$\text{DMT}_{\text{fc}}$"~transducers, and similarly for
subclasses of these grammars and transducers.  Then, checking
the proofs above, we obtain that $\text{$m$"~MCFT} =
\text{DMT}_{\text{fc}(m)}(\text{RT})$~and~$\text{$m$"~MRT} =
\text{DT}_{\text{fc}(m)}(\text{RT})$ for every~$m \in \nat$.  For the
preservation of the $m$"~copying property in
Proposition~\ref{pro:mtt-repfree} we additionally need to inspect the
proof of~\cite[Lemma~6.10]{engman99}).
For $m=1$ we obtain that $\text{CFT}_{\text{sp}} = \text{DMT}_{\text{fc}(1)}(\text{RT})$.
A $\text{DMT}_{\text{sp}}$"~transducer is \emph{simple} (in short,
$\text{DMT}_{\text{si,sp}}$"~transducer) if it is also simple (i.e.,
linear and nondeleting) in the input variables.  Clearly, 
$\text{DMT}_{\text{si,sp}}$"~transducers are $1$"~copying. 
Checking again the proofs
above, it is easy to see that $\text{CFT}_{\text{sp}} =
\text{DMT}_{\text{si,sp}}(\text{RT})$.\footnote{The only small
  technical problem is the deletion of all input variables in the
  dummy rules in the proof of Lemma~\protect{\ref{lem:valismtt}}.
  This can be easily remedied by introducing an additional state~$q$
  of rank~$1$, changing the dummy rules into~$\langle C, \rho(\seq
  y1k) \rangle(\seq x1m) \to \langle q, y_1 \rangle(\dotsm (\langle q,
  y_k \rangle(t_m)) \dotsm)$ and adding additionally all dummy rules
  of the form~$\langle q, \rho(\seq y1k) \rangle(x_1) \to \langle q,
  y_1 \rangle(\dotsm(\langle q, y_k \rangle(x_1)) \dotsm)$.}

In the remainder of this section we discuss the consequences of 
the characterization result in Theorem~\ref{thm:charact}. 
One immediate consequence is that MCFT is closed under intersection 
with regular tree languages: If $M$ is a $\text{DMT}_{\text{fc}}$"~transducer
and $R_1$ and $R_2$ are in RT, then $M(R_1)\cap R_2= M(R_1\cap M^{-1}(R_2))$. Moreover,
$M^{-1}(R_2)$ is in RT by~\protect{\cite[Theorem~7.4]{engvog85}} 
and so $R_1\cap M^{-1}(R_2)$ is in RT. 

From Theorem~\ref{thm:charact} and Corollary~\ref{cor:ymcft} we obtain
two known results.  First, $\text{MCF} =
\text{yDT}_{\text{fc}}(\text{RT})$.  Since it is easy to check from
the proof of Corollary~\ref{cor:ymcft} that $\text{$m$"~MCF} =
\text{y($m$"~MRT)}$, we even obtain that $\text{$m$"~MCF} =
\text{yDT}_{\text{fc}(m)}(\text{RT})$ for every~$m \in \nat$.  It was,
in fact, proved in~\cite{wei92} that $m$"~MCF equals the class of
output languages of deterministic tree-walking transducers with
``crossing number''~$m$, which equals
$\text{yDT}_{\text{fc}(m)}(\text{RT})$
by~\cite[Corollary~4.11]{engrozslu80}.  Second,
$\text{yDMT}_{\text{fc}}(\text{RT}) =
\text{yDT}_{\text{fc}}(\text{RT})$, which was proved
in~\cite[Corollary~7.10]{engman99}.  Vice versa, this equality and
Theorem~\ref{thm:charact} imply that $\text{yMCFT} = \text{yMRT}$
(Corollary~\ref{cor:ymcft}).  We also observe that this equality is a
restricted version of $\text{yDMT}_{\text{sp}}(\text{RT}) =
\text{yDT}(\text{RT})$, which was proved
in~\cite[Theorem~15]{engman02} (cf.\@ the last sentence before
Theorem~\ref{thm:mcft-to-mcf}) and will follow from the results in
Section~\ref{sec:parallel}.

More interestingly, Theorem~\ref{thm:charact} implies three other
characterizations of MCFT~and~MCF (of which those of~MCF are already
known).  First, they can be characterized in monadic second-order
logic~(MSO).  Let~$\text{DMSOT}$ be the class of deterministic (or
parameterless) MSO"~definable tree transductions (see, e.g.,
\cite[Chapter~8]{coueng12}), and let~$\text{DMSOTS}$ be the analogous
class of tree-to-string transductions.  Since regular look-ahead can
be simulated by a relabeling of the input tree,  it follows
from~\cite[Theorem~7.1]{engman99} that~$\text{DMSOT}(\text{RT}) =
\text{DMT}_{\text{fc}}(\text{RT})$ and
from~\cite[Theorem~7.7]{engman99} that $\text{DMSOTS}(\text{RT}) =
\text{yDT}_{\text{fc}}(\text{RT})$.

\begin{corollary}
  \label{cor:mcft-mso}
  $\textup{MCFT} = \textup{DMSOT}(\textup{RT})$~and~$\textup{MCF} =
  \textup{DMSOTS}(\textup{RT})$.  
\end{corollary}

Since MSO"~definable transductions are closed under
composition~\cite[Theorem~7.14]{coueng12},  this implies that MCFT~is
closed under DMSOT"~transductions, and hence under
$\text{DMT}_{\text{fc}}$"~transductions even when they are equipped
with regular look-ahead by~\cite[Theorem~7.1]{engman99}.  Similarly,
MCF~is closed under deterministic MSO"~definable string transductions,
which are the transductions realized by two-way deterministic
finite-state transducers~\cite{enghoo01}. 
In particular, it follows from Lemma~\ref{lem:valismtt} 
that MCFT is closed under control, in the following sense.
Let $G$ be an MCFTG and let $C$ be a (``control'') tree language in MCFT.
Then $\val(L(G_\der) \cap C)$ is in MCFT. Intuitively, the derivation trees 
of the grammar $G$ are restricted to be an element of $C$; in that way 
$C$ ``controls'' the derivation trees (and hence the derivations) of $G$. 

Second, MCFT~and~MCF
can be characterized in terms of context-free graph grammars.  It is
known that $\text{DMSOT}(\text{RT})$ equals the class of tree
languages that can be generated by (either hyperedge-replacement or
vertex-replacement) context-free graph grammars (see, e.g.,
\cite[Section~6]{eng97} or the introduction
of~\cite[Section~8.9]{coueng12}).  Similarly,
$\text{DMSOTS}(\text{RT})$~is the class of string languages generated
by such grammars.  These facts were also used to 
obtain~\cite[Corollaries~7.3~and~7.8]{engman99}.

\begin{corollary}
  \label{cor:mcft-cfgg}
  $\textup{MCFT}$ (resp.~$\textup{MCF}$) is the class of tree languages (resp.~string
  languages) generated by context-free graph grammars. 
\end{corollary}

\begin{remark}
\label{rem:mcft-cfgg}
\upshape
For completeness' sake we show here how easy it is to simulate an MCFTG by 
a context-free graph grammar, in particular a hyperedge-replacement grammar (HRG).
We assume the reader to be familiar with HRGs (see, e.g., \cite{drekrehab97,eng97,baucou87}). 
Let us first recall how trees and forests can be represented as hypergraphs. 
Let $\Omega$ be a ranked alphabet. 
A forest $t=(\seq t 1n)\in P_\Omega(X)\plus$ is represented by the hypergraph 
$\gr(t)$ that has the set of nodes $\pos(t)$ and the set of hyperedges $\{e_p\mid p\in\pos_\Omega(t)\}$ 
such that $e_p$ has label $t(p)$ and sequence of incident nodes $\inc_t(p)=(p1,\dotsc,pk,p)$ 
where $k=\rk(t(p))$. 
Moreover, $\gr(t)$ has the sequence of external nodes $\ext(t)=\ext(t_1)\cdots\ext(t_n)$
such that $\ext(t_j)=(p_ {j,1},\dotsc,p_{j,k_j},\#^{j-1})$ 
where $t_j(p_{j,\ell})=x_\ell\in X$ for every $j\in[n]$ and $\ell\in[k_j]$ with $k_j=\rk(t_j)$. 
We say that an HRG is \emph{tree generating} (or, generates a tree language) if 
its terminal alphabet is a ranked alphabet $\Sigma$
and the generated hypergraph language is a subset of $\{\gr(t)\mid t\in T_\Sigma\}$.  

Now let $G = (N, \N, \Sigma, S, R)$ be an MCFTG.
We construct an HRG $G'$ that has the set of nonterminals~$\N$, with initial nonterminal $S$, 
and the set of terminals~$\Sigma$.  
Let $A\to (u,\LL)$ be a rule in $R$.
Then $G'$ has the rule $A\to \gr(u,\LL)$, 
where $\gr(u,\LL)$ is the hypergraph obtained from $\gr(u)$ as follows. 
For every $B=(u(p_1),\dotsc,u(p_m))\in\LL$ with $\seq p 1m \in \pos_N(u)$,
remove the hyperedges $\seq e {p_1}{p_m}$
and replace them by one new hyperedge $e_B$ that has label~$B$
and sequence of incident nodes $\inc_u(p_1)\cdots\inc_u(p_m)$.\footnote{Every terminal 
or nonterminal symbol $\alpha$ of an HRG should have a ``rank''. For 
every hyperedge $e$ with label $\alpha$ the ``rank'' of $\alpha$ should be equal to 
the number of nodes that are incident with $e$. Moreover, 
for every rule $A\to g$ the ``rank'' of $A$ should be equal to 
the number of external nodes of the hypergraph $g$.
In the grammar $G'$, every terminal $\sigma\in\Sigma$ has ``rank'' $\rk(\sigma)+1$
and every nonterminal $A=(A_1,\dots,A_n)$ has ``rank'' $\sum_{i=1}^n(\rk(A_i)+1)$.} 
Intuitively, the hyperedge $e_B$ explicitly links the occurrences in $u$ of the 
nonterminals $u(p_1),\dotsc,u(p_m)$ of the link $B$. 
Now let $t_B\in P_\Sigma(X)\plus$ 
be a forest with $\rk(t_B)=\rk(B)$, for every $B\in\LL$.
Then it is straightforward to check that $\gr(u[B\gets t_B\mid B\in\LL])$ is equal to 
the result of simultaneously substituting $\gr(t_B)$ 
for the hyperedge $e_B$ in $\gr(u,\LL)$ for every $B\in\LL$.
Thus, using the least fixed point semantics of the HRG $G'$ 
(see \cite[Theorem~2.4.2]{drekrehab97}), 
we obtain that $L(G')=\{\gr(t)\mid t\in L(G)\}$. 
It can also easily be checked that the derivations of $G$, as
defined in Section~\ref{sub:deriv}, can be simulated by
the derivations of~$G'$: for every $t\in T_{(N\times\nat^*)\cup\Sigma}$ and $n\in\nat_0$, 
if $S\otimes\varepsilon \Rightarrow^n_G t$ then $\gr(S)\Rightarrow^n_{G'} \gr(t,\LL)$,
where $\gr(t,\LL)$ is defined similarly to $\gr(u,\LL)$ above 
using the set $\LL\subseteq \N\otimes \nat^*$ mentioned at the end of Section~\ref{sub:deriv}.
Moreover, these are all possible derivations in $G'$. 
Intuitively, the role of the link identifiers 
in the derivation $S\otimes\varepsilon \Rightarrow^n_G t$ is taken over by explicit hyperedges. 

We say that an HRG is in \emph{tree generating normal form} if it can be obtained from
an MCFTG in the way described above, eventually followed by a renaming of its nonterminals
and an identification of nonterminals that are aliases.\footnote{It
can be checked that this is equivalent to~\protect{\cite[Definition~6]{engman00}}, provided that
the MCFTG is assumed to be nonerasing.} 
Then the above, together with Lemma~\ref{lem:nonerasing} and Corollary~\ref{cor:mcft-cfgg}, 
proves that every tree generating HRG  
has an equivalent HRG in tree generating normal form (see~\cite[Theorem~7]{engman00}). 
We finally note that there is a similar easy construction showing that every string
language in~MCF can be generated by an HRG (see~\cite[Theorem~6.4]{eng97}). 

As an example of the above construction 
we consider the MCFTG $G$ of Example~\protect{\ref{exa:main}}. 
The rules of the HRG $G'$ are shown in Figure~\ref{fig:hrgrammar} 
(without the rules for the alias $B'$ of $B$) and the derivation of $G'$ 
corresponding to the one of $G$ in Figure~\ref{fig:derivation} is shown 
in Figure~\ref{fig:hrderivation}. By definition, $G'$~is in tree generating normal form. 
Note that the sequence of external nodes of the right-hand side of rule $\rho_4$ 
(and of rule $\rho_6$) of $G'$ is not repetition-free, 
which allows $G'$ to erase hyperedges (or ``parts'' of hyperedges). 
For a nonerasing MCFTG $G$ the above construction
results in an HRG $G'$ for which all sequences of external nodes 
(and all sequences of incident nodes) are repetition-free. Thus by 
Lemma~\ref{lem:nonerasing}, this requirement can be added to the 
tree generating normal form. 
\fin
\end{remark}

\begin{figure}
  \centering
  \includegraphics{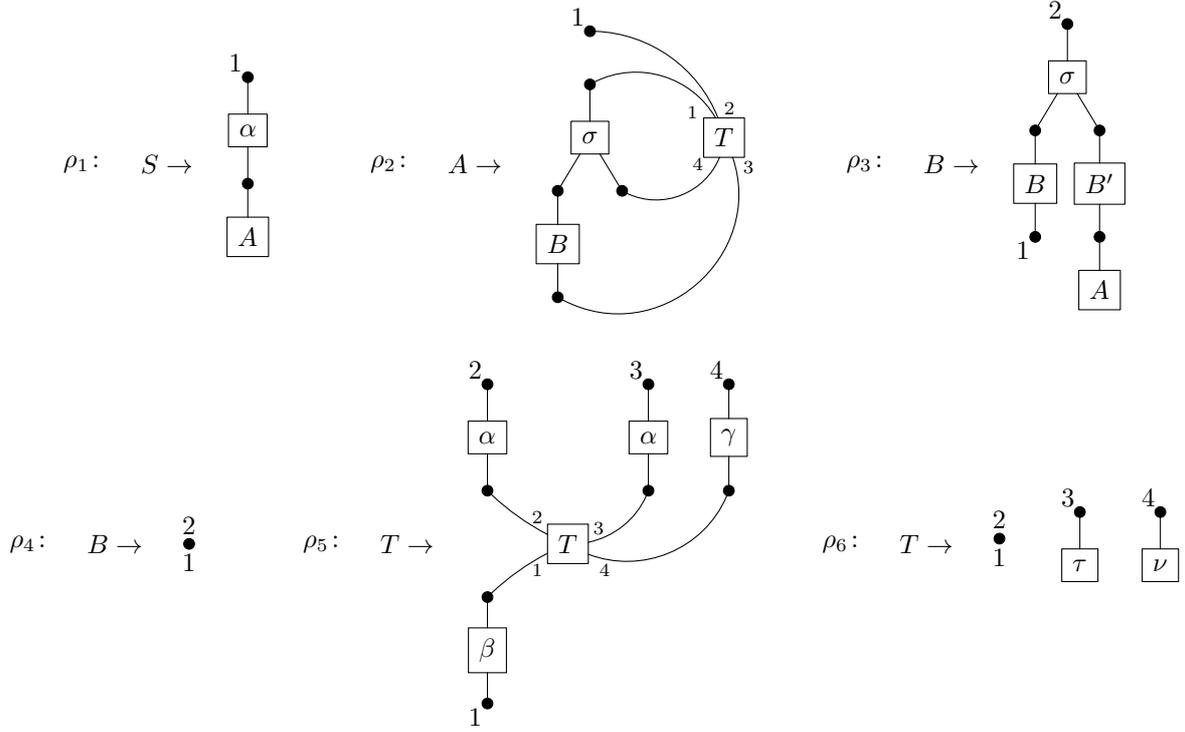}
  \caption{Rules of the HRG~$G'$ corresponding to the MCFTG~$G$ of
    Example~\protect{\ref{exa:main}} (without the rules for~$B'$).
    Hypergraphs are drawn as in~\protect{\cite{drekrehab97,eng97}}.
    A~hyperedge~$e$ is drawn as a box containing the label of~$e$.  A
    line with label~$i$ connects~$e$ with its $i$"~th incident node.
    If the label of~$e$ is in~$N \cup \Sigma$ with rank~$k$, then the
    labels of the incidence lines are dropped; by convention, the
    first $k$~incident nodes of~$e$ are below the box, from left to
    right, and the last incident node is above the box.  The $j$"~th
    external node of the hypergraph is labeled~$j$.} 
  \label{fig:hrgrammar}
\end{figure}

\begin{figure}
  \centering
  \includegraphics{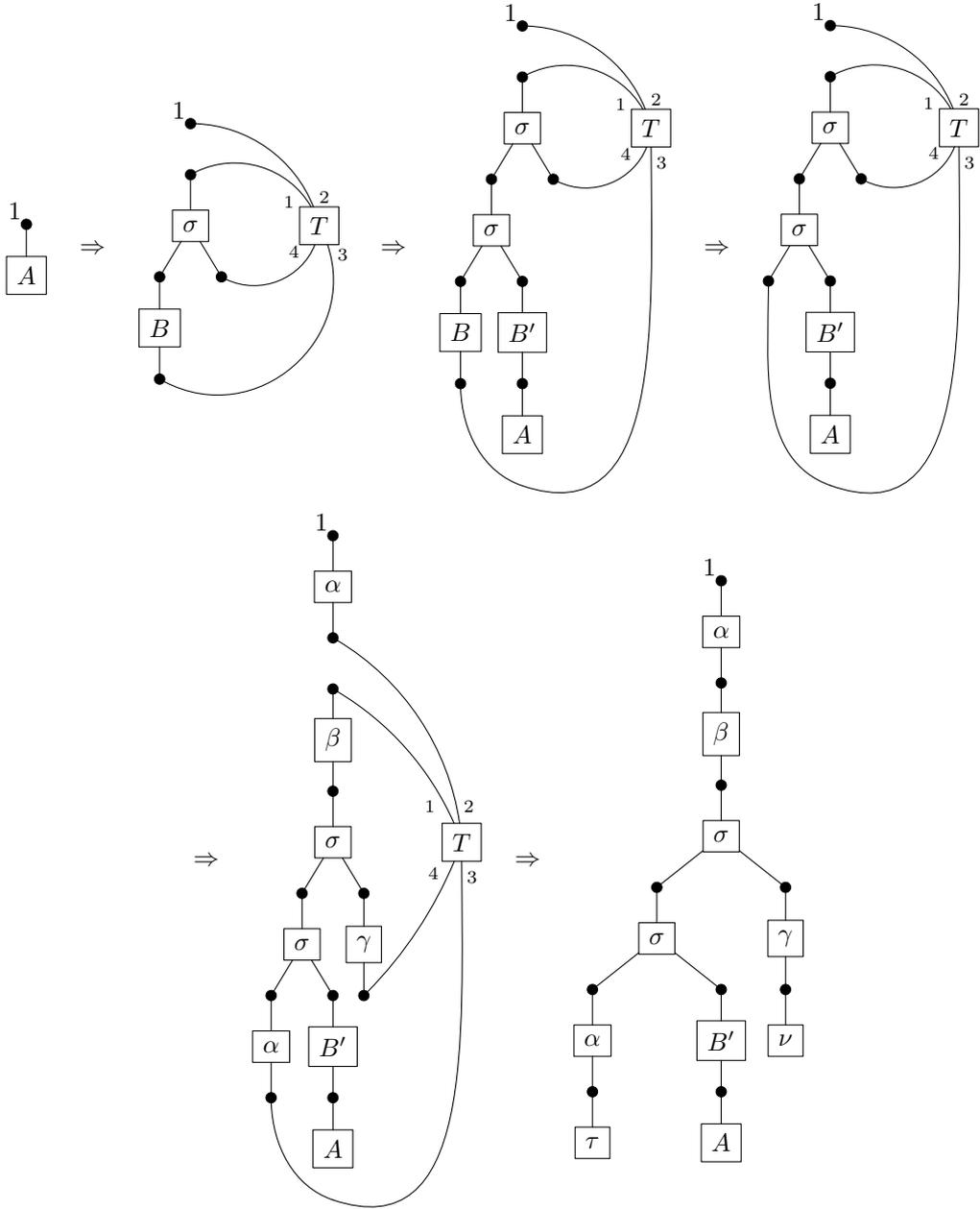}
  \caption{Derivation of the HRG of Figure~\protect{\ref{fig:hrgrammar}} 
    corresponding to the MCFTG derivation of
    Figure~\protect{\ref{fig:derivation}}.} 
  \label{fig:hrderivation}
\end{figure}

Third,
MCFT~and~MCF can be characterized in terms of second-order abstract
categorial grammars.  It is shown in~\cite{kan10} that such grammars
have the same tree and string generating power as
hyperedge-replacement context-free graph grammars, which was already
known for strings from earlier results as discussed in~\cite{kan10}. 

\begin{corollary}
  \label{cor:mcft-acg}
  $\textup{MCFT}$ (resp.~$\textup{MCF}$) is the class of tree languages (resp.~string
  languages) generated by second-order abstract categorial grammars.
\end{corollary}

\begin{proof}
  Let $\text{TR}(\text{2AC})$ denote the class of tree languages
  generated by second-order abstract categorial grammars (in short,
  2ACGs).  It is shown in~\cite{kan10} that $\text{TR}(\text{2AC})$ is
  included in the class of tree languages generated by
  hyperedge-replacement context-free graph grammars (HRG), and hence
  $\text{TR}(\text{2AC}) \subseteq \text{MCFT}$ by
  Corollary~\ref{cor:mcft-cfgg}.  In the other direction, it is shown
  in~\cite{kan10} by a simple construction that every tree language
  generated by an HRG in tree generating normal form 
  (as in~\cite[Definition~6]{engman00} or equivalently in~Remark~\ref{rem:mcft-cfgg})
  is in~$\text{TR}(\text{2AC})$. Note that together with the construction 
  in Remark~\ref{rem:mcft-cfgg} this also shows that there is 
  a simple construction to transform every MCFTG into an equivalent 2ACG. 
\end{proof}

We finally observe (cf.~the paragraph
after~\cite[Corollary~7.10]{engman99}) that MRT~is properly included
in~MCFT.  The tree language~$\{a^n b^n \e \mid n \in \nat_0\}$
over~$\Sigma = \{a^{(1)}, b^{(1)}, \e^{(0)}\}$ is in~MCFT and even
in~$\text{CFT}_{\text{sp}}$, but not
in~$\text{DT}_{\text{fc}}(\text{RT})$ because all tree languages
over~$\Sigma$ in this class are regular~\cite[Theorem~4]{rou70}.  Also
$\text{CFT}_{\text{sp}}$~is properly included in~MCFT since it is
shown in~\cite[Section~5]{engfil81} that the tree language~$L(G)$,
where $G$~is the MRTG of Example~\ref{exa:copy}, is not
in~$\text{CFT}_{\text{sp}}$.  Thus, MRT~and~$\text{CFT}_{\text{sp}}$
are incomparable subclasses of~MCFT.

\section{Translation}
\label{sec:trans}
\noindent
As observed in~\cite{rao97} for~MRTGs, 
MCFTGs are not only a natural generation device but also a natural translation device. 
In general, we can also use an MCFTG~$G$
to define a forest language (i.e., an $n$"~ary relation on~$T_\Sigma$)
by considering~$L(G, A)$ for a big nonterminal~$A = (\seq A1n)$
with~$\rk(A_i) = 0$ for every~$i \in [n]$.  In particular, for the
case~$n = 2$, the MCFTG~can be used as a synchronous translation device, 
which we will call an MCFT"~transducer. After defining MCFT"~transducers we present 
two results analogous to those in~\cite{nedvog12} (see also~\cite{mal13}).  Namely, we prove a
characterization of the corresponding MCFT"~transductions by macro tree transducers,
similar to the one for MCFT tree languages in Theorem~\ref{thm:charact} (in the previous section), 
and we present a solution to the parsing and translation problem for MCFT"~transducers,
similar to the one for MCFTGs in Theorem~\ref{thm:mcftgparse} (in Section~\ref{sec:mcfg}).  

A \emph{multiple context-free tree transducer} (in
short,~MCFT"~transducer) is a system~$G = (N, \N, \Sigma, S, R)$,
where~$N$,~$\N$,~$\Sigma$, and~$R$ are as in
Definition~\ref{def:mcftg} and $S = (S_1, S_2) \in \N$ is the
\emph{initial big nonterminal} with~$S_1, S_2 \in N^{(0)}$.  We
require (without loss of generality) that $G$~is start-separated;
i.e., that $S_1$~and~$S_2$ do not occur in the right-hand sides of
rules.  Moreover, we require that $N$~is partitioned into two subsets
$N_1$~and~$N_2$ of input nonterminals and output nonterminals,
respectively, such that 
\begin{compactenum}[\indent (1)]
\item $S_1 \in N_1$~and~$S_2 \in N_2$, and 
\item for every rule~$(\seq A1n) \to ((\seq u1n), \LL)$ in~$R$, every $j \in
[n]$, and every $i \in [2]$ we have $\alp_N(u_j) \subseteq N_i$ if $A_j \in N_i$.  
\end{compactenum}
Intuitively this requirement means that the
nonterminals in~$N_1$ generate the input tree, and those in~$N_2$
generate the output tree.  For every~$A \in \N$, the forest
language~$L(G, A)$ is defined as for MCFTGs, and the \emph{tree
  transduction realized by~$G$} is the binary relation $\tau(G) = L(G)
= L(G,S) \subseteq T_\Sigma \times T_\Sigma$.  We also define~$G_\der$
as for MCFTGs.  Thus, the initial nonterminal of~$G_\der$ is~$S =
(S_1, S_2)$.  Consequently,~$\tau(G) = \val(L(G_\der))$ by
Theorem~\ref{thm:dtree}.  Note that the input and output alphabet of $G$ are
the same ranked alphabet~$\Sigma$.  This is a slight restriction that
could be solved by allowing symbols in a ranked alphabet to have more
than one rank.  The latter feature is easy to implement, but
technically rather tiresome.  We will say that $G$~is an
MCFT"~transducer \emph{over~$\Sigma$} and that $\tau(G)$ is 
an MCFT"~transduction over $\Sigma$.   The synchronous context-free
tree grammar of~\cite{nedvog12} is the special case of the
MCFT"~transducer in which~$\N \subseteq N_1 \times N_2$.

Our characterization of MCFT"~transductions by macro tree transducers uses 
a generalization of the notion of bimorphism. Bimorphisms are a classical 
symmetrical way to characterize classes of string and tree transductions 
(see, e.g., \cite{niv68,arndau76,mal07e}). 
Let $\X$ be a class of tree transductions. For a finite
ranked alphabet~$\Sigma$, we define an
$\X$"~\emph{bimorphism} over~$\Sigma$ to be a
transduction~$\tau \subseteq T_\Sigma \times T_\Sigma$ such that~$\tau
= \{(M_1(s), M_2(s)) \mid s \in L\}$, where $L$~is a regular tree
language over a finite ranked alphabet~$\Omega$ and $M_1$~and~$M_2$
are $\X$"~transductions with input alphabet~$\Omega$ and output alphabet~$\Sigma$. 
In the classical case $\X$ is a class of tree homomorphisms (or string homomorphisms 
in the similar case of strings); cf.\@ the proof of Proposition~\ref{pro:LDTRinv}. 
In the present case we take $\X=\text{DMT}_{\text{fc}}$ and we 
show that MCFT"~transductions are as expressive as
$\text{DMT}_{\text{fc}}$"~bimorphisms. 
Clearly, if $M_1$ and $M_2$ are $\text{DMT}_{\text{fc}}$"~transductions, 
then the domain $L_1=\{M_1(s)\mid s\in L\}$ and the range $L_2=\{M_2(s)\mid s\in L\}$
of the $\text{DMT}_{\text{fc}}$"~bimorphism $\tau$ 
are tree languages in MCFT by Theorem~\ref{thm:charact}
and $\tau$ can be viewed as translating $L_1$ into $L_2$. 
The inverse of $\tau$ is the $\text{DMT}_{\text{fc}}$"~bimorphism 
$\tau^{-1} = \{(M_2(s), M_1(s)) \mid s \in L\}$
which translates $L_2$ into $L_1$. 
Thus, $\text{DMT}_{\text{fc}}$"~bimorphisms are a natural symmetrical model 
for the translation of MCFT languages. 
To prove the characterization we need a few more definitions. 

We first modify the notion of 
$\text{DMT}_{\text{fc}}$"~transducer in such a way that it translates
trees into forests of length~$2$.  We define a $\text{DMT}_{\text{sp},
  2}$"~transducer to be a system~$M = (Q, \Omega, \Sigma, q_0, R)$,
where the only difference to a $\text{DMT}_{\text{sp}}$"~transducer is
that $q_0 = q_1q_2$~is the \emph{initial state sequence} with~$q_1,
q_2 \in Q^{(0)}$.  For $s \in T_\Omega$~and~$q \in Q$, the
tree~$M_q(s)$ is defined as for $\text{DMT}_{\text{sp}}$"~transducers,
and $M(s) = M_{q_0}(s)$ which equals $(M_{q_1}(s), M_{q_2}(s))$ by the
definition before Lemma~\ref{lem:valismtt}.  The tree transduction
realized by~$M$ is defined as for
$\text{DMT}_{\text{sp}}$"~transducers; i.e., it is the total
function~$M = \{(s, M(s)) \mid s \in T_\Omega\}$ from~$T_\Omega$
to~$T_\Sigma \times T_\Sigma$.  The state sequences of a
$\text{DMT}_{\text{sp}, 2}$"~transducer are defined in the same way as
for $\text{DMT}_{\text{sp}}$"~transducers with $\sts(s, \varepsilon) =
q_0$, and finite-copying $\text{DMT}_{\text{sp}, 2}$"~transducers are
called~$\text{DMT}_{\text{fc}, 2}$"~transducers. 

We now define the product of two $\text{DMT}_{\text{sp}}$"~transducers
$M_1$~and~$M_2$ with the same input and output alphabets to be the
$\text{DMT}_{\text{sp}, 2}$"~transducer~$M_1 \otimes M_2$ given as
follows.  Let~$M_i = (Q_i, \Omega, \Sigma, q_i, R_i)$ with~$i \in
[2]$, where we assume that $Q_1$~and~$Q_2$ are disjoint.  Then $M_1
\otimes M_2 = (Q_1 \cup Q_2, \Omega, \Sigma, q_1q_2, R_1 \cup
R_2)$.  It should be clear that for every~$s \in T_\Omega$ we
have~$(M_1 \otimes M_2)(s) = (M_1(s), M_2(s))$.  It should also be
clear that for every~$p \in \pos(s)$ the state sequence of~$M_1
\otimes M_2$ at~$p$ is the concatenation of the state sequences of
$M_1$~and~$M_2$ at~$p$.  This implies that $M_1 \otimes M_2$~is
finite-copying (repetition-free) if and only if $M_1$~and~$M_2$ are
both finite-copying (repetition-free).  Vice versa, for every
$\text{DMT}_{\text{sp}, 2}$"~transducer~$M$ there are
$\text{DMT}_{\text{sp}}$"~transducers $M_1$~and~$M_2$ such that
$M$~and~$M_1 \otimes M_2$ realize the same tree transduction.
Clearly, if~$M = (Q, \Omega, \Sigma, q_1q_2, R)$, then we can take
$M_1 = (Q, \Omega, \Sigma, q_1, R)$~and~$M_2 = (Q', \Omega, \Sigma,
q'_2, R')$, where the primes indicate a consistent renaming of the
states of~$M$ such that~$Q \cap Q' = \emptyset$.  The transducer~$M_1
\otimes M_2$ is obviously equivalent to~$M$ and it is finite-copying
if~$M$ is. Thus we have shown that 
$\text{DMT}_{\text{fc}, 2} = \{M_1 \otimes M_2 \mid M_1,M_2\in\text{DMT}_{\text{fc}}\}$. 
Note that it follows from
Proposition~\ref{pro:mtt-repfree} that this proposition also holds for
$\text{DMT}_{\text{fc}, 2}$"~transducers.

With these preparations, we can now prove our
characterization of MCFT"~transductions as bimorphisms of macro tree transductions. 

\begin{theorem}
  \label{thm:bimorphism}
  Let $\Sigma$~be a finite ranked alphabet.  A transduction~$\tau
  \subseteq T_\Sigma \times T_\Sigma$ is an $\textup{MCFT}$"~transduction
  over~$\Sigma$ if and only if it is a
  $\textup{DMT}_{\textup{fc}}$"~bimorphism over~$\Sigma$.
\end{theorem}

\begin{proof}
  Exactly the same proofs as those of Lemmas
  \ref{lem:valismtt}~and~\ref{lem:mttrt-mcft} show that the class of
  MCFT"~transductions equals the class~$\text{DMT}_{\text{fc},
    2}(\text{RT})$.  The latter class coincides with the class
  of~$\text{DMT}_{\text{fc}}$"~bimorphisms because 
  if $M = M_1 \otimes M_2$,
  where $M$ is a $\text{DMT}_{\text{fc}, 2}$"~transducer and 
  $M_1$~and~$M_2$ are $\text{DMT}_{\text{fc}}$"~transducers, 
  then $M(L) = \{(M_1(s),
  M_2(s)) \mid s \in L\}$ for every regular tree language~$L \in
  \text{RT}$.  Note that if $M_1$~and~$M_2$ have the disjoint
  sets of states $Q_1$~and~$Q_2$, then the set~$N'$ of nonterminals of
  the MCFT"~transducer~$G'$ constructed in the proof of
  Lemma~\ref{lem:mttrt-mcft} is partitioned into the set~$Q_1 \times
  N$ of input nonterminals and the set~$Q_2 \times N$ of output
  nonterminals, where $N$~is the set of nonterminals of the given RTG.
\end{proof} 

We note that we can define MRT"~transducers and
$\text{DT}_{\text{fc}}$"~bimorphisms in the obvious way, and prove as
a special case of Theorem~\ref{thm:bimorphism} that the
MRT"~transductions (which are the binary rational tree translation of~\cite{rao97})
coincide with the
DT$_{\mathrm{fc}}$-bimorphisms.  In~\cite{mal13} the MRT"~transducers
are called synchronous forest substitution grammars, and it is shown
in~\cite[Theorem~3]{mal13} that the MRT"~transductions are the
ld"~MBOT"~bimorphisms, where ld"~MBOT is the class of transductions
realized by linear deterministic multi bottom-up tree
transducers~\cite{englilmal09}.\footnote{The restriction to
  \emph{linear} d"~MBOT is implicit in~\protect{\cite{mal13}}.}
By~\cite[Theorem~18]{englilmal09} and~\cite[Theorems
5.10~and~7.4]{engman99}, this is essentially the same result.  We also
note that we can define
$\text{DMT}_{\text{fc}}^{\text{R}}$"~bimorphisms in the obvious way,
where $\text{DMT}_{\text{fc}}^{\text{R}}$"~transducers are defined
just as $\text{DMT}_{\text{fc}}$"~transducers, but with regular
look-ahead as in the definition of \LDTR"~transducer.  Since regular
look-ahead can be simulated by a relabeling of the input tree, the
$\text{DMT}_{\text{fc}}^{\text{R}}$"~bimorphisms are the same as the
$\text{DMT}_{\text{fc}}$"~bimorphisms.  In other words, the addition
of regular look-ahead does not increase the power of these
bimorphisms.  Moreover, the class of
$\text{DMT}_{\text{fc}}^{\text{R}}$"~transductions coincides with the
class~$\text{DMSOT}$ of deterministic MSO"~definable tree transductions
(cf.\@ Corollary~\ref{cor:mcft-mso} and the preceding paragraph).
Thus, the $\text{MCFT}$"~transductions are the $\text{DMSOT}$"~bimorphisms.  
The notion of $\text{DMSOT}$"~bimorphism is quite natural
as it is a transduction of the form~$\{(M_1(s), M_2(s)) \mid s \in
L\}$, where $L$~is an MSO"~definable tree language and $M_1$~and~$M_2$
are deterministic MSO"~definable tree transductions.
Even if we assume that $\text{DMSOT}$ transductions need not be total (cf.\@ footnote~\ref{foo:dom}),
it follows  that the class of 
$\text{MCFT}$"~transductions properly includes the class $\text{DMSOT}$.
To see this note that, in particular, every $\text{DMSOT}$ transduction and its inverse are 
$\text{DMSOT}$"~bimorphisms. Thus, since $\text{DMSOT}$ is not closed under inverse
(see~\cite[Remark~7.23]{coueng12}),
$\text{DMSOT}$ is properly included in the class of $\text{DMSOT}$"~bimorphisms. 

We now turn to the parsing and translation problem for
MCFT"~transducers, generalizing the parsing algorithm for MCFTGs in Theorem~\ref{thm:mcftgparse}. 
Let~$G$ be an MCFT"~transducer over~$\Sigma$, and let~$\Delta \subseteq \Sigma^{(0)}
\setminus \{e\}$ be a set of lexical symbols.  
We can view $G$ as translating input strings into output strings, 
thereby realizing the string transduction 
$\{(\yield_\Delta(t_1), \yield_\Delta(t_2)) \mid (t_1, t_2) \in \tau(G)\}$.
In such a case the parsing and translation problem
for~$G$ amounts to finding the syntactic trees for a given string over
$\Delta$ and finding its possible translations 
together with their syntactic trees. 
In the next result we show that this can be done in polynomial time. 
For its proof we need some more terminology. 
It is straightforward to prove the analogue of Lemma~\ref{lem:cover} for
MCFT"~transducers, which shows that MCFT"~transductions are closed
under tree homomorphisms.  For a given MCFT"~transducer~$G$ and tree
homomorphism~$h$, the MCFT"~transducer~$G_h$ has the same initial big
nonterminal as~$G$.  Moreover, the lemma implies that~$\tau(G_h) =
\{(\hat{h}(t_1), \hat{h}(t_2)) \mid (t_1, t_2) \in \tau(G)\}$.  As 
before, $(G, h)$~is said to be a cover of~$G_h$ if $h$~is a
projection.  An MCFT"~transducer~$G$ over $\Sigma$ is \emph{i/o"~disjoint} if
$\Sigma$~is partitioned into subsets $\Sigma_1$~and~$\Sigma_2$ of
input and output terminal symbols, and 
\begin{compactenum}
\item[\indent ($2'$)] for every rule~$(\seq A1n) \to ((\seq u1n),
  \LL)$ in~$R$, every $j \in [n]$, and every $i \in [2]$ we have 
$\alp_{N\cup \Sigma}(u_j) \subseteq N_i \cup \Sigma_i$ if $A_j \in N_i$. 
 \end{compactenum}
This guarantees that~$\tau(G) \subseteq
T_{\Sigma_1} \times T_{\Sigma_2}$. It should be clear that every
MCFT"~transducer~$G$ over $\Sigma$ has a cover~$(G^{\text{u}}, h)$ such that
$G^{\text{u}}$~is i/o"~disjoint, the terminal alphabet~$\Sigma
\cup \Sigma'$ of $G^{\text{u}}$ is partitioned into $\Sigma_1 =
\Sigma$~and~$\Sigma_2 = \Sigma'$, and the restriction of $h$ to 
$\Sigma$ is `in'. To construct~$G^{\text{u}}$ from~$G$, change every~$u_j$
with~$A_j \in N_2$ in the above rule into~$u'_j$, where $u'_j$~is
obtained from~$u_j$ by changing every label~$\sigma$ into its primed 
version~$\sigma'$, and define~$h(\sigma') = h(\sigma) =
\init(\sigma)$ for every $\sigma \in \Sigma$.

\begin{theorem}
  \label{thm:transparse}
  For every MCFT"~transducer~$G$ over~$\Sigma$ and
  every~$\Delta \subseteq \Sigma$, there is a polynomial time
  algorithm that, on input~$w \in \Delta^*$, outputs an RTG~$H_w$ and
  an MCFT"~transducer~$G_w$ such that 
  \[ L(H_w) = \{d \in L(G_\der) \mid \val(d) \in \tau(G_w)\} \qquad
  \text{and} \qquad \tau(G_w) = \{(t_1, t_2) \in \tau(G) \mid
  \yield_\Delta(t_1) = w\} \enspace. \]
  The degree of the polynomial is $\mu(G)\cdot(\wid(G)+1)\cdot(\lambda(G)+1)$.
\end{theorem}

\begin{proof}
  We first show how to construct the RTG~$H_w$.  Note that
  $L(H_w)$~should consist of all derivation trees~$d \in L(G_\der)$
  such that~$\yield_\Delta(\val(d)_1) = w$, where $\val(d)_1$~is the
  first tree of the forest~$\val(d)$.  To show this, we may assume
  that $G$~is i/o"~disjoint with $\Sigma$~partitioned into
  $\Sigma_1$~and~$\Sigma_2$ and with~$\Delta \subseteq \Sigma_1$.  In
  fact, let $(G^{\text{u}}, h)$~be an i/o"~disjoint cover of~$G$ with
  the properties described before this theorem.  Now let
  $H_w^{\text{u}}$~be an RTG such that~$L(H_w^{\text{u}}) = \{d \in
  L(G_\der^{\text{u}}) \mid \yield_\Delta(\val(d)_1) = w\}$.  By the
  proof of Lemma~\ref{lem:cover} there is a projection~$\pi$ such that
  $\hat{\pi}(L(G^{\text{u}}_\der)) =
  L(G_\der)$~and~$\val(\hat{\pi}(d)) = \hat{h}(\val(d))$ for every~$d
  \in L(G^{\text{u}}_\der)$.  Applying~$\hat{\pi}$ to the rules
  of~$H_w^{\text{u}}$, we obtain an RTG~$H_w$ such that~$L(H_w) =
  \hat{\pi}(L(H_w^{\text{u}}))$.  Clearly, $H_w$~satisfies the above
  requirement.

  Assuming that $G =(N, \N, \Sigma, (S_1, S_2), R)$ is i/o"~disjoint
  with $\Sigma$~partitioned into $\Sigma_1$~and~$\Sigma_2$ and
  with~$\Delta \subseteq \Sigma_1$, we construct the MCFTG~$G^\# = (N
  \cup \{S'\}, \N \cup \{S'\}, \Sigma \cup \{\#^{(2)}\}, S', R^\#)$,
  where $S'$~is a new nonterminal, $\#$~is a new terminal, and
  $R^\#$~contains all rules of~$R$ and the rule~$\rho_\# =
  S' \to (\#(S_1, S_2), \LL)$ with~$\LL = \{(S_1, S_2)\}$.  Note that 
  \[ L(G^\#) = \{\#(t_1, t_2) \mid (t_1, t_2) \in \tau(G)\} \qquad
  \text{and} \qquad L(G_\der^\#) = \{\rho_\#(d) \mid d \in L(G_\der)\}
  \enspace. \]
  By Theorem~\ref{thm:mcftgparse} there is a polynomial time algorithm
  that, on input~$w \in \Delta^*$, outputs an RTG~$H_w^\#$ such
  that~$L(H_w^\#) = \{d \in L(G_\der^\#) \mid \yield_\Delta(\val(d)) =
  w\}$.  We construct the RTG~$H_w$ from~$H_w^\#$ by
  removing~$\rho_\#$; i.e., changing every initial rule $S\to \rho_\#(C)$
  of~$H_w^\#$ into all rules $S\to \rho(\seq C 1k)$ such that 
  $C\to \rho(\seq C 1k)$ is a rule of~$H_w^\#$. 
  Then~$L(H_w) = \{d \in L(G_\der) \mid
  \yield_\Delta(\val(d)) = w\}$ because~$\# \notin \Delta$.  Clearly,
  since $\Sigma_1$~and~$\Sigma_2$ are disjoint and~$\Delta \subseteq
  \Sigma_1$, we have~$\yield_\Delta(\val(d)) = w$ if and only
  if~$\val(d) = (t_1, t_2)$ with~$\yield_\Delta(t_1) = w$.  Thus,
  $H_w$~satisfies the requirement.  

  Finally, we construct~$G_w$ from
  $G$~and~$H_w$ as in the proof of Theorem~\ref{thm:mcftgparse} with
  initial big nonterminal~$(S_1, S_2) \otimes S_w = (\langle S_1, S_w
  \rangle, \langle S_2, S_w \rangle)$.  Then~$\tau(G_w) =
  \val(L(H_w))$, and hence $G_w$~satisfies the requirement.
\end{proof}

Remarks similar to those following Theorem~\ref{thm:mcftgparse} are also valid here.  
For $\Delta = \Sigma$, Theorem~\ref{thm:transparse} solves the parsing and translation problem 
for MCFTG"~transducers as tree transducers in polynomial time. 
For every input tree $t\in T_\Sigma$ the algorithm 
produces as output an RTG $H_t$ such
that $L(H_t) = \{d \in L(G_\der) \mid \exists t'\in T_\Sigma\colon \val(d) = (t,t')\}$.  
The algorithm can be extended to test in linear time whether or not $t$
is in the domain of~$\tau(G)$, by testing whether $L(H_t)$ is nonempty.  
Additionally, if $L(H_t) \neq \emptyset$, then it can also compute in linear time a
derivation tree $d\in L(H_t)$ and a tree $t'\in T_\Sigma$ such that 
$\val(d) = (t,t')$. Thus, $t'$ is a possible translation of $t$. 

For~$\Delta \subseteq \Sigma^{(0)} \setminus \{e\}$, 
we are in the situation described before Theorem~\ref{thm:transparse}. 
For every input string $w\in\Delta^*$ the algorithm outputs 
an MCFT"~transducer $G_w$ such that $\tau(G_w)$~is the set of all
pairs of syntactic trees~$(t_1,t_2) \in \tau(G)$ such that $t_1$ is a syntactic tree for $w$;
i.e., $\yield_\Delta(t_1)=w$. 
Using $H_w$ as before, the algorithm can be extended to test in linear time whether~$w$ 
is in the domain of the string transduction 
$\{(\yield_\Delta(t_1), \yield_\Delta(t_2)) \mid (t_1, t_2) \in \tau(G)\}$
realized by $G$, and if so compute a 
derivation tree~$d\in L(H_w)$, its value $(t_1,t_2)$ such that~$\yield_\Delta(t_1) = w$,
and the string $w'= \yield_\Delta(t_2)$.  
Thus, $t_1$ is a syntactic tree of $w$ and $t_2$ is a syntactic tree of a possible translation
$w'$ of $w$. 
Note that, since the
proof of Theorem~\ref{thm:transparse} is based on
Theorem~\ref{thm:mcftgparse}, these parsing and translation algorithms for
MCFT"~transducers are, again, based on a parsing algorithm for MCFGs.

Let us finally consider the class of string transductions realized by MCFT"~transducers
as discussed above. We first restrict attention to 
the case $\Delta = \Sigma^{(0)} \setminus \{e\}$, 
which means that each MCFT"~transducer~$G$ realizes the
string transduction~$\{(\yield(t_1), \yield(t_2)) \mid (t_1, t_2) \in
\tau(G)\}$.  Let us call this a yMCFT"~transduction.  We can define
MCF"~transducers in the obvious way, with $S_1$~and~$S_2$ being the
only nonterminals of rank~$0$.  It should now be clear that we can
generalize Corollary~\ref{cor:ymcft} as follows:  The
yMCFT"~transductions coincide with the MCF"~transductions
(and with the yMRT"~transductions).  These MCF"~transductions
can also be characterized as the
$\text{yDT}_{\text{fc}}$"~bimorphisms, or equivalently, as the
bimorphisms determined by deterministic tree-walking transducers
(cf.\@ the third paragraph after Theorem~\ref{thm:charact}).
Since there is an analogue of Lemma~\ref{lem:cover} for
MCFT"~transducers (as discussed before Theorem~\ref{thm:transparse}),
the MCF"~transductions are closed under string homomorphisms.
This implies that, for every MCFT"~transducer $G$ and every set  
$\Delta \subseteq \Sigma^{(0)} \setminus \{e\}$ of lexical symbols, 
the string transduction 
$\{(\yield_\Delta(t_1), \yield_\Delta(t_2)) \mid (t_1, t_2) \in \tau(G)\}$
is also a yMCFT"~transduction.

\section{Parallel and general MCFTG}
\label{sec:parallel}
\noindent
In this last section we consider two natural extensions of the MCFTG 
that allow the grammar to make an unbounded number of copies of subtrees.
The definitions of the syntax and semantics of these extensions 
are easy variants of those for the MCFTG. 
The first extension is the \emph{parallel} MCFTG (or PMCFTG), 
which is the obvious generalization of 
the well-known parallel MCFG of~\cite{sekmatfujkas91}. 
In a parallel MCFTG (or parallel MCFG), two or more occurrences 
of the same nonterminal may appear in the right-hand side of a rule.
In the least fixed point semantics the terminal tree 
generated by that nonterminal is therefore copied. 
In the derivation semantics, after application of the rule,  
the occurrences must be rewritten in exactly the same way 
in the remainder of the derivation.
The second generalization, which we only briefly consider, is the \emph{general} 
(P)MCFTG, for which we drop the restriction that the rules must be linear. 
Thus, two or more occurrences of the same variable may appear 
in the same tree of the right-hand side of a rule and, 
when the rule is applied in a derivation step,
the tree that is the current value of the variable is copied. 
The classical (nondeleting) IO context-free tree grammar 
is the general MCFTG of multiplicity~1. 
 
A \emph{parallel multiple
  context-free tree grammar} (in short,~PMCFTG) is a system~$G = (N,
\N, \Sigma, S, R)$ as in Definition~\ref{def:mcftg} except that the
right-hand side~$u$ of a rule~$A \to (u, \LL) \in R$ is not required
to be uniquely $N$"~labeled.  The least fixed point semantics of~$G$ is defined just as
for an MCFTG.  As an example, the PMCFTG~$G$ with $N = \N =
\{S\}$~and~$\Sigma = \{\sigma^{(2)}, a^{(0)}, b^{(0)}\}$ using the
rules
\[ S \to (\sigma(S, S), \{S\}) \qquad S \to (a, \emptyset) \qquad
\text{and} \qquad S \to (b, \emptyset) \]
generates the tree language~$L(G)$ consisting of all full binary trees
over~$\Sigma$ of which all leaves have the same label.  Thus, 
$\yield(L(G)) = \{a^{2^n} \mid n \in \nat_0\} \cup \{b^{2^n} \mid n
\in \nat_0\}$.  In fact, from the least fixed point semantics we first
obtain that $a$~and~$b$ are in~$L(G)$.  Next, we obtain that
the trees $\sigma(S, S)[S \gets a] = \sigma(a, a)$ and $\sigma(S, S)[S
\gets b] = \sigma(b, b)$ are in~$L(G)$, and then we confirm that
$\sigma(S, S)[S \gets \sigma(a, a)] = \sigma(\sigma(a, a), \sigma(a,
a))$~is in~$L(G)$, etc.  Here we use the trivial
fact that a tree homomorphism 
(and hence a second-order substitution) replaces different occurrences 
of the same nonterminal by the same tree. 
Since $\yield(L(G))$ is not semi-linear, PMCFTGs are more powerful than~MCFTGs, 
even when they are used to define string languages via the yields of
the generated tree languages.

Intuitively, for a rule~$A \to (u, \LL)$ of~$G$, it is still the case
that every big nonterminal~$B \in \LL$ occurs ``spread-out'' exactly
once in~$u$, but now each nonterminal of~$B$ may occur more than once
in~$u$. More precisely, for each big nonterminal
$B = (\seq C1m) \in \LL$ with~$\seq C1m \in N$, there is a unique
set $P_B \subseteq \pos_N(u)$ of 
positions such that $\{u(p)\mid p\in P_B\} = \{\seq C1m\}$,
and we have that $P_B \cap P_{B'}=\emptyset$ for every other $B'\in\LL$
and $\pos_N(u) = \bigcup_{B\in \LL} P_B$.
After the application of the rule, all occurrences of each nonterminal~$C_i$
must be rewritten in the same way.  This idea was first introduced for
context-free grammars in~\cite{rose64} with a least fixed point
semantics; for a rewriting semantics similar to the one in
Section~\ref{sub:deriv} we refer to~\cite{skyum74}.  

Derivation trees
can be defined for~$G$ as in Section~\ref{sub:dtrees} with the same
results, which are proved in the same way, with one notable exception.
Statements (1)~and~(2) of Lemma~\ref{lem:valocc} do not hold and must
be reformulated.  For our purposes here it suffices to replace them by
the following weaker statements:
\begin{compactenum}[\indent (1)]
\item $\alp_\Delta(\val(d)) = \bigcup_{\rho \in \alp_R(d)}
  \alp_\Delta(\rhs(\rho))$ for every~$\Delta \subseteq \Sigma$, and
\item $\alp_N(\val(d)) = \bigcup_{B \in \alp_\N(d)} \alp(B)$, 
\end{compactenum}
which can easily be proved by induction on the structure of~$d$.  The
rewriting semantics in Section~\ref{sub:deriv} also applies to~PMCFTGs
without change.  For instance, the tree~$\sigma(\sigma(a, a),
\sigma(a, a))$ is derived by the above grammar in three derivation
steps:
\[ S^\varepsilon \Rightarrow_G^{\rho_1, \varepsilon} \sigma(S^1, S^1)
\Rightarrow_G^{\rho_1, 1} \sigma(\sigma(S^{11}, S^{11}),
\sigma(S^{11}, S^{11})) \Rightarrow_G^{\rho_2, 11} \sigma(\sigma(a,
a), \sigma(a, a)) \enspace, \]
where $\rho_1$~is the first rule of~$G$ and $\rho_2$~is the second.

The results and proofs of Section~\ref{sub:basicnf} on basic normal
forms are also valid for PMCFTGs. The same is true for 
Lemmas~\ref{lem:epsilon-free} and~\ref{lem:nonerasing}.
However, we did not further study the
lexicalization of~PMCFTGs.  Thus, we leave it as an open problem
whether finitely ambiguous~PMCFTGs can be lexicalized, which we
conjecture to be true.  The results and proofs of 
Section~\ref{sec:monadic} are also valid for PMCFTGs (without the
statements on lexicalization). Thus, for every
PMCFTG there are an equivalent monadic PMCFTG, an equivalent footed PMCFTG,
and an equivalent ``parallel'' MC"~TAG (provided that 
the generated tree language is root consistent). 

Parallel MCFGs (in short,~PMCFGs) can be defined as in Section~\ref{sec:mcfg}, and all
the results and proofs in that section are also valid for the parallel
case, except Corollary~\ref{cor:ambmcf} on lexicalization.  Thus, we
have that $\text{yPMCFT} = \text{PMCF} = \text{yPMRT}$.  Moreover,
PMCFGs and PMCFTGs~can be parsed in polynomial time; i.e.,
Lemma~\ref{lem:cyk} and Theorem~\ref{thm:mcftgparse} also hold in the parallel case
(cf.~\cite{sekmatfujkas91,lju12}). 
However, as observed in~\cite{sekmatfujkas91}, the degree of the polynomial 
is one more than in those results
because in the proof of Lemma~\ref{lem:cyk}, in the construction of the rules of $H_w$, 
it must be checked additionally in linear time that 
$w[\ell(C_{i_1}),r(C_{i_1})]=w[\ell(C_{i_2}),r(C_{i_2})]$ whenever $C_{i_1}=C_{i_2}$
(where $C_{i_1}$ may occur in a different $u_j$ than $C_{i_2}$). 
It should also be noted that, for a given derivation tree~$d$, the syntactic
tree~$t = \val(d)$ can no longer be computed in linear time.  Instead,
it should be clear that in linear time a directed acyclic graph~$g$ can be computed 
that represents the tree~$t$ with shared nodes.  In the case
where~$\Sigma^{(0)} \subseteq \Delta$, this graph~$g$ can be unfolded
into~$t$ in time linear in the size of~$g$ plus the size of~$w =
\yield_\Delta(t)$, and thus $t$~is obtained in the required polynomial time from
the string~$w$ by the parsing algorithm.  
 
The results of
Section~\ref{sec:charact} (except Corollaries~\ref{cor:mcft-mso},
\ref{cor:mcft-cfgg} and~\ref{cor:mcft-acg}) as well as those of Section~\ref{sec:trans} are also
valid for the parallel case provided that we
change~$\text{DMT}_{\text{fc}}$ into~$\text{DMT}_{\text{sp}}$,
and~$\text{DT}_{\text{fc}}$ into~DT.  The proofs are also the same,
except that in the proof of Lemma~\ref{lem:valismtt} we do not have to
consider the state sequences of~$M$, and for the proof of
Lemma~\ref{lem:mttrt-mcft} we do not need
Proposition~\ref{pro:mtt-repfree} and we have to redefine state
sequences, as follows.  Roughly speaking, the new state sequences are
the old ones from which repetitions have been removed; thus, they can
be viewed as `state sets' (cf.~\cite[Definition~3.1.8]{engrozslu80}).
Formally, let~$M = (Q, \Omega, \Sigma, q_0, R)$ be a
$\text{DMT}_{\text{sp}}$"~transducer,  and consider a fixed order~$p_1
\sqsubset  \dotsb \sqsubset  p_r$ on the set~$Q=\{\seq p1r\}$ of states of~$M$.  For a
subset~$Q' = \{\seq p{i_1}{i_m}\}$ of~$Q$ with~$i_1 < \dotsb < i_m$,
we define the state sequence~$\sequ(Q') = \word p{i_1}{i_m}$.  Now let
$\seq q1n \in Q$~and~$n \in \nat_0$, and let $\omega \in
\Omega^{(k)}$~with~$k \in \nat_0$.  For~$i \in [k]$ we
(re-)define~$\sts_{\omega, i}(\seq q1n) \in Q^*$ to be the sequence of
states 
\[ \sts_{\omega, i}(\seq q1n) = \sequ(\{q' \in Q \mid \exists j \in[n]
\colon \langle q', y_i \rangle \in \alp_{Q\times Y}(\rhs_M(q_j, \omega))\})
\enspace. \]
Then $\sts(s,p)$~and~$\sts(M)$ can be defined as in
Section~\ref{sec:charact}, and with these definitions the proof of
Lemma~\ref{lem:mttrt-mcft} is valid.  Note that $\sts(M)$~is now
finite for every $\text{DMT}_{\text{sp}}$"~transducer.
Consequently, we have that $\text{PMCFT} =
\text{DMT}_{\text{sp}}(\text{RT})$~and~$\text{PMRT} =
\text{DT}(\text{RT})$.  As further consequences we obtain the known
result~$\text{yDMT}_{\text{sp}}(\text{RT}) =
\text{yDT}(\text{RT})$, which was proved
in~\cite[Theorem~15]{engman02}, and the known result~$\text{PMCF} =
\text{yDT}(\text{RT})$, which was proved in~\cite[Theorem~3.1]{vug96}
by taking into account the well-known fact that string-valued
attribute grammars without inherited attributes
generate~$\text{yDT}(\text{RT})$.  As in Section~\ref{sec:charact},
the multiplicity of the grammars corresponds to the copying power of
the transducers.  Thus, $\text{$m$"~PMCFT} = \text{DMT}_{\text{sp},
  (m)}(\text{RT})$ and $\text{$m$"~PMRT} = \text{DT}_{(m)}(\text{RT})$
and $\text{$m$"~PMCF} = \text{yDT}_{(m)}(\text{RT})$, where the
prefix~`$m$-' means that the grammars have multiplicity at most~$m$
and the subscript~`$(m)$' means that the transducers are $m$"~copying
(with the new definition of state sequence).  As shown
in~\cite[Theorem~3.2.5]{engrozslu80} by a pumping lemma
for~$\text{yDT}_{(m)}(\text{RT})$, the language~$L_m = \{a_1^n a_2^n
\dotsm a_{2m+2}^n \mid n \in \nat_0\}$ is in~$(m+1)$"~MCF but not
in~$m$"~PMCF.  
As results analogous to those in
Section~\ref{sec:trans} we obtain that the PMCFT"~transductions are
the same as the $\text{DMT}_{\text{sp}}$"~bimorphisms, and the
PMRT"~transductions are the same as the DT"~bimorphisms, and hence
by~\cite{fulkuhvog04} they coincide with the d"~MBOT"~bimorphisms, where
the d"~MBOTs are not necessarily linear.  Moreover,
PMCFT"~transductions can be parsed and translated in polynomial time
(with the degree of the polynomial one more than in Theorem~\ref{thm:transparse}).

Finally we consider a further extension of~PMCFTGs.  Until now we have
restricted our grammars to be simple (i.e., linear and nondeleting),
which means that for every rule~$(\seq A1n) \to ((\seq u1n), \LL)$ and
every~$j \in [n]$, the tree~$u_j$ contains every variable
in~$X_{\rk(A_j)}$ exactly once.  We now drop the linearity condition
and just require every such variable to occur at least once.
Technically it is convenient to achieve this by \emph{redefining the
  notion of pattern} (see the first paragraph of
Section~\ref{sub:sub}).  Thus, we redefine the set~$P_\Sigma(X_k)$ of
patterns of rank~$k$ to consist of all trees~$t \in T_\Sigma(X_k)$
such that~$\alp_X(t) = X_k$; i.e., each~$x \in X_k$ occurs at least
once in~$t$.  It should be noted that this also changes our definition
of tree homomorphism, which is now only required to be nondeleting,
and hence that of second-order substitution.  Clearly,
Lemma~\ref{lem:treehom} is not true anymore.  For our purposes here it
can be replaced by the following weaker statements:
\begin{compactenum}[\indent (1)]
\item $\alp_X(\hat{h}(t)) = \alp_X(t)$, and 
\item $\alp_\Sigma(\hat{h}(t)) = \bigcup_{\tau \in \alp_\Sigma(t)}
  \alp_\Sigma(h(\tau))$.  
\end{compactenum}
The remaining definitions and results of Section~\ref{sub:sub} can be
taken over without change.  

The definition of a \emph{general parallel multiple
context-free tree grammar} (in short,~gPMCFTG) is identical to the one of
a~$\text{PMCFTG}$ with the new meaning of~$P_{N \cup \Sigma}(X)$ as above.  The
semantics of a $\text{gPMCFTG}$~$G$ is defined just as for an~MCFTG.
The class of tree languages generated by gPMCFTGs is denoted by $\text{PMCFT}_\text{g}$. 
Derivation trees are defined for~$G$ just as for an~MCFTG, and
Section~\ref{sub:dtrees} is valid for~gPMCFTGs with the same change of
Lemma~\ref{lem:valocc} as stated above for~PMCFTGs.  The rewriting
semantics in Section~\ref{sub:deriv} is also valid for~gPMCFTGs.  The
semantics of a~PMCFTG is essentially an ``inside-out'' semantics in
the sense of~\cite{engsch77}.  In fact, consider a classical
IO~context-free tree grammar~$G$ such that (i)~$G$ is nondeleting
(i.e., every variable in the left-hand side of a rule also occurs  
in the right-hand side) and (ii)~the right-hand side of each rule is
uniquely $N$"~labeled (i.e., every nonterminal occurs at most once in
the right-hand side of each rule).  Viewing~$G$ as a~gPMCFTG  
in the obvious way, it is easy to see that the least fixed point semantics  
of~$G$ as a~gPMCFTG coincides with the least fixed point semantics
of~$G$ as an IO~context-free tree grammar as stated
in~\cite[Theorem~3.4]{engsch77}.  Since requirements (i)~and~(ii) are a
normal form for IO~context-free tree grammars
(cf.~\cite[Theorem~3.1.10]{fis68b}), this shows that all 
IO~context-free tree languages can be generated by gPMCFTGs.
More precisely, they are the tree languages generated by 
the (nonparallel) gMCFTGs of multiplicity~1. 

As an example, the gPMCFTG~$G$ with $N = \N = \{S^{(0)}, A^{(1)},
B^{(1)}\}$~and~$\Sigma = \{\sigma^{(2)}, a^{(0)}, b^{(0)}\}$ using the
rules
\[ S \to A(b) \qquad A(x_1) \to B(A(\sigma(a, x_1))) \qquad A(x_1) \to
x_1 \qquad \text{and} \qquad B(x_1) \to \sigma(x_1, x_1) \enspace, \]
generates the tree language~$L(G)$ consisting of all trees~$t_1[x_1
\gets t_2]$, where $t_1$~is a full binary tree over~$\{\sigma, x_1\}$
of height~$n$ and $t_2$~equals~$(\sigma a)^n b$. Thus,~$\yield(L(G)) =
L_{\text{ec}} = \{(a^n b)^{2^n} \mid n \in \nat\}$.  For~$n = 2$, the
tree~$t = \sigma(\sigma(\sigma a \sigma ab, \sigma a \sigma ab),
\sigma(\sigma a\sigma ab, \sigma a \sigma ab))$ is obtained by the
derivation
\begin{align*}
  S^\varepsilon 
  &\Rightarrow_G^{\rho_1, \varepsilon} A^1(b) \Rightarrow_G^{\rho_2,
    1} B^{11}(A^{12}(\sigma ab)) \Rightarrow_G^{\rho_2, 12}
    B^{11}(B^{121}(A^{122}(\sigma a \sigma ab))) \\
  &\Rightarrow_G^{\rho_3, 122} B^{11}(B^{121}(\sigma a \sigma ab))
    \Rightarrow_G^{\rho_4, 121} B^{11}(\sigma(\sigma a \sigma ab,
    \sigma a \sigma ab)) \Rightarrow_G^{\rho_4, 11} t \enspace,
\end{align*}
which corresponds to the ``inside-out'' derivation of the
IO~context-free tree grammar~$G$, but is, for instance, also obtained by
the ``outside-in'' derivation  
\begin{align*}
  S^\varepsilon 
  &\Rightarrow_G^{\rho_1, \varepsilon} A^1(b) \Rightarrow_G^{\rho_2,
    1} B^{11}(A^{12}(\sigma ab)) \Rightarrow_G^{\rho_4, 11}
    \sigma(A^{12}(\sigma ab), A^{12}(\sigma ab)) \\
  &\Rightarrow_G^{\rho_2, 12} \sigma(B^{121}(A^{122}(\sigma a \sigma
    ab)), B^{121}(A^{122}(\sigma a \sigma ab))) \\
  &\Rightarrow_G^{\rho_4, 121} \sigma(\sigma(A^{122}(\sigma a \sigma
    ab), A^{122}(\sigma a \sigma ab)), \sigma(A^{122}(\sigma a \sigma
    ab), A^{122}(\sigma a \sigma ab))) \Rightarrow_G^{\rho_3, 122} t
    \enspace.
\end{align*}
The language~$L_{\text{ec}}$ is the well-known example of an
IO~context-free tree language that is not an OI~context-free tree
language (see~\cite[Section~4.3]{fis68b}).  It is shown
in~\cite[Theorem~3.16]{eng82b}, using again the pumping lemma
for~$\text{yDT}(\text{RT})$, that $L_{\text{ec}}$~is not
in~$\text{yDT}(\text{RT})$, and hence not in~PMCF.  Thus,
gPMCFTGs are more powerful than~PMCFTGs, even when they 
are used to define string languages via the yields of
the generated tree languages.
Note that the above grammar is even a~gMCFTG because the
right-hand sides of its rules are uniquely $N$"~labeled.\footnote{We do not know
whether there is a tree language in PMCFT that is not in $\text{MCFT}_\text{g}$; 
i.e., we do not know whether
PMCFT~and~$\text{MCFT}_\text{g}$ are incomparable subclasses of~$\text{PMCFT}_\text{g}$.}  
The multiple context-free tree grammars in~\cite{boukalsal12} are the~gMCFTGs, whereas
our MCFTGs are there called \emph{linear} multiple context-free tree grammars. 
It is shown in~\cite{boukalsal12} that the closure 
of MCF under IO-substitution is included in $\text{yMCFT}_\text{g}$
and that the string languages in this closure 
satisfy the constant-growth property and can be recognized in polynomial time. 

The only result we have for~gPMCFTGs is their characterization in
terms of macro tree transducers.  Let $\text{DMT}_{\text{np}}$~denote
the class of tree transductions realized by macro tree transducers
with the new definition of pattern (where `np'~stands for~`nondeleting
in the parameters').  The semantics of such transducers is as in
Section~\ref{sec:charact}.  Using the redefined notion of state
sequence as for~PMCFTGs, the proofs of Lemmas
\ref{lem:valismtt}~and~\ref{lem:mttrt-mcft} are still valid.  Thus, we
obtain that~$\text{PMCFT}_\text{g} = \text{DMT}_{\text{np}}(\text{RT})$.  Now
let DMT~denote the class of tree transductions realized by all (total
deterministic) macro tree transducers as known from the literature,
which means that also deletion of parameters is allowed; i.e., for a
rule~$\langle q, \omega(\seq y1k) \rangle(\seq x1m) \to \zeta$, it is
just required that~$\zeta \in T_{(Q \times Y_k) \cup \Sigma}(X_m)$.
Their semantics is still the same as in Section~\ref{sec:charact}.
It is proved in~\cite[Lemma~6.6]{engman99} that for every
DMT"~transducer with regular look-ahead there is an equivalent one
that is nondeleting in the parameters.  Since regular look-ahead can
be simulated by relabeling the input tree, this implies
that~$\text{DMT}(\text{RT}) = \text{DMT}_{\text{np}}(\text{RT})$.
Thus we obtain the characterization~$\text{PMCFT}_\text{g} =
\text{DMT}(\text{RT})$.  
We observe that the two types (P and g) of copying subtrees that can be realized by gPMCFTGs,
correspond for macro tree transducers to the copying of input variables (from~$Y$)
and the copying of output variables (or parameters, from~$X$), respectively. 

At the end of this section we discuss the class S"~CF of synchronized-context-free tree languages 
introduced in~\cite{chacheret06} and applied, e.g., in~\cite{boicharet15}. 
The logic programs generating these tree languages are
essentially tree-valued attribute grammars, which means that
$\text{S"~CF}=\text{AT}(\text{RT})$, where AT~denotes the class of attributed tree transductions 
(see, e.g.,~\cite{fulvog98,engman99}).
It was shown in~\cite{enghey92}
that $\text{AT}(\text{RT})$~is the class of tree languages obtained by
unfolding the term graphs generated by a context-free graph grammar,
where a term graph is a directed acyclic graph representing a tree
with shared subtrees (cf.~Corollary~\ref{cor:mcft-cfgg}).
It is well known
that~$\text{DT} \subsetneq \text{AT} \subsetneq \text{DMT}$ (see,
e.g.,~\cite{fulvog98}).  Thus, the class~$\text{AT}(\text{RT})$ is
included in~$\text{PMCFT}_\text{g}$.  It seems to be unknown whether the
inclusion is proper.  It follows from~\cite[Theorem~7.1]{engman99}
that~$\text{DMT}_{\text{fc}}(\text{RT}) \subseteq
\text{AT}(\text{RT})$.  Thus, MCFT~is included
in~$\text{AT}(\text{RT})$, but the relationship
of~$\text{AT}(\text{RT})$ to~PMCFT is not clear.  However,
$\text{PMCF} = \text{yDT}(\text{RT}) \subsetneq
\text{yAT}(\text{RT})$, because~$L_{\text{ec}} \in
\text{yAT}(\text{RT})$.  Hence we have 
$\text{MCFT} \subsetneq \text{AT}(\text{RT}) \subseteq \text{PMCFT}_\text{g}$ and
$\text{MCF} \subsetneq \text{PMCF} \subsetneq \text{yAT}(\text{RT}) \subseteq
\text{yPMCFT}_\text{g}$. 
We finally note that the class~$\text{CFT}_{\text{sp}}$ is characterized in terms
of a special type of attributed tree transducers in~\cite{mon10}.

\section{Conclusion} 
\noindent
We have proved in Theorem~\ref{thm:main} that 
every finitely ambiguous MCFTG can be lexicalized, for an arbitrary set $\Delta$ of lexical symbols.
A remaining question is whether the given bounds on the multiplicity and width of the 
resulting MCFTG are optimal. In the particular case where all lexical symbols in $\Delta$ have rank~0,
the multiplicity stays the same, but the width increases by~1. 
By Theorems~\ref{thm:main} and~\ref{thm:monadic} together, there is also 
an equivalent lexicalized grammar of width at most~1 but with increased multiplicity. 
A similar question is relevant for the transformation of an MCFTG into an equivalent MC"~TAG
(Theorem~\ref{thm:mctal}), and for the lexicalization of MC"~TAGs (Theorem~\ref{thm:lexmctag}).
As shown in~\cite{maleng17}, the factor $\mr_\Sigma^2$ in Theorems~\ref{thm:mctal}, \ref{thm:lexmctag}, 
and~\ref{thm:monadic} can be reduced to $\mr_\Sigma$ by combining the two constructions in the proofs of 
Theorem~\ref{thm:footed} and Lemma~\ref{lem:foottotag} into one.

All our grammar transformations produce an MCFTG that is grammatically close (i.e., \LDTR-equivalent) 
to the given MCFTG, except for the transformation of an MCFTG into a monadic MCFTG
(Lemma~\ref{lem:nlex} and~Theorem~\ref{thm:monadic}), for which we could only prove 
\LDTR-equivalence in the special case in which all lexical symbols in $\Delta$ have rank~0. 
As already observed in footnotes~\ref{foo:fc1} and~\ref{foo:fc2}, 
this problem can be ``solved'' by considering the weaker notion of \DTRfc-equivalence 
instead of \LDTR-equivalence, where \DTRfc\ is the class of transductions realized by finite-copying
top-down tree transducers with regular look-ahead. The definition of \DTRfc-equivalence is
the same as that of \LDTR-equivalence in Definition~\ref{def:LDTReq}. 
Since \DTRfc\ is closed under composition (see, e.g., \cite[Theorem~5.4]{engrozslu80}), 
this is indeed an equivalence relation. Actually, we feel that \DTRfc-equivalence is a better 
formalization of the notion of grammatical closeness than \LDTR-equivalence because it 
can also handle the combination of rules as needed, e.g., in the proof of Lemma~\ref{lem:nlex}. 
Such a combination of rules is also needed for the binarization of grammars 
(which we did not study for MCFTGs), to transform the derivation trees of the binarized grammar 
into those of the original one. An MCFTG $G$ is \emph{binary} if its rule-width $\lambda(G)$
is at most $2$. 
In view of Lemma~\ref{lem:cyk} and Theorem~\ref{thm:mcftgparse}, 
binarization is important for parsing (see, e.g., \cite{ramsat99,gomsat09}). 
We note that most of our constructions preserve $\lambda(G)$. 
The two exceptions are Lemmas~\ref{lem:terminal-removal} and~\ref{lem:nlex} 
which decrease and increase $\lambda(G)$, respectively. 

In Theorem~\ref{thm:charact} we have proved a characterization of MCFTGs in terms of finite-copying
macro tree transducers, and from that we have deduced characterizations in terms of 
monadic second-order logic (Corollary~\ref{cor:mcft-mso}), 
context-free graph grammars (Corollary~\ref{cor:mcft-cfgg}), 
and abstract categorial grammars (Corollary~\ref{cor:mcft-acg}). 
It would be worthwhile to investigate whether there are more results 
from the literature on macro tree transducers that can be applied to MCFTGs. 

In Section~\ref{sec:trans} we have introduced the $\text{MCFT}$"~transducer and we have 
shown that they realize the $\textup{DMT}_{\textup{fc}}$"~bimorphisms and 
hence the $\text{DMSOT}$"~bimorphisms. This class of $\text{MCFT}$"~transductions 
deserves further study. Only subclasses have been investigated in the literature. 
As stated in~\cite[Example~5]{rad08}, the $\text{MRT}$"~transductions are not closed under composition.
We do not know whether the $\text{MCFT}$"~transductions are closed under composition or 
whether composition gives rise to a proper hierarchy. Another question is whether or not
every functional $\text{MCFT}$"~transduction is a composition of 
deterministic macro tree transductions. 

Our remaining problems concern the extensions of MCFTGs
discussed in Section~\ref{sec:parallel}: the PMCFTGs and the g(P)MCFTGs.  
As observed in that section it is open whether PMCFTGs can be lexicalized, 
and the same is true for g(P)MCFTGs. 
Although Theorem~\ref{thm:charact} can be generalized to PMCFTGs and gPMCFTGs,
it is not clear whether there are natural generalizations of the three corollaries
mentioned above. Also, a characterization of $\text{MCFT}_\text{g}$ is missing. 
Finally, it would be interesting to determine the correctness (or incorrectness) 
of the obvious Hasse diagram of the six classes 
MRT, MCFT, PMRT, PMCFT, $\text{MCFT}_\text{g}$, $\text{PMCFT}_\text{g}$. 
The tree language~$\{a^n b^n \e \mid n \in \nat_0\}$, 
which we considered at the end of Section~\ref{sec:charact}, is
in~MCFT (even in $\text{CFT}_{\text{sp}}$) but not
in PMRT because all monadic tree languages
in the class $\text{DT}(\text{RT})$ are regular~\cite[Theorem~4]{rou70}.
The IO context-free tree language $L_{\text{ec}}$ that we considered 
in Section~\ref{sec:parallel} is in~$\text{MCFT}_\text{g}$ but not in PMCFT.
The PMRTG (of multiplicity~1) that we considered in the second paragraph 
of Section~\ref{sec:parallel}, generates a tree language that is not in MCFT.
However, we do not know whether there exists a tree language in PMCFT
(or even in PMRT) that is not in $\text{MCFT}_\text{g}$. 
If we also add the six classes (as above) with multiplicity~1, then 
the situation is less clear. In view of~\cite[Corollary~3.5]{engsky82} 
we guess that $\text{1-PMRT}=\text{HOM(RT)}$ where HOM is the class of all 
(not necessarily simple) tree homomorphisms. Thus, apart from the trivial inclusions,
we obtain the additional inclusion  
$\text{1-PMRT} \subseteq \text{1-MCFT}_{\text{g}}$ because the class of IO context-free
tree languages is closed under arbitrary tree homomorphisms~\cite[Corollary~6.4]{engsch78}.
The tree language of Example~\ref{exa:copy},
which we also considered at the end of Section~\ref{sec:charact}, 
is in MRT but not in $\text{1-MCFT}_{\text{g}}$
because it cannot be generated by an IO context-free tree grammar
as shown in~\cite[Section~5]{engfil81}. However, we do not know 
whether there exists a tree language in MRT that is not in $\text{1-PMCFT}_{\text{g}}$;
i.e., that cannot be generated by a parallel IO context-free tree grammar.

\section*{References}
\bibliographystyle{plain}
\bibliography{multi-extra}

\end{document}